\documentclass[sigconf]{acmart}

\title{A Dichotomy for the Generalized Model Counting Problem for Unions of Conjunctive Queries}

\author{Batya Kenig}
\affiliation{
	\institution{University of Washington}
}

\author{Dan Suciu}
\affiliation{
	\institution{University of Washington}
}


\usepackage{enumerate}
\usepackage{amsfonts}
\usepackage{multirow}
\usepackage{graphicx}
\usepackage{listings}
\usepackage{mathtools}

\usepackage{subcaption}
\usepackage{stmaryrd}
\usepackage[noend]{algorithmic}

\usepackage{color}
\usepackage{colortbl}
\usepackage{multirow}
\usepackage{graphicx,fancyvrb}
\usepackage[colorinlistoftodos]{todonotes}
\usepackage[noend]{algorithmic}
\usepackage{capt-of}
\usepackage{xcolor}
\usepackage{multirow}

\usepackage{float}
\usepackage{subfloat}

\usepackage{array}
\usepackage{arydshln}
\usepackage{todonotes}

\definecolor{mygreen}{rgb}{0,0.6,0}
\definecolor{mygray}{rgb}{0.5,0.5,0.5}
\definecolor{mymauve}{rgb}{0.58,0,0.82}

\usepackage{listings}
\lstnewenvironment{VerbatimText}[1][]{
    
    \lstset{fancyvrb=true,basicstyle=\footnotesize,captionpos=b,xleftmargin=2em,#1}
}{}

\newcommand{\pr}{{\tt \mathrm{Pr}}}

\newcommand{\ignore}[1]{}

\newcommand{\calR}{\mathcal R}

\newcommand{\Z}{\mathbb Z} 
\newcommand{\N}{\mathbb N} 
\newcommand{\Q}{\mathbb Q} 
\newcommand{\R}{\mathbb R} 

\newcommand{\functionname}[1]{\text{\sf #1}}
\newcommand{\Dom}{\functionname{Dom}}
\newcommand{\Tup}{\functionname{Tup}}

\newcommand{\batya}[1]{{\texttt{\color{blue} Batya: [{#1}]}}}
\newcommand{\blue}[1]{{\color{black}{#1}}}
\newcommand{\dan}[1]{{\texttt{\color{red} Dan: [{#1}]}}}

\newcommand{\true}{{\tt true}}
\newcommand{\false}{{\tt false}}

\newcommand*{\rom}[1]{\expandafter\@slowromancap\romannumeral #1@}
\newcommand{\RNum}[1]{\uppercase\expandafter{\romannumeral #1\relax}}

\newcommand{\mb}[1]{{\mathbf{#1}}}

\newcommand{\sel}[1]{{\sigma}}

\newcommand{\cut}[1]{}
\newcommand{\eat}[1]{}

\newcommand{\defeq}{\stackrel{\text{def}}{=}}
\newcommand{\setof}[2]{\{{#1}\mid{#2}\}}        

\def\set#1{\mathord{\{#1\}}}

\def\eqdef{\mathrel{\stackrel{\textsf{\tiny def}}{=}}}

\def\J{\mathcal{J}}

\def\e#1{\emph{#1}}

\def\implies{\Rightarrow}

\newenvironment{repeatresult}[2]
{\vskip0.5em\par\textsc{#1} #2.\em}
{\vskip1em}

\def\appendix{\par
	\section*{APPENDIX}
	\setcounter{section}{0}
	\setcounter{subsection}{0}
	\def\thesection{\Alph{section}} }

\def\eqdef{\mathrel{\stackrel{\textsf{\tiny def}}{=}}}

\def\J{\mathcal{J}}
\def\B{\mathcal{B}}
\def\e#1{\emph{#1}}

\def\b{\boldsymbol{b}}

\def\c{\boldsymbol{c}}

\newcommand{\algname}[1]{{\sf #1}}
\def\myrulewidth{3.25in}
\def\therule{\rule{\myrulewidth}{0.2pt}}

\def\myrulewidthwide{4in}
\def\therulewide{\rule{\myrulewidthwide}{0.2pt}}

\newenvironment{insidecode}[3]
{
	\begin{tabular}{p{\myrulewidth}}
		\multicolumn{1}{c}{\rule{0mm}{3mm}{\bf #3} $\algname{#1}(\mbox{#2})$\vspace{-0.6em}}\\
		\therule\vskip-0.8em\therule
		\vspace{-1em}
		\begin{algorithmic}[1]}
		{\end{algorithmic}
		\vskip-0.4em\therule
\end{tabular}}

\newenvironment{insidecodewide}[3]
{
	\begin{tabular}{p{\myrulewidthwide}}
		\multicolumn{1}{c}{\rule{0mm}{3mm}{\bf #3} $\algname{#1}(\mbox{#2})$\vspace{-0.6em}}\\
		\therulewide\vskip-0.8em\therulewide
		\vspace{-1em}
		\begin{algorithmic}[1]}
		{\end{algorithmic}
		\vskip-0.3em\therulewide
\end{tabular}}

\newcommand{\vars}{\mathbf{Vars}}

\def\M{\mathcal{M}}

\newcommand{\edges}{\texttt{edges}}
\newcommand{\nodes}{\texttt{nodes}}

\newcommand{\gmc}{\texttt{GFOMC}}

\newcommand{\fomc}{\texttt{FOMC}}

\newcommand{\pqe}{\texttt{PQE}}

\newcommand{\symb}{\texttt{Symb}}

\newcommand{\pdb}{{\mb{\Delta}}}  

\newcommand{\bk}{\boldsymbol{k}}

\newcommand{\bu}{\boldsymbol{u}}
\newcommand{\p}{p}
\newcommand{\var}{{\tt \mathrm{var}}}

\newcommand{\veck}{\#\boldsymbol{k}}

\newcommand{\typea}{\texttt{Type-I}}

\newcommand{\fTypebb}{\texttt{f-Type-II-II}}

\newcommand{\lin}[2]{\mbox{$\Phi_{#1}(#2)$}}

\newcommand{\minus}{\scalebox{0.75}[1.0]{$-$}}
\newcommand{\asn}{{:\mathrel{\minus}}}
\newcommand{\bp}{\textbf{p}}

\newcommand{\zg}[1]{\texttt{zg}(#1)}

\newcommand{\ccp}{\texttt{CCP}}

\newcommand{\diag}{\textbf{diag}}

\newcommand{\minor}{{\tt \mathrm{\bf min}}}

\newtheorem{claim}{Claim}

\begin{abstract}
  We study the {\em generalized model counting} problem, defined as
  follows: given a database, and a set of deterministic tuples, count
  the number of subsets of the database that include all deterministic
  tuples and satisfy the query. This problem is computationally
  equivalent to the evaluation of the query over a tuple-independent
  probabilistic database where all tuples have probabilities in
  $\set{0,\frac{1}{2},1}$. Previous work has established a dichotomy
  for Unions of Conjunctive Queries (UCQ) when the probabilities are
  arbitrary rational numbers, showing that, for each query, its
  complexity is either in polynomial time or \#P-hard.  The query is
  called {\em safe} in the first case, and {\em unsafe} in the second
  case.  Here, we strengthen the hardness proof, by proving that an
  unsafe UCQ query remains \#P-hard even if the probabilities are
  restricted to $\set{0,\frac{1}{2},1}$.  This requires a complete
  redesign of the hardness proof, using new techniques.  A related
  problem is the {\em model counting problem}, which asks for the
  probability of the query when the input probabilities are restricted
  to $\set{0,\frac{1}{2}}$.  While our result does not extend to model
  counting for all unsafe UCQs, we prove that model counting is
  \#P-hard for a class of unsafe queries called Type-I forbidden
  queries.
  
\end{abstract}

\copyrightyear{2021}
\acmYear{2021}
\setcopyright{acmcopyright}\acmConference[PODS '21]{Proceedings of the 40th ACM SIGMOD-SIGACT-SIGAI Symposium on Principles of Database Systems}{June 20--25, 2021}{Virtual Event, China}
\acmBooktitle{Proceedings of the 40th ACM SIGMOD-SIGACT-SIGAI Symposium on Principles of Database Systems (PODS '21), June 20--25, 2021, Virtual Event, China}
\acmPrice{15.00}
\acmDOI{10.1145/3452021.3458313}
\acmISBN{978-1-4503-8381-3/21/06}

\begin{document}
\fancyhead{}
	
\maketitle

\section{Introduction}\label{sec:intro}
Fix a First Order (FO) sentence $Q$.  The {\em generalized model
  counting problem} for $Q$ is the following: given a database $DB$,
and a subset of tuples $D_1\subseteq DB$, count the number of models
of $Q$ that are subsets of $DB$ and include all tuples in $D_1$.  In
the {\em model counting problem}, we set $D_1 = \emptyset$, and the
problem is to count the number of models of $Q$ that are subsets of
$DB$.


An equivalent formulation to the generalized model counting problem is
to state it as a special case of the query evaluation problem on
Tuple-Independent probabilistic Database
(TID)~\cite{DBLP:series/synthesis/2011Suciu}.  In that setting, each
tuple in the domain has an associated probability, and the problem is
to compute the probability that a query $Q$ is true over a possible
world obtained by including randomly and independently each tuple in
the domain.  It is not hard to see that the generalized model counting
problem corresponds to restricting the probabilities to
$\set{0,1/2,1}$, while the model counting problem corresponds to the
restricting them to $\set{0,1/2}$.

In this paper we will restrict the query $Q$ to be a Union of
Conjunctive Queries or, equivalently, to a dual of a UCQ (see below).
The following dichotomy is known~\cite{DBLP:journals/jacm/DalviS12}:
either $Q$ is computable in PTIME over any TID, or the query
evaluation problem is provably \#P-hard over arbitrary TIDs.  In the
first case the query is called {\em safe}, in the second case it is
called {\em unsafe}.  Moreover, one can decide by static analysis over
the expression of the query $Q$ whether it is safe or unsafe.  In this
paper we strengthen that result by proving a dichotomy theorem for the
generalized model counting problem: for any UCQ $Q$, if $Q$ is safe
then the generalized model counting problem is in PTIME, and if $Q$ is
unsafe then the problem is \#P-hard.  The syntactic classification
into safe/unsafe queries remains the same as for arbitrary TIDs.

One side of this result is trivial.  If the query $Q$ is safe, then it
can be evaluated in PTIME over any TID, even if all probabilities are
in $\set{0,1/2,1}$.  This paper is about the other side of the proof:
if $Q$ is unsafe, then we show that the query evaluation problem is
\#P-hard even if the TID is restricted to have probabilities in
$\set{0,1/2,1}$.  As we explain below, some parts of the hardness
proof in~\cite{DBLP:journals/jacm/DalviS12} (namely Sections 6 and 7)
continue to hold even when the probabilities are restricted to
$\set{0,1/2,1}$, but the most difficult part (Section 8) requires an
entirely new proof.  The reason is that the previous
proof~\cite{DBLP:journals/jacm/DalviS12} relies on multiple distinct
probabilities in $(0,1)$, a number that depends on the size of the
database.  In our paper we develop entirely new proof techniques for
this most critical piece of the hardness proof.  Before we present the
technical material, we comment on several aspects of our
contributions.

\subsection{Significance}

If $Q$ is an unsafe query, then the evaluation problem over arbitrary
TIDs is \#P-hard.  But this problem may become tractable if one
restricts the input TID.  For example, Amarilli et
al.~\cite{DBLP:conf/icalp/AmarilliBS15} prove that {\em any} query $Q$
can be evaluated in PTIME if the input TID has bounded tree-width, and
Van den Broeck et
al.~\cite{DBLP:conf/kr/BroeckMD14,DBLP:conf/pods/BeameBGS15} prove
that any query in $FO^2$ can be evaluated in PTIME if the TID is {\em
  symmetric}.  This leads to a natural question: could the query
evaluation problem become easier if we restrict what probabilities can
appear in the TID?  Our result answers this negatively: query
evaluation remains hard even if the probabilities are restricted to
$\set{0,1/2,1}$.  In fact, it remains hard even if the probabilities
are restricted to $\set{0,c,1}$, where $c \in (0,1)$ is any fixed
constant. The only property needed in our proof is the following
simple fact:

\begin{lemma} \label{lemma:three:values} Let $f(x_1, \ldots, x_n)$ be
  a multivariate polynomial, not identically 0, where each variable
  has degree $\leq 2$.  Let $c_1, c_2, c_3 \in \R$ be three distinct
  constants.  Then there exists an assignment $\theta$ of the
  variables $x_1, \ldots, x_n$ with values in $\set{c_1,c_2,c_3}$,
  such that $f[\theta]\neq 0$.
\end{lemma}

The proof is by induction on $n$: if $f = g x_n^2 + h x_n + k$, where
$g,h,k$ are multi-variate polynomials in the other variables, at least
one not identically 0, then there exists an assignment $\theta$
s.t. at least one of $g[\theta],h[\theta],k[\theta]$ is $\neq 0$;
since a degree 2 polynomial in $x_n$ has at most 2 roots, there exists
a ``non-root'' $c_i \in \set{c_1, c_2, c_3}$, thus
$f[\theta; x_n:=c_i]\neq 0$.  While our hardness proof is complex, the
only probabilities that we need to set in the proof are either 0 or 1
or non-roots of a degree 2 multivariate polynomial, where we use
Lemma~\ref{lemma:three:values}.  Thus, if $Q$ is unsafe, then it
remains \#P-hard even if all probabilities are in $\set{0,c,1}$ for
some fixed $c \in (0,1)$.  This ruins any hope of improving query
evaluation by restricting the probability values.

\subsection{Generalized- v.s. Standard Model Counting}

Our result states that, if a query is unsafe, then the evaluation
problem over TID's with probabilities in $\set{0,1/2,1}$ is \#P-hard.
But what if the probabilities were restricted to $\set{0,1/2}$?  This
corresponds to the model counting problem, and is a natural setting in
probabilistic databases, where tuples not present in the database have
probability 0, while those in the database could be associated with
probability $1/2$.  Amarilli and
Kimelfeld~\cite{DBLP:journals/corr/abs-1908-07093} study precisely
this problem and prove a dichotomy for the model counting problem for
conjunctive queries without self-joins.  We also prove \#P-hardness
for the model counting problem, but only for a restricted class of
queries called {\em final, type I} queries.  This complements the
result proven by Amarilli and Kimelfeld.  We leave open the question
whether any unsafe UCQ is hard for model counting.  Thus, with the exception
of forbidden type I queries, in this paper we study the {\em
  generalized} model counting problem for a query $Q$, which we denote
by $\gmc(Q)$.  We argue next that this is, in fact, a more natural
problem than the model counting problem.

\subsection{Dual Queries}

The {\em dual} of a First Order query $Q$ is obtained by switching
$\exists$ and $\forall$, and switching $\vee$ and $\wedge$.  Many
problems over Boolean formulas are closed under duals, for example the
satisfiability for a class of Boolean formulas is in PTIME iff the
validity for the class of duals is also in PTIME; similarly, model
counting has the same complexity for a class of Boolean formulas and
for its dual.  The same property holds for $\gmc$.  The query
evaluation problem for $Q$ on a TID is essentially the same as the
evaluation problem of its dual $Q'$ on the TID where each probability
$p$ is replaced by $p' \defeq 1-p$, because $\Pr(Q) = 1-\Pr'(Q')$.
Thus, $\gmc(Q)$ and $\gmc(Q')$ have the same complexity.  However,
duality does not preserve the complexity for model counting: if $Q$ is
hard on TIDs with probabilities $\set{0,1/2}$, it doesn't follow that
is dual $Q'$ is also hard on TIDs with probabilities in $\set{0,1/2}$.
For that reason, we find the {\em generalized}
model counting problem a more robust notion than the model counting
problem, and will focus on it in this paper.

Throughout the paper we will discuss duals of UCQs instead of UCQs.
We denote the class of sentences that are duals of UCQs by {\em
  $\forall$CNF}.  Since {\em model counting} for UCQs correspond to
restricting probabilities to $\set{0,1/2}$, we define {\em model
  counting} for $\forall$CNF to mean restricting the probabilities to
$\set{1/2,1}$. 


\subsection{Final Queries}

The hardness proofs in~\cite{DBLP:journals/jacm/DalviS12} follows a
simple high level structure.  First, show that if $Q$ is unsafe, then
it can be rewritten to a simpler query $Q'$, which is still unsafe, such the
computation problem for $\Pr(Q')$ can be reduced in polynomial time to
that for $\Pr(Q)$; in particular, hardness of $Q'$ implies hardness
of $Q$.  Second, provide a direct \#P-hardness proof for any unsafe
query $Q'$ in the simpler class.  Usually, the polynomial time
reduction consists of adding to the database tuples with probabilities
0 or 1, never different probability values (see~\cite[Definition
4.13]{DBLP:journals/jacm/DalviS12}).  For example, to prove that
$\forall x \forall y (R(x) \vee S(x,y) \vee T(y) \vee A(x))\wedge
\forall yB(y)$ is hard it suffices to set all probabilities in $A$ to
0 and all probabilities in $B$ to 1, and obtain the query
$\forall x \forall y(R(x) \vee S(x,y) \vee T(y))$ then prove that the
latter is hard.  A {\em forbidden query} is an unsafe query where no
more rewriting to a simpler unsafe query is possible.  

The first step of the proof in~\cite{DBLP:journals/jacm/DalviS12} also
applies to $\gmc$.  Thus, in order to prove that $\gmc(Q)$ is hard for
any unsafe query $Q$, it suffices to prove it for ``forbidden''
queries $Q$.  The definition of forbidden queries
in~\cite{DBLP:journals/jacm/DalviS12} is quite technical.  To prove hardness for Type I queries,
we do not need forbidden queries, but instead prove hardness for a
slightly larger, and easier to describe class of queries $Q$, called
{\em final} queries; we will return to forbidden queries in the \blue{full version of this} paper to prove hardness for Type II queries~\cite{DBLP:journals/corr/abs-2008-00896}.
\eat{(Sec.~\ref{sec:background:factorization}).}  All queries considered in
this paper are {\em bipartite queries}, meaning that they have only
two variables, $x,y$, and three kinds of atoms: two unary atoms
$R(x), T(y)$, and an arbitrary number of binary atoms $S_j(x,y)$,
$j=1,2,\ldots$.  A type I query contains both atoms $R(x)$ and $T(y)$,
and is {\em unsafe} iff these two atoms are connected by the clauses
of $Q$.  If setting any atom to $0$ or to $1$ makes the query safe,
then we call it {\em final}; if the query is not final, then we can
simply set that atom to 0 or 1 and prove hardness for the simpler
query $Q'$.  Our main result in this paper consists of proving that,
for every final query, $\gmc(Q)$ is \#P-hard.



\subsection{Reducing \#P2CNF to \gmc}

A positive 2CNF formula is
$\Phi = \bigwedge_{(i,j) \in E}(X_i \vee X_j)$.  The model counting
problem, denoted \#P2CNF is \#P-hard.  Provan and
Ball~\cite{DBLP:journals/siamcomp/ProvanB83} proved that it remains
\#P-hard even when the graph $E$ is bipartite, in which case the
problem is denoted \#PP2CNF.  We use reductions from these problems to
prove hardness for $\gmc(Q)$.

Our reduction is a polynomial-time reduction, also called {\em
  Cook-reduction} (as opposed to the many-one polynomial time
reduction, or {\em Karp-reduction}).  Specifically, we construct a
sequence of databases $DB_1, DB_2, \ldots$ with probabilities in
$\set{0,c,1}$ (i.e., where $c\in (0,1)$ is fixed), and show how to compute $\#\Phi$ from
$\Pr_1(Q), \Pr_2(Q), \ldots$ To compute $\#\Phi$, we need to solve a
linear system of equations, and the crux of the hardness proof consists
of showing that the matrix of this system is non-singular.  We call
this matrix the {\em big matrix}, since its size is polynomial in the
size of $\Phi$.


This is the place where our proof diverges from that
in~\cite{DBLP:journals/jacm/DalviS12}.  There, the databases
$DB_1, DB_2, \ldots$ were isomorphic, but used different
probabilities, hence the need to use very many distinct probability
values.  In contrast, we construct databases whose probabilities are
only among $\set{0,c,1}$, and, instead, we vary the number of
tuples.  Each database $DB_i$ consists of {\em disjoint paths} (formal
definition in Sec.~\ref{sec:BlockTID}) of lengths that depends on $i$.
Instead of varying the probability values in $(0,1)$ as
in~\cite{DBLP:journals/jacm/DalviS12}, here we fix the probability
values in $\set{0,c,1}$ and vary the length of the paths.

\subsection{Connecting Logic and Algebra}

We show that the non-singularity of the big matrix follows by proving
that a certain ``small matrix'', $A$, is non-singular.  The small
matrix is a $2\times 2$ matrix that describes the probabilities of a
single link in the path.  Intuitively, when the small matrix is
non-singular, then by increasing the path length $i$ in $DB_i$ we gain
more information about $\#\Phi$ from probability $\Pr(Q)$ on $DB_i$.
The small matrix depends only on the query $Q$, more precisely, on the arithmetization of its lineage $Y$ on one link of the path.  The
{\em arithmetization} of a Boolean formula $Y$ is a multilinear
polynomial $y$ that agrees with $Y$ at all points in $\set{0,1}^n$.
For example, if
$Q = \forall x \forall y (R(x) \vee S(x,y))\wedge (S(x,y) \vee
T(y))$, then the lineage is $Y = (R \vee S) \wedge (S \vee T)$, and
its arithmetization is $y(r,s,t) = rt + s - rst$. Equivalently, the
arithmetization is the formula for the probability of $Y$; if
$\pr(R)=\pr(S)=\pr(T)=1/2$, then $\pr(Y)=y(1/2,1/2,1/2)=5/8$.  The
small matrix $A$ is $\left(
\begin{array}{cc}
 y[r=0,t=0] & y[r=0,t=1] \\ y[r=1,t=0] & y[r=1,t=1]
\end{array}
\right)=
\left(
\begin{array}{cc}
s & s \\ s & 1
\end{array}
\right)
$.  At its core, our hardness proof relies on the following
connection between logic and algebra: the small matrix is non-singular
iff the bipartite query $Q$ \e{connects} the atoms $R$ and $S$.  We state
here the formal lemma that captures this connection.

\begin{lemma} \label{lemma:determinant:connected} Let $y$ be the
  arithmetization of a Boolean formula $Y$, and let $R, T$ be two
  Boolean variables.  Denote the following matrix of polynomials:
  \begin{align}
  \label{eq:smallMatrix}
    \mb y \defeq &
    \begin{bmatrix}
    y_{00} & y_{01}\\ y_{10} & y_{11}
    \end{bmatrix}
  \end{align}
  where $y_{00} \defeq y[r:=0; t := 0]$ is obtained by setting $r,t$
  to 0, and similarly for $y_{01}, y_{10}, y_{11}$.  Then the
  following are equivalent: (1) $Y$ disconnects $R,T$, i.e.
  $Y \equiv F \wedge G$ such that
  $R \in \vars(F), T \in \vars(G), \vars(F)\cap \vars(G) =\emptyset$.
  (2) $\det(\mb y) \equiv 0$.
\end{lemma}

The direction (1) $\Rightarrow$ (2) is immediate, because $y$
factorizes as $y=f \cdot g$ where $r$ occurs only in $f$ and $t$ only
in $g$, thus $y_{ij} = f_i \cdot g_j$ and it follows that
$\det(\mb y) \equiv 0$.  For (2) $\Rightarrow$ (1) assume that
$y_{00} \cdot y_{11} \equiv y_{01} \cdot y_{10}$.  Assume\footnote{If
  $h \defeq \gcd(y_{00}, y_{01}, y_{10}, y_{11})\neq 1$, then we
  define $y'_{ij} \defeq y_{ij}/h$, obtain $y'_{ij} = f'_i \cdot g_j$,
  and define $f_i \defeq f'_ih$.}  w.l.o.g. that
$\gcd(y_{00}, y_{01}, y_{10}, y_{11})=1$, which also implies
$\gcd(y_{00}, y_{01}, y_{10})=1$ because $y_{00}$ is a multi-linear
polynomial\footnote{Assuming $p=\gcd(y_{00}, y_{01}, y_{10})$ then
  $p^2 | y_{00} \cdot y_{11} \equiv y_{01} \cdot y_{10}$, but
  $p \not| y_{11}$, implying $p^2 | y_{00}$, which is impossible since
  $y_{00}$ is multilinear.}.  Then define
$f_i \defeq \gcd(y_{i0},y_{i1})$ and $g_j \defeq\gcd(y_{0j},y_{1j})$, and
we have $y_{ij} = f_i\cdot g_j$ for all $i,j =0,1$, because
$\gcd(f_0, g_0) = \gcd(y_{00},y_{01},y_{10})=1$, and similarly for
$\gcd(f_i,g_j)=1$\footnote{To see this, consider $y_{00}$, and let $p$
  be an irreducible factor of $y_{00}$ (i.e.,
  $y_{00}=y'_{00}p$). Therefore, since
  $y_{00}y_{11}\equiv y_{01}y_{10}$ then $p|y_{01}$ or $p|y_{10}$. If
  it is the former, then $p\in\gcd(y_{00},y_{01})=f_0$, and if the
  latter then $p\in\gcd(y_{00},y_{10})=g_0$. Hence,
  $y_{00}\equiv f_0g_0$.}. The claim follows from
$y= (1-r)(1-t)y_{00} + (1-r)ty_{01} + r(1-t)y_{10}+rty_{11} =
((1-r)f_0+rf_1)((1-t)g_0+tg_1)\defeq f \cdot g$, which implies
condition (1) by defining $F,G$ the Boolean formulas associated to
$f,g$ respectively.

Finally, our proof falls into place by the following argument.  Let
$Q$ be a bipartite query.  If $Q$ connects the atoms $R(x), T(y)$,
then its lineage $Y$ is connected, and thus
$\det(\mb y) \not\equiv 0$.  By Lemma~\ref{lemma:three:values}, there
exist probability values in $\set{0,c,1}$ that ensures that the small
matrix $A$ is non-singular.  Then, the big matrix is non-singular, and
we can compute in polynomial time $\#\Phi$ from the probabilities
$\Pr_1(Q), \Pr_2(Q), \ldots$, completing the reduction.

\subsection{Type II Queries}

When a bipartite query $Q$ contains the atoms $R(x), T(y)$, then we
call it a type I query.  Otherwise, we call it a type II query.  Our
discussion so far has been restricted to type I queries; we prove
their hardness in the main body of the paper.  We discuss type II queries,
and prove their hardness in the full version of this paper~\cite{DBLP:journals/corr/abs-2008-00896}.  The proof for Type II
queries is more complex than for Type I queries.  While the two proofs
share many common techniques, they are sufficiently distinct to
justify a completely separate proof for Type II queries.  In
particular, the proof for type I queries is by reduction from \#P2CNF,
and that for type II queries by reduction from \#PP2CNF.

\subsection{Organization}

In Section~\ref{sec:probStatement} we provide background on final
queries, and formally state the problem and main result.  We prove the
hardness for queries of type I in Section~\ref{sec:BlockTID}, and
defer the proof for type II queries to the full version of this paper~\cite{DBLP:journals/corr/abs-2008-00896}.

\section{Problem Statement}\label{sec:probStatement}
\label{sec:problem:statement}

Fix a finite domain $\Dom$ and let $\Tup(\Dom)$ denote the set of ground
tuples consisting of relation names from some fixed vocabulary, and
constants from $\Dom$.  With some abuse, we write $\Tup$ instead of
$\Tup(\Dom)$ when the domain is clear from the context.

A {\em tuple-indepepndent probabilistic database} is a pair
$\Delta = (\Dom,p)$ where $p : \Tup \rightarrow [0,1]$ associates a
probability to each tuple.  With some abuse, we simply say that
$\Delta$ is a {\em probabilistic database}.  Given a Boolean query
$Q$, its probability is defined as
$\Pr(Q) \defeq \sum_{W \subseteq \Tup: W \models Q} \Pr(W)$, where the
probability of a world $W$ is
$\Pr(W) \defeq \prod_{t \in W} p(t) \times \prod_{t \in
  \Tup-W}(1-p(t))$.  The {\em probabilistic query evaluation problem},
$\pqe(Q)$, is the following: given a probabilistic database
$(\Dom, p)$, where $p$ takes rational values, compute $\Pr(Q)$.  The
following dichotomy theorem was shown
in~\cite{DBLP:journals/jacm/DalviS12}:

\begin{theorem}~\cite{DBLP:journals/jacm/DalviS12} \label{th:old:dichotomy}
  Let $Q$ be any UCQ,
  then one of the following holds: either $\pqe(Q)$ is in PTIME, or
  $\pqe(Q)$ is \#P-hard.  Moreover, there exists a syntactic condition
  on the query $Q$ called {\em safety} such that $\pqe(Q)$ is in PTIME
  when $Q$ is safe, and is \#P-hard when $Q$ is unsafe.
\end{theorem}

For the purpose of this paper we do not need the general definition of
safety, and will omit it, except for the special case of bipartite
queries defined below.  If $Q$ is not safe, then we call it {\em
  unsafe}.

The {\em generalized model counting problem} for $Q$, denote
$\gmc(Q)$, is the following restriction: given a probabilistic
database $(\Dom,p)$ where $p(t) \in \set{0,1/2,1}$ for every tuple
$t$, compute $\Pr(Q)$.  Obviously, if $\pqe(Q)$ is in PTIME, then so
is $\gmc(Q)$.  We prove that the converse holds too:

\begin{theorem} \label{th:new:dichotomy}
  For any unsafe UCQ $Q$, $\gmc(Q)$ is \#P-hard.  This result
  continues to hold even if the probability $1/2$ is replaced by some
  constant probability $c \in (0,1)$.
\end{theorem}

The high level structure of the proof is similar to that
in~\cite{DBLP:journals/jacm/DalviS12}.  Starting with an unsafe query
$Q$, first simplify it using simple rewrite rules, as long as the
query is still unsafe, until one reaches an unsafe query where every
further simplification is safe.
\eat{This last query is called a {\em  forbidden} query.  }
Second, prove that each such simplified query $Q$ (called
in~\cite{DBLP:journals/jacm/DalviS12} a {\em forbidden} query),
$\gmc(Q)$ is \#P-hard.  The first part of the proof is identical
to~\cite{DBLP:journals/jacm/DalviS12}; we will only state the main
result in Theorem~\ref{th:proof:step:1} below, and refer the reader
to~\cite{DBLP:journals/jacm/DalviS12} for the proof.  The second part
is novel and will be presented in the rest of the paper.

In this paper we will prove Theorem~\ref{th:new:dichotomy} for the
duals of UCQs, and denote their class $\forall$CNF.  We briefly review
$\forall$CNF here.  A {\em clause} is a disjunction of atoms,
$C \equiv R_1(\mb x_1) \vee R_2(\mb x_2) \vee \cdots$; we note that
this is the dual notion of a conjunctive query.  A {\em homomorphism}
$C\rightarrow C'$ is a function $\vars(C) \rightarrow \vars(C')$ that
maps every atom in $C$ to an atom in $C'$.  We say that $C$ is
minimized if every homomorphism $C \rightarrow C$ is a bijection.  It
is known that every clause is equivalent to a minimized
clause\footnote{This follows from the same property for conjunctive
  queries.}.  A $\forall$CNF formula is a conjunction of clauses,
$Q \equiv C_1 \wedge C_2 \wedge \cdots$; this is the dual of a UCQ.
If there exists a homomorphism $C_i \rightarrow C_j$ then $C_j$ is
redundant, and can be removed from $Q$.  In this paper we always
assume, unless otherwise stated, that all clauses are minimized, and
there is no redundant clause.  All logical variables are universally
quantified, and we will freely switch between prenex normal form of
the entire sentence, or of each clause individually, e.g.
$\forall x \forall y (C_1 \wedge C_2)$ v.s.
$\forall x \forall y C_1 \wedge \forall x \forall y C_2$.

The \e{lineage} of a $\forall$CNF query $Q$ on a tuple independent database $\pdb=(\Dom,p)$ is the propositional formula $\lin{\pdb}{Q}$ computed as usual, by induction\footnote{\begin{align*}
	\Phi_\pdb(\forall x Q){\defeq} & \bigwedge_{u \in  \Dom} \Phi_\pdb(Q[u/x]) && \\
	\Phi_\pdb(Q_1 {\wedge} Q_2){\defeq} & \Phi_\pdb(Q_1) {\wedge} \Phi_\pdb(Q_2) & \Phi_\pdb(Q_1 {\vee} Q_2){\defeq} & \Phi_\pdb(Q_1) {\vee} \Phi_\pdb(Q_2)
	\end{align*}} on the structure of $Q$. 
We assume that each tuple in the domain $\Dom$ is associated with a Boolean
variable, thus $\Phi_\pdb(Q)$ is a Boolean function over the tuples. We remark that $\lin{\pdb}{Q}$ is a Boolean CNF whose size is polynomial in the size of the domain $\Dom$.

We consider a restricted vocabulary consisting of two unary symbols
$R, T$, and one or more binary symbols $S_1, S_2,\ldots$. We call a
probabilistic database ${\pdb} = (\Dom, p)$ {\em bipartite} if $\Dom$ is
the disjoint union $\Dom = U \cup V$, and every tuple $t$ with
probability $p(t) \neq 1$ is either $t = R(u)$ or $t = S_j(u,v)$ or
$t = T(v)$ for $u \in U$ and $v \in V$.  We denote by
$\gmc_{\text{bi}}(Q)$ the $\gmc$ problem where the probabilistic
database is restricted to be bipartite.

We define next a bipartite query.  It has only two variables $x,y$,
its atoms are restricted to be $R(x)$, or $T(y)$, or $S_j(x,y)$ for
some $j=1,2,\ldots$  
Denote by $S_J(x,y) \defeq \bigvee_{j \in J} S_j(x,y)$.

\begin{definition}\label{def:types:of:clauses}
  We define the following types of sentences:
\begin{itemize}
	\item A {\em left clause of Type I} is $\forall x \forall y(R(x) \vee S_J(x,y))$, where
	$J \neq \emptyset$.
	\item A {\em left clause of Type II} is
	$\forall x \left(\bigvee_{\ell=1}^m (\forall y S_{J_\ell}(x,y))\right)$ where
	$\emptyset \neq J_1, \ldots, J_m \subseteq [p]$, and $m > 1$.
	\item A {\em middle clause} is $\forall x \forall y S_J(x,y)$ for $J \neq \emptyset$.
	\item A {\em right clause of Type I} is $\forall y \forall x(S_J(x,y) \vee T(y))$, where
	$J \neq \emptyset$.
	\item A {\em right clause of Type II} is
	$\forall y \left(\bigvee_{\ell=1}^n (\forall x S_{J_\ell}(x,y))\right)$ where
	$\emptyset \neq J_1, \ldots, J_n \subseteq [p]$, and $n > 1$.	
\end{itemize}
For every $A, B \in \set{\text{I}, \text{II}}$, we define a {\em
  bipartite query} of type $A-B$ to be a $\forall$CNF query
$Q \defeq Q_{\text{left}} \wedge Q_{\text{middle}} \wedge
Q_{\text{right}}$ where $Q_{\text{left}}$ is a conjunction of left
clauses of type $A$, $Q_{\text{middle}}$ is a conjunction of middle
clauses, and $Q_{\text{right}}$ is a conjunction of right clauses of
type $B$.
\end{definition}

An example of a left clause of type I is
$\forall x \forall y(R(x) \vee S_1(x,y) \vee S_2(x,y))$.  An example
of a left clause of type II is
$\forall x (\forall y S_1(x,y) \vee \forall y S_2(x,y))$.  For Type II
clauses the term ``clause'' is used with some abuse, since this
sentence is not in prenex normal form; the prenex normal form of our
example is
$\forall x \forall y_1 \forall y_2 (S_1(x,y_1) \vee S_2(x,y_2))$.

Definition 4.14 in~\cite{DBLP:journals/jacm/DalviS12} calls a UCQ
query {\em safe} if it satisfies a certain syntactic condition.  We
review below the safety definition only for the special case of
bipartite queries and, for the intuition behind this definition, we
make two observations.  First, let
$Q = Q_{\text{left}} \wedge Q_{\text{middle}}$ be a bipartite query
without any right clauses.  Then one can compute $\Pr(Q)$ in
polynomial time on a probabilistic database $\Delta=(\Dom,p)$, in
other words $Q$ is {\em safe}.  Indeed,
$Q \equiv \bigwedge_{a \in \Dom} Q[a/x]$, $x$ occurs in each atom of
$Q$, hence the events $Q[a_1/x], Q[a_2/x], \ldots$ are independent,
and therefore $\Pr(Q) = \prod_{a \in \Dom} \Pr(Q[a/x])$.  It is then
easy to check that each $\Pr(Q[a/x])$ can be computed in polynomial
time.\footnote{It has only unary atoms.}  Second, suppose $Q$ is a
bipartite query that can be partitioned into $Q \equiv Q' \wedge Q''$,
where $\symb(Q') \cap \symb(Q'') = \emptyset$, $Q'$ has no right
clauses, and $Q''$ has no left clauses.  Then $Q$ is again safe,
because $\Pr(Q) = \Pr(Q') \cdot \Pr(Q'')$.  This justifies:

\begin{definition}
	\label{def:unsafeQueries}
        A {\em bipartite query} is {\em unsafe} if it contains a left
        clause $C_0$ and a right clause $C_k$ that are connected by a
        path, i.e. there exists a sequence $C_0, C_1, \ldots, C_k$,
        called a {\em left-to-right path}, such that every two
        consecutive clauses share an atom.  The {\em length} of the
        unsafe query is the minimal $k$ for which there exists a
        left-to-right path of length $k$.
\end{definition}


Let $A, B$ be two problems.  A {\em polynomial time many-one
  reduction} from $A$ to $B$, in notation $A \leq^P_m B$, is a
polynomial time computable function $f$ such that, for every instance
$x$ of $A$, $f(x)$ is an instance of $B$ and the answer to problem $A$
on $x$ is the same as the answer to problem $B$ on $f(x)$.  A {\em
  polynomial time reduction} from $A$ to $B$, in notation
$A \leq^P B$, is a polynomial time algorithm for solving instances of
$A$ given an oracle for solving instances of $B$.  The first part of
the proof of Theorem~\ref{th:new:dichotomy} is given by the following
theorem, where
$H_0 \defeq \forall x \forall y (R(x) \vee S(x,y) \vee T(y))$.

\begin{theorem}\label{th:proof:step:1}
  Let $Q$ be a $\forall$CNF query that is {\em unsafe} (according to
  Definition 4.14 in~\cite{DBLP:journals/jacm/DalviS12}).  Then,
  either $\gmc_{\text{bi}}(H_0) \leq^P_m \gmc_{\text{bi}}(Q)$, or
  there exists a bipartite, unsafe query $Q'$ such that
  $\gmc_{\text{bi}}(Q') \leq^P_m \gmc_{\text{bi}}(Q)$.
\end{theorem}

\begin{proof}
  (Sketch) The proof follows directly
  from~\cite{DBLP:journals/jacm/DalviS12}, more precisely from
  Theorems 6.3 and 7.3; note that a bipartite query is called a
  two-leveled query in~\cite{DBLP:journals/jacm/DalviS12}.  The only
  tool used in those proofs is Lemma 4.17, and it continues to hold if
  all probabilities are in $\set{0,1/2,1}$.
\end{proof}

\eat{
the rewriting that sets all tuples in the domain of $S$
to $1$, and by $S\rightarrow 0$ the rewriting that sets all tuples in
the domain of $S$ to $0$. Effectively, the rewriting $S\rightarrow 1$
removes all clauses of $Q$ that contain $S$. Observe that all remaining
clauses are non-redundant. The rewriting $S\rightarrow 0$ removes the
symbol $S$ from $Q$.
}

Thus, in order to prove Theorem~\ref{th:new:dichotomy}, it suffices to
show that $\gmc_{\text{bi}}(H_0)$ is \#P-hard, and $\gmc_{\text{bi}}(Q)$
is \#P-hard for any bipartite, unsafe query $Q$.  Hardness for $H_0$ is
already shown in~\cite{DBLP:journals/jacm/DalviS12}, because that proof
only uses probabilities in $\set{0,1/2,1}$; furthermore, Amarilli and
Kimelfeld~\cite{DBLP:journals/corr/abs-1908-07093} strengthened this
result by showing that model counting for $H_0$ is \#P-hard.  Thus, we
will not consider $H_0$ any further in this paper; we only consider
bipartite queries.

Next, we need a technical lemma, whose proof is in the full version of the paper~\cite{DBLP:journals/corr/abs-2008-00896}.
\def\lemmalong{Let $Q$ be a bipartite, unsafe query of type $A-B$,
  with $A,B \in \set{\text{I},\text{II}}$, and of length $k$.  Then
  there exists a bipartite, unsafe query $Q'$ of type $A-A$ of length
  $\geq 2k$ such that
  $\gmc_{\text{bi}}(Q') \leq_m^P \gmc_{\text{bi}}(Q)$.  }
\begin{lemma}\label{lemma:long}
  \lemmalong
\end{lemma}

We introduce now two simplification rules for queries (these are subsets
of the rules in Definition 4.13~\cite{DBLP:journals/jacm/DalviS12}).
Fix a bipartite query $Q$, and let $S$ be any relational symbol. We
denote by $Q[S:=\texttt{false}]$ the query obtained by replacing every
occurrence of the atom $S$ by $\texttt{false}$.  Similarly,
$Q[S:=\texttt{true}]$ is obtained by replacing $S$ by $\texttt{true}$.
We sometimes abbreviate these rewritings by $Q[S:=0]$, $Q[S:=1]$.  As
discussed earlier, we always assume that the rewritten query is reduced,
by minimizing its clauses and removing redundant clauses.

\begin{lemma} \label{lemma:supersimple}
  Let $Q$ be a bipartite query, and let $Q'$ be either $Q[S:=0]$ or
  $Q[S:=1]$.  Then (1)
  $\gmc_{\text{bi}}(Q') \leq^P_m \gmc_{\text{bi}}(Q)$, (2) $Q$ and $Q'$
  have the same types (I or II), (3) if $Q'$ is unsafe, then so is $Q$
  (but the converse does not hold in general), and (4) the length of
  $Q'$ is $\geq$ the length of $Q$.
\end{lemma}

The proof is immediate and omitted.  This justifies:

\begin{definition}
\label{def:finalQueries}
A \e{final} query is a bipartite, unsafe query $Q$ with the following
property: for any symbol $S$ of $Q$, both $Q[S:=0]$ and $Q[S:=1]$ are
safe queries. 
\eat{For every $A,B \in \set{I,II}$ we denote by $\texttt{f-Type-A-B}$ a final query of type $A-B$.}
\end{definition}

Intuitively, if we want to prove that an unsafe query $Q$ is \#P-hard,
we can simplify it first to $Q' = Q[S:=0]$ or $Q' = Q[S:=1]$ and prove
that $Q'$ is \#P-hard. A {\em final} query is one where no further
simplifications are possible.  A ``forbidden query'' (Definition 7.2
in~\cite{DBLP:journals/jacm/DalviS12} \blue{and Definition~C.10 in the
  full version of this
  paper~\cite{DBLP:journals/corr/abs-2008-00896}}) is defined
similarly, but considers some additional simplifications, thus, every
forbidden query is final, but the converse does not hold.  In this
paper we only discuss final queries, but will return to forbidden
queries in the full version~\cite{DBLP:journals/corr/abs-2008-00896}.

We say that $Q$ is disconnected if $Q \equiv Q' \wedge Q''$, where
$Q', Q''$ use disjoint sets of symbols.  Every final query is connected.
Indeed, assuming the contrary, one of $Q',Q''$ must must be unsafe, and,
assuming $Q'$ is unsafe, we can set to $\texttt{true}$ all symbols $S$
occurring in $Q''$ and obtain the simpler query $Q'$ which is still
unsafe, contradicting the fact that $Q$ is final.


\eat{Let $G(V,E)$ be a graph where $V=\set{X_1,...,X_n}$.  The positive 2CNF
(P2CNF) associated with the graph is
$\Phi = \bigwedge_{(i,j) \in E}(X_i \vee X_j)$.  When $G$ is bipartite,
then we call $\Phi$ a {\em partitioned} P2CNF, in notation PP2CNF:
$\Phi = \bigwedge_{(i,j) \in E}(X_i \vee Y_j)$, where
$X_i, i=1,2,\ldots$ and $Y_j, j=1,2,\ldots$ are disjoint sets of
variables. The model counting problems associated with P2CNF and PP2CNF
are called \#P2CNF and \#PP2CNF respectively, and both are \#P-hard.}

Finally, we can now state the main technical result of this paper, which
immediately implies Theorem~\ref{th:new:dichotomy}.  Recall that the
model counting problem for UCQ's is defined as $\pqe(Q)$ where the
probabilities are restricted to $\set{0,1/2}$.  Since our discussion is for
$\forall$CNF rather than UCQ's, we define the {\em model counting
  problem}, $\fomc(Q)$, as $\pqe(Q)$ where the probabilities are
restricted to $\set{1/2,1}$.  In this paper we prove:

\begin{theorem} \label{th:dichotomy:main:lemma} (1) If $Q$ is a final
  query of type I, then $\fomc_{\text{bi}}(Q)$ is \#P-hard. (2) If $Q$
  is a final query of type II of length $\geq 5$, then
  $\gmc_{\text{bi}}(Q)$ is \#P-hard.
\end{theorem}
\eat{
\batya{We should add some discussion that the type-1 proof establishes hardness of FOMC (not GFOMC)}
}

In the rest of the paper we prove
Theorem~\ref{th:dichotomy:main:lemma} (1), and defer part (2) to the
full version~\cite{DBLP:journals/corr/abs-2008-00896}.  We end this
section by showing how this theorem proves our main result,
Theorem~\ref{th:new:dichotomy}, which we prove for $\forall$CNF
queries rather than UCQs. Let $Q$ be any unsafe $\forall$CNF query.
By Theorem~\ref{th:proof:step:1}, either
$\gmc_{\text{bi}}(H_0) \leq_m^P \gmc(Q)$, in which case the theorem
follows from the fact that $\gmc_{\text{bi}}(H_0)$ is \#P-hard, or
$\gmc_{\text{bi}}(Q') \leq_m^P \gmc(Q)$ for some bipartite, unsafe
query $Q'$ of some type $A-B$.  If $A$ is I, then by
Lemma~\ref{lemma:long}
$\gmc_{\text{bi}}(Q'') \leq^P_m \gmc_{\text{bi}}(Q')$ for some
bipartite, unsafe query $Q''$ of type I-I: w.l.o.g. we may assume that
$Q''$ is final (by Lemma~\ref{lemma:supersimple}), then the result
follows from Theorem~\ref{th:dichotomy:main:lemma} (1).  If $A$ is II,
then we apply Lemma~\ref{lemma:long} three times, to obtain a
bipartite, unsafe query $Q''$ of type II and of length $\geq 8$, such
that $\gmc_{\text{bi}}(Q'') \leq^P_m \gmc_{\text{bi}}(Q')$.  As before,
we can assume w.l.o.g. that $Q''$ is final, hence the result follows
from Theorem~\ref{th:dichotomy:main:lemma} (2).

\section{Hardness of Final Queries of Type-I}\label{sec:BlockTID}

In this section we prove the first item of
Theorem~\ref{th:dichotomy:main:lemma}.  A Positive 2CNF, or P2CNF, is
a formula $\Phi=\bigwedge_{(i,j)\in E}(X_i\vee X_j)$ with $n$
variables and $|E|=m$ clauses.  The problem ``given a P2CNF $\Phi$,
compute the number of satisfying assignments $\#\Phi$'' is denoted
\#P2CNF and is known to be \#P-hard.  In this section we prove:

\begin{theorem} \label{th:hard:type:1}
  For every final query $Q$ of type I, \#P2CNF $\leq^P$
  $\fomc_{\text{bi}}(Q)$.
\end{theorem}

Let $\theta:\set{X_1,\dots,X_n}\rightarrow \set{0,1}^n$ be an
assignment to $\Phi$'s variables. We define its \e{signature} to be
the mapping \blue{$\bk(\theta): \set{0,1}^n\rightarrow \set{0,\ldots,m}^4\times \set{0,\ldots,n}^2$:
$$\bk(\theta) \eqdef
(k_{00}(\theta),k_{01}(\theta),k_{10}(\theta),k_{11}(\theta),q_0(\theta),q_1(\theta))$$}
where
\begin{align*}
   k_{ab}(\theta)\eqdef|\set{(i,j)\in E : \theta(X_i)=a, \theta(X_j)=b}| && ab\in \set{0,1}^2 \\
   q_a(\theta)\eqdef |X_i\in \set{X_1,\dots,X_n} : \theta(X_i)=a| && a\in \set{0,1}
\end{align*}
Thus, $k_{00}$ is the number of clauses where both $X_i, X_j$ are
\false, $k_{11}$ the number of clauses where both are \true, \blue{ and $q_1(\theta)$ is the number of \true~variables $X_i$}.  We
assume that $E$ is a directed graph, and contains at most one of
$(i,j)$ or $(j,i)$ for all $i\neq j$.  Thus, $k_{01}$ and $k_{10}$ may
be different, and their sum $k_{01,10}\defeq k_{01}+k_{10}$ represents
the number of clauses with exactly one variable set to \true, and we
write $\bk'(\theta) \defeq (k_{00}, k_{01,10}, k_{11},q_0,q_1)$ for the
undirected signature.  
For any vector
$\bk=(k_{00},k_{01},k_{10},k_{11},q_0,q_1)$, its \e{count}, $\#\bk$, is the
number of assignments
$\theta:\set{X_1,\dots,X_n}\rightarrow \set{0,1}^n$ with signature
$\bk$, and $\#\bk'$ is the similar \e{undirected counts}:
\begin{align}
\forall \bk\in \set{0,\dots,m}^4\times \set{0,\ldots,n}^2: &&  \veck\eqdef |\set{\theta : \bk(\theta)=\bk}|\\
\forall \bk'\in \set{0,\dots,m}^3\times \set{0,\ldots,n}^2: &&  \veck'\eqdef |\set{\theta : \bk'(\theta)=\bk'}|
\end{align}
Thus, there are \blue{$(m+1)^4(n+1)^2$ counts, and $(m+1)^3(n+1)^2$ undirected counts}, of
which at most $(m+1)^2(n+1)$ are non-zero, because, for any signature,
$k_{00}+k_{01,10}+k_{11}=m$ \blue{ and $q_0+q_1=n$}.
\eat{
\batya{Likewise, For any vector
	$\bk=(k_{00},k_{01},k_{10},k_{11},q_0,q_1)$, its \e{probability}, $\Pr(\bk)$, is the
	total probability mass of all assignments 
	$\theta:\set{X_1,\dots,X_n}\rightarrow \set{0,1}^n$ with signature
	$\bk$, and $\Pr(\bk')$ is the similar \e{undirected probability}. Formally: $\forall \bk \in \set{0,\dots,m}^4\times \set{0,\dots,n}^2$:  $\Pr(\bk)\eqdef \sum_{\theta: \bk(\theta)=\bk}\Pr(\theta)$ and where $\Pr(\theta)=\prod_{i=1}^nc^{\theta(X_i)}(1-c)^{1-\theta(X_i)}$.}
}
To prove Theorem~\ref{th:hard:type:1}, we start from a P2CNF $\Phi$
with $m$ clauses, and construct \blue{$(m+1)^{2}(n+1)$} bipartite TIDs
$\pdb$. Then, we describe an algorithm that, using all probabilities
$\Pr_{\pdb}(Q)$, computes all \blue{$(m+1)^{2}(n+1)$} consistent undirected counts
$\#\mb{k}'$ (the others are $=0$).  The reduction from \#P2CNF
immediately follows by noting that
$\#\Phi=\sum_{\bk':k'_{00}=0}\veck'$.  This strategy requires
computing a polynomial number of counts. The crux of the proof
consists in constructing the databases $\pdb$ to make this computation
possible.  Each such database is a union of {\em blocks}; we describe
next how to compute $\Pr_{\pdb}(Q)$ when the TID $\pdb$ is a union of
blocks.  In Section~\ref{sec:reduction} we present the reduction from
\#P2CNF, and in Section~\ref{sec:blocks} we design the blocks such
that the answers $\Pr_{\pdb}(Q)$ allow us to compute the counts
$\#\mb{k}$.

\subsection{The Block TID}\label{sec:zigzagDB}

Throughout this section we fix a final query $Q$ of type-I:
\begin{equation}
\label{eq:fTypeIQ}
Q{=}\forall x \forall y (\bigwedge_{i=1}^\ell(R(x){\vee} S_{J_i}(x,y))){\wedge} C(x,y) {\wedge} (\bigwedge_{k=1}^r(T(y) {\vee} S_{J_k}(x,y)))
\end{equation}

Its vocabulary is
$\calR = \set{R(x),S_1(x,y), \ldots, S_p(x,y), T(y)}$.  Recall that a
bipartite TID $\pdb=(\Dom,p)$ has a bipartite domain $\Dom=U \cup V$.  
\begin{definition}
	\label{eq:BlockTID}
	A {\em block} $B(u,v)$ is a bipartite TID with two
        distinguished constants $u,v$.  We call $u,v$ the {\em end
          points} of $B(u,v)$, and call any other constant occurring
        in $B(u,v)$ an {\em internal point}.
\end{definition}

\blue{In what follows, }we assume that both end-points $u,v$ of a block are in
its left domain, i.e. there are atoms $R(u), R(v)$ but not
$T(u), T(v)$ (more precisely, the latter have probabilities 1); in the
\blue{full version of this
  paper~\cite{DBLP:journals/corr/abs-2008-00896}} we discuss Type II
queries and there we allow $u,v$ to be on the left or right.  Two blocks
$B(u,v)$ and $B(u',v')$ are called {\em disjoint} if they share at
most their end points:
$\Dom(B(u,v)) \cap \Dom(B(u',v')) \subseteq \set{u,v} \cap
\set{u',v'}$.
%
Therefore, when the blocks are disjoint, they can only share tuples
$R(u)$ or $R(v)$.


\begin{definition}
\label{def:linQ(u,v)}
We denote by $Y(u,v)$ the lineage of $Q$ over the bipartite TID
$B(u,v)$: $Y(u,v) \eqdef \lin{B(u,v)}{Q}$.
For any pair $(a,b)\in \set{0,1}^2$ we define
\begin{equation}
\label{eq:Yab}
Y_{ab}(u,v)\eqdef Y(u,v)[R(u)\asn a, R(v) \asn b]
\end{equation}
\end{definition}
In other words, $Y_{ab}$ is the lineage $Y(u,v)$ where we substitute the Boolean variables $R(u), R(v)$ with the values $a,b$ respectively. \blue{Using~\eqref{eq:Yab}, we define:
\begin{equation}
	\label{eq:Ya}
\eat{	Y_{a}(u)\eqdef Y(u,v)[R(u)\asn a, R(v) \asn 0]\vee Y(u,v)[R(u)\asn a, R(v) \asn 1]}
	Y_{a}(u)\eqdef Y_{a0}(u,v)\vee Y_{a1}(u,v)
\end{equation}
(we observe that due to monotonicity $Y_{a0}(u,v)\implies Y_{a1}(u,v)$, and hence $Y_{a}(u)\equiv Y_{a1}(u,v)$. In this paper, we are interested in counting models for the query, and hence the representation of~\eqref{eq:Ya}).
}

\eat{In what follows, we map every symbol $u$ in $\Dom=U \cup V$ to the unary atom $R(u)$. That is, every symbol $u\in U \cup V$ is associated with the Boolean RV $R(u)$.  In particular, we note that for any two disjoint blocks $B(u,v)$, $B(u',v')$ it holds that $\var(Y(u,v))\cap \var(Y(u',v'))\subseteq \set{R(u),R(v)}$.
We let $\theta: U\cup V \rightarrow \set{0,1}^n$ be an assignment to the Boolean RVs $\set{R(u) : u \in U\cup V}$. Since the tuples are independent, and are assigned a uniform probability of $\frac{1}{2}$, then $\Pr(\theta)=\frac{1}{2^n}$. }

\begin{theorem} \label{th:disjointBlocks} Let $U$ be a domain of size
  $n$, and let $\Delta = \bigcup_{u, v \in U} B(u,v)$ be a bipartite
  TID that is a disjoint union of blocks.  Assume its bipartite domain
  is $V_1 \cup V_2$ s.t. $U \subseteq V_1$, and that all tuples
  $\setof{R(u)}{u \in U}$ have probability $c$.  Then:
	%
	\begin{align}
	\label{eq:PrQpdb}
	\mathrm{\Pr_\pdb}(Q) = \sum_{\theta:U \rightarrow \set{0,1}}\Pr(\theta)\prod_{u, v \in U}Pr(Y_{\theta(u)\theta(v)}(u,v))
	\end{align}
\end{theorem}
where $\Pr(\theta)=\prod_{u \in U}c^{\theta(u)}(1-c)^{1-\theta(u)}$.
\begin{proof}
  Since $\lin{\pdb}{Q}=\bigwedge_{u, v \in  U}Y(u,v)$ and any
  two Boolean formulas $Y(u,v)$, $Y(u',v')$ share at most the boolean
  variables $R(u)$ or $R(v)$ (when $u=u'$ or $v=v'$), we have:
\begin{flalign}
&\Pr_\pdb(Q)=\Pr(\lin{\pdb}{Q})=\Pr(\bigwedge_{u, v \in  U}Y(u,v))& \nonumber\\
&{=}\sum_{\theta:U{\rightarrow} \set{0,1}^n} \Pr(\theta)\Pr(\bigwedge_{u, v \in  U}Y(u,v)[R(u){\asn}\theta(u),R(v){\asn}\theta(v)])& \nonumber \\
&{=}\sum_{\theta:U{\rightarrow} \set{0,1}^n}\Pr(\theta) \Pr(\bigwedge_{u, v \in  U}Y_{\theta(u)\theta(v)}(u,v)) &\nonumber \\
&{=}\sum_{\theta:U{\rightarrow} \set{0,1}^n}\Pr(\theta) \prod_{u, v \in  U}\Pr(Y_{\theta(u)\theta(v)}(u,v))&
\label{eq:disjointBlocks1}
\end{flalign}
where~\eqref{eq:disjointBlocks1} follows because the Boolean functions
$Y_{\theta(u)\theta(v)}(u,v)$ and $Y_{\theta(u')\theta(v')}(u',v')$
are over disjoint sets of Boolean variables, and hence they are independent.
\end{proof}
\eat{
Since we seek to prove hardness for $\fomc(Q)$, then we  set the tuple probabilities of the block disjoint TID $\pdb$ to values in $\set{\frac{1}{2},1}$. 
In particular, every unary tuple in $R(x)$ is assigned a probability of $\frac{1}{2}$, thus $\Pr(\theta)=\frac{1}{2^n}$.
\dan{do we need this discussion here?  I would simply say that we set
  all probabilities of $R(u), T(v)$ to $1/2$, without making other
  claims.  Notice also that we have not introduced the notation
  $\fomc$: I don't think we need it.}
}
\subsubsection*{Block TID associated with a graph}
\blue{ Let $G(U,E)$ be the undirected graph associated with the P2CNF $\Phi=\bigwedge_{(u,v)\in E}(X_u\vee X_v)$, where
	$n=|U|$ and $m=|E|$.  To define the TID, we consider an isomorphic copy
	of the set $U$, $U' = \setof{u'}{u \in U}$, and define the graph
	$G'(U \cup U', E \cup E')$, where
	$E' = \setof{(u,u')}{u \in U, u' \in U'}$.  That is, $E'$ contains,
	for each node $u \in U$, one edge connecting it to its isomorphic
	copy $u'$ in $U'$.}
\eat{
Let \sout{$G(U,E)$} \blue{$G(U\cup U',E \cup E')$} be a directed graph where \blue{$U'=\set{u' : u\in U}$, $U\cap U'=\emptyset$, $|U|=|U|'=n$, $E \subseteq U \times U$, and $E'=\set{(u,u'): u\in U, u'\in U'}$}. For each \blue{pair of vertices $\set{u,v}\subseteq U\cup U'$ where $u\neq v$}
at most one of $(u,v), (v,u)$ is in $E\cup E'$.}  We associate
every edge $(u,v)\in E$ \blue{and $(u,u')\in E'$} with blocks $B(u,v)$ \blue{and $B(u,u')$ respectively}, and define the
block-disjoint TID\footnote{Not to be confused with Block-Disjoint-TIDs that refer to TIDs consisting of disjoint sets of mutual exclusive tuples~\cite{DBLP:series/synthesis/2011Suciu}.} associated with $G$:
\blue{
$\pdb = \bigcup_{(u,v) \in E} B(u,v) \cup \bigcup_{(u,u') \in E'} B(u,u')$}, where for every non-edge
$(a,b) \not\in E\cup E'$ we define $B(a,b)$ to be the trivial block
consisting of all tuples $S_1(a,b), \ldots, S_p(a,b)$ with probability
$1$. In this setting, by~\eqref{eq:PrQpdb} we have:
\blue{
\begin{align}
\Pr_\pdb(Q)&=\sum_{\theta:U{\rightarrow}\set{0,1}^n}\Pr(\theta)\prod_{(u,v){\in}E\cup E'}\Pr(Y_{\theta(u)\theta(v)}(u,v))\nonumber \\
&=\sum_{\theta:U{\rightarrow}\set{0,1}^n}\Pr(\theta)\prod_{(u,v){\in}E}\Pr(Y_{\theta(u)\theta(v)}(u,v))\nonumber \\
&\hspace{0.8cm}{\times}\prod_{(u,u'){\in}E'}\left(\Pr(Y_{\theta(u)0}(u,u')){+}\Pr(Y_{\theta(u)1}(u,u'))\right) \label{eq:prTIDQ1}\\
&=\sum_{\theta:U{\rightarrow}\set{0,1}^n}\Pr(\theta)\prod_{(u,v){\in}E}\Pr(Y_{\theta(u)\theta(v)}(u,v))\prod_{u\in U}\Pr(Y_{\theta(u)}(u)) \label{eq:prTIDQ2}\\
&=\sum_{\theta:U{\rightarrow}\set{0,1}^n}\prod_{(u,v){\in}E}\Pr(Y_{\theta(u)\theta(v)}(u,v))\nonumber \\
&\hspace{0.8cm}{\times}\prod_{u\in U}(c\Pr(Y_{1}(u)))^{\theta(u)}((1{-}c)\Pr(Y_{0}(u)))^{1{-}\theta(u)} \label{eq:prTIDQ3}
\end{align}
where the transition from~\eqref{eq:prTIDQ1} to~\eqref{eq:prTIDQ2} is by the definition of $Y_a(u)$ in~\eqref{eq:Ya}. The transition from~\eqref{eq:prTIDQ2} to~\eqref{eq:prTIDQ3} follows from the fact that $\Pr(\theta)=\prod_{u\in U}c^{\theta(u)}(1-c)^{1-\theta(u)}$.
}

\subsection{The Reduction from \#P2CNF to $\fomc(Q)$}
\label{sec:reduction}
In this section, we show that using an oracle to $\fomc(Q)$ over a block disjoint TID $\pdb$ allows us to construct a system of linear equations \blue{$\B$} whose solution allows us to solve \#P2CNF. We establish three conditions on the blocks of $\pdb$, which guarantee that \blue{$\B$} is non-singular, and thus has a unique solution. In section~\ref{sec:blocks} we show how to construct the blocks such that these conditions hold.

Fix a final query $Q$ of type-I, and an instance of \#P2CNF
$\Phi=\bigwedge_{(i,j)\in E}(X_i\vee X_j)$. We let $U=\set{u_1,\dots,u_n}$, \blue{$U'=\set{u'_1,\dots,u'_n}$,
$|E|=m$, and $E'=\set{(u_i,u_i')|u_i\in U, u_i'\in U')}$.}  We create $m$ blocks $B(u,v)$, for all $(u,v) \in E$ and \blue{$n$ blocks $B(u,u')$, for all $(u,u') \in E'$}, and
define $\pdb = \bigcup_{(u,v)\in E\cup E'} B(u,v)$.  The blocks will be
isomorphic, and therefore, the following quantities do not depend on
$u,v$:
\begin{align}
y_{ab} \defeq\Pr(Y_{ab}(u,v)) && (u,v)\in E \mbox{ and } ab\in \set{0,1}^2 \label{eq:yab}\\
y_a \eqdef \Pr(Y_a(u)) && u\in U \mbox{ and } a\in \set{0,1} \label{eq:ya}
\end{align}
A block $B(u,v)$ is \e{symmetric} if $\Pr(Y_{ab}(u,v))=\Pr(Y_{ba}(u,v))$. In our reduction, we construct symmetric blocks, and thus can assume that $y_{ab}=y_{ba}$.

Consider assignment $\theta$ occurring the sum $\sum_\theta$
of~\eqref{eq:prTIDQ3}, and let
\blue{$\bk(\theta)=\set{k_{00}(\theta),\dots,k_{11}(\theta),q_0(\theta),q_1(\theta)}$} be its
signature.  The factor $y_{ab}$ occurs precisely $k_{ab}$ times in the
product, hence its exponent is $k_{ab}$, i.e. the number of 
edges $(u,v)\in E$ where $\theta(u)=a$, and $\theta(v)=b$.
\blue{
Likewise, the factor $y_{a}$ occurs precisely $q_{a}$ times in the
product, hence its exponent is $q_{a}$, i.e. the number of 
edges $(u,u')\in E'$ where $\theta(u)=a$.
}
\eat{ . The factor $y_{aa}$ will occur precisely $k_{aa}$ times (i.e.,
  with exponent $k_{aa}$) where $k_{aa}$ is the number of edges
  $(u,v)\in E$ where $\theta(u)=\theta(v)=a$. The factor $y_{10}$ will
  occur precisely $k_{10}$ times (i.e., with exponent $k_{10}$) where
  $\theta(u)=1$ (and $\theta(v)=0$) or $\theta(v)=1$ (and
  $\theta(u)=0$).}
Therefore,~\eqref{eq:prTIDQ3} becomes:
\blue{
\begin{align}
\Pr_{\pdb}(Q){=}&	\sum_{\substack{k_{00}+k_{01}+k_{10}+k_{11}=m\\q_0+q_1=n}}\#\bk\left(y_{00}\right)^{k_{00}}\left(y_{01}\right)^{k_{01}}\left(y_{10}\right)^{k_{10}}\left(y_{11}\right)^{k_{11}}\nonumber \\
&\hspace{2.2cm}\cdot \left(cy_1\right)^{q_1}\left((1{-}c)y_0\right)^{q_0}
\end{align}
}
\eat{
\begin{equation}
\label{eq:probeqReuction}
\Pr_{\pdb}(Q){=}\frac{1}{2^n}\sum_{k_{00}+k_{01}+k_{10}+k_{11}=m}\veck\left(y_{00}\right)^{k_{00}}\left(y_{01}\right)^{k_{01}}\left(y_{10}\right)^{k_{10}}\left(y_{11}\right)^{k_{11}}
\end{equation}
}
Our unknowns are $\veck$; there is one unknown for every signature $\bk$. Since the blocks are symmetric, then $y_{01}{=}y_{10}$, \blue{ then we may write:
\begin{align}
\Pr_{\pdb}(Q){=}&\sum_{\substack{k_{00}+k_{01,10}+k_{11}=m\\q_0+q_1= n}}\#\bk\left(y_{00}\right)^{k_{00}}\left(y_{10}\right)^{k_{01,10}}\left(y_{11}\right)^{k_{11}}\nonumber \\
&\hspace{2.2cm}\cdot \left(cy_1\right)^{q_1}\left((1{-}c)y_0\right)^{q_0}
\label{eq:probeqReuctionSymmetric}
\end{align}
}
\eat{
\begin{align}
\Pr_{\pdb}(Q){=}&\frac{1}{2^n}\sum_{k_{00}+k_{01}+k_{10}+k_{11}=m}\veck\left(y_{00}\right)^{k_{00}}\left(y_{10}\right)^{k_{01}+k_{10}}\left(y_{11}\right)^{k_{11}}\nonumber\\
{=}&\frac{1}{2^n}\sum_{k_{00}+k_{1}+k_{11}=m}\veck'\left(y_{00}\right)^{k_{00}}\left(y_{10}\right)^{k_{01,10}}\left(y_{11}\right)^{k_{11}}\label{eq:probeqReuctionSymmetric}
\end{align}
}where $k_{01,10}\eqdef k_{01}+k_{10}$.
Eq.~\eqref{eq:probeqReuctionSymmetric} has $(n+1)(m+1)^2$ unknowns
$\veck$ because $k_{00}+k_{11}+k_{01,10}=m$  and $q_0+q_1=n$. The coefficient associated with $\veck$ is
$\left(y_{00}\right)^{k_{00}}\left(y_{11}\right)^{k_{11}}\left(y_{10}\right)^{k_{01,10}}\left(cy_1\right)^{q_1}\left((1{-}c)y_0\right)^{q_0}$. To construct a system of linear equations that will allow us
to solve for the $(n+1)(m+1)^2$ unknowns $\veck$, we need to create
$(n+1)(m+1)^2$ linearly independent equations corresponding to the coefficients.

\blue{
To that end, we construct $\B$ to be an $(n+1)(m+1)^2\times (n+1)(m+1)^2$ matrix whose rows correspond to $(n+1)(m+1)^2$ distinct, block disjoint TIDs, and whose columns correspond to the $(n+1)(m+1)^2$ signatures $\bk=k_{00},k_{01,10},k_{11},q_0,q_1$. 
Every such block-disjoint-TID is comprised of $m+n$ disjoint blocks, one for every edge in $E \cup E'$. \eat{As previously mentioned, the blocks are isomorphic, and hence are both interchangeable and disjoint once values are assigned to $R(u)$ for all $u\in U$ (see Figure~\ref{fig:doubleBuv}).} Every block corresponding to an edge $(u,v)\in E$ is parameterized by one of $(m+1)^2$ values $p\in \set{1,\dots,(m+1)^2}$, while every block corresponding to an edge $(u,u')\in E'$ is parameterized by one of $(n+1)$ values $q\in \set{1,\dots,n+1}$ (See Section~\ref{sec:blocks}). 
Taking all possible combinations of parameters, we arrive at $(n+1)(m+1)^2$ distinct, block-disjoint-TID, one for every combination of parameters in $\set{1,\dots,(m+1)^2}\times \set{1,\dots,n+1}$. Likewise, we observe that every coefficient $y_{00}^{k_{00}}y_{11}^{k_{11}}y_{10}^{k_{01,10}}(cy_1)^{q_1}((1-c)y_0)^{q_0}$ corresponds to a combination of two signatures, one for the edges in $E$ : $\bk_1=(k_{00},k_{01,10},k_{11})$, and one for the edges in $E'$: $\bk_2=(q_0,q_1)$. 

From this construction, along with the linear equations specified in~\eqref{eq:probeqReuctionSymmetric}, we observe that $\B$ is the Kronecker product of two matrices: an $(n+1)\times (n+1)$ matrix denoted $\mathcal{N}$, and  an $(m+1)^2\times (m+1)^2$ matrix denoted $\M$. Hence, showing that $\B$ is non-singular reduces to showing that both $\mathcal{N}$ and $\M$ are non-singular.
\eat{
We do this in two steps. First, we construct an $(n+1)\times(n+1)$ linear system $\mathcal{N}$ whose rows correspond to $n+1$ distinct, block disjoint TIDs, one for each edge $(u,u')\in E'$. These TIDs induce $n+1$ distinct probabilities $y_1^{(1)},\dots,y_1^{(n+1)}$. Its columns correspond to the value of $q_1\in \set{0,\dots,n}$. 
Second, we construct an $(m+1)^2\times(m+1)^2$ linear system $\M$ whose rows correspond to $(m+1)^2$ distinct, block-disjoint TIDs that induce distinct probabilities $y_{00}^{(p)},y_{10}^{(p)},y_{11}^{(p)}$ where $p\in \set{1,\dots,(m+1)^2}$, and whose columns correspond to the values $k_{00},k_{11}\in \set{0,\dots,m+1}^2$.
Then, we create the $(n+1)(m+1)^2\times (n+1)(m+1)^2$ system $\B$ by taking the Kronecker product: $\B \eqdef \mathcal{N}\bigotimes \M$.
To show that $\B$ is invertible, we prove that both $\mathcal{N}$ and $\M$ are invertible. Hence, $\B$ has a single solution.
}Since
the size of the system $\B$ is polynomial in the size of the data, then this allows us to solve for the counts $\veck$ in PTIME, thus proving hardness. 
\begin{proposition}
	\label{prop:NInvertible}
	The $(n+1)\times (n+1)$ system of linear equations $\mathcal{N}$ is invertible.
	\end{proposition}
\begin{proof}
	Let $t\in \set{1,\dots,n+1}$ denote $n+1$ distinct, block disjoint TIDs. Each such TID induces a pair of probability values $y_0^{(t)},y_1^{(t)}$ (The construction of these TIDs is deferred to Section~\ref{sec:blocks}).
	The columns of $\mathcal{N}$ represent the values $q\in \set{0,\dots,n}$. Therefore, $[\mathcal{N}]_{tq}=\left(y_1^{(t)}c\right)^q\left((1-c)y_0^{(t)}\right)^{n-q}$. Factoring out $\left(y_0^{(t)}(1-c)\right)^n$ for each row of $\mathcal{N}$, we get $[\mathcal{N}]_{tq}=\left(\frac{c}{1-c} \cdot \frac{y_1^{(t)}}{y_0^{(t)}}\right)^q$. Hence, the quotient of each row $t$ is $\frac{c}{1-c} \cdot \frac{y_1^{(t)}}{y_0^{(t)}}$. 
	To prove the claim we need to show that $y_a^{(t)}>0$ for all $t\geq 1$ and all $a\in \set{0,1}$. Further, we need to show that for every pair $t_2 > t_1 \geq 1$ it holds that $\frac{y_1^{(t_1)}}{y_0^{(t_1)}}\neq \frac{y_1^{(t_2)}}{y_0^{(t_2)}}$. 
	We prove this claim in Lemma~\ref{lem:NVandermondeQuotient} in Section~\ref{sec:blocks} because it relies on the structural properties of the blocks. 
\end{proof}
}

To prove that $\M$ is invertible, we show that it 
meets three
conditions that characterize non-singular matrices.
\eat{ allow us to prove that the
resulting system of linear equations, denoted $\M$, is invertible.
}In the rest of this section we present the three conditions on the
probabilities $y_{ab}$, and prove that if they are met then the
resulting system $\M$ of $(m+1)^2$ linear equations is non-singular.\eat{ Since
the size of the system is polynomial in the size of the data, then
this allows us to solve for the counts $\veck'$ in PTIME, thus proving
hardness.} In fact, we prove something more general required for \eat{ that we also need
in} proving hardness of Type-II queries.

Let $h \geq 1$ be a natural number.  Let
$\lambda_1, \lambda_2 \in \R$, and $a_i,b_i$, $i=0,\ldots, h$ be real
numbers satisfying the following conditions:
\begin{align}
\lambda_1\neq \pm \lambda_2 && \mbox{ and } && \lambda_1\neq 0, \lambda_2 \neq 0 \label{eq:conditionLambda}\\
b_i \neq 0 && && \forall i=0,h \label{eq:conditiononZero}\\
a_ib_j{\neq}a_jb_i&&  && i\neq j \label{eq:conditionCoefficients}
\end{align}
Let $\mb p= \set{p_1, \dots,p_h}$ be $h$ natural numbers where
$p_i \geq 1$.  For $i=0,h$ we define:
\begin{align}
  y_{i}(\mb p)=\prod_{j=1}^h\left(a_i\lambda_1^{p_j} + b_i\lambda_2^{p_j}\right) && \forall i\in\set{0,\ldots,h}
\label{eq:yi}
\end{align}
Finally, define the following $(m+1)^h \times (m+1)^h$ matrix $\M$:

\begin{align*}
  \M_{\mb p, \mb k} \defeq & \prod_{i=0,h} y_i^{k_i}(\mb p), & \mb p \in & \set{1,2,\ldots,m+1}^h \\
\mb k \in & \set{0,1,\ldots,m}^h && k_0 + \ldots + k_h = m
\end{align*}

\begin{example}
  For a simple example, assume $h=2$, $m=2$, then:
  \begin{align*}
    y_i(p_1,p_2) = & (a_1 \lambda_1^{p_1} + b_1 \lambda_2^{p_1})(a_1 \lambda_1^{p_2} + b_1 \lambda_2^{p_2})
\ \ \ i =  0,1,2 \\
\mb \M =& \left[
    \begin{array}{cccc}
      y_0^2(1,1)&y_0(1,1)y_1(1,1)&\ldots&y_2^2(1,1)\\
      y_0^2(1,2)&y_0(1,2)y_1(1,2)&\ldots&y_2^2(1,2)\\
               & \ldots & & \\
      y_0^2(3,3)&y_0(3,3)y_1(3,3)&\ldots&y_2^2(3,3)
    \end{array}
\right]
  \end{align*}
  Each row of $\mb \M$ has all products $y_0^{k_0}y_1^{k_1}y_2^{k_2}$
  where $k_0+k_1+k_2=2$.
\end{example}

\begin{theorem} \label{th:non-singular}
  The matrix $\M$ is non-singular
\end{theorem}

Before we prove the theorem, we show how to use it to prove
Theorem~\ref{th:hard:type:1}.  We set $h=2$.  Given the P2CNF formula
defined by the graph $(U,E)$, we will construct a TID obtained as a
disjoint union of blocks $\pdb = \bigcup_{(u,v) \in E}B(u,v)$.  The
probability of $Q$ is given by Eq.~\eqref{eq:probeqReuctionSymmetric}.
Fix two numbers $p_1, p_2 \geq 1$.  We will describe in the next
section how to construct a symmetric block $B(u,v)$ such that its
probabilites are given by expresions similar to~\eqref{eq:yi}, more
precisely:
\begin{align*}
y_{00} = & \prod_{j=1,2}(a_{00} \lambda_1^{p_j} + b_{00}  \lambda_2^{p_j}) \\
y_{10} = & \prod_{j=1,2}(a_{10} \lambda_1^{p_j} + b_{10}  \lambda_2^{p_j}) \\
y_{11} = & \prod_{j=1,2}(a_{11} \lambda_1^{p_j} + b_{11}  \lambda_2^{p_j})
\end{align*}
where the coefficients satisfy conditions
~\eqref{eq:conditionLambda}-\eqref{eq:conditionCoefficients}.  We use
repeatedly the oracle for $\pr(\pdb)$, once for each $(m+1)^2$
combination of values $p_1, p_2 \in \set{1,2,\ldots,m+1}$, and obtain
a system of linear equations with unknowns $\#\mb k'$.  By
Theorem~\ref{th:non-singular} this system has a unique solution, which
can be computed in polynomial time using Gaussian elimination.  This
gives us all the undirected counts $\#\mb k'$, from which we extract
$\#\Phi = \sum_{\#\mb k': k_{00}=0} \#\mb k'$.

In the rest of this section we prove Theorem~\ref{th:non-singular}, by
proving a series of lemmas in calculus, of possible independent
interest.

\begin{lemma}\label{lem:independentPolynomials}
  For each $\mb k = (k_1,\dots,k_h){\in} \set{0,\dots,m}^h$, define
  the following polynomial in variables $\mb y = (y_1, \ldots, y_h)$:
  $g_{\mb k}(\mb y)\eqdef y_1^{k_1}\cdots y_h^{k_h}$.  Then, the
  polynomials $g_{\mb k}$, $\mb k \in \set{0,\dots,m}^h$ are linearly
  independent.
\end{lemma}

\begin{proof}
  Assume the contrary, that there exist coefficients
  $a_{k_1, \ldots, k_h}$ such that, denoting
  $f(\mb y) \defeq \sum_{\mb k} a_{\mb k} \prod_i y_i^{k_i}$, the
  polynomial $f$ is identically 0, $f\equiv 0$.  For each $i=1,h$, let
  $S_i \defeq \set{v_{i,0}, \ldots, v_{i,m}}$ be a set of $m+1$
  distinct values, and denote by $A_i$ the Vandermonde matrix defined
  by the set $S_i$: $(A_i)_{k\ell} = v_{i,\ell}^k$, for
  $0 \leq k,\ell \leq m$.  By our assumption that $f\equiv 0$, then for any combination of values $(u_1,\dots,u_h)\in S_1\times \dots \times S_h$ we have: \eat{For any values $u_i \in S_i$, the
  polynomial $f$ is 0, which implies:}
  \begin{align*}
     \sum_{\mb k} a_{\mb k} \prod_i u_i^{k_i} = & 0
  \end{align*}
  By using all $(m+1)^h$ combinations of values $u_1, \ldots, u_h$, we
  obtain a linear system of $(m+1)^h$ unknowns $a_{k_1, \ldots, k_h}$,
  whose matrix is the Kronecker product
  $A \defeq A_1 \otimes A_2 \otimes \cdots \otimes A_h$.  Since
  $\det(A_i)\neq 0$ for all $i$, it follows that $\det(A)\neq 0$.
  This implies that the system has a single solution,
  $a_{k_1, \ldots, k_h}=0$ for all $k_1, \ldots, k_h$, proving that
  the polynomials are linearly independent.
\end{proof}

\begin{lemma}
\label{lem:polynomialNonZero}
Let $f(z_1,\dots,z_h)$ be a polynomial of degree $\leq m$ in every
$z_i$, such that $f\not\equiv 0$.  For all $i=1,h$, let
$A_i\subseteq \R$ be a set of $m+1$ distinct, real values.  Then there
exists values $u_i \in A_i$, for $i=1,h$, such that $f(u_1,\dots,u_h)\neq 0$.
\end{lemma}
\begin{proof}
  By induction on $h$. When $h=1$, then $f$ has at most $m$ roots,
  hence there is at least one value $u_1\in A_1$ such that
  $f(u) \neq 0$. When $h > 1$, consider any value $a\in A_h$, and set
  $z_h=a$.  If $f[z_h := a]\not\equiv 0$, then by induction on $h$ we
  get values $u_1 \in A_1,\dots, u_{h-1} \in A_{h-1}$ for the other
  $h-1$ variables such that
  $f[z_1:= u_1, \dots, z_{h-1}:= u_{h-1}, z_h:= a]\neq 0$.  If
  $f[z_h:= a]\equiv 0$ for all $m+1$ values $a \in A_h$, then $f$ can
  be divided by the polynomial $\prod_{a \in A_h}(z_h-a)$, which has
  degree $m+1$, contradiction.
\end{proof}

\begin{lemma}
\label{lem:fLInearlyIndependent}
Let $g_1, g_2, \ldots, g_N$ be linearly independent multivariate
polynomials, in $h$ variables $y_1, \ldots, y_h$.  Let
$H:A (\subseteq \R^h)\rightarrow \R^h$ be a differentiable function
such that its Jacobian $\J(H)\neq 0$ at some interior point in $A$,
and define $f_i(z_1,\dots,z_h)\eqdef g_i(H(z_1,\dots,z_h))$ for every
$i$. Then the functions $f_1,\dots,f_N$ are linearly independent.
\end{lemma}

\begin{proof}
  Assume, by contradiction, that there exist constants $a_1,\dots,a_N$
  not all of which are $0$ such that:
  $F =a_1f_1 + a_2f_2 + \dots + a_Nf_N \equiv 0$. We prove that
  $G =a_1g_1 + a_2g_2 + \dots + a_Ng_N \equiv 0$, which is a
  contradiction.  Let $\mb v = (v_1, \ldots, v_h) \in A$ be a point
  where $\J(H) \neq 0$.  By the inverse function theorem, the function
  $H$ is invertible in some neighborhood of $H(\mb v)$.  Assume
  w.l.o.g. that this neighborhood is a product of open intervals,
  $\prod_i (\alpha_i, \beta_i)$, where $\alpha_i < \beta_i$, and let
  $A_i \subseteq (\alpha_i,\beta_i)$ be any finite set with $m+1$
  distinct values, for each $i=1,h$.  Since $H$ is invertible on
  $A_1 \times \cdots \times A_h$, for any combination of values
  $\mb u \in A_1 \times \cdots \times A_h$ there exists
  $\mb w \in \R^h$ such that $H(\mb w) = \mb u$.  By assumption,
  $F(\mb w)=0$, and this implies
  $G(\mb u) = G(H(\mb w)) = F(\mb w) =0$.
  Lemma~\ref{lem:polynomialNonZero} implies that $G \equiv 0$, which
  is a contradiction.
\eat{
Since $\J(H) \not\equiv 0$, there exists a point $(u_1,\dots,u_h)$ where the transformation $H$ is invertible in the neighborhood of $(u_1,\dots,u_h)$. Consider the image of the neighborhood of $(u_1,\dots,u_h)$ under $H$. 
That image contains a hypercube $[a_1,b_1]\times [a_2,b_2]\times  \cdots \times [a_h,b_h]$.  Choose $m+1$ distinct values in each of the $h$ intervals.
For every $h$-tuple $(y_1,\dots,y_h) \in [a_1,b_1]\times [a_2,b_2]\times  \cdots \times [a_h,b_h]$ we have that $G(y_1,\dots,y_h)=F(H(z_1,\dots,z_h))=0$, since the latter is, by assumption, identically $0$. Since the interval contains at least $h(m+1)$ values (where  $h\geq 1$), then by Lemma~\ref{lem:polynomialNonZero}, $G$ is identically $0$, and we arrive at a contradiction.}
\end{proof}

\begin{lemma}
\label{lem:HInvertible}
Let $c_1, \ldots, c_h$ be distinct real values, and let $H : \R^h
\rightarrow \R^h$ be the following function: $H(\mb z) = \mb y$,
where:
\begin{align}
y_i \eqdef & \prod_{j=1,h}(c_i + z_j) \label{eq:transformationH}
\end{align} 
Let $\mb u = (u_1, \ldots, u_h)$ be any point with distinct
coordinates, i.e. $u_i \neq u_j$, such that $u_i + c_j \neq 0$ for all
$i,j$.  Then the Jacobian of $H$ at $\mb u$ is nonzero,
$\J(H)(\mb u) \neq 0$.
\end{lemma}

\begin{proof}
Recall the definition of the Jacobian:
$$\J(H) \eqdef 
\begin{bmatrix}
\frac{\partial y_{1}}{\partial z_1} & \cdots &  \frac{\partial y_{1}}{\partial z_h}
\\
\ddots& \ddots & \ddots
\\
\frac{\partial y_{h}}{\partial z_1} & \cdots &\frac{\partial y_{h}}{\partial z_h}
\end{bmatrix}
$$
Since $y_i=\prod_{j=1}^h(c_i+z_j)$ then $\frac{\partial y_i}{\partial z_k}=\prod_{j\neq k}(c_i+z_j)$. For each row $i\in [1,h]$, we factor out the product $\prod_{j=1}^h(c_i+z_j)$. This results in the following matrix:
$\J' \eqdef 
\begin{bmatrix}
\frac{1}{c_1+z_1} & \cdots &  \frac{1}{c_1+z_h}
\\
\ddots& \ddots & \ddots
\\
\frac{1}{c_h+z_1} & \cdots &\frac{1}{c_h+z_h}
\end{bmatrix}
$. 
We note that $\det(\J(H))=\det(\J')\cdot \prod_{i=1}^h \prod_{j=1}^h(c_i+z_j)$. Therefore, if  $\det(\J')\neq 0$ then $\det(\J)\neq 0 $.
The expression for $\det(\J')$ has a closed form~\cite{Krattenthaler}:
\begin{equation}\label{eq:CauchyDblAlternat}
\det(\J')=\frac{\prod_{1\leq i<j\leq h}(c_i-c_j)(z_i-z_j)}{\prod_{1\leq i<j\leq h}(c_i+z_j)}
\end{equation}
When $z_j= u_j$ for all $j$, then this value is $\neq 0$ because the
$c_i$'s are distinct, and the $u_i$'s are distinct.
\end{proof}

For the next two statements we fix $c_1, \ldots, c_h$ to distinct real
values and, for each $k_1,\dots,k_h{\in} \set{0,\dots,m}$, we define the
following polynomial in variables $\mb z = (z_1, \ldots, z_h)$:
\begin{align}
\label{eq:fPolynomials}
f_{k_1,\dots,k_h}(\mb z)\eqdef \prod_{i=1}^h\prod_{j=1}^h(c_i+z_j)^{k_i}
\end{align}

\begin{corollary}
\label{corr:LinearlyIndependent}
The polynomials $f_{k_1,\dots,k_h}$ of~\eqref{eq:fPolynomials} are
linearly independent.
\end{corollary}
\begin{proof}
  By Lemma~\ref{lem:independentPolynomials}, the polynomials
  $g_{k_1,\dots,k_h}(\mb y) \defeq \prod_{i=1}^hy_i^{k_{i}}$ are
  linearly independent.  By Lemma~\ref{lem:HInvertible}, the
  transformation $\mb z \mapsto \mb y$ given by
  ~\eqref{eq:transformationH} has a non-zero Jacobian (at some point).
  Then, by Lemma~\ref{lem:fLInearlyIndependent} it follows that the
  polynomials $f_{k_1,\dots,k_h}(\mb z)$ in~\eqref{eq:fPolynomials}
  are also linearly independent.
\end{proof}

\begin{lemma} \label{lemma:another:nonsingular} For each $i=1,h$, let
  $A_i\subseteq \R$ be a set of $m+1$ distinct, real values.  Consider
  the following $(m+1)^h \times (m+1)^h$ matrix, whose rows are
  indexed by $\mb u \in A_1, \ldots, A_h$, and whose columns are
  indexed by $\mb k \in \set{0,1,2,\cdots,m}^h$:
  \begin{align*}
    M_{\mb u, \mb k} \defeq \prod_{i=1}^h\prod_{j=1}^h(c_i+u_j)^{k_i}
  \end{align*}
  Then $M$ is non-singular.
\end{lemma}

\begin{proof}
  We notice that each row of the matrix $M$ consists of the
  polynomials $f_{k_1,\dots,k_h}(\mb z)$ applied to some point in
  $\mb u$.  We construct a matrix $M'$ such that $\det(M')\neq 0$ and
  $M'$ differs from $M$ only by permutations of rows and columns.  We
  construct, by induction on $k$, a $k \times (m+1)^h$ matrix $M_k$ such
  that each row consists of the values of the polynomials
  $f_{k_1,\dots,k_h}(\mb z)$ applied to some point in
  $\mb u \in A_1 \times \cdots \times A_h$, and such that the
  $k\times k$ minor consisting of the first $k$ columns in $M_k$ is
  non-singular.  When $k=1$ we choose any
  $\mb u \in A_1 \times \cdots \times A_h$, and the statement holds
  because $M_{\mb u, (0,0,\ldots,0)} = 1$ (i.e.  all entries of the
  column $k_1=\cdots =k_h=0$ are 1).  We show now how to construct
  $M_{k+1}$ by extending $M_k$.  First, extend $M_k$ with a row
  consisting of the polynomials $f_{k_1,\dots,k_h}(\mb z)$.  Let $D$
  be the $(k+1)\times (k+1)$ minor $M_{k+1}$ defined by the first
  $k+1$ columns.  $D$ is a linear combination of these polynomials,
  where the coefficient of each polynomial is the $k\times k$ minor
  consisting of the first $k+1$ columns except that containing the
  polynomial.  By induction, the coefficient given by the minor
  consisting of the first $k$ columns is non-zero.  By
  Corollary~\ref{corr:LinearlyIndependent}, $D(\mb z)$ (viewed as a
  polynomial in $\mb z$) is non-zero, hence by
  Lemma~\ref{lem:polynomialNonZero} there exists
  $\mb u \in A_1 \times \cdots \times A_h$ such that
  $D(\mb u) \neq 0$, proving the claim.  Thus, we obtain
  $M' \defeq M_{(m+1)^h}$, a matrix that is non-singular.  Since the
  matrix is non-singular, no two rows in $M'$ can use the same value
  $\mb u \in A_1 \times \cdots \times A_h$, and since both the number
  of rows in $M'$ and the cardinality of the set
  $A_1 \times \cdots \times A_h$ are the same, $(m+1)^h$, it follows
  that $M'$ contains precisely the same rows as $M$, up to a
  permutation, proving the lemma.
\end{proof}

Finally, we can now prove Theorem~\ref{th:non-singular}.  For that we
use the fact that $k_0 = m - (k_1+\cdots+k_h)$ and write:
\begin{align*}
  M_{\mb p, \mb k} = & y_0^m(\mb p) \prod_{i=1,h} \left(\frac{y_i(\mb p)}{y_0(\mb p)}\right)^{k_i}\\
\end{align*}
Since every row $\mb p$ in $M$ has the same factor $y_0^m(\mb p)$, it
suffices to prove that the matrix $M'$ without this factor is
non-singular:
\begin{align*}
  M'_{\mb p, \mb k} = &\prod_{i=1,h} \left(\frac{y_i(\mb p)}{y_0(\mb p)}\right)^{k_i}=  \prod_{i=1,h} \prod_{j=1,h}\left(\frac{a_i\lambda_1^{p_j} + b_i\lambda_2^{p_j}}{a_0\lambda_1^{p_j} + b_0\lambda_2^{p_j}}\right)^{k_i}\\
= & \prod_{i=1,h}\prod_{j=1,h} \left(\frac{a_i+ b_i\left(\frac{\lambda_2}{\lambda_1}\right)^{p_j} }{a_0+ b_0\left(\frac{\lambda_2}{\lambda_1}\right)^{p_j} }\right)^{k_i}\\
= & \prod_{i=1,h}\prod_{j=1,h} \left(\frac{b_i}{b_0} + \frac{a_i - \frac{a_0b_i}{b_0}}{a_0 + b_0\left(\frac{\lambda_2}{\lambda_1}\right)^{p_j}}\right)^{k_j}\\
\defeq & \prod_{i=1,h}\prod_{j=1,h} \left(c_i + \frac{d_i}{a_0 + b_0\left(\frac{\lambda_2}{\lambda_1}\right)^{p_j}}\right)^{k_j}
\end{align*}
We now use Lemma~\ref{lemma:another:nonsingular}.  Since
$\lambda_2/\lambda_1 \not\in \set{-1, 0, +1}$, the function
$p_j \mapsto z_j \defeq \frac{d_i}{a_0 +
  b_0\left(\frac{\lambda_2}{\lambda_1}\right)^{p_j}}$ is injective,
therefore the $m+1$ distinct values $p_j= 1,2,3, \ldots, m+1$ will
yield $m+1$ distinct values of $z_j$.  By
Lemma~\ref{lemma:another:nonsingular}, the matrix $M'$ is
non-singular, completing the proof of Theorem~\ref{th:non-singular}.

\subsection{Designing the blocks $B(u,v)$}
\label{sec:blocks}
In this section, we design the block TID $B(u,v)$ such that the probability of the lineage of $Q$ over $B(u,v)$ can be expressed as in~\eqref{eq:yi}, and that it meets the conditions of~\eqref{eq:conditionLambda}-\eqref{eq:conditionCoefficients}. 

We design our blocks to prove hardness for final type-I queries, where every endpoint (i.e., $u$ and $v$) is assigned one of two values $\set{0,1}$. Since we design symmetric blocks, then $h=\set{00,10,11}$.
In our construction, each block is parameterized by a pair $\mb p=\set{p_1,p_2}$. The main focus of this section is a block that is parameterized by a single value $p$, denoted $B_p(u,v)$. We show how two blocks $B_{p_1}(u,v),B_{p_2}(u,v)$ can be combined in \e{parallel} to create a block $B_{\mb p}(u,v)$ that has the desired form and properties.

\subsubsection*{The Block $B_p(u,v)$.}
The bipartite domain of $B_\p(u,v)$ is $V_1(\p)\cup V_2(\p)$, and is defined as follows:
\begin{align}
V_1(\p)&=\set{u,v} \cup \set{r_k : k\in [1,\p-1]} \label{eq:V1} \\
V_2(\p)&=\set{t_k: k\in [1,\p]} \label{eq:V2}
\end{align}
where the $r_k$ and $t_k$ are fresh constants.
The tuple probabilities of $B_\p(u,v)$ are:
\begin{align*}
	\Pr(R(u))=\begin{cases}
		c & u\in V_1(p) \\
		1 & \mbox{ otherwise }
	\end{cases} 
	&&
	\Pr(T(v))=\begin{cases}
		c & v\in V_2(p) \\
		1 & \mbox{ otherwise }
	\end{cases}
\end{align*}
where $c\in (0,1)$ is a constant.
For every binary symbol $S \in \calR$ the probabilities are:
\begin{itemize}
	\item $\Pr(S(u,t_1))=c$.
	\item $\Pr(S(v,t_\p))=c$.
	\item $\Pr(S(r_k,t_k))=c$ for all $k\in [1,p-1]$.
	\item $\Pr(S(r_k,t_{k+1}))=c$ for all $k\in [1,p-1]$.
	\item Otherwise, $\Pr(S(a,b))=1$.
\end{itemize}
\begin{example} \label{ex:zig:zag}
	The graphical representation of $B_{p}(u,v)$ is:
	\begin{align*}
		u\defeq r_0 - t_1 - r_1 - \cdots - r_{p-1} - t_p - r_p\defeq v
	\end{align*}
	The binary atoms corresponding to the edges of the path have a probability of $c$, and atoms corresponding to non-edges (e.g., $S(t_1,t_3)$) have
	a probability of $1$. Also note that $r_0\eqdef u$, and $r_p\eqdef v$.
\end{example}

Let $p\geq 1$ be a natural number. The lineage of $Q$ over the block TID $B_p(u,v)$ is denoted $Y^{(p)}(u,v)$, and its arithmetization $y(p)$. That is, $y(p)$ is the multilinear polynomial representing the probability $\Pr(Y^{(p)}(u,v))$.
We associate the endpoints $u$ and $v$ with the unary atoms $R(u)$  and $R(v)$ respectively. Thus $Y_{ab}^{(p)}(u,v)\eqdef Y^{(p)}(u,v)[R(u) \asn a,R(v) \asn b]$, and $y_{ab}(p)\eqdef y(p)[R(u)\asn a, R(v)\asn b]$ is defined accordingly. 
\blue{
The following proposition follows directly from the way we construct the blocks.
\begin{proposition}
	\label{prop:symmetry}
	Let $Y^{(p)}(u,v)$ denote the lineage of $Q$ over a $p$-length block-TID $B_p(u,v)$. Then $Y^{(p)}(u,v)=Y^{(p)}(v,u)$.	
\end{proposition}
Since $Y^{(p)}(u,v)=Y^{(p)}(v,u)$, then $Y_{ab}^{(p)}(u,v)=Y_{ab}(v,u)=Y_{ba}(u,v)$. Accordingly, we have that $y_{ab}(p)=y_{ba}(p)$.
}

\eat{
Let $N$ be the number of Boolean variables in $Y_{ab}^{(p)}$.
Given an assignment $\bu:\var(y_{ab}^{(\p)})\rightarrow [0,1]^N$ we denote by $y^{(\p)}_{ab}(\bu)$ the valuation of the multi-linear polynomial $y^{(\p)}_{ab}$ at point $\bu$. In particular, we observe that  $y^{(\p)}_{ab}(\bu)=\Pr(Y_{ab}^{(\p)})$, the probability that $Y_{ab}^{(\p)}$ is satisfied  when the probabilities of its Boolean RVs are assigned according to $\bu$.
}

We define $z_{ab}(p)$ to be the valuation of $y_{ab}(p)$ when all of the RVs in $Y_{ab}^{(p)}(u,v)$ have a probability of $c$. For any natural $p \geq 1$:
\begin{equation}
\label{eq:zab}
z_{ab}(p)\eqdef y_{ab}(p)[c,\dots,c]
\end{equation}
In the rest of this section we design the block $B_p(u,v)$ and prove that it meets the following conditions. 
\begin{theorem}
\label{thm:designSection}
Let $B_{\p}(u,v)$ be a block TID, parameterized by a parameter $p\geq 1$.
For every $i\in \set{00,10,11}$ it holds that:
\begin{align}
z_{i}(p)=(a_i\lambda_1^{p} + b_i\lambda_2^{p})&\hfill&\forall i\in \set{00,10,11} \label{eq:yiDesign}
\end{align}
where $\lambda_1$, $\lambda_2$ and $a_i,b_i$ are constants, independent of $p$, satisfying:
\begin{align}
\lambda_1\neq \pm \lambda_2 && \mbox{ and } && \lambda_1\neq 0, \lambda_2 \neq 0 \label{eq:designyi}\\
b_i \neq 0 && && \forall i\in \set{00,10,11} \label{eq:nonZerocondition} \\ 
a_ib_j{\neq}a_jb_i&&  && i \neq j \label{eq:designcondition}
\end{align}
\end{theorem}
Let $\mb p=\set{p_1,p_2}$. We now describe how to create the block  $B_{\mb p}(u,v)$ that has probability $y_i(\mb p)$ (as in~\eqref{eq:yi}) for every $i\in \set{00,10,11}$.
To do so, we construct two disjoint blocks between nodes $u$ and $v$ with parameters $p_1,p_2$ respectively (see Figure~\ref{fig:doubleBuv}).
The lineage of this block is $Y_{i}^{(p_1)}(u,v)\wedge Y_{i}^{(p_2)}(u,v)$.
Since the blocks are disjoint, then $\var(Y^{(p_1)}(u,v))\cap \var(Y^{(p_2)}(u,v)) =\set{R(u),R(v)}$. In particular, the Boolean functions $Y_i^{(p_1)}(u,v)$ and $Y_{i}^{(p_{2})}(u,v)$ are disjoint for every $i\in \set{00,10,11}$. So, we get that:
\begin{align}
\label{eq:overlineProbab}
y_{ab}(\mb p)=\Pr(\bigwedge_{j=1}^2Y_{ab}^{(p_j)}(u,v))=\prod_{j=1}^2\Pr(Y_{ab}^{(p_j)}(u,v))=y_{ab}(p_1)y_{ab}(p_2)
\end{align}
Consequently, once we prove Theorem~\ref{thm:designSection}, the conditions of \eqref{eq:conditionLambda}-\eqref{eq:conditionCoefficients} follow from~\eqref{eq:overlineProbab}.
\eat{ 
	For a natural number $p\geq 1$, we show that $y_{i}^{(p)}=(a_i\lambda_1^{p-1}+b_i\lambda_2^{p-1})$ where $i\in \set{00,10,11}$, thus proving~\eqref{eq:yi}. We also show that conditions~\eqref{eq:conditionCoefficients} and~\eqref{eq:conditionLambda} hold, thus proving the claim. 
}

To prove Theorem~\ref{thm:designSection} we first show that
$Y^{(p)}(u,v)$ is a connected Boolean function. Consequently, we can
apply Lemma~\ref{lemma:determinant:connected} to argue that the
determinant of the ``small matrix'' associated with $y(p)$
(see~\eqref{eq:smallMatrix}) is not identically zero, and, by
Lemma~\ref{lemma:three:values}, there are  probability values in
$\set{0,c,1}$ for which the small matrix is non-singular.
In fact, we prove something stronger: if $Q$ is final, then the matrix
remains nonsingular even if we set all probabilities to $c$.  Some
further analysis of this matrix, along with two simple properties of
the lineage $Y^{(p)}(u,v)$ allow us to prove
Theorem~\ref{thm:designSection}.

We recall that $Q(s,t)$ is the Boolean formula that results from mapping $x \mapsto s$ and $y\mapsto t$. We note that $Q(s,t)$ is identically $\true$ whenever $(s,t){\notin} \set{(u,t_1),(v,t_p),(r_k,t_k),(r_k,t_{k+1})\mid k{\in}[1,p{-}1]}$.
Therefore:
\begin{equation}
Y^{(p)}(u,v)=\bigwedge_{i=1}^{p}Q(r_{i-1},t_i)\wedge Q(r_i,t_i)
\end{equation}
where $r_0{=}u$ and $r_p{=}v$. In particular, when $p=1$ then $Y^{(1)}(u,v)=Q(u,t_1)\wedge Q(v,t_1)$.

\eat{
We note that for every pair $(u,v) \in V_1 {\times} V_2$ exactly one of the following holds: Either all binary atoms $S(u,v)$ have a probability of $1$,
or they all have a probability $\frac{1}{2}$.
We provide a graphical interpretation to $\pdb_{E,\p}$, denoted $G_{\pdb_{E,\p}}$, that will make the exposition clearer. 
The values $V_1\cup V_2$ are the nodes of $G_{\pdb_{E,\p}}$ (i.e., $\nodes(G_{\pdb_{E,\p}})=V_1\cup V_2$) and $\edges(G_{\pdb_{E,\p}})=\set{(u,v): \forall S\in \calR, \Pr(S(u,v))=\frac{1}{2}}$. This means that every pair of nodes $(u,v)\notin \edges(G_{\pdb_{E,\p}})$ indicates that $\Pr(S(u,v))=1$  for all binary symbols in $\calR$. 
}

\begin{figure}
\centering
\includegraphics[width=0.4\textwidth]{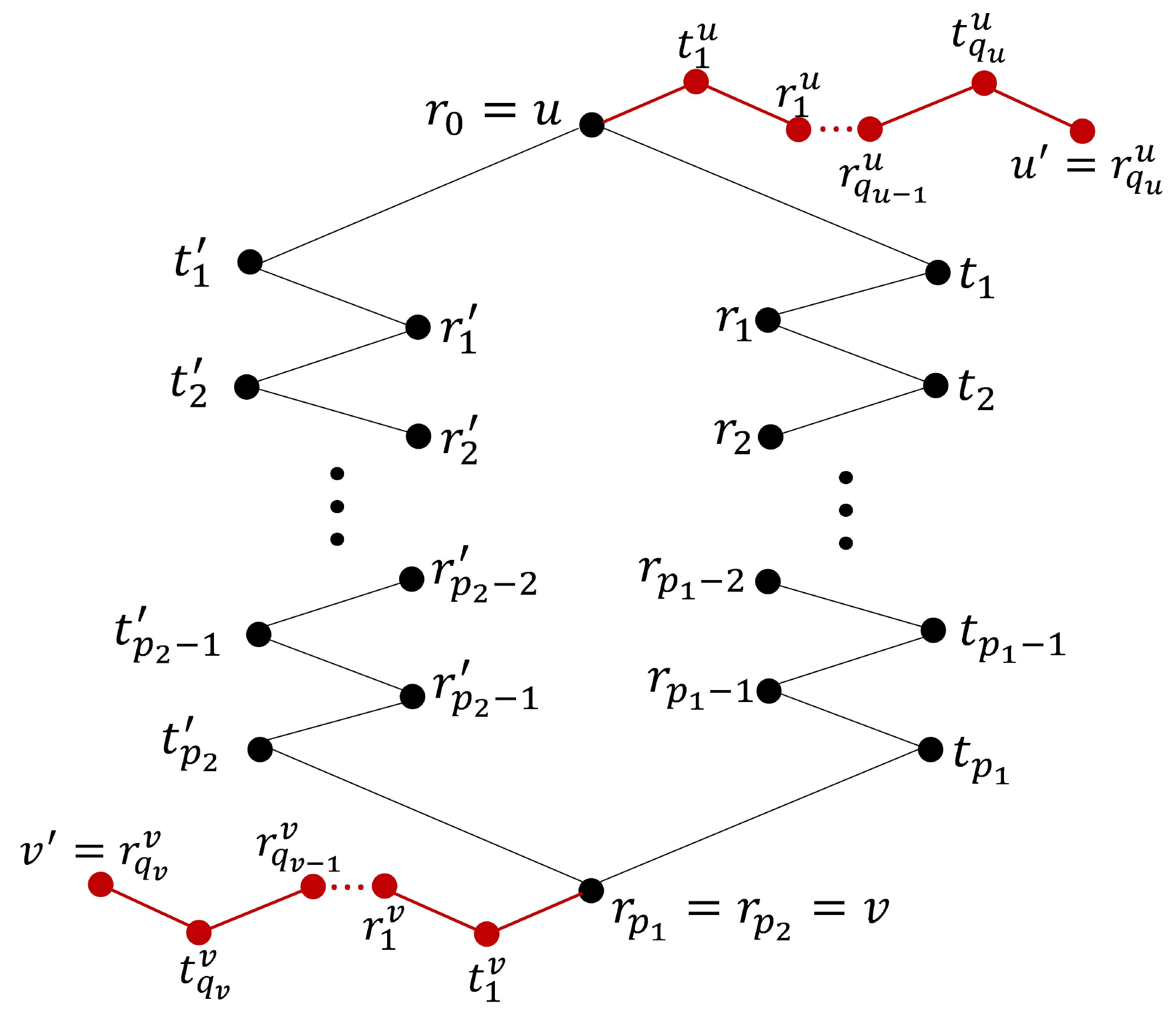}
\caption{Block $B_{\bp}(u,v)$ where $\bp=\set{p_1,p_2}$\blue{, and blocks $B_{q_u}(u,u')$ and $B_{q_v}(v,v')$}.}
\label{fig:doubleBuv}		
\end{figure}

\eat{
\begin{figure*}
	\begin{subfigure}[t]{0.3\textwidth}
		\centering
		\includegraphics[width=0.7\textwidth]{BlockBuv.pdf}
		\caption{Block $B_p(u,v)$.}
		\label{fig:Buv}
	\end{subfigure}		
	\hfill
	\begin{subfigure}[t]{0.4\textwidth}
		\centering
		\includegraphics[width=0.8\textwidth]{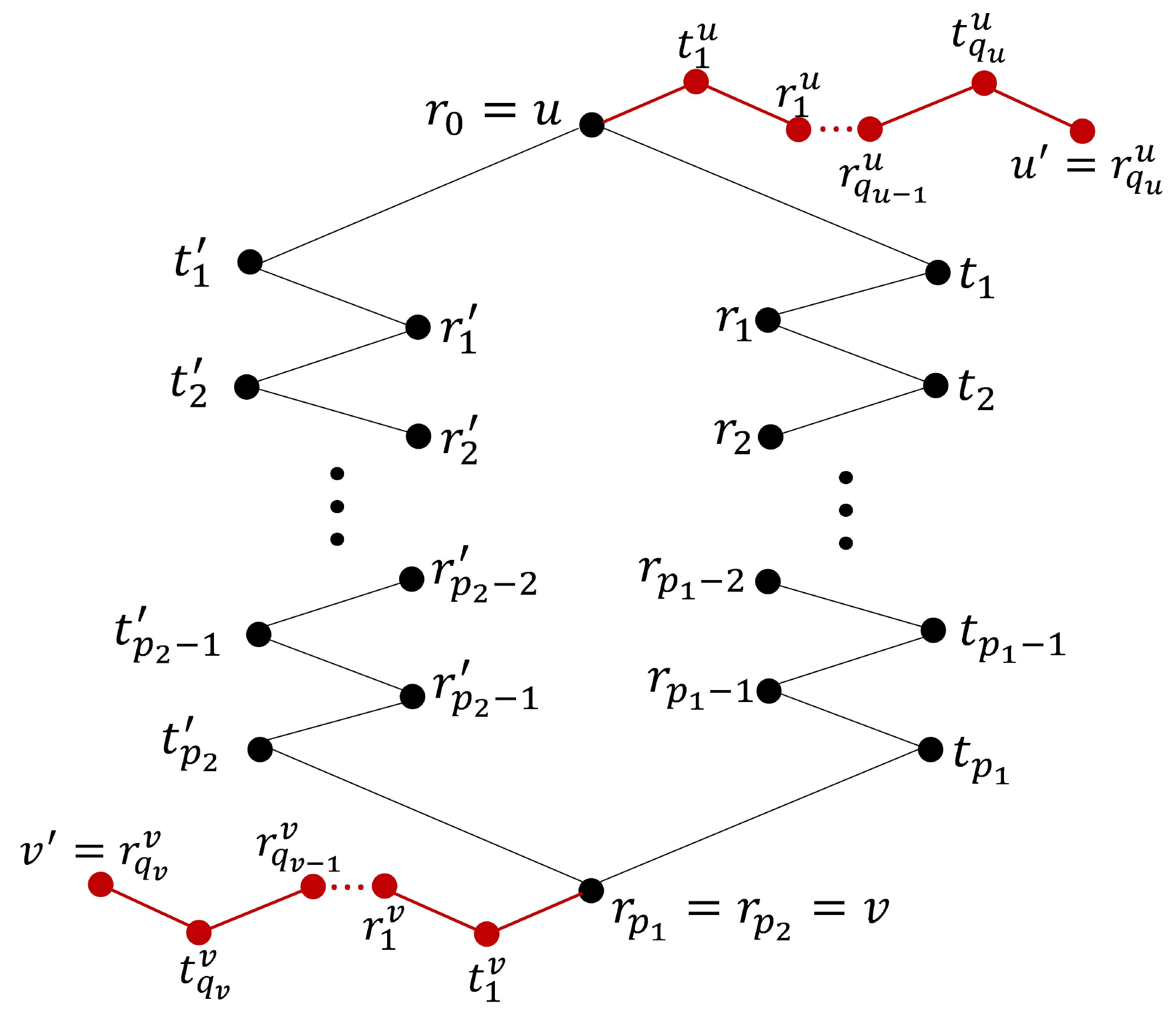}
		\caption{Block $B_{\bp}(u,v)$ where $\bp=\set{p_1,p_2}$. }
		\label{fig:doubleBuv}		
	\end{subfigure}
\caption{Disjoint block TID.}
\end{figure*}
}
We say that a Boolean function $f$ is \e{disconnected} if $f=f_1 \wedge f_2$ where $f_1$ and $f_2$ are non-constant, disjoint Boolean functions (i.e., $\var(f_1){\cap}\var(f_2){=}\emptyset$). Otherwise, we say that $f$ is \e{connected}. 
\def\lemLineageConnected{
	If $Q$ is an unsafe type-I query, then $Y^{(p)}(u,v)$ is connected.
}
\begin{lemma}\label{lem:lemLineageConnected}
		\lemLineageConnected
\end{lemma}
\begin{proof}
By induction on $\p$.
Since $Q$ is unsafe, it is connected. Therefore, it must hold that $h_1\eqdef Q(u,t_1)$ and $h_2\eqdef Q(v,t_1)$ are connected.
Observe that when $\p=1$ then $Y^{(1)}(u,v)=h_1\wedge h_2$. Since $T(t_1)\in \var(h_1){\cap \var(h_2)}$ then $h_1$ and $h_2$ are not disjoint.
Therefore, $Y^{(1)}(u,v)$ is disconnected only if there is some homomorphism $C(u,t_1)\rightarrow C_R(u,t_1)$ (or $C(v,t_1)\rightarrow C_R(v,t_1)$) where $C$ is a clause in $Q_{\texttt{left}}{\wedge}Q_{\texttt{middle}}$ and $C_R$ is a right clause. We note that $C \notin Q_{\texttt{left}}$ because all left clauses, and only left clauses, contain the unary symbol $R$. If $C \in Q_{\texttt{middle}}$, it means that $C_R$ is a redundant clause, which is a contradiction (we assume that all clauses of $Q$ are non-redundant).
Symmetrically, for any $k\in[1,p{-}1]$ it holds that $Q(r_k,t_k){\wedge}Q(r_{k},t_{k+1})$ are connected via the unary atom $R(r_{k})$, and thus connected.

Let $\p{>}1$. 
Observe that
$Y^{(p)}(u,v){=}Y^{(p{-}1)}(u,r_{p-1})\wedge Y^{(1)}(r_{\p-1},v)
$
 where  $Y^{(1)}(r_{\p-1},v)=Q(r_{\p-1},t_{\p}){\wedge} Q(v,t_{\p})$. By the induction hypothesis, $Y^{(p{-}1)}(u,r_{p-1})$ and $Y^{(1)}(r_{\p-1},v)$ are both connected. In addition, $Y^{(p{-}1)}(u,r_{p-1})$ and $Y^{(1)}(r_{\p-1},v)$ are  connected via the unary left atom $R(r_{\p-1})$ that appears in both Boolean formulas. By the previous reasoning none of the clauses containing $R(r_{\p-1})$ become redundant and hence, $Y^p(u,v)$ is connected. \eat{ a connected Boolean formula.}
\end{proof}
Note, in particular, that Lemma~\ref{lem:lemLineageConnected} holds for final queries.

Let $N$ be the number of RVs in $y_{ab}(1)$. We define the matrix of polynomials:
\begin{equation}\label{eq:A1}
  A^{(1)}\eqdef \begin{bmatrix} y_{00}(1) & y_{01}(1) \\
    y_{10}(1)  & y_{11}(1)\end{bmatrix}
\end{equation}
and the polynomial $f_A{:}[0,1]^N{\rightarrow}[0,1]$ associated with
the determinant of $A^{(1)}$:
\begin{equation}
f_A(u_1,\dots,u_N)\eqdef \det(A^{(1)})=y_{00}(1)y_{11}(1){-}y_{01}(1)y_{10}(1)
\label{eq:fA}
\end{equation}
and observe that $f_A$ is a degree-2 polynomial.  The previous lemma,
and Lemma~\ref{lemma:determinant:connected} from the introduction,
imply that $f_A \not\equiv 0$.  Therefore, by
Lemma~\ref{lemma:three:values}, there exists an assignment $\bu$ of the
variables of $f_A$ with probabilities in $\set{0,c,1}$ such that
$f_A(\bu) \neq 0$.  Next, we show a stronger result: if $Q$ is a {\em
  final} query, then $f_A(\bu)\neq 0$ for \e{any} assignment
$\bu{:}\set{u_1,\dots,u_N}\rightarrow (0,1)^N$, and in particular,
$f_A(c,\dots,c)\neq 0$.
\def\thmDetANon0{
	If $u_i{\in}(0,1)$ for all $i{\in}[1,N]$ then 
	 $f_A(u_1,\dots,u_N){\neq} 0$.
}
\begin{theorem}\label{thm:DetANon0}
\thmDetANon0
\end{theorem}

To prove Theorem~\ref{thm:DetANon0}, we use the fact that $Q$ is
final, and prove:
%


\def\type1Dichotomy{
	Let $Q$ be a $\typea$ query. Then $f_A\equiv 0$ if and only if $\lin{B_\p(i,j)}{Q}$ is reducible.
}\eat{
\begin{theorem}\label{thm:type1Dichotomy}
	\type1Dichotomy
\end{theorem}
}


\def\LemforbiddenLineage{
Let $Q$ be a final Type-I query, and let $X$ be a Boolean RV in $Y^{(p)}(u,v)$ corresponding to any atom other than $R(u)$ and $R(v)$.
Then $Y^{(p)}[X\asn a]$is disconnected for $a\in \set{0,1}$.
}
\begin{lemma}\label{lem:LemforbiddenLineage}
\LemforbiddenLineage	
\end{lemma}
\begin{proof}
For $i{\geq}1$, we define: $f(i){\eqdef}Q(r_{i{-}1},t_{i}){\wedge}Q(r_i,t_i)$.

Case 1: $X{=}R(r_k)$ for some $k{\in}[1,p{-}1]$ (we let $r_0{=}u$, and $r_p{=}v$).
We note that $Y^{(\p)}{=}h_1{\wedge}h_2$ where
$h_1{\eqdef}\bigwedge_{i=1}^{k}f(i)$ and
$h_2{\eqdef}\bigwedge_{i={k+1}}^{\p}f(i)$.
Since $\var(h_1){\cap}\var(h_2){=}\set{X}$, the result follows.

Case 2:  $X{=}T(t_k)$ for some $k{\in}[1,p]$.
We express $Y^{(\p)}{=}h_1{\wedge}h_2$ where $h_1{\eqdef}\bigwedge_{i=1}^{k{-}1}f(i){\wedge}Q(r_{k{-}1},t_k)$ and
$h_2{\eqdef}Q(r_{k},t_k) \bigwedge_{i={k{+}1}}^{\p}f(i)$. The result follows from noting that 
$\var(h_1){\cap}\var(h_2){=}\set{X}$.

Case 3: $X{=}S(r_k,t_k)$ for some binary symbol $S{\in}\symb(Q)$ where $k{\in}[1,\p]$. Consider the Boolean function $Q(r_k,t_k)$.  
Since $Q$ is final, then setting a value to $S(r_k,t_k)$ makes it safe. Hence, it does one of the following: (1) Makes $Q(r_k,t_k)$ disconnected, (2) Makes $Q_{\texttt{left}}(r_k,t_k)$ redundant or, (3) Makes $Q_{\texttt{right}}(r_k,t_k)$ redundant.
If (1) then we are done because if $Q(r_k,t_k)[S(r_k,t_k)]$ is disconnected, then clearly so is $Y^{(\p)}$. Otherwise,
since (2) is equivalent to setting $R(r_k){\gets}1$ and since (3) is equivalent to setting $T(t_k){\gets}1$ then the result follows from cases 1 and 2 respectively.
\end{proof}

\begin{corollary} There exists some constant $\alpha \neq 0$ such that:
\label{corr:fADet}
$f_A=\alpha\prod_{i=1}^Nu_i(1-u_i)$
\end{corollary}
\begin{proof}
  From Lemma~\ref{lemma:determinant:connected} it follows that $f_A[u_i:=0]=f_A[u_i:=1]=0$. Therefore, it follows that $f_A$ is
  divisible by $u_i(1-u_i)$, for every $i$.  Hence
  $f_A = \alpha\prod_{i=1}^Nu_i(1-u_i)$ and, since each variable has degree
  $\leq 2$ in $f_A$, it follows that $\alpha$ is a constant.
%
\end{proof}
Theorem~\ref{thm:DetANon0} follows from Corollary~\ref{corr:fADet} because $f_A(\bu)\neq0$ for any assignment $\bu:\var(f_A)\rightarrow (0,1)^N$.
\def\lemfAZero
{
	Let $f_A(u_1,\dots,u_N)$ denote the polynomial associated with $\det(A^{(1)})$ where $u_i$ represents the probability $\Pr(b_i)$ for every binary variable $b_i\in \var(\lin{B_\p(i,j)}{Q})$. 
	We define:
	\begin{align}
	g_1(b_1,\dots, b_N)&=\neg b_1\wedge \cdots \wedge \neg b_N\\
	g_2(b_1,\dots, b_N)&= b_1\wedge \cdots \wedge b_N
	\end{align}
	Then $f_A(\sigma(u_1),\dots,\sigma(u_N))=0$ iff $M_{g_1}(\sigma(u_1),\dots,\sigma(u_N))=0$ or $M_{g_2}(\sigma(u_1),\dots,\sigma(u_N))=0$
	where $\sigma{:}\set{u_1,\dots,u_N}{\rightarrow}[0,1]^N$. \eat{ is an assignment of probabilities to the binary random variables $\set{b_1,\dots,b_N}$.}
}
Thus, we have established that $f_A(c,\dots,c)\neq 0$. That is, $\det(A^{(1)})\neq 0$ when the real variables in $y(1)$ have a uniform value of $c$. In other words, we have shown that the matrix:
\begin{equation}\label{eq:A1z}
A^{(1)}\eqdef \begin{bmatrix} z_{00}(1) & z_{01}(1) \\
z_{10}(1)  & z_{11}(1)\end{bmatrix}
\end{equation}
is invertible.
\eat{
In what follows, $z_{ab}^{(1)}$ represents the valuation of the polynomial $y_{ab}^{(1)}$ when all variables are assigned $\frac{1}{2}$. That is:
\begin{equation}
\label{eq:zab}
z_{ab}^{(1)}\eqdef y_{ab}^{(1)}(\frac{1}{2},\dots,\frac{1}{2})
\end{equation}
}

So far, we have discussed only the ``small'' matrix $A^{(1)}$ that
corresponds to one step of the zig-zag block $B_p(u,v)$; in other
words, this is the matrix that we have for $B_1(u,v)$, when $p=1$.
Next, we show how to compute $A^{(p)}$.
We define:
\begin{equation} \label{eq:Ap}
A^{(p)}\eqdef \begin{bsmallmatrix} z_{00}(p) & z_{01}(p) \\
z_{10}(p)  & z_{11}(p)\end{bsmallmatrix}
\end{equation}

\eat{
\def\MatrixPowerthm{
	$A^{(\p)}=\begin{bsmallmatrix}z_{00}^{(\p)} & z_{01}^{(\p)} \\
	z_{10}^{(\p)} & z_{11}^{(\p)}\end{bsmallmatrix}=
	\frac{1}{2^{p-1}}\begin{bsmallmatrix}z_{00}^{(1)} & z_{01}^{(1)} \\
	z_{10}^{(1)} & z_{11}^{(1)}\end{bsmallmatrix}^p=\frac{1}{2^{p-1}}\left(A^{(1)}\right)^p$
}
}
\def\MatrixPowerlem{
	When $p\geq 1$ then $A^{(\p)}=[A^{(1)}\cdot C]^{p-1}\cdot A^{(1)}$ where $C=\begin{bsmallmatrix}1-c & 0 \\
		0 & c\end{bsmallmatrix}$. When $p=0$ then $A^{(\p)}= \mathbb{I}$ where $\mathbb{I}$ is the identity matrix.
	\eat{
		\begin{equation}
			\begin{bsmallmatrix}z_{00}(1) & z_{01}(1) \\
				z_{10}(1) & z_{11}(1)\end{bsmallmatrix}^p=\left[\begin{bsmallmatrix}z_{00}(1) & z_{01}(1) \\
				z_{10}(1) & z_{11}(1)\end{bsmallmatrix}\begin{bsmallmatrix}1-c & 0 \\
				0 & c\end{bsmallmatrix} \right]^{p-1}\begin{bsmallmatrix}z_{00}(1) & z_{01}(1) \\
				z_{10}(1) & z_{11}(1)\end{bsmallmatrix} \\		
		\end{equation}
	}
}
\begin{lemma}\label{lem:MatrixPower}
\MatrixPowerlem
\end{lemma}
\begin{proof}
	For the case where $\p\geq 1$, the proof is by induction on $\p$. The base case for $\p{=}1$ is immediate. So, we assume correctness for $\p{-}1$ and prove for $\p$.
	\begin{align*}
	Y^{(p)}(u,v)&\eqdef \bigwedge_{(s,t)\in B_p(u,v)}Q(s,t)\\
	&=\underbrace{\bigwedge_{(s,t)\in B_{p-1}(r_0,r_{\p{-}1})}Q(s,t)}_{Y^{(p-1)}(r_0,r_{\p{-}1})}\underbrace{\bigwedge_{(s,t)\in B_1(r_{\p{-}1},r_{\p})}Q(s,t)}_{Y^{(1)}(r_{\p{-}1},r_\p)}
	\end{align*}
	We note that the only atom common to both $Y^{(p-1)}(r_0,r_{\p{-}1})$ and $Y^{(1)}(r_{\p{-}1},r_\p)$ is $R(r_{\p-1})$. 
	Since $\Pr(R(r_{\p{-}1}){=}1)=c$, we have that:
	\begin{align}
	z_{ab}(p)&=(1{-}c)\left(z_{a0}(p{-}1)\right)\left(z_{0b}(1)\right)+c\left(z_{a1}(p{-}1)\right)\left(z_{1b}(1)\right)
	\label{eq:MatrixProof}
	\end{align}	
	Writing~\eqref{eq:MatrixProof} in matrix terms and applying the induction hypothesis, we get that:
	\begin{align}
		A^{(p)}&=\begin{bsmallmatrix} z_{00}(p) & z_{01}(p) \\
			z_{10}(p)  & z_{11}(p)\end{bsmallmatrix} \nonumber \\
			   &=\begin{bsmallmatrix} z_{00}(p-1) & z_{01}(p-1) \\
			   	z_{10}(p-1)  & z_{11}(p-1)\end{bsmallmatrix}\begin{bsmallmatrix}1-c & 0 \\ 0 & c\end{bsmallmatrix}\begin{bsmallmatrix} z_{00}(1) & z_{01}(1) \\
			   	z_{10}(1)  & z_{11}(1)\end{bsmallmatrix} \label{eq:writeInMatrixTerms} \\
		   	   &=\left[\begin{bsmallmatrix}z_{00}(1) & z_{01}(1) \\
		   	   	z_{10}(1) & z_{11}(1)\end{bsmallmatrix}\begin{bsmallmatrix}1-c & 0 \\
		   	   	0 & c\end{bsmallmatrix} \right]^{p-2}\begin{bsmallmatrix}z_{00}(1) & z_{01}(1) \\
		   	   	z_{10}(1) & z_{11}(1)\end{bsmallmatrix} \nonumber \\
	   	   	    &~~~\cdot \begin{bsmallmatrix}1-c & 0 \\ 0 & c\end{bsmallmatrix}\begin{bsmallmatrix}z_{00}(1) & z_{01}(1) \\
	   	   	    	z_{10}(1) & z_{11}(1)\end{bsmallmatrix} \label{eq:applyInductionHypothesis}\\
   	   	    	&=\left[\begin{bsmallmatrix}z_{00}(1) & z_{01}(1) \\
   	   	    		z_{10}(1) & z_{11}(1)\end{bsmallmatrix}\begin{bsmallmatrix}1-c & 0 \\
   	   	    		0 & c\end{bsmallmatrix} \right]^{p-1}\begin{bsmallmatrix}z_{00}(1) & z_{01}(1) \\
   	   	    		z_{10}(1) & z_{11}(1)\end{bsmallmatrix} \label{eq:simplify}
	\end{align}
where~\eqref{eq:writeInMatrixTerms} follows from writing $z_{ab}(p)$ in matrix terms, ~\eqref{eq:applyInductionHypothesis} follows from the induction hypothesis, and~\eqref{eq:simplify} from simplification. 

When $p=0$, we are basically looking at an empty lineage expression $Y^{(0)}(u,u)$, that is trivially $\true$ for $R(u)\in \set{0,1}$. In other words, $Y_{00}^{(0)}=Y^{(0)}_{11}=1$ and $Y_{01}^{(0)}=Y^{(0)}_{10}=0$.
\eat{
	\begin{align}
	A^{(p)} &= 
	\frac{1}{2}
	\begin{bsmallmatrix}
	z_{00}(\p{-}1) & z_{01}(\p{-}1)\\
	z_{10}(\p{-}1) & z_{11}(\p{-}1)
	\end{bsmallmatrix} 
	\begin{bsmallmatrix}
	z_{00}(1) & z_{01}(1)\\
	z_{10}(1) & z_{11}(1)
	\end{bsmallmatrix} \nonumber
	=\frac{1}{2^{p-1}}\begin{bsmallmatrix}
	z_{00}(1) & z_{01}(1) \\
	z_{10}(1) & z_{11}(1)\end{bsmallmatrix}^p \nonumber
	\end{align}	
}
\end{proof}
An immediate corollary of Lemma~\ref{lem:MatrixPower} is that $A^{(\p)}=[A^{(1)} C]^{p} C^{-1}$. It is easy to see that $C^{-1}=\begin{bsmallmatrix} \frac{1}{1-c} & 0 \\
	0  & \frac{1}{c}\end{bsmallmatrix}$.

Let $\lambda_1$ and $\lambda_2$ be the eigenvalues of $A^{(1)}$
(defined in~\eqref{eq:A1z}).  We prove now
condition~\eqref{eq:designyi}, which is the first of the three conditions
that we need to establish. We require the following simple proposition.
\begin{proposition}\label{prop:lineageProps}
	The following hold: 
	(1) $z_{00}(1)< z_{01}(1)=z_{10}(1)< z_{11}(1)$, and
	(2) $0<z_{ab}(1)\leq 1$ for all $a,b\in \set{0,1}$
\end{proposition}
\begin{proof}
	We note that $Y^{(1)}(u,v)$ and $Y^{(1)}(v,u)$ are identical. Consequently, $Y_{ab}^{(1)}\equiv Y_{ba}^{(1)}$ up to the renaming of the variables. Therefore, $z_{01}(1)=z_{10}(1)$.
	Since $Q$ does not contain negations, then its lineage is a monotonic Boolean function. Further, since the Boolean function $Y^{(1)}$ depends on both atoms $R(u)$ and $R(v)$, then (1) follows.	
	Item (2) follows by noting that $Y_{00}^{(1)}$ is satisfiable for any final type-I query, thus $z_{00}(1)> 0$.
\end{proof}
\blue{
In what follows, we define 
\begin{equation}
	\label{eq:B}
	B\eqdef [A^{(1)}\cdot C]=\begin{bsmallmatrix}
		(1-c)z_{00}(1) & cz_{01}(1) \\
		(1-c)z_{10}(1) & cz_{11}(1)\end{bsmallmatrix}
\end{equation}
}
\def\DetBNon0{ Let $\lambda_1,\lambda_2$ be the eigenvalues of
  $B$ (see~\eqref{eq:B}). Then: $\lambda_1\neq 0$, $\lambda_2\neq 0$, and
  $\lambda_1 {\neq} \pm \lambda_2$.  }
\begin{lemma}\label{lem:DetBNon0}
	\DetBNon0
\end{lemma}
\begin{proof}  This follows immediately from $\lambda_1+\lambda_2 =
  \text{Tr}(B)=(1-c)z_{00}+cz_{11}>0$ because $c\in (0,1)$ and by Proposition~\ref{prop:lineageProps} $0<z_{00}<z_{11}$. Also, $\lambda_1\lambda_2 =
  \det(B)=\det(A^{(1)})\det(C)\neq 0$, and the fact that a matrix where all
  eigenvalues are equal is a diagonal matrix (which $B$
  obviously is not).
  
   For a more elementary argument, recall that
	the characteristic polynomial of $B$ (see~\eqref{eq:B}) is (we drop the parameter $(1)$):
	\begin{align}
	\det(\lambda I{-}B)&=\lambda^2-\lambda(cz_{11} + (1{-}c)z_{00})+ c(1{-}c)(z_{00}z_{11}-z_{01}z_{10}) \label{eq:DetBNon0_1}
	\end{align}
	From~\eqref{eq:DetBNon0_1} we see that $\lambda{=}0$ is a root of the characteristic polynomial iff $\det(A^{(1)})=z_{00}z_{11}{-}z_{01}z_{10}{=}0$. By Theorem~\ref{thm:DetANon0}, this cannot be the case. Therefore, $\lambda_1,\lambda_2 {\neq} 0$. Also from~\eqref{eq:DetBNon0_1} we get that the two roots of $\det(\lambda I-A^{(1)})$ are: $$\lambda_{12}{=}\frac{(cz_{11} + (1{-}c)z_{00})\pm \sqrt{(cz_{11}- (1{-}c)z_{00})^2+4c(1{-}c) z_{01}z_{10}}}{2}.$$
	Since $z_{01}{=}z_{10}{>} 0$ (Proposition~\ref{prop:lineageProps})) and since $c\in (0,1)$, it follows that $\lambda_1 {\neq} \lambda_2$. Since, by Proposition~\ref{prop:lineageProps}, $(1{-}c)z_{00}+cz_{11}> 0$ it follows that $\lambda_1 {\neq} \pm \lambda_2$.
\end{proof}
\blue{
	\begin{corollary}
		\label{corr:DetApNonZero}
	For every $p\geq 1$ it holds that $\det(A^{(p)})\neq 0$. Or, that $y_{00}^{(p)}y_{11}^{(p)}-y_{01}^{(p)}y_{10}^{(p)}\neq 0$.
	\end{corollary}
	\begin{proof}
		By Lemma~\ref{lem:MatrixPower}, we have that $A^{(p)}=B^pC^{-1}$. Hence $\det(A^{(p)})=0$ iff $\det(B)=0$. By Lemma~\ref{lem:DetBNon0}, the two eigenvalues of $B$ are non-zero. Hence, $\det(B)\neq 0$.
	\end{proof}
}
\blue{The following Lemma is required to establish the fact that $\mathcal{N}$ is invertible (Proposition~\ref{prop:NInvertible}). We recall the definition of $y_a^{(t)}$ in~\eqref{eq:Ya}.
\begin{lemma}
	\label{lem:NVandermondeQuotient}
	For every $a\in \set{0,1}$ and $t\geq 1$ it holds that $y_a^{(t)}>0$.
	Also, for every pair of integers $t_2>t_1 \geq 1$ it holds that $\frac{y_1^{(t_1)}}{y_0^{(t_1)}}\neq \frac{y_1^{(t_2)}}{y_0^{(t_2)}}$.
\end{lemma}
\begin{proof}
		The claim that $y_a^{(t)}>0$ follows by noting that $Y_{00}^{(1)}$ is satisfiable for any final type-I query and that $y_a^{(t)} =y_{a0}^{(t)}+y_{a1}^{(t)}\geq y_{00}^{(t)}$ by Proposition~\ref{prop:lineageProps}.
		
	Now, we show that for any $d>0$ it holds that $y_1^{(t+d)}y_0^{(t)}\neq y_1^{(t)}y_0^{(t+d)}$. We observe that for any $a\in \set{0,1}$ it holds that:
	\begin{align*}
		y_a^{(t+d)}&=y_{a0}^{(t+d)}+y_{a1}^{(t+d)} \\
		&=y_{a0}^{(t)}y_{00}^{(d)}+y_{a1}^{(t)}y_{10}^{(d)}+y_{a0}^{(t)}y_{01}^{(d)}+y_{a1}^{(t)}y_{11}^{(d)}
	\end{align*}
	Hence, after simplifying the expression and using the fact that $y_{01}^{(t)}=y_{10}^{(t)}$ (Proposition~\ref{prop:symmetry}), we get that 
	\begin{equation*}
		y_1^{(t+d)}y_0^{(t)}- y_1^{(t)}y_0^{(t+d)}=
		(y_{11}^{(d)}-y_{00}^{(d)})(y_{00}^{(t)}y_{11}^{(t)}-y_{01}^{(t)}y_{10}^{(t)})
	\end{equation*}
	By monotonicity and non-redundancy we get that for any $d\geq 1$ it holds that $(y_{11}^{(d)}-y_{00}^{(d)})>0$. By Corollary~\ref{corr:DetApNonZero} we have that $(y_{00}^{(t)}y_{11}^{(t)}-y_{01}^{(t)}y_{10}^{(t)})\neq 0$ for all $t\geq 1$. This completes the proof.
\end{proof}
}

Since $B$ has two distinct, non-zero eigenvalues, then it has two linearly independent eigenvectors.
In other words, $B$ is \e{diagonizable}, and thus $B^p=PD^kP^{-1}$ where $P$ is the matrix whose columns are the eigenvectors of $B$, and $D$ is the diagonal matrix of its eigenvalues. 
Let $v_1=(c_{11}, c_{21})^T$, and  $v_2=(c_{12}, c_{22})^T$ be the two linearly independent eigenvectors corresponding to eigenvalues $\lambda_1$ and $\lambda_2$ respectively. 
By Lemma~\ref{lem:MatrixPower}, we have that $A^{(p)}=B^p\cdot C^{-1}$, and hence has the following form.

\begin{align}
A^{(p)}&{=}B^pC^{-1}{=}(PD^\p P^{-1})C^{-1}{=}(P)(D^\p)(P^{-1}C^{-1}) \\
&{=}\begin{bsmallmatrix} c_{11} & c_{12} \\c_{21} & c_{22}\end{bsmallmatrix}\begin{bsmallmatrix} \lambda_1^{\p} & 0 \\0 & \lambda_2^{\p}\end{bsmallmatrix}\begin{bsmallmatrix} b_{11} & b_{12} \\b_{21} & b_{22}\end{bsmallmatrix} \nonumber
\\
&{=}\begin{bsmallmatrix}c_{11}b_{11}\lambda_1^{\p}{+}c_{12}b_{21}\lambda_2^{\p} & c_{11}b_{12}\lambda_1^{\p}{+}c_{12}b_{22}\lambda_2^{\p}\\
c_{21}b_{11}\lambda_1^{\p}{+}c_{22}b_{21}\lambda_2^{\p} & c_{21}b_{12}\lambda_1^{\p}{+}c_{22}b_{22}\lambda_2^{\p}\end{bsmallmatrix}\label{eq:BMatrix_1} 
\\
&{=}\begin{bsmallmatrix}a_{00}\lambda_1^{\p}{+}b_{00}\lambda_2^{\p} & a_{01}\lambda_1^{\p}{+}b_{01}\lambda_2^{\p} \\
a_{10}\lambda_1^{\p}{+}b_{10}\lambda_2^{\p} & a_{11}\lambda_1^{\p}{+}b_{11}\lambda_2^{\p} 
\end{bsmallmatrix}
\label{eq:BMatrix_2}
\\
&{=}\begin{bsmallmatrix}
a_{00}\lambda_1^{\p}{+}b_{00}\lambda_2^{\p} & a_{10}\lambda_1^{\p}{+}b_{10}\lambda_2^{\p} \\
a_{10}\lambda_1^{\p}{+}b_{10}\lambda_2^{\p} & a_{11}\lambda_1^{\p}{+}b_{11}\lambda_2^{\p} 
\end{bsmallmatrix}
\label{eq:BMatrix_3} 
\end{align}
where the transition from~\eqref{eq:BMatrix_1} to~\eqref{eq:BMatrix_2} is by defining $a_{00}{\eqdef}c_{11}b_{11}$, $b_{00}{\eqdef}c_{12}b_{21}$ etc. 
The transition from~\eqref{eq:BMatrix_2} to~\eqref{eq:BMatrix_3} follows from the fact that,\blue{
by construction, $A^{(p)}$ is symmetric (see Proposition~\ref{prop:symmetry}).}
Notice that this establishes Eq.\eqref{eq:yiDesign}, which we need to
prove as part of Theorem ~\ref{thm:designSection}.
\eat{
For every $i{\in} \set{00,10,11}$ we define $c_i{\eqdef} \frac{a_{i}}{b_i}$, and let $w_p{\eqdef} \left(\frac{\lambda_2}{\lambda_1}\right)^{\p}$. This allows us to write the matrix of~\eqref{eq:BMatrix_3} as follows:
\begin{equation}
A^{(\p)}=\lambda_1^{\p}\begin{bsmallmatrix}
b_{00}(c_{00}+w_p) & b_{10}(c_{10}+w_p) \\ b_{10}(c_{10}+w_p) & b_{11}(c_{11}+w_p)
\end{bsmallmatrix}
\label{eq:PowerpMatrix}
\end{equation}
}
In particular, we have that:
\begin{equation}
\label{eq:A1NewExpression}
A^{(1)}=\begin{bsmallmatrix}
a_{00}\lambda_1+ b_{00}\lambda_2 & 
a_{10}\lambda_1+ b_{10}\lambda_2 \\
a_{10}\lambda_1+ b_{10}\lambda_2  & 
a_{11}\lambda_1+ b_{11}\lambda_2
\end{bsmallmatrix}
\end{equation}
\subsubsection*{Proof of Theorem~\ref{thm:designSection}}
From~\eqref{eq:BMatrix_3} it follows that $z_i(p)$ can be written in the form of~\eqref{eq:yiDesign} for all $i\in \set{00,10,11}$. 
In Lemma~\ref{lem:DetBNon0}, we have shown that $\lambda_1,\lambda_2 \neq 0$ , and that $\lambda_1\neq \pm \lambda_2$, thus proving condition ~\eqref{eq:designyi}.

By Lemma~\ref{lem:MatrixPower} we have that $A^{(0)}=(A^{(1)})^0=\mathbb{I}$ where $\mathbb{I}$ is the identity matrix.
Therefore, from~\eqref{eq:BMatrix_3}, we get the following three equations:  
\begin{align}
\label{eq:getToContradiction}
a_{00}+b_{00}=1 && a_{11}+b_{11}=1 && a_{10}+b_{10}=0
\end{align}
We show that $b_i\neq 0$ for all $i\in \set{00,10,11}$.
If $b_{10}=0$ then, since $a_{10}=-b_{10}$ it follows that  $z_{10}(1)=0$, which is a contradiction (Proposition~\ref{prop:lineageProps}). Thus, $b_{10}\neq 0$.
Assume that $b_{00}=0$. This means that $a_{00}=1$, and that for any $p \geq 1$, we have that $z_{00}(p)=\lambda_1^p$ (see~\eqref{eq:BMatrix_3}).
Recall that $z_{00}(p)=(1{-}c)z_{00}(p-1)z_{00}(1)+cz_{01}(p-1)z_{10}(1)$. By proposition~\ref{prop:lineageProps} we have that $z_{01}(p-1)>z_{00}(p-1)$, and that $z_{10}(1)>z_{00}(1)$. Therefore, $z_{00}(p)> z_{00}(p-1)z_{00}(1)= \lambda_1^{p-1}\lambda_1=\lambda^p$, and we arrive at a contradiction. 
Similarly, if $b_{11}=0$ then $a_{11}=1$, and $z_{11}(p)=\lambda_1^p$. Now, since $z_{11}(p)=(1{-}c)z_{10}(p-1)z_{01}(1)+cz_{11}(p-1)z_{11}(1)$, then since by Proposition~\ref{prop:lineageProps} $z_{11}(p)>z_{10}(p)=z_{01}(p)$ for all $p\geq 1$, we have that  $z_{11}(p)<\lambda_1^p$, which is a contradiction.

Finally, we show that $a_ib_j\neq a_jb_i$ for $i\neq j$. 
Assume, by contradiction, that $a_{00}b_{11}=a_{11}b_{00}$. Substituting $b_{11}=(1-a_{11})$ and $b_{00}=(1-a_{00})$, 
this implies that $a_{00}=a_{11}$ and thus $b_{00}=b_{11}$. But, by~\eqref{eq:A1NewExpression}, this means that $z_{00}(1)=z_{11}(1)$ which, by Proposition~\ref{prop:lineageProps}, is a contradiction.
Now, assume,  by contradiction, that $a_{00}b_{10}=a_{10}b_{00}$. From~\eqref{eq:getToContradiction} we have that $-a_{10}=b_{10}$. Substituting, this gives us that $a_{00}b_{10}=-b_{10}(1-a_{00})$ or, that $b_{10}=0$. But then $z_{10}(1)=0$, which, by Proposition~\ref{prop:lineageProps}, is a contradiction.
Symmetrically, it is shown that $a_{11}b_{10}\neq a_{10}b_{11}$. This completes the proof.

\section{Conclusions}

One can think of the model counting problem as: given a set of tuples
to {\em exclude}, compute the number of models of a sentence that do
not use any of the excluded tuples.  In this paper we studied the {\em
  generalized} model counting problem, where we are also given a set
of tuples to {\em include}, and need to count only those models that
contain all these tuples, and none of the excluded ones.  We have
established a dichotomy for Unions of Conjunctive Queries or,
equivalently, for $\forall$CNF formulas.  For a special case, called
{\em final queries of type I} we have also established a dichotomy for
the model counting problem; this complements a result by Amarilli et
al.~\cite{DBLP:conf/icalp/AmarilliBS15} that prove a dichotomy for
model counting for conjunctive queries without self-joins.  We leave
open the question whether UCQs admit a dichotomy for model counting.

\balance 
\bibliography{bib}

\newpage
\onecolumn	
\appendix
\section{Proof of Lemma~\ref{lemma:long}}

\label{appendix:proof:lemma:long}

\begin{figure}
  \centering
      \includegraphics[width=0.5\textwidth]{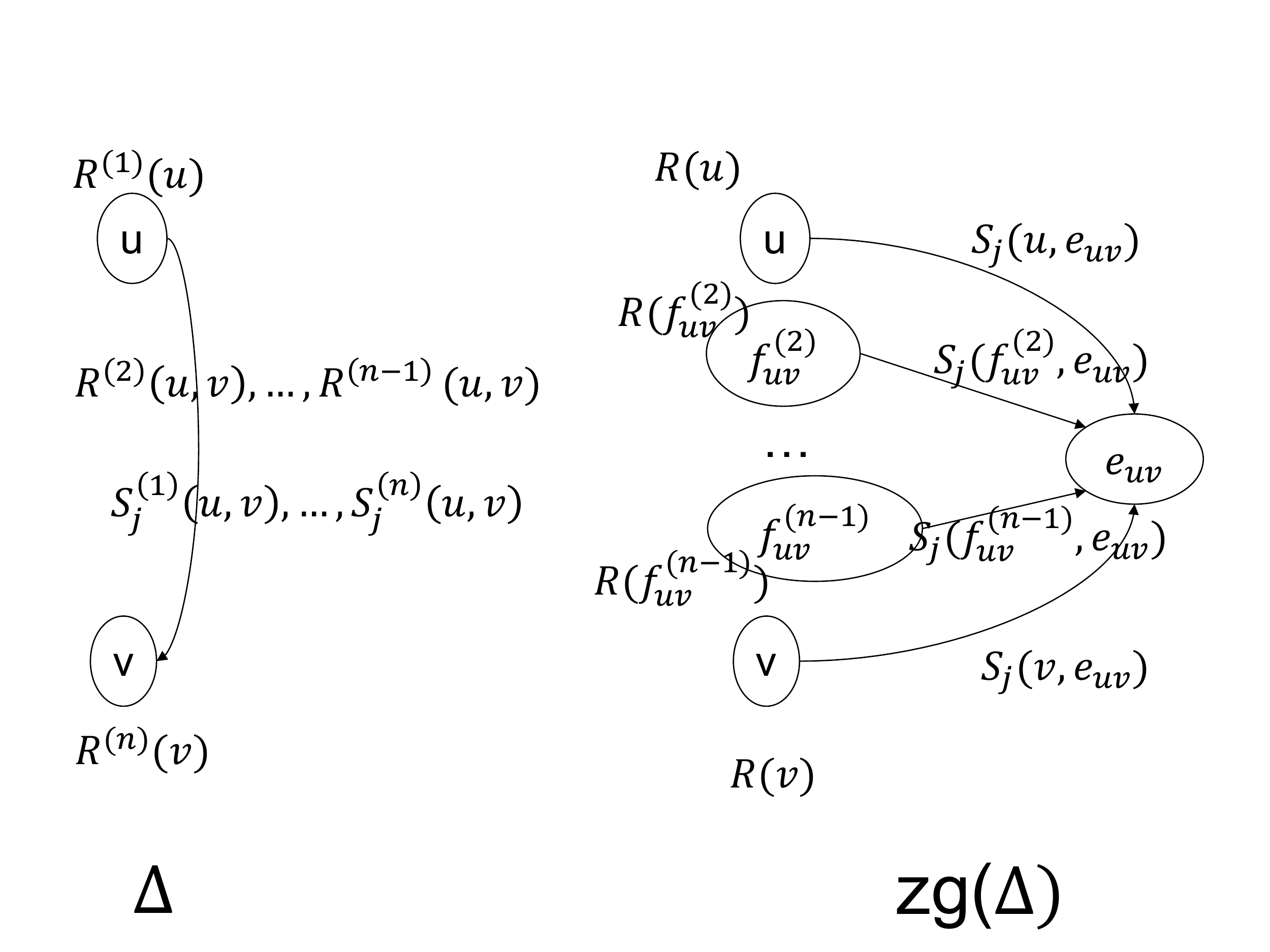}
      \caption{Illustration for a reduction
        $\pqe(\zg{Q}) \leq^P_m \pqe(Q)$ where $Q$ is a Type I-II
        query.  Given a query $Q$ we construct the query $\zg{Q}$.
        Then, given the database $\Delta$ for $\zg{Q}$ on the left, we
        construct the database $\zg{\Delta}$ for $Q$ on the right.
        Here $n$ is one plus the longest right clause of $Q$; for
        example if
        $Q_{\text{right}}= \forall y (\forall x S_1(x,y) \vee \forall
        x S_2(x,y)$, then $n=3$ and there is a single node
        $f_{uv}^{(2)}$, so that in total there are three edges
        incoming to $e_{uv}$.}
  \label{fig:reduction:to:long}
\end{figure}

Let $Q$ be an unsafe, bipartite query of length $k$, and let
$C_0, C_1, \ldots, C_k$ be a minimal left-right path.  Let $\calR$
denote its vocabulary.  We describe (a) a new unsafe, bipartite query
$\zg{Q}$ of length $2k$, over a new vocabulary $\zg{\calR}$ and (b) a
polynomial-time mapping that takes a bipartite TID $\pdb$ over the
vocabulary $\zg{\calR}$ and returns bipartite TID $\zg{\pdb}$ over the
vocabulary $\calR$ such that $\Pr_{\pdb}(Q)=\Pr_{\zg{\pdb}}(\zg{Q})$
and $\zg{Q}$ is long.  The probabilities values in $\pdb$ and
$\zg{\pdb})$ are the same, which proves that
$\gmc_{\text{bi}}(\zg{Q}) \leq^P_m \gmc_{\text{bi}}(Q)$.  Moreover, if
$Q$ is of type $A-B$, then $\zg{Q}$ is of type $A-A$.  These
properties prove Lemma~\ref{lemma:long}.

Define a number $n$ as follows.  If $Q_{\texttt{right}}$ is of Type I,
then $n=2$.  Otherwise, recall from Def.~\ref{def:types:of:clauses}
that the right clauses of type II have the form
$\forall y \left(\bigvee_{i=1}^\ell (\forall x
  S_{J_\ell}(x,y))\right)$.  We define $n$ to be the maximum between
3, and the largest value of $\ell$ of all right clauses.  Thus, by
definition $n \geq 3$.
%

{\bf The vocabulary $\zg{\calR}$.} We start by describing the mapping
$\texttt{zg}$ between the vocabularies.  The vocabulary $\calR$
contains a set of binary symbols $S_1,\dots,S_t$. In addition, it may
contain a single unary symbol $R$, and a single unary symbol $T$.  The
new vocabulary $\zg{\calR}$ consists of $m+n$ disjoint copies of
$\calR$, as follows:
\begin{itemize}
	\item If $R \in \calR$ then $\zg{\calR}$ contains symbols
	$R^{(1)}, R^{(2)}, \ldots, R^{(n)}$, where the first and last one
	are unary, and all others are binary.
	\item For every binary symbol $S \in \calR$, $\zg{\calR}$ contains the following binary symbols:	 
	$S^{(1)}, \ldots, S^{(n)}$.	
	\item If $T \in \calR$, then $\zg{\calR}$ contains a binary
	symbol $T^{(12)}$.
\end{itemize}

Notice that the only unary symbols are $R^{(1)}$ and $R^{(n)}$; the
symbol $T$ became a binary symbol $T^{(12)}$.  These will be the unary
symbols ``$R$ and $T$'' of the new query $\zg{Q}$.  Thus, if $Q$ was
of type I-I or I-II then $\zg{Q}$ will be of type I-I because it has
the two unary symbols $R^{(1)}$ and $R^{(n)}$.  If $Q$ was of type
II-I or II-II, then it has no $R$-symbol, hence $\zg{Q}$ has no unary
symbols, i.e. it will be of type II-II.

{\bf The database $\zg{\Delta}$.}  Next, given a bipartite
probabilistic database $\pdb = (\Dom,p)$ over vocabulary $\zg{\calR}$,
we describe how to construct $(\zg{\pdb},p')$.  Since $\pdb$ is
bipartite, is domain is $\Dom = V_1 \cup V_2$ and the only tuples $t$
with $p(t) \neq 1$ are those of the form $S_j(u,v)$,
$u \in V_1, v \in V_2$.  Define the following bipartite database,
$\zg{\pdb} = (\Dom', p')$ over the vocabulary $\calR$.  Its domain
$\Dom' = U_1 \cup U_2$ consists of the following constants:
\begin{itemize}
	\item For every $u \in V_1$,  $u$ is in $U_1$.
	\item For every $v \in V_2$,  $v$ is in $U_1$.
	\item For every pair $u\in V_1, v\in V_2$, there is a fresh element
	denoted $e_{uv}$, which is in $U_2$.
	\item For every pair $u\in V_1, v\in V_2$, there are $n-2$ fresh
	elements $f^{(i)}_{uv}$, $i=2,\ldots,n-1$, all of which are in
	$U_1$.
\end{itemize}

To define the tuple probabilities $p'$ in $\pdb$, we establish a
1-to-1 correspondence between the tuples in $\pdb$ and those in
$\zg{\pdb}$, which, in turn, defines $p'$ in terms of $p$:
\begin{itemize} 
	\item $p'(R(u)) \defeq p(R^{(1)}(u))$ for all $u \in V_1$.
	\item $p'(R(f^{(i)}_{uv})) = p(R^{(i)}(u,v))$ for all
	$u \in V_1, v \in V_2$, $i\in [2,n-1]$.
	\item $p'(R(v)) \defeq p(R^{(n)}(v))$ for all $v \in V_2$.
	\item $p'(S(u,e_{uv})) \defeq p(S^{(1)}(u,v))$ for all $u \in V_1, v \in V_2$.
	\item $p'(S(f^{(i)}_{uv},e_{uv})) \defeq p(S^{(i)}(u,v))$ for all  $u \in V_1, v \in V_2$, $i=2,\ldots,n-1$.
	\item $p'(S(v,e_{uv})) \defeq p(S^{(n)}(u,v))$ for all $u \in V_1, v \in V_2$.   
	\item $p'(T(e_{uv})) \defeq p(T^{(12)}(u,v))$
\end{itemize}
All other tuples have probability $p'(t) = 1$.

{\bf The querry $\zg{Q}$.}  Finally, we define the zig-zag query
$\zg{Q}$.

\begin{itemize}
	\item For every left clause $C$ in $Q$ there are $n$ clauses in $\zg{Q}$,
	of which the first is a left clause, the last is a right clause, and
	the rest are middle clauses.  More precisely, if $C$ is of type I clause,
	$C = R(x) \vee S_J(x,y)$ then $\zg{Q}$ contains the following clauses:
	\begin{align}
	\forall x \forall y (R^{(1)}(x) \vee & S_J^{(1)}(x,y)) \label{eq:new:left:1} \\
	\forall x \forall y (R^{(i)}(x,y) \vee & S_J^{(i)}(x,y)) &	i=2,\ldots,n-1\nonumber\\
	\forall x \forall y (R^{(n)}(y) \vee & S_J^{(n)}(x,y)) \label{eq:new:right:1} 
	\end{align}
	If $C$ is of type II,
	$C = \forall x \left(\bigvee_{i=1}^m \forall y(S_{J_i}(x,y))\right)$
	then $\zg{Q}$ contains the following clauses:  
	\begin{align}
	& \forall x \left(\bigvee_{i=1}^m\forall y(S_{J_i}^{(1)}(x,y))\right)\label{eq:new:left:2:beg}
	\\
	& \forall x \forall y \left(\bigvee_{i=1}^m S_{J_i}^{(j)}(x,y)\right) & i=2,\ldots,n-1\nonumber\\ 	
	& \forall y \left(\bigvee_{i=1}^m\forall x(S_{J_i}^{(n)}(x,y))\right)	  \label{eq:new:right:2} 
	\end{align}
	\item Every middle clause $C = S_J(x,y)$ in $Q$ becomes $n$ middle
	clauses in $\zg{Q}$:
	\begin{align}
	&\forall x \forall y S^{(i)}_J(x,y) & i=1,\dots,n \label{eq:middle-beg}
	\end{align}
      \item Every right clause becomes several middle clauses.  There
        are two cases.  If the right part of $Q$ is of Type I, then
        every right clause has the form
        $C = \forall x \forall y(S_J(x,y) \vee T(y))$.  In that case
        recall that $n=2$ and there will be exactly two middle clauses
        in $\zg{Q}$:
	\begin{align}
	& \forall x \forall y(S_J^{(1)}(x,y) \vee T^{(12)}(x,y)) \label{eq:new:middle:1}\\
	& \forall x \forall y(S_J^{(2)}(x,y) \vee T^{(12)}(x,y)) \label{eq:new:middle:2}
	\end{align}
	If the right part of $Q$ is of Type II, then every right
        clause has the form
        $C = \forall y(\bigvee_{i=1}^\ell \forall x S_{J_i}(x,y))$.
        In this case, we create $n^\ell$ middle clauses in $\zg{Q}$,
        as follows.  For every function $\phi: [\ell] \rightarrow [n]$
        there will be one middle clause (for a total of $n^\ell$
        middle clauses):
	\begin{align}
	C^{(\phi)} \defeq &\forall x \forall y \left(S_{J_1}^{(\phi(1))}(x,y) \vee \cdots \vee S_{J_\ell}^{(\phi(l))}(x,y)\right) &\label{eq:new:middle} 
	\end{align}  
\end{itemize}
$\zg{Q}$ is defined as the conjunction of all clauses above.  Recall
that, by assumption, we minimize every $\forall$CNF expression, that
means that the clauses described above need to be minimized, and
redundant ones need to be removed. 

We prove several properties of $\zg{Q}$:

\begin{lemma}\label{lem:shortToLong}
	$\Pr_{\pdb}(\zg{Q}) = \Pr_{\zg{\pdb}}(\zg{Q})$
\end{lemma}

\begin{proof}
  We prove a stronger statement: the lineage of $Q$ on $\zg{\pdb}$ is
  equivalent to the lineage of $\zg{Q}$ on $\pdb$, up to the 1-to-1
  correspondence between the tuples described above.  For that, we
  will show that for each clause $C$ of $Q$, it's lineage is
  equivalent to the conjunction of lineages of clauses derived from
  $C$ in $\zg{Q}$.  There are several cases. In all cases we exploit
  the fact that $x$ in $C$ can be mapped only to constants of the form
  $u, f^{(2)}_{uv}, \ldots, f^{(n-1)}_{uv}, v$, and $y$ can be mapped
  only to constants of the form $e_{uv}$, where
  $u \in V_1, v \in V_2$:
	\begin{itemize}
		\item Suppose $C$ is a left clause.  If it is of Type I,
		$C = R(x) \vee S_J(x,y)$, then its lineage (in $\zg{\pdb}$) is the conjunction
		(over $u \in V_1$ and $v \in V_2$) of:
		\begin{align*}
		&  \left(R(u) \vee S_J(u,e_{uv})\right) 
		\wedge \bigwedge_{i=2}^{n-1}(R(f^{(i)}_{uv}) \vee S_J(f^{(i)}_{uv},e_{uv}))
		\wedge  \left(R(v) \vee S_J(v,e_{uv})\right) &		
		\end{align*}
		This is precisely the lineage of the clauses
		(\ref{eq:new:left:1})-(\ref{eq:new:right:1}).  If $C$ is of type
		II,
		$C = \bigvee_{i=1}^{\ell} \forall y S_{J_i}(x,y)$ then its lineage is the conjunction of the following Boolean formulas:
		\begin{align*}
		& \bigwedge_{u\in V_1}\left(\bigvee_{i=1}^{\ell}
		\left(
		\bigwedge_{v \in V_2}S_{J_i}(u,e_{uv})     
		\right)\right) & \\
		& \bigwedge_{v\in V_2}\left(\bigvee_{i=1}^{\ell}
		\left(
		\bigwedge_{u \in V_1}S_{J_i}(v,e_{uv})     
		\right)\right) & \\	
		&\bigwedge_{j=2}^{n-1}\left(\bigvee_{i=1}^{\ell} S_{J_i}(f_{uv}^{(j)},e_{uv})\right)	
		\end{align*}
		which is precisely the lineage of the clauses
		(\ref{eq:new:left:2:beg})-(\ref{eq:new:right:2}).
		\item If $C$ is a middle clause, $C = S_J(x,y)$, then it's lineages in $\zg{\pdb}$ is the Boolean expression:
		\begin{align*}
		S_J(u,e_{uv}) \wedge 
		\bigwedge_{j \in [2,n-1]} S_J(f^{(j)}_{uv},e_{uv})  
		\wedge S_J(v,e_{uv})
		\end{align*}
		This is precisely the lineage of clause~\eqref{eq:middle-beg}.	
		\item If $C$ is a right clause, then we distinguish two cases.  If
		it is of Type I, $C = S_J(x,y) \vee T(y)$, then recall that $n=2$,
		and its lineage is the conjunction over all $u, v$ of the
		expression:
		\begin{align*}
		(S_J(u,e_{uv}) \vee T(e_{uv})) \wedge (S_J(v,e_{uv}) \vee T(e_{uv}))
		\end{align*}
		which is precisely the lineage of the expressions
		(\ref{eq:new:middle:1})-(\ref{eq:new:middle:2}).  If it is of Type
		II,
		$C = \forall y \left(\bigvee_{i=1}^\ell \forall x (S_{J_i}(x,y)) \right)$, then we note that $y$ must be mapped to some
		value $e_{uv}$. Fixing $y$ to $e_{uv}$ implies exactly $n$ possibilities for
		each $x$. Namely, $\set{u,v,f_{uv}^{(2)},\dots,f_{uv}^{(n)}}$.  The lineage of $C$ is the conjunction, over all
		$u \in U_1, v \in V_2$, of the Boolean formula:
		\begin{align*}
		&  (S_{J_1}(u,e_{uv}) \wedge S_{J_1}(f^{(2)}_{uv},e_{uv}) \wedge
		\cdots S_{J_1}(f^{(n-1)}_{uv},e_{uv}) \wedge S_{J_1}(v,e_{uv}))  \\
		& \cdots \\
		&  (S_{J_\ell}(u,e_{uv}) \wedge S_{J_\ell}(f^{(2)}_{uv},e_{uv}) \wedge
		\cdots S_{J_\ell}(f^{(n-1)}_{uv},e_{uv}) \wedge S_{J_\ell}(v,e_{uv}))
		\end{align*}
		which is equivalent to:
		\begin{align*}
		& (\bigwedge_{i \in \ell} S^{(i)}_{J_1}(u,v) \vee \cdots \vee \bigwedge_{i \in \ell} S^{(i)}_{J_\ell}(u,v))
		\end{align*}
		We apply the distributivity law of $\vee$ over $\wedge$ to convert
		this expression into a CNF expression, and obtain the conjunction
		of all lineages of clauses (\ref{eq:new:middle}).
	\end{itemize}
	This completes the proof.
\end{proof}

\begin{lemma}
	If $Q$ is an unsafe query, then $\zg{Q}$ is also unsafe.
\end{lemma}

\begin{proof}
	We start by observing that the following clauses in $\zg{Q}$ are both
	minimized and non-redundant:
	\begin{itemize}
        \item Every left clause of the form (\ref{eq:new:left:1}) or
          (\ref{eq:new:left:2:beg}) is minimized and non-redundant.
        \item If $C$ is a middle clause in $Q$, then the middle clause
          $C^{(i)}$ in $\zg{Q}$ is minimized and non-redundant.
        \item If
          $C =\forall y \bigvee_{i=1}^\ell \forall x S_{J_i}(x,y)$ is
          a right clause in $Q$, then every middle clause $C^{(\phi)}$
          of the form (\ref{eq:new:middle}) where $\phi$ is injective
          (i.e. the indices $\phi(1), \phi(2), \ldots, \phi(\ell)$ are
          mapped to distinct elements of $[n]$) is minimized and
          non-redundant.  Indeed, if there was some homomorhism
          $C_0' \rightarrow (S_{J_1}^{(\phi(1))}(x,y) \vee
          S_{J_2}^{(\phi(2))}(x,y) \vee \cdots
          S_{J_\ell}^{(\phi(\ell))}(x,y))$, where $C_0'$ is a clause
          in $\zg{Q}$ derived from some clause $C_0$ in $Q$, then we
          can construct a homorphism
          $C_0 \rightarrow (\forall x S_{J_1}(x,y) \vee \forall x
          S_{J_2}(x,y) \vee \ldots)$, implying that $C$ was redundant
          in $Q$ , which is a contradiction.  To see why, note that
          such a homomorphism would imply that
          $\symb(C_0) \subseteq J_i$ for some $i \in [\ell]$.  We
          remark here that, if $\phi$ is not injective, then
          $C^{(\phi)}$ may be redundant.  Since we have chosen $n$
          such that $n \geq \ell$, for for every right clause $C$
          there exists some; injective $\phi$, and in that case
          $C^{(\phi)}$ is non-redundant.  
	\end{itemize}
	
        \begin{example}  \label{ex:why:we:need:dead:ends}
          We justify here the reason for introducing the ``dead end''
          branches $f^{(i)}_{uv}$.  Consider $\zg{Q}$ where $Q$ is the
          following query:
            \begin{align*}
              Q&=\forall x\forall y (R(x) \vee S_0(x,y))
              \\&\wedge \forall x \forall y(S_0(x,y) {\vee}
                  S_1(x,y)) \wedge \underbrace{(S_1(x,y) {\vee}  S_2(x,y) {\vee} S_3(x,y))}_{\defeq D}
              \\&\wedge 
                  \underbrace{\forall y (\forall x (U(x,y) {\vee}
                  S_1(x,y) {\vee} S_2(x,y)) {\vee} \forall x (U(x,y)
                  {\vee} S_1(x,y) {\vee} S_3(x,y)) {\vee} \forall x
                  (U(x,y) {\vee} S_2(x,y) {\vee} S_3(x,y)))}_{\defeq C}
		\end{align*}
The middle clause $D$ generates two clauses in $\zg{Q}$:
\begin{align*}
  D^{(1)} \defeq &  (S_1^{(1)}(x,y) \vee  S_2^{(1)}(x,y) \vee  S_3^{(1)}(x,y)) &
  D^{(2)} \defeq &  (S_1^{(2)}(x,y) \vee  S_2^{(2)}(x,y) \vee  S_3^{(2)}(x,y))
\end{align*}
Assuming $n=2$ (i.e. no ``dead end'' branches $f_{uv}^{(i)}$), the
right clause $C$ is mapped to the conjunction of $2^3=8$ clauses
$C^{(\phi)}$ and {\em all} become redundant.  For example, when
$\phi(1)=\phi(2)=1$, $\phi(3)=2$ then:
\begin{align*}
C^{(\phi)} = & \forall x \forall y (U^{(1)}(x,y) \vee S_1^{(1)}(x,y) \vee S_2^{(1)}(x,y) \vee S_3^{(1)}(x,y)
\vee U^{(2)}(x,y) \vee S_2^{(2)}(x,y)  \vee S_3^{(2)}(x,y))
\end{align*}
And this is redundant because of $D^{(1)}$.  It is easy to check that
all 8 clauses $C^{(\phi)}$ are made redundant by either $D^{(1)}$ or
$D^{(2)}$.  We notice that the query $Q$ is even a {\em forbidden
  query}, to be discussed in Sec.~\ref{sec:background:factorization},
where $U$ is called a {\em ubiquitous symbol}.  Thus, the dead end
branches cannot be avoided even if $Q$ were guaranteed to be a
forbidden query.
\end{example}

Next, we prove that $\zg{Q}$ is an unsafe bipartite query, by showing
the existence of a left-to-right path.  By assumption $Q$ is an
unsafe, bipartite query of length $k$, hence there exists a
left-to-right path $C_0, C_1, \ldots, C_k$ in $Q$.  We define a
left-to-right path in $\zg{Q}$ as follows.  The path starts with:
\begin{align*}
  \underbrace{C_0^{(1)}}_{\mbox{\eqref{eq:new:left:1} or \eqref{eq:new:left:2:beg}}}, \underbrace{C_1^{(1)}, \ldots, C_{k-1}^{(1)}}_{\mbox{\eqref{eq:middle-beg}}}
\end{align*}
that is, we start with the translation of $C_0$ into a left clause
$C_0^{(1)}$ using expression \eqref{eq:new:left:1} if $C_0$ is of type
I, or \eqref{eq:new:left:2:beg} if $C_0$ is of type II, then continue
with the translations of the middle clauses, all using branch $i=1$.
The path ends with the following:
\begin{align*}
  \underbrace{C_{k-1}^{(n)}, \ldots, C_1^{(n)}}_{\mbox{\eqref{eq:middle-beg}}}, \underbrace{C_0^{(n)}}_{\mbox{\eqref{eq:new:right:1} or \eqref{eq:new:right:2}}}
\end{align*}
That is, it ends at the translation of $C_0$ into a right clause
$C_0^{(n)}$, as per \eqref{eq:new:right:1} or \eqref{eq:new:right:2}.
So far, all clauses we have used are non-redundant.  It remains to
connect $C_{k-1}^{(1)}$ with $C_{k-1}^{(n)}$, and for that we use the
right clause $C_k$.  Assume first that $C_k$ is of type I, i.e.
$C_k = \forall x \forall y (S_J(x,y) \vee T(y))$; in that case $n=2$.
By assumption $C_{k-1}$ and $C_k$ have some common symbol, call it
$S$; obviously $S$ must be among the symbols $S_J$, since it cannot be
$T$ (because $C_{k-1}$ is not a right clause).  Define
$C_k^{(1)}, C_k^{(2)}$ to be the middle clauses
\eqref{eq:new:middle:1} and \eqref{eq:new:middle:2} respectively.  The
left-to-right path in $\zg{Q}$ is:
\begin{align*}
  &  C_0^{(1)}, C_1^{(1)}, \ldots, C_{k-1}^{(1)}, C^{(1)}_k, C^{(2)}_k, C_{k-1}^{(2)}, \ldots, C_0^{(2)}
\end{align*}
The clauses $C_{k-1}^{(1)}, C^{(1)}_k$ share the common symbol
$S^{(1)}$; the clauses $C^{(1)}_k, C^{(2)}_k$ share the common symbol
$T^{(12)}$, and the clauses $C_k^{(2)}, C^{(2)}_{k-1}$ share the
common symbol $S^{(2)}$. Thus, the path above is a left-right path of
length $2k+1$ in $\zg{Q}$.  Assume now that $C_k$ is of type II, i.e.
$C_k = \forall x S_{J_1}(x,y) \vee \cdots \forall x S_{J_\ell}(x,y)$
and let $S \in \symb(C_{k-1}) \cap \symb(C_k)$; assume wlog that
$S \in \symb(S_{J_1})$.  We consider two derived middle clause
$C^{(\phi_1)}_k, C^{(\phi_2)}_k$ of the form (\ref{eq:new:middle}), as
follows.  Both $\phi_1, \phi_2$ will be injective, ensuring that both
clauses are non-redundant, and are defined as follows:
\begin{align*}
  \phi_1(1) \defeq & 1 & \phi_1(2) \defeq & 2 & \phi_1 : \set{3,\ldots,\ell}\rightarrow \set{3,\ldots,n} \mbox{ any injective function}\\
  \phi_2(1) \defeq & n & \phi_2(2) \defeq & 2 & \phi_2 : \set{3,\ldots,\ell}\rightarrow \set{1,3,\ldots,n-1} \mbox{ any injective function}\\
\end{align*}
Recall that we have defined $n \geq 3$, henc $\phi_2$ is injective.
Thus,
$C_k^{(\phi_1)} = (S_{J_1}^{(1)} \vee S_{J_2}^{(2)} \vee \cdots)$ and
$C_k^{(\phi_2)}= (S_{J_1}^{(n)} \vee S_{J_2}^{(2)} \vee \cdots)$.  Now
we observe that $C^{(1)}_{k-1}, C^{(\phi_1)}_k$ share the common
symbol $S^{(1)}$, the clauses $C_k^{(\phi_1)},C_k^{(\phi_2)}$ share
all common symbols in $S_{J_2}^{(2)}$, and
$C_k^{(\phi_2)},C_{k-1}^{(2)}$ share the common symbols $S^{(n)}$.
Thus, the following is a left-to-right path in $\zg{Q}$:
	\begin{align*}
	& C_0^{(1)}, C_1^{(1)}, \ldots, C_{k-1}^{(1)}, C^{(\phi_1)}_k, C^{(\phi_2)}_k, C_{k-1}^{(n)}, \ldots, C_0^{(n)}
	\end{align*}
\end{proof}

This proves that $\zg{Q}$ is an unsafe, bipartite query of length
$\leq 2k+1$.

Finally, it remains to prove that every left-to-right path in $\zg{Q}$
has length $\geq 2k$, where $k$ is the length of the unsafe query $Q$.
Consider any left-ot-right path in $\zg{Q}$:
$C_0', C_1', \ldots, C_m'$.  Since $C_0'$ is a left clause it must be
of the form \eqref{eq:new:left:1} or \eqref{eq:new:left:2:beg}, hence
all its symbols are from the branch 1, like $S^{(1)}$.  Let
$C_{\ell}'$ be the first clause that contains an index of a branch
other than $1$, i.e. either some $S^{(i)}$, $i>1$ or $T^{(12)}$.
$C_\ell'$ must also have a symbol on branch 1, because it shares a
symbol with $C_{\ell-1}'$ whose symbols are all on branch 1, hence
$C_\ell'$ can only be of the form \eqref{eq:new:middle:1} or
\eqref{eq:new:middle}.  Therefore, the fragment
$C_0', C_1', \ldots, C_\ell'$ of the path in $\zg{Q}$ corresponds to a
left-to-right path $C_0, C_1, \ldots, C_\ell$ in $Q$, namely
consisting of the clauses in $Q$ that generated the clauses
$C_0', C_1', \ldots, C_\ell'$ in $\zg{Q}$.  Since $k$ is the length of
the shortest left-right path in $Q$, we have $\ell \geq k$.  We reason
similarly about the suffix of the path.  Since $C_m'$ is a right
clause, it is of the form \eqref{eq:new:right:1} or
\eqref{eq:new:right:2} (note that it originates from a {\em left}
clause $C_m$), hence all its symbols are on branch $n$,
i.e. $S^{(n)}$.  Let $C_p'$ be the last clause that has some symbol on
a branch other than $n$.  Using the same argument as above, we
conclude that $C_p'$ originates from a right clause $C_p$, hence the
fragment $C_p', C_{p+1}', \ldots, C_m'$ of the path in $\zg{Q}$
corresponds to a left-to-right path in $Q$ (in reverse order):
$C_m, C_{m-1}, \ldots, C_{p+1}, C_p$.  Therefore, $m-p \geq k$.  Since
$\ell \leq p$ we conclude that $m \geq \ell + (m-p) \geq 2k$.

\section{Background on Factorization and Independence}

\label{sec:background:factorization}

Every multivariate polynomial $f$ admits a unique decomposition into
irreducible factors: $f = g_1^{k_1} g_2^{k_2} \cdots g_m^{k_m}$.  In
this paper we use repeatedly the following:

\begin{theorem} \label{th:rank:1:matrix:of:polynomials}
  Let $f_{00}, f_{01}, f_{10}, f_{11}$ be four multivariate
  polynomials, such that the following determinant is identically
  zero:
  \begin{align*}
\det\left(
      \begin{array}{cc}
        f_{00} & f_{01} \\
        f_{10} & f_{11}
      \end{array}
\right) \equiv & 0
  \end{align*}
  Then, there exists polynomials $g_0, g_1$ and $h_0, h_1$ such that
  the following identity holds:
  \begin{align*}
    \left[
    \begin{array}{c}
      g_0 \\ g_1
    \end{array}
\right] \cdot
    \left[
    \begin{array}{cc}
      h_0 & h_1
    \end{array}
\right] \equiv &
\left[
      \begin{array}{cc}
        f_{00} & f_{01} \\
        f_{10} & f_{11}
      \end{array}
\right]
  \end{align*}
\end{theorem}

\begin{proof}
  By induction on the total degree of $f_{00}f_{11}$.  Let $p$ be any
  irreducible factor of $f_{00}f_{11} \equiv f_{01}f_{10}$.  Assume
  that $p | f_{00}$ and $p | f_{01}$ (the other three cases are
  similar and omitted).  Then the polynomials
  $f_{00}/p, f_{01}/p, f_{10}, f_{11}$ also satisfy the condition of
  the theorem, therefore, by induction hypothesis, there exists
  $g_0, g_1, h_0, h_1$ such that:
  \begin{align*}
    \left[
    \begin{array}{c}
      g_0 \\ g_1
    \end{array}
\right] \cdot
    \left[
    \begin{array}{cc}
      h_0 & h_1
    \end{array}
\right] \equiv &
\left[
      \begin{array}{cc}
        f_{00}/p & f_{01} \\
        f_{10}/p & f_{11}
      \end{array}
\right]
  \end{align*}
  Then the polynomials $g_0, g_1, ph_0, h_1$ satisfy the condition of
  the theorem:
  \begin{align*}
    \left[
    \begin{array}{c}
      g_0 \\ g_1
    \end{array}
\right] \cdot
    \left[
    \begin{array}{cc}
      ph_0 & h_1
    \end{array}
\right] \equiv &
\left[
      \begin{array}{cc}
        f_{00} & f_{01} \\
        f_{10} & f_{11}
      \end{array}
\right]
  \end{align*}
\end{proof}

The next two results concern variations on
Lemma~\ref{lemma:three:values}.  We are given four monotone Boolean
functions $F, G, H, K$ over $n$ variables, with arithmetizations
$f, g, h, k$, and seek a valuation $\varphi \in \set{0,1/2,1}^n$ such
that, on one hand $f[\varphi]g[\varphi] \neq h[\varphi]k[\varphi]$, on
the other hand $\varphi$ avoids the value 0 as much as possible.  We
start with a simple case, when $\varphi$ avoids the value 0
completely.

If $\theta$ is a truth assignment of the Boolean variables of $F$,
then we denote by $F_\theta$ the Boolean function obtained from $F$ by
setting $X:=1$ for all variables where $\theta(X)=1$, and leaving the
other variables free.  For example, if $F(X_1,X_2,X_3,X_4)$ and
$\theta(X_1)=\theta(X_3)=1$, $\theta(X_2)=\theta(X_4)=0$ then
$F_\theta \defeq F(1, X_2, 1, X_4)$. We denote by $\#F_\theta$ the number of models of the residual formula $F_\theta$.

\begin{lemma}\label{lemma:values:half:or:one}
  Let $F, G, H, K$ be four monotone Boolean functions with $n$
  variables and $f,g,h,k$ their arithmetization. If
  $F \wedge G \not\equiv H \wedge K$ then there exists
  $\varphi \in \set{1/2,1}^n$ such that
  $f[\varphi]g[\varphi] \neq h[\varphi]k[\varphi]$.
\end{lemma}

\begin{proof} 
  Let $\theta$ be an assignment that distinguishes the two. Assume
  w.l.o.g. that:
  \begin{align*}
    F[\theta]= &0, \ \ G[\theta]\in \set{0,1}, \ \  H[\theta]= K[\theta]=1
  \end{align*}
  If $m$ is the number of variables in
  $F_\theta, G_\theta, H_\theta, K_\theta$, then $\#F_\theta < 2^m$
  because $F[\theta]=0$, $\#G_\theta \leq 2^m$, and
  $\#H_\theta=\#K_\theta = 2^m$, which implies
  $\#H_\theta\#K_\theta - \#F_\theta\#G_\theta = 2^{2m} - \#F_\theta
  \#G_\theta > 0$.
\end{proof}

\begin{corollary}\label{cor:values:half:or:one}
  With the notations in Lemma~\ref{lemma:values:half:or:one}, if there
  exists any assignment $\varphi \in \set{0,1/2,1}$ such that
  $f[\varphi]g[\varphi]\neq h[\varphi]k[\varphi]$, then there exists
  an assignment $\varphi' \in \set{0,1/2,1}$ such that
  $0 \neq f[\varphi']g[\varphi']\neq h[\varphi']k[\varphi']\neq 0$.
\end{corollary}

\begin{proof}
  If $f[\varphi]\neq 0, g[\varphi]\neq 0$, and
  $h[\varphi], k[\varphi] \neq 0$, then we are done.	
  Otherwise, assume w.l.o.g. that $f[\varphi]=0$ while
  $h[\varphi], k[\varphi] \neq 0$.  Let $\theta$ be the following
  assignment: $\theta(X)=0$ if $\varphi(X)=0$, and $\theta(X)=1$.  
  We claim that $F[\theta]=0$. To see this, let $\theta'$ be $\theta$ except the variables  that are assigned $1/2$ in $\varphi$ remain unassigned. 
  We now observe that $F_{\theta'}(1/2,\dots,1/2)=c\cdot \#F_{\theta'}$ where $c>0$. Since $\#F_{\theta'}=0$  then no assignment to the variables in the residual formula $F_{\theta'}$ satisfies it, including the assignment that sets all of its variables to $1$.
  
  So, we have that $F[\theta]=0$, and
   $H[\theta]=K[\theta]=1$, implying
  that $F\wedge G \not\equiv H\wedge K$, and we can use
  Lemma~\ref{lemma:values:half:or:one}.
\end{proof}

When $F,G$ are equivalent to $H,K$ then obviously no $\varphi$ exists
that satisfies Lemma~\ref{lemma:values:half:or:one}.  But even when
$F,G$ are different from $H,K$, such a $\varphi$ may still not exists,
as illustrated by the following:

\begin{example} \label{ex:no:good:varphi}
Consider the following four functions:
  \begin{align*}
    F = & (X_1 \vee X_3) \vee (Y_1 \vee Y_3) \vee T &  H = & (X_1 \vee X_3) \vee (Y_2 \vee Y_3) \vee T \\
    G = & (X_2 \vee X_3) \vee (Y_2 \vee Y_3) \vee T &  K = & (X_2 \vee X_3) \vee (Y_1 \vee Y_3) \vee T\\
    T = & (Y_1\vee Y_2 \vee Y_3)(X_1\vee X_2 \vee X_3)
  \end{align*}
  They are distinct, yet for any $\varphi \in \set{1/2,1}^6$,
  $f[\varphi]g[\varphi]=h[\varphi]k[\varphi]$.  Indeed, if we set any
  variable $X_i$ to $1$ then $F=H$ and $G=K$, if we set any $Y_j$ to 1
  then $F=K$ and $G=H$.  Finally, if assign probabilities $1/2$ to all
  variables, then this also doesn't help because $\#F=\#G=\#H=\#K$ by
  symmetry.
\end{example}

%
%
%

\begin{lemma} \label{lemma:values:half:or:one:part2} Let $F, G, H, K$
  be monotone Boolean functions with $n$ variables, and let
  $U_1, \ldots, U_k$ be some Boolean variables with the property that
  $F[U_i:=1] \equiv G[U_i:=1] \equiv H[U_i:=1] \equiv K[U_i:=1]$, for
  all $i=1,k$.  Assume that there exists $\varphi_0 \in \set{0,1/2,1}$
  such that $f[\varphi_0]g[\varphi_0]\neq h[\varphi_0]k[\varphi_0]$.
  Then there exists $\varphi \in \set{0,1/2,1}$ such that
  $0 \neq f[\varphi]g[\varphi]\neq h[\varphi]k[\varphi]\neq 0$ and
  $\varphi(U_i) \in \set{1/2,1}$ for all $i=1,k$.
\end{lemma}

\begin{proof}
  The multilinear polynomial $f$ admits the following expansion in
  terms of the variables $u_1, \ldots, u_k$: 
  \begin{align*}
    f = & \sum_{\tau \in \set{0,1}^k} \mb u^\tau f[\tau]
  \end{align*}
  where
  $\mb u^\tau \defeq \prod_{i: \tau(u_i)=0}(1-u_i) \cdot \prod_{i:
    \tau(u_i)=1} u_i$.  By assumption, for every
  $\tau\neq (0,0,\ldots,0)$, we have
  $f[\tau]=g[\tau]=h[\tau]=k[\tau]$, thus we can write:
  \begin{align*}
    f = & \prod_i(1-u_i)f(0,\ldots,0) + \Delta & h = & \prod_i(1-u_i)h(0,\ldots,0) + \Delta\\
    g = & \prod_i(1-u_i)g(0,\ldots,0) + \Delta & k = & \prod_i(1-u_i)k(0,\ldots,0) + \Delta
  \end{align*}
  where $\Delta$ is the same quantity for all four polynomials.  

  We can assume w.l.o.g. that $F \wedge G \equiv H \wedge K$,
  otherwise the claim follows immediately from
  Lemma~\ref{lemma:values:half:or:one}.  We consider two cases.
  First, if $F \vee G \not\equiv H \vee K$.  In that case we can
  assume w.l.o.g. that there exists a truth assignment $\theta$ such
  that $F[\theta]=G[\theta]=H[\theta]=0$ and $K[\theta]=1$; notice that
  $\theta(U_i)=0$ for all $U_i$ (otherwise $F[\theta]=K[\theta]$).
  Define $\varphi$ as follow: $\varphi(U_i)=1/2$ and
  $\varphi(X) = \theta(X)$ when $X\not\in\set{U_1,\ldots,U_k}$. 
  We first note that since $K$ is monotonic, then $k[\varphi]=1$.  
   We
  compute $f[\varphi]$ using the formula above, noting that for
  $\tau=(0,0,\ldots,0)$, $f=0$, and for any other $\tau$,
  $f[\tau]=k[\tau]=1$ because $K[\theta]=1$.  Therefore,
  $v \defeq \prod(1-u_i) = 1/2^k$ and $\Delta = (2^k-1)/2^k = 1 - v$.
  We obtain:
  \begin{align*}
    f[\varphi] = & (1-v) > 0 & h[\varphi] = & (1-v) > 0 \\
    g[\varphi] = & (1-v) > 0 & k[\varphi] = & 1
  \end{align*}
  and we obtain $h[\varphi]k[\varphi]-f[\varphi]g[\varphi]>(1-v) -
  (1-v)^2 > 0$.

  Second, assume that both identities hold:
  $F \wedge G \equiv H \wedge K$ and $F \vee G \equiv H \vee K$.
  Then $f+g \equiv h+k$ because:
  \begin{align*}
    \pr(F \vee G) = & f+g-\pr(F \wedge G) = \pr(H \vee K) = h+k-\pr(H \wedge K)
  \end{align*}
  Let $\varphi_0$ be an assignment such that
  $f[\varphi_0]g[\varphi_0]\neq h[\varphi_0]k[\varphi_0]$.  Assume
  w.l.o.g. that $\varphi_0(U_i)=0$ for all $i$; otherwise, if
  $\varphi_0(U_i)=1/2$ then we simply remove the variable $U_i$ from
  the list $U_1,\ldots, U_k$ and decrease $k$. We define $\varphi$ as
  follows: $\varphi(U_i)=1/2$ for all variables $U_i$, and
  $\varphi(X) = \varphi_0(X)$ for all other variables.  To prove the
  claim of the lemma, let $f_0,g_0,h_0,k_0$ be the polynomials in the
  variables $u_1, \ldots, u_k$ obtained by substituting all other
  variables $X$ with the value $\varphi_0(X)$.  Thus,
  $f_0(0,0,\ldots,0)=f[\varphi_0]$, and similarly for $g,h,k$,
  therefore,
  \begin{align*}
    k[\varphi]h[\varphi]-f[\varphi]g[\varphi] = 
    & (\frac{1}{2^k}h_0(0,\ldots,0) + \Delta)(\frac{1}{2^k}k_0(0,\ldots,0) + \Delta)- (\frac{1}{2^k}f_0(0,\ldots,0) + \Delta)(\frac{1}{2^k}g_0(0,\ldots,0) + \Delta)\\
= & \frac{1}{2^{2k}}\left(h[\varphi_0]k[\varphi_0]-f[\varphi_0]g[\varphi_0]\right)+\frac{\Delta}{2^k}\left(h[\varphi_0]+k[\varphi_0]-f[\varphi_0]-g[\varphi_0]\right)\\
= & \frac{1}{2^{2k}}\left(h[\varphi_0]k[\varphi_0]-f[\varphi_0]g[\varphi_0]\right)\neq 0
  \end{align*}
  because $f+g=h+k$.  Also, $\Delta>0$, because at least one of the
  four values $f[\varphi_0], g[\varphi_0], h[\varphi_0], k[\varphi_0]$
  is $>0$, hence setting any $u_i:=1$ can only increase that value.
  This proves that $f[\varphi], g[\varphi], h[\varphi], k[\varphi]$
  are $\neq 0$.
\end{proof}

Next, we discuss tight connections between Boolean formulas, their
arithmetizations, and conditional independence statements, of possible
independent interest.

\begin{definition} \label{def:properties:of:boolean:functions} Fix a
  Boolean formula $F$.

  \begin{itemize}
  \item We say that $F$ is {\em connected} if, whenever
    $F \equiv F_1 \wedge F_2$ where $F_1, F_2$ do not share any common
    Boolean variables, then either $F_1 \equiv \texttt{true}$ or
    $F_2 \equiv \texttt{true}$.  Otherwise we say that $F$ decomposes
    into $F_1, F_2$.
  \item We say that $F$ {\em disconnects} two sets of variables
    $\mb U, \mb V$, if $F\equiv F_1 \wedge F_2$ where $F_1, F_2$ do
    not share any common Boolean variables,
    $\mb V \cap \vars(F_1) = \emptyset$,
    $\mb U \cap \vars(F_2) = \emptyset$.  Otherwise, we say that
    $\mb U, \mb V$ are connected.
  \item A Boolean variable $X$ in $F$ {\em disconnects}
    $\mb U, \mb V$ if both $F[X:=0]$ and $F[X:=1]$ disconnect
    $\mb U, \mb V$.
  \end{itemize}
\end{definition}

Notice that, if $F$ does not depend on $\mb U$, then it trivially
disconnects $\mb U,\mb V$, by writing $F = \texttt{true} \wedge F$.
We describe now the equivalent notions on multi-variate polynomials
$f$.

\begin{definition}  Let $f$ be a multivariate polynomial.
  \begin{itemize}
  \item We say that $f$ is {\em irreducible} if, whenever
    $f \equiv f_1 \cdot f_2$, then either $f_1$ or $f_2$ is a constant
    polynomial.
  \item We say that $f$ disconnects two sets of variables
    $\mb u,\mb v$ if $f \equiv f_1 \cdot f_2$ and
    $\mb v \cap \vars(f_1) = \mb u \cap \vars(f_2)=\emptyset$.
  \item We say that $x$ disconnects $\mb u,\mb v$ if both $f[x:=0]$
    and $f[x:=1]$ disconnect $\mb u,\mb v$.
  \end{itemize}
\end{definition}

We prove that, if the  arithmetization $f$ factorizes  $f = g\cdot h$,
then the associated Boolean function decomposes.

\begin{lemma} \label{lemma:from:polynomial:to:boolean} Let $F$ be a
  Boolean formula, and assume its arithmetization factorizes as
  $f = g \cdot h$, where both $g,h$ are non-constant multi-linear
  polynomials.  Then, there exists two Boolean formulas $G,H$ such
  that $F = G \wedge H$ and $g,h$ are, up to a constant, the
  arithmetization of $G,H$; in other words, there exists some constant
  $c \neq 0$ such that $\pr(G) = c\cdot g$ and $\pr(H)= h/c$.
\end{lemma}

\begin{proof}
  We can assume w.l.o.g. that , $F \not\equiv \texttt{false}$, and let
  $\theta$ be any assignment such that $F[\theta]=\texttt{true}$.
  Then $f[\theta] = g[\theta_1] h[\theta_2] = 1$, where
  $\theta_1, \theta_2$ are the restrictions of $\theta$ to $\vars(g)$
  and $\vars(h)$ respectively.  Denote $c=h[\theta_2]\neq 0$, then
  $g[\theta_2]=1/c$.  Define $G,H$ as follows.  $\vars(G)=\vars(g)$,
  $\vars(H)=\vars(h)$, and for any assignment $\theta_1'$ of
  $\vars(G)$, define $G[\theta_1'] \defeq c\cdot g[\theta_1']$
  (similarly for $H$).  We claim that this is correct, i.e.
  $c \cdot g[\theta_1'] \in \set{0,1}$.  Indeed, consider the full
  assignment $\theta' = \theta_1' \cup \theta_2$.  If
  $F[\theta']=\texttt{false}$ then
  $0 = f[\theta'] = g[\theta_1']h[\theta_2]=c \cdot g[\theta_1']$, and
  if $F[\theta']=\texttt{true}$ then
  $1 = f[\theta'] = g[\theta_1']h[\theta_2]=c \cdot g[\theta_1']$,
  proving the claim.  Thus, $\Pr(G)= c\cdot g$ and similarly
  $\Pr(H) = h/c$.
\end{proof}

The following are easy to check:

\begin{lemma}
  Let $F$ be a Boolean formula and $f$ be the multilinear polynomial
  representing its arithmetization.  Then:
  \begin{itemize}
  \item $F$ is connected iff $f$ is irreducible.
  \item $F$ disconnects $\mb U, \mb V$ iff $f$ disconnects
    $\mb u, \mb v$.
  \item $X$ disconnects $\mb U, \mb V$ iff $x$ disconnects
    $\mb u, \mb v$.
  \end{itemize}
\end{lemma}

In this paper we are concerned only with monotone Boolean formulas
$F$, which admit a unique, canonical CNF representation, where no
clause is redundant (i.e. is not a superset of some other clause).  Then,
connectedness can be viewed as a graph-theoretic property, since it is
equivalent to saying that there exists clauses $C_0, C_1, \ldots, C_k$
such that $\mb U \cap \vars(C_0)\neq \emptyset$,
$\mb V \cap \vars(C_k)\neq \emptyset$ and
$\vars(C_{i-1})\cap \vars(C_i) \neq \emptyset$ for $i=1,k$.  Define
the {\em distance}, $d(\mb U, \mb V)$, to be the minimum such $k$.
Notice that we may have $d(U, V)=0$ even if $U, V$ are single
variables.    Connectedness is also an algebraic property, and related
to polynomial factorization.  We will make use repeatedly of these
equivalent formulations.  

Fix a set of variables $\mb U$.  For any number $m \geq 0$, define the
ball $B(\mb U, m) \defeq \setof{Z}{d(\mb U, Z) \leq m}$.  The
following is easy to check:

\begin{lemma} \label{lemma:disconnect:balls} Fix a monotone Boolean
  formula $F$.
  \begin{enumerate}
  \item If $X$ disconnects $\mb U_1, \mb V$ and $X$ also disconnects
    $\mb U_2, \mb V$, then it disconnects $\mb U_1 \cup \mb U_2, \mb V$.
  \item \label{item:disconnect:balls} If $X$ disconnects
    $\mb U, \mb V$ and $m = d(\mb U, X)$, $n = d(\mb V, X)$, then $X$
    disconnects $B(\mb U, m-2)$ and $B(\mb V, n-2)$.
  \end{enumerate}
\end{lemma}

A third characterization uses conditional independence.  Let $\Pr(-)$
denote the distribution where each random variable $X$ is set to
$\texttt{true}$ independently, with probability $\Pr(X) = x$.  For a
fixed Boolean formula $F$, define $\Pr_F$ the probability space
$\Pr_F \defeq \Pr(-|F)$: that is, its outcomes are assignments that
satisfy $F$.  We write $\mb U \perp_F \mb V$ when $\mb U, \mb V$ are
independent in the probability space $\Pr_F$.  We prove the following:

\begin{lemma} \label{lemma:migration:independence}
  A Boolean variable $X$  disconnects $\mb U, \mb V$, iff
  $\mb U \perp_F \mb V | X$.
\end{lemma}

\begin{proof}
  We start by establishing the connection between $\Pr_F$ and the
  arithmetization $f$.  Consider some partial assignment $\theta$ of
  the variables in $F$ with values in $\set{0,1}$.  Then
  $\Pr_F(\theta)=f[\theta]/f$, where the polynomial $f$ is evaluated
  over values $x = \Pr(X=1), y=\Pr(Y=1), \ldots$ For example, if
  $F= X \vee Y$, then $\Pr(F)= f(x,y) = x+y-xy$ and
  $\Pr_F(X=1) = \Pr(X=1|F) = f(1,y)/f(x,y) = 1/(x+y-xy)$.  Assume
  first that $X$ disconnects $\mb U, \mb V$. Then $F= G \wedge H$ and
  $f=g\cdot h$, where
  $\mb u \subseteq \vars(g), \mb v \subseteq \vars(h)$.  For any
  assignments $\theta_1, \theta_2$ of the variables $\mb U, \mb V$
  respectively, with values $\set{0,1}$ we write $\mb U = \theta_1$
  for the event that the outcomes of $\mb U$ are those given by
  $\theta_1$.  Then we have
  $\Pr_F(\mb U=\theta_1, \mb V=\theta_2) = f[\mb U = \theta_1, \mb V =
  \theta_2]/f = g[\mb U = \theta_1]\cdot h[\mb V = \theta_2]/(g\cdot
  h) = \Pr_G(\mb U=\theta_1) \cdot \Pr_H(\mb V=\theta_2)$.
  Conversely, assume
  $\Pr_F(\mb U = \theta_1, \mb V = \theta_2) = \Pr_F(\mb U = \theta_1)
  \Pr_F(\mb V = \theta_2)$.  Then
  $f[\theta_1,\theta_2]/f = (f[\theta_1]/f) \cdot (f[\theta_2]/f)$,
  or, equivalently,
  $f[\theta_1,\theta_2] \cdot f = f[\theta_1] f[\theta_2]$.  By
  Theorem~\ref{th:rank:1:matrix:of:polynomials} there exists
  polynomials $g_0, g_1, h_0, h_1$ such that:
  \begin{align*}
    f \equiv & g_0 h_0 & f[\theta_1] \equiv & g_1 h_0 & f[\theta_2] \equiv & g_0 h_1  & f[\theta_1 \theta_2]\equiv & g_1 h_1
  \end{align*}
  Since $f$ is multi-linear, $\vars(g_0) \cap \vars(h_0)=\emptyset$.
  From the first two identities we conclude that
  $\mb u \subseteq \vars(g_0)$ (since, recall, $\theta_1$ assigns the
  variables $\mb u$ to $0,1$), and from identities one and three we
  conclude that $\mb v \subseteq \vars(h_0)$.  Thus, the factorization
  $f \equiv g_0 h_0$ disconnects $\mb u, \mb v$, which proves the
  lemma.
\end{proof}

Recall that, for every joint distribution of random variables, if
$U_1 U_2 \perp V | X$ then $U_1 \perp V | X$ and $U_2 \perp V | X$.
The converse does not hold in general, but it holds for $\perp_F$, by
Lemma~\ref{lemma:disconnect:balls}.

\begin{definition}
  Let $F$ be a monotone Boolean formula, where $X$ disconnects
  $\mb U, \mb V$, and let $Y$ be any other variable in $F$.  If $F$
  does not disconnect either $\mb UY, \mb V$, nor $\mb U, \mb V Y$,
  then we say that $Y$ is a {\em migrating} variable w.r.t.
  $X, \mb U, \mb V$.
\end{definition}

Suppose $X$ disconnects $\mb U, \mb V$ in $F$, in other words
$F[X=0] = G_0 \wedge H_0$ and $F[X=1] = G_1 \wedge H_1$, where
$\vars(G_0)\cap \vars(H_0)=\emptyset$,
$\vars(G_1)\cap \vars(H_1)=\emptyset$ and $\mb U$ occurs only in
$G_0, G_1$ while $\mb V$ occurs only in $H_0, H_1$.  Consider where
$Y$ occurs.  If it occurs in $G_0, G_1$ then $X$ separates $Y$ and
$\mb V$; if it occurs in $H_0, H_1$ then $X$ separates $\mb U$ and
$Y$.  If none of these hold, then we say that it is migrating.

\begin{lemma}
	\label{lemma:migrating:d}
  Assuming $m = d(\mb U, X)$, $n = d(\mb V, X)$, if $Y$ migrates
  w.r.t. $X, \mb U, \mb V$, then $d(\mb U,Y) \geq m-1$, and
  $d(\mb V, Y) \geq n-1$.
\end{lemma}

\begin{proof}
  Follows from Lemma~\ref{lemma:disconnect:balls},
  item~\ref{item:disconnect:balls}.
\end{proof}

\begin{example} \label{ex:migrating} Consider the following Boolean function:
  \begin{align*}
    F =  (U \vee Z_0) \wedge \underbrace{\overbrace{(Z_0 \vee Z_1 \vee Z_2 \vee Z_3)}^{\defeq C_1}\wedge
    \overbrace{(Z_3 \vee X \vee Y)}^{\defeq C_2} \wedge
    \overbrace{(X \vee Y \vee Z_4)}^{\defeq C_3} \wedge
    (X\vee Z_1) \wedge (Y \vee Z_2)}_{\defeq C}\wedge (Z_4 \vee V)
  \end{align*}
  $X$ disconnects $U,V$, and we have:
  \begin{align*}
    F[X:=0] = & \overbrace{(U \vee Z_0) \wedge Z_1}^{\defeq G_0} \wedge \overbrace{(Z_3 \vee Y) \wedge (Y \vee Z_4) \wedge (Y \vee Z_2) \wedge (Z_4 \vee V)}^{\defeq H_0}\\
    F[X:=1] = & \underbrace{(U \vee Z_0) \wedge (Z_0 \vee Z_1 \vee Z_2 \vee  Z_3) \wedge (Y \vee Z_2)}_{\defeq G_1} \wedge \underbrace{ (Z_4 \vee V)}_{\defeq H_1}
  \end{align*}
  Here $Y, Z_2$, and $Z_3$ migrate from the right to the left.
\end{example}

As we shall see, migrating variables add complexity to our proof.
However, we prove that the migration property is symmetric: if $X$
causes $Y$ to migrate, then $Y$ causes $X$ to migrate.  To prove this
we use the following result of independent interest.

\begin{lemma} \label{lemma:independence:binary}
  Let $X,Y,\mb U,\mb V$ be jointly distributed random variable, and
  assume that $Y$ is binary (i.e. it has only two outcomes).  Then,
  the following implication holds: if $(\mb U \perp \mb V | X)$ and
  $(\mb UX \perp \mb V | Y)$ then either $(\mb V \perp Y)$ or
  $(\mb U \perp Y | X)$.
\end{lemma}

This implication does not hold in general, but it holds when $Y$ is
binary.  A similar, but different example is given by Geiger and
Pearl~\cite{DBLP:journals/amai/GeigerP90}, in Corollary 8.

\begin{proof}
  Let $\pr(-)$ denote the joint distribution.  As usual we write
  $\pr(X)$ or $\pr(XY)$ etc, for the marginal distribution.  The first
  condition, $(\mb U \perp \mb V | X)$ says that
  $\pr(\mb UX) \cdot \pr(\mb VX) - \pr(\mb U\mb VX) \cdot \pr(X)=0$.
  We use the fact that $Y$ has only two outcomes, and expand each
  probability, using the formula
  $\pr(-) = \pr(-|Y=0)(1-y) + \pr(-|Y=1)y$, where $y \defeq \pr(Y=1)$,
  and further abbreviate $p_0(-) = \Pr(- | Y=0)$ and
  $p_1(-) = \Pr(- | Y=1)$, thus, for example
  $\pr(\mb U) = p_0(\mb U)(1-y)+p_1(\mb U)y$.  We also use the second
  condition, $(\mb UX \perp \mb V | Y)$, which implies
  $\pr(\mb UX\mb V|Y)= \pr(\mb UX|Y) \cdot \pr(\mb V|Y)$ and
  $\pr(\mb VX|Y)=\pr(\mb V|Y)\cdot\pr(X|Y)$, and derive:
  \begin{align*}
0 = &  \pr(\mb UX) \cdot \pr(\mb VX) -  \pr(\mb U\mb VX) \cdot \pr(X) \\
 = & \big(p_0(\mb UX)(1-y)+p_1(\mb UX)y\big)\cdot \big(p_0(\mb VX)(1-y)+p_1(\mb VX)y\big)
 -  \big(p_0(\mb U\mb VX)(1-y)+p_1(\mb U\mb VX)y\big)\cdot \big(p_0(X)(1-y)+p_1(X)y\big)\\
 = & \big(p_0(\mb UX)(1-y)+p_1(\mb UX)y\big)\cdot \big(p_0(\mb V)p_0(X)(1-y)+p_1(\mb V)p_1(X)y\big)
 -  \big(p_0(\mb UX)p_0(\mb V)(1-y)+p_1(\mb UX)p_1(\mb V)y\big)\cdot \big(p_0(X)(1-y)+p_1(X)y\big)
  \end{align*}
  We multiply out both products.  The term
  $p_0(\mb UX)p_0(\mb V)p_0(X)(1-y)^2$ occurs in both products,
  and cancels out, and similarly for the term
  $p_1(\mb UX)p_1(\mb V)p_1(X)y^2$, thus we obtain:
\begin{align*}
0 = & \big(p_0(\mb UX)p_1(\mb V)p_1(X)y(1-y) + p_1(\mb UX)p_0(\mb V)p_0(X)y(1-y)\big)
 -  \big(p_0(\mb UX)p_0(\mb V)p_1(X)y(1-y) + p_1(\mb UX)p_1(\mb V)p_0(X)y(1-y)\big)\\
= &  \big(p_0(\mb UX)p_1(X)-p_1(\mb UX)p_0(X)\big)
\cdot
\big(p_1(\mb V)-p_0(\mb V)\big) y (1-y)
\end{align*}
We can assume w.l.o.g. that $y(1-y)\neq 0$.  If
$(p_1(\mb V)-p_0(\mb V)) = 0$ then $\mb V \perp Y$.  Suppose
$p_1(\mb UX)p_0(X)-p_0(\mb UX)p_1(X)=0$, then:
\begin{align*}
  \frac{p_1(\mb UX)}{p_1(X)} = &   \frac{p_0(\mb UX)}{p_0(X)}
\end{align*}
This is equivalent to $\pr(\mb U|X, Y=1) = \pr(\mb U |X, Y=0)$, or
$\mb U\perp Y |X$.
\end{proof}

We prove:

\begin{corollary} \label{cor:migrating:symmetric}
  Let $F$ be a monotone, connected Boolean formula.  Suppose $X$
  disconnects $\mb U,\mb V$, and $Y$ also disconnects $\mb U,\mb V$.
  Then $Y$ is migrating w.r.t. $X, \mb U, \mb V$ iff $X$ is migrating
  w.r.t. $Y, \mb U, \mb V$.
\end{corollary}

\begin{proof} We prove the counterpositive: if $X$ does not migrate
  w.r.t. $Y, \mb U, \mb V$, then $Y$ does not migrate w.r.t.
  $X, \mb U, \mb V$.  Since $X$ disconnects $\mb U, \mb V$, then, by
  Lemma~\ref{lemma:migration:independence}, we have
  $\mb U\perp_F \mb V | X$.  If $X$ does not migrate w.r.t.
  $Y, \mb U, \mb V$, then $Y$ disconnects either $\mb U X, \mb V$ or
  $\mb U, \mb V X$.  Assuming the former, we have
  $\mb UX \perp_F \mb V | Y$, Lemma~\ref{lemma:independence:binary}
  implies that either $(\mb V \perp_F Y)$ or $(\mb U \perp_F Y | X)$
  holds.  The first is not possible because $F$ is connected, hence we
  have $(\mb U \perp_F \mb Y | X)$.  Then we also have
  $(\mb U \perp_F \mb V Y | X)$, proving that $Y$ does not migrate.
\end{proof}

\section{Proof of Theorem~\ref{th:dichotomy:main:lemma} (2)}

In this section we prove Theorem~\ref{th:dichotomy:main:lemma} (2): if
$Q$ is a bipartite, unsafe query of type II of length $\geq 5$, then
$\#PP2CNF \leq^P \gmc_{\text{bi}}(Q)$.  For type I queries we proved
that all probabilities in a block can be set to 1/2, and therefore we
obtained a symmetry (more precisely, we obtained $y_{01}=y_{10}$),
allowing us to reduce from \#P2CNF.  For type II queries we need to
use all three probability values $0,1/2,1$ and may not have
symmetries.  This makes the proof a bit more complicated, and our
reduction will be from \#PP2CNF rather than \#P2CNF.

Recall the definition of a bipartite query $Q$ in
Def.~\ref{def:types:of:clauses}.  In this section $Q$ is of type
II-II, and here we simply call it of type II.  Recall the definition
of a left-rigth path from Def.~\ref{def:unsafeQueries}: it is a
sequence of clauses $C_0, C_1, \ldots, C_k$ such that $C_0$ is a left
clause, $C_k$ is a right clause, and every consecutive clauses share a
common relational symbol.  We will use repeatedly this simple fact:

\begin{lemma} \label{lemma:everything:on:path}
  Fix a left-to-right path $C_0, \ldots, C_k$.  If $Q$ is a final
  query, then every symbol occurring in $Q$ must also occur in the
  left-to-right path.
\end{lemma}

Indeed, if a symbol $S$ does not occur on the path, then the query
$Q[S:=1]$ still contains the path $C_0, \ldots, C_k$, since none of
these clauses contained $S$, and therefore $Q[S:=1]$ is an unsafe
query, contradicting the assumption that it is final.

\subsection{The Coloring Count Problem}
\label{sec:CCP}

For Type I queries, our reduction was not directly from \#P2CNF, but
from a more general problem, asking for all signature counts.  Here,
too, we need to consider a more general problem, which we define
formally, and call it the Coloring Count Problem.

Fix two numbers $m \geq 2, n \geq 2$.  For every bipartite graph
$G = (U, V, E \subseteq U \times V)$, a {\em coloring} is a pair of
functions $\sigma : U \rightarrow [m]$ and $\tau : V \rightarrow [n]$.
The coloring associates a color to each node.  Let
$M \defeq \max(|U|,|V|,|E|)$.  Given a coloring $\sigma, \tau$, we
denote by $k_{\alpha \beta}(\sigma, \tau)$ the number of edges whose
endpoints are colored with $\alpha$ and $\beta$ respectively; denote
by $k_{\alpha,\hat 1}(\sigma,\tau)$ the number of nodes $u$ colored
$\alpha$, and by $k_{\hat 1,\beta}(\sigma,\tau)$ the number of nodes
colored $\beta$.  (The choice for the notation $\hat 1$ will be come
clear in the next section.)  All these numbers are $\leq M$.  Taking
together, these numbers form the {\em signature} of a coloring
$\sigma, \tau$, which is a mapping
$\mb{k}(\sigma, \tau) : ([m] \cup \set{\hat 1}) \times ([n] \cup
\set{\hat 1}) \rightarrow \set{0,\ldots,M}$ where
$k_{\hat 1,\hat 1}(\sigma, \tau)\eqdef 0$; equivalently, the signature
is a vector with $(m+1)(n+1)$ dimensions and with values in $[0,M]$,
i.e.  $\in \set{0,\ldots,M}^{(m+1)(n+1)}$.  Formally:
\begin{align*}
\forall \alpha\in [m], \forall \beta \in [n]: 
&& k_{\alpha \beta}(\sigma, \tau) \defeq & |\setof{(u,v) \in E}{\sigma(u) = \alpha, \tau(v) = \beta}|\\
&& k_{\alpha \hat 1}(\sigma, \tau) \defeq & |\setof{u \in U}{\sigma(u) = \alpha}|\\
&& k_{\hat 1,\beta}(\sigma, \tau) \defeq & |\setof{v \in V}{\tau(v) = \beta}|\\
\mbox{{\em Signature} of }  \sigma, \tau:
&& \mb{k}(\sigma, \tau)  \defeq & (k_{\alpha, \beta}(\sigma,\tau))_{\alpha \in [m]\cup\set{\hat 1}, \beta\in [n]\cup\set{\hat 1}}\\
\end{align*}

Conversely, given such a vector
$\mb{k} :
([m]\cup\set{\hat 1})\times([n]\cup\set{\hat 1})\rightarrow\set{0,\ldots,M}$,
its {\em coloring count}, $\#\mb{k}$, is the number of colorings
$\sigma, \tau$ with signature $\mb{k}$:
\begin{align*}
\forall \mb{k} \in \set{0,\ldots,M}^{(m+1)(n+1)}:
&& \#\mb{k} \defeq & |\setof{(\sigma, \tau)}{\mb{k}(\sigma,\tau)=\mb{k}}|
\end{align*}

\begin{definition}[Coloring Counting Problem] \label{def:CCP}
	We assume
	$m \geq 2, n \geq 2$ to be fixed.  The Coloring Counting Problem,
	$\ccp(m,n)$, is the following: given a bipartite graph
	$(U,V,E\subseteq U \times V)$, compute all coloring counts:
	$\setof{\#\mb{k}}{\mb{k} : ([m]\cup \set{\hat 1})\times ([n]\cup\set{\hat 1})
		\rightarrow \set{0,\ldots,M}}$, where $M=\max(|U|,|V|,|E|)$
\end{definition}

Notice that the number of coloring counts is $(M+1)^{(m+1)(n+1)}$,
hence, under the assumption that $m$ and $n$ are constant (in other words, $m, n= O(1)$), then the size of the output is polynomial in the size of the graph.

\begin{theorem} \label{th:ccp:hard}
	For all $m, n \geq 2$, $\ccp(m,n)$ is \#P-hard.
\end{theorem}

\begin{proof}
	Assuming we have an oracle for $\ccp(m,n)$, we describe a PTIME
	algorithm for computing \#PP2CNF.  Let
	$\Phi = \bigwedge_{(u,v) \in E}(X_u \vee Y_v)$ be a PP2CNF, where
	$E \subseteq U \times V$, and let $M = |E|$.  Use the oracle to
	compute all coloring counts $\#\mb{k}$, for all vectors $\mb{k}$.
	Call a coloring $\sigma, \tau$ {\em valid} if it uses only two
	colors, i.e. $\sigma(U) \subseteq \set{1,2}$ and
	$\tau(V) \subseteq \set{1,2}$, and call a signature $\mb{k}$ {\em
		valid} if $k_{\alpha\beta}=0$ whenever $\alpha >2$ or $\beta >2$.
	A valid coloring corresponds to a truth assignment, e.g. by
	associating \texttt{false} to color 1 and \texttt{true} to color 2.
	The assignment satisfies $\Phi$ iff its signature satisfies
	$k_{11}=0$.  Thus, $\#\Phi$ is the sum of $\#\mb{k}$ over all valid,
	satisfying signatures $\mb{k}$. 
\end{proof}

For example, assume we have 3 colors for the left, $a,b,c$ and 3
colors for the right, $u,v,w$.  Then one can think of a vector $\mb k$
as a $4 \times 4$ matrix:
\begin{align*}
  \mb k = &
\left[
 \begin{array}{cccc}
  k_{au} & k_{av} & k_{aw} & k_{a\hat 1}\\
  k_{bu} & k_{bv} & k_{bw} & k_{b\hat 1}\\
  k_{cu} & k_{cv} & k_{cw} & k_{c\hat 1}\\
  k_{\hat 1u} & k_{\hat 1v} & k_{\hat 1w} & k_{\hat 1\hat 1}\\
 \end{array}
\right]
\end{align*}
Given a bipartite graph with $M=|E|$ edges (we assume no isolated
vertices), there are $M^{16}$ matrices $\mb k$.  An oracle for the
$\ccp(3,3)$ will compute all $M^{16}$ numbers $\#\mb k$.  To answer
the \#PP2CNF formula, we consider only colorings that use the colors
$a,b$ on the left, and $u,v$ on the right, e.g.
$k_{cu}=k_{cv}=\ldots=0$.  Using the counts $\#\mb k$ for these
matrices, we can obtain \#PP2CNF.

In the rest of this appendix we prove:

\begin{theorem} \label{th:ccp:type2}
  if $Q$ is a bipartite, unsafe query of type II-II of length
  $\geq 5$, then there exists numbers $m,n$ that depend only on $Q$
  such that $\ccp(m,n) \leq^P \gmc_{\text{bi}}(Q)$
\end{theorem}

Theorem~\ref{th:ccp:hard} and Theorem~\ref{th:ccp:type2} prove
Theorem~\ref{th:dichotomy:main:lemma} (2).  In the rest of this
appendix we prove Theorem~\ref{th:ccp:type2}.

\subsection{M\"obius Inversion formula for Type II Queries}

The principle behind the formula for Type I queries was the Shannon
expansion formula: $\pr(F) = \pr(F[X:=0])(1-p) + \pr(F[X:=1])p$, where
$X$ is a boolean variable and $p= \Pr(X)$.  Essentially, we applied
this formula repeatedly, once for each Boolean variable associated to
a unary atom $R(u)$.  For a Type II query, we no longer have unary
atoms.  Instead, we will use a different principle: the
inclusion/exclusion formula,
$\pr(F \vee G) = \pr(F) + \pr(G) - \pr(F \wedge G)$.  We start by
expressing $Q$ as a union, so we can enable the inclusion exclusion
formula.

Recall the definition of a bipartite query $Q$ in
Def.~\ref{def:types:of:clauses}.   $Q$ is of type II-II, and here we
simply call it of type II.  We will rewrite $Q$ as follows:
\begin{align}
Q = & Q_{\text{left}} \wedge Q_{\text{middle}} \wedge Q_{\text{right}} \label{eq:q:left:middle:right} \\
Q_{\text{left}} = & \forall x (\forall y G_1(x,y) \vee \ldots \vee \forall y G_m(x,y)) \defeq \forall x G(x) \label{eq:q:left}\\
Q_{\text{middle}} = & \forall x \forall y C(x,y)\label{eq:q:middle} \\
Q_{\text{right}} = & \forall y(\forall x H_1(x,y) \vee \ldots \vee \forall x H_n(x,y)) \defeq \forall y H(y) \label{eq:q:right}
\end{align}
where $m \geq 2, n \geq 2$, and $G_i(x,y)$, $C(x,y)$, $H_j(x,y)$ are
CNF formulas, i.e. conjunctions of one or more clauses.  Indeed, each
type II query can be written this way, because $Q_{\text{left}}$ is a
conjunction of left clauses, each of the form
$\forall x \left(\bigvee_{\ell=1}^m \forall y S_{J_\ell}(x,y)\right)$,
hence we obtain \eqref{eq:q:left} by distributing $\wedge$ over
$\vee$, in other words converting from CNF to DNF.  We apply similar
reason to $Q_{\text{right}}$ and obtain \eqref{eq:q:right}

\begin{example}
	\label{ex:QLeft}
 We illustrate this transformation on a
  $Q_{\text{left}}$ with two left clauses:
  \begin{align*}
    Q_{\text{left}} = & \underbrace{\forall x(\forall y (S_1(x,y) \vee S_2(x,y)) \vee (\forall y (S_1(x,y) \vee S_3(x,y))))}_{\mbox{left clause 1}}
    \wedge \underbrace{\forall x((\forall y (S_1(x,y))) \vee \forall y (S_2(x,y) \vee S_3(x,y)))}_{\mbox{left clause 2}}\\
    =  & \forall x(\forall y(\underbrace{S_1(x,y)}_{\defeq G_1(x,y)}) \vee  \forall y(\underbrace{(S_1(x,y) \vee S_2(x,y))\wedge(S_2(x,y) {\vee} S_3(x,y))}_{\defeq G_2(x,y)})  \vee \forall y(\underbrace{(S_1(x,y) \vee S_3(x,y))\wedge(S_2(x,y) \vee S_3(x,y))}_{\defeq G_3(x,y)})) \\
    = & \forall x( \forall y G_1(x,y) \vee \forall y G_2(x,y) \vee \forall y G_3(x,y))  
  \end{align*}
\end{example}

Thus, we can write $Q$ as:
\begin{align}
Q = (\forall x \bigvee_i \forall y G_i) \wedge (\forall x  \forall y C) \wedge (\forall y \bigvee_j \forall x H_j)\label{eq:q:gi:hj}
\end{align}
which gets us closer to using the inclusion/exclusion formula.  Here,
each $G_i(x,y)$ and every $H_j(x,y)$ is a CNF formula, i.e. a
conjunction of clauses.  At this point we observe that some of the
terms in the inclusion/exclusion formula can be logically equivalent.
For example, given three Boolean formulas $F_1, F_2, F_3$, the
expansion of $\pr(F_1 \vee F_2 \vee F_3)$ has 7 terms, but some may be
logically equivalent, e.g. we may have
$F_1 \wedge F_2 \equiv F_1 \wedge F_2 \wedge F_3$.  In our proof it is
important to ensure that all terms are logically in-equivalent (we
make this precise in Lemma~\ref{lemma:simple:lattices:1} below), and
for that reason we consider next the lattice consisting of all
logically inequivalent conjunctions:

\begin{definition}\label{def:lattice}
  Let $\mb{F} = \set{F_1, \ldots, F_m}$ be a set of formulas; For each
  set $\alpha \subseteq [m]$ denote by
  $F_\alpha \defeq \bigwedge_{i \in \alpha} F_i$.  The {\em closure}
  of $\alpha$ is:
  $\bar \alpha \defeq \setof{i}{F_\alpha \Rightarrow F_i}$.  A set
  $\alpha$ is closed if $\alpha = \bar \alpha$.  The {\em lattice
    associated with $\mb{F}$} denoted $\hat L(\mb{F})$ consists of all
  closed sets ordered by reverse set inclusion
  $\bar \alpha \leq \bar \beta$ if $\bar \beta \subseteq \bar \alpha$.
  We denote by $\hat 1$ the top element of $\hat L(\mb{F})$ (that is,
  $\hat 1 = \emptyset$), and define
  $F_{\hat 1} \defeq F_1 \vee \ldots \vee F_m$.  The {\em M\"obius} function
  $\mu :\hat L(\mb{F}) \rightarrow \Z$ is defined as $\mu(\hat 1)=1$,
  $\mu(\alpha) = - \sum_{\beta > \alpha}\mu(\beta)$.  The {\em support
    of $\mb{F}$} is
  $L(\mb{F}) \defeq \setof{\alpha \in \hat L(\mb{F})}{\mu(\alpha)\neq
    0}$.
\end{definition}

Intuitively, the lattice is obtained as follows. Compute all $2^m$
conjunctions $F_\alpha$, then group them into equivalence classes
based on logical equivalence.  The lattice consists of all equivalence
classes.  In each class there exists a maximal $\alpha$ such that
$F_\alpha$ is in that class, and this $\alpha$ is closed; we use it as
representative for the class.  By convention, the top element of the
lattice is defined as $F_{\hat 1} = F_1 \vee \cdots \vee F_m$; this is
consistent to what we need in the rest of this section, and also
standard in the context of the M\"obius function,
see~\cite[pp.117]{10.5555/2124415}.  The {\em support} $L(\mb{F})$ is
obtained by removing all elements $\alpha$ where $\mu(\alpha)=0$.

The M\"obius inversion formula generalizes inclusion/exclusion, by
combining equivalent terms.  More precisely, if
$\mb{Y} = \set{Y_1, \ldots, Y_m}$ is a set of $m$ formulas, then,
M\"obius' inversion formula is:
\begin{align*}
\Pr(Y_1 \vee \ldots \vee Y_m) = & -\sum_{\alpha < \hat 1}\mu(\alpha)\Pr(Y_\alpha)
\end{align*}
Obviously it suffices to sum only over the support, less $\hat 1$,
i.e. $\alpha \in L(\mb F) - \set{\hat 1}$, a fact that we will exploit
later.  As before, we write $y$ for the arithmetization of $Y$,
i.e. the probability above expressed in terms of the probabilities of
the Boolean variables, $z_1, z_2, \ldots$ Then:
\begin{align*}
\pr(y_1 \vee \ldots \vee y_m) = & - \sum_{\alpha < \hat 1}\mu(\alpha)\cdot y_\alpha
\end{align*}

\begin{example}
	Consider the following set $\mb{Y}=\set{Y_1, Y_2, Y_3}$:
	\begin{align*}
	Y_1 = & Z_1 Z_2 & Y_2 = & Z_1 Z_3 & Y_3 = & Z_2 Z_3
	\end{align*}
	where $Z_1 Z_2$ means $Z_1 \wedge Z_2$.  We notice that
        $Y_{12} = Y_{13} = Y_{23} = Y_{123} = Z_1 Z_2 Z_3$ and
        therefore $\hat L(\mb{Y}) = \set{\emptyset,1,2,3,123}$. The
        M\"obius function is $\mu(\emptyset)=1$,
        $\mu(1)=\mu(2)=\mu(3)=-1$, $\mu(123) = 2$, thus the support is
        $\set{\emptyset,1,2,3,123}$, and we obtain:
	\begin{align*}
	\Pr(Z_1Z_2 \vee Z_1Z_3 \vee Z_2 Z_3) = & \Pr(Y_1) + \Pr(Y_2) + \Pr(Y_3) - 2\Pr(Y_1Y_2Y_3)
	\end{align*}
	For another example, consider $\mb{Y}=\set{Y_1, Y_2, Y_3}$ where:
	\begin{align*}
	Y_1 = & Z_1 Z_2 & Y_2 = & Z_2 Z_3 & Y_3 = & Z_3 Z_4
	\end{align*}
	Then $\hat L(\mb{Y}) = \set{\emptyset,1,2,3,12,23,123}$, and
        $\mu(\emptyset)=1$, $\mu(1)=\mu(2)=\mu(3)=-1$,
        $\mu(12)=\mu(23)=1$, $\mu(123)=0$. The support consists of
        $\emptyset, 1,2,3,12,23$ and thus:
	\begin{align*}
	\Pr(Y_1 \vee Y_2 \vee Y_3) =& \Pr(Y_1) + \Pr(Y_2) + \Pr(Y_3) - \Pr(Y_1Y_2) - \Pr(Y_2Y_3)
	\end{align*}
\end{example}

Given a bipartite query $Q$ of type II, we denote the following sets
of formulas:
\begin{align*}
\mb G \defeq & \set{G_1(x,y) \wedge C(x,y), \ldots, G_m(x,y) \wedge  C(x,y)}
& \mb H \defeq & \set{C(x,y) \wedge H_1(x,y), \ldots, C(x,y) \wedge  H_n(x,y)}
\end{align*}
where $G_1, \ldots, G_m, H_1, \ldots, H_n$ are the CNF formulas that
occur in~\eqref{eq:q:left}, \eqref{eq:q:middle},
and~\eqref{eq:q:right} respectively.  We define two lattices:

\begin{definition}
  Fix a bipartite, unsafe query query $Q$ of type II, as in
  Eq.\eqref{eq:q:left:middle:right}.  The {\em left and right lattice
    supports} of $Q$ are $L(\mb{G})$ and $L(\mb{H})$, where $\mb{G}$
  and $\mb{H}$ are the sets defined above.  The {\em strict} supports
  are $L_0(\mb G) \defeq L(\mb G) - \set{\hat 1}$ and
  $L_0(\mb H) \defeq L(\mb H) - \set{\hat 1}$, and we denote by
  $\bar m = |L_0(\mb{G})|$ and $\bar n = |L_0(\mb{H})|$ their sizes.
  Notice that $3 \leq \bar m \leq 2^m-1$ and
  $3 \leq \bar n \leq 2^n-1$.
\end{definition}

We define the following, where $\alpha  \in L(\mb G)$ and $\beta  \in L(\mb H)$:

\begin{align}
  G_\alpha(x) \defeq & \forall y G_\alpha(x, y)  \label{eq:g:alpha:u} \\
  H_\beta(y) \defeq & \forall x H_\beta(x, y)  \label{eq:h:beta:v}\\
  Q_{\alpha \beta}(x,y) \defeq & G_\alpha(x) \wedge Q \wedge  H_\beta(y)
\label{eq:q:alpha:beta}
\end{align}

We notice that, if $\alpha, \beta \neq \hat 1$, then
$\forall x \forall y Q_{\alpha \beta}(x,y)$ is equivalent to the
following expression (recall that
$Q_{\text{middle}}=\forall x \forall y C(x,y)$).
\begin{align}
  \forall x \forall y Q_{\alpha \beta}(x,y) = & \forall x \forall y (G_\alpha(x,y) \wedge C(x,y) \wedge H_\beta(x,y))
 \label{eq:q:alpha:beta:2}
\end{align}
because, for every $\alpha \neq \hat 1$,
$\forall x G_\alpha(x) \Rightarrow Q_{\texttt{left}}$.  Indeed, each
CNF expression $G_i$ in~\eqref{eq:q:gi:hj} is a conjunction of
subclauses $S_{J_k}$, one from each left clause.  Therefore, for any
left clause $\forall x (\bigvee_k \forall y S_{J_k}(x,y))$ in
$Q_{\text{left}}$, the logical implication
$G_i(x) \Rightarrow \bigvee_k \forall y S_{J_k}(x,y)$ holds, and
therefore,
$\forall x G_\alpha(x) \wedge Q_{\text{left}} \equiv \forall x
G_\alpha(x)$ (since $\alpha \neq \emptyset$).  On the other hand, if
$\alpha = \hat 1$, then:
\begin{align}
  \forall x \forall y Q_{\hat 1 \beta}(x,y) = & Q \wedge \forall y H_\beta(y)
&  \forall x \forall y Q_{\alpha \hat 1}(x,y) = & \forall x G_\alpha(x) \wedge Q
\label{eq:q:alpha:1:2}
\end{align}
because $Q_{\texttt{left}} \Rightarrow \forall x G_{\hat 1}(x)$ and
$Q_{\texttt{right}} \Rightarrow \forall y H_{\hat 1}(y)$.

\begin{example} \label{ex:running:query:type:2} Consider:
  \begin{align*}
    Q = & \forall x (\underbrace{\forall y S_1(x,y)}_{G_1(x,y)} \vee \underbrace{\forall y S_2(x,y)}_{G_2(x,y)})
\wedge \forall x \forall y (S_1(x,y) \vee  S_3(x,y))
\wedge \forall y (\underbrace{\forall x S_3(x,y)}_{H_1(x,y)} \vee \underbrace{\forall x S_4(x,y)}_{H_2(x,y)})
  \end{align*}
Then:
\begin{align*}
  G_1(x) = & \forall y S_1(x,y) & G_2(x) = & \forall y S_2(x,y) & G_{12}(x) = & \forall y (S_1(x,y)\wedge S_2(x,y))& G_{\hat 1}(x) = & \forall y (S_1(x,y) \vee S_2(x,y))\\
  H_1(y) = & \forall x S_3(x,y) & H_2(y) = & \forall x S_4(x,y) & H_{12}(x) = & \forall x (S_3(x,y)\wedge S_4(x,y)) & H_{\hat 1}(y) = & \forall x (S_3(x,y) \vee  S_4(x,y))
\end{align*}
We show now a few examples of $Q_{\alpha \beta}$:
\begin{align*}
  \forall x \forall y Q_{1,1}(x,y) = & \forall x \forall y \left(S_1(x,y) \wedge S_3(x,y)\right) 
& \forall x \forall y Q_{1,2}(x,y) = & \forall x \forall y \left(S_1(x,y) \wedge S_4(x,y)\right) \\
\forall x \forall y Q_{1,12}(x,y) = & \forall x \forall y \left(S_1(x,y) \wedge S_3(x,y) \wedge S_4(x,y)\right)
& \forall x \forall y Q_{2,2}(x,y) = & \forall x \forall y \left(S_2(x,y)\wedge (S_1(x,y) \vee S_3(x,y)) \wedge S_4(x,y)\right)
\end{align*}
Notice that the middle clause became redudant in all queries except
$Q_{2,2}$.
\end{example}

\subsection{Forbidden Queries of Type II}

In our hardness proof we need all queries $Q_{\alpha \beta}$ to be
connected: in example~\ref{ex:running:query:type:2} {\em none} of
these queries is connected.  To ensure this property, we need to
restrict our queries to a strict subclass of final queries of Type II,
which are called {\em forbidden queries} in
~\cite{DBLP:journals/jacm/DalviS12}.  Every final query of type II can
be simplified to a forbidden query, hence it suffices to prove
hardness for forbidden queries.  In this section we give the formal
definition of forbidden queries, prove the connectedness property,
then prove that every final query of type II can be simplified to a
forbidden query.  The key results in this section are based
on~\cite{DBLP:journals/jacm/DalviS12}.  Here we expand those results
and give a complete characterization of the forbidden queries.  Our
setting here is slightly simpler than that
in~\cite{DBLP:journals/jacm/DalviS12}, because we assume that the
bipartite query $Q$ is long, i.e. the shortest left-right path has
length $k \geq 2$.

Recall that a left clause is a union of subclauses
$C(x) = \forall y S_{J_1}(x,y) \vee \forall y S_{J_2}(x,y) \vee
\cdots$ A binary symbol $U(x,y)$ is {\em $C$-ubiquitous} if it occurs
in all its subclauses $S_{J_1}, S_{J_2}, \ldots$ A binary symbol
$U(x,y)$ is {\em left ubiquitous} if it is $C$-ubiquitous for all left
clauses $C$.  We define similarly right ubiquitous symbols, and denote
them with $V$.

\begin{definition} \label{def:forbidden}
  Let $Q$ be a query of type II.  $Q$ is called a {\em forbidden
    query} if it is a final query, and, for every left-right path
  $C_0, C_1, \ldots, C_k$ of minimal length, every symbol in $C_0$ is
  either ubiquitous, or occurs in $C_1$; similarly, every symbol in
  $C_k$ is either ubiquitous or occurs in $C_{k-1}$.
\end{definition}

The following was shown in~\cite{DBLP:journals/jacm/DalviS12}, and we
included its proof here, later in this section:

\begin{theorem} \label{th:forbidden} Let $Q$ be a final query of type
  II of length $k$, where $k\geq 2$.  Then there exists a query $Q'$
  such that $\gmc_{\text{bi}}(Q') \leq^P_m \gmc_{\text{bi}}(Q)$ and
  $Q'$ is either of type I, or $Q'$ is a forbidden query of type II of
  length $\geq k$.
\end{theorem}

We also give a complete syntactic characterization of the forbidden
queries, which is novel:

\begin{theorem}\label{th:forbidden:syntax}
  Let
  $Q=Q_{\text{left}} \wedge Q_{\text{middle}} \wedge Q_{\text{right}}$
  be a forbidden query, and let
  $\mb U=\set{U_1, \ldots, U_r}, \mb V = \set{V_1, \ldots, V_t}$
  denote the sets of left-ubiquitous and right-ubiquitous symbols
  respectively, and let $C_0, C_1, \ldots, C_k$ be a left-to-right
  path of minimal length.  Then, every clause in $Q$ has one of the
  following forms:
  \begin{itemize}
  \item A left clause, $\bigvee_i \forall y S_{J_i}(x,y)$, where each
    subclause $S_{J_i}$ is:
    \begin{align*}
      U_1(x,y) \vee \cdots \vee U_r(x,y) \vee S_{J_1}(x,y)\vee S_{J_2}(x,y) \vee\cdots
    \end{align*}
   where $S_{J_1}, S_{J_2}, \ldots \subseteq \symb(C_1)$.  Notice that the left clause $C_0$ 
   contains all left ubiquitous symbols.
 \item The middle clause $C_1$ does not contain any left-ubiquitous
   symbol.
 \item If $C$ is any middle clause that contains a left ubiquitous
   symbol, then $\symb(C) \subseteq \symb(C_0) \cup \symb(C_1)$.
 \item A middle clause without any ubiquitous symbols, of the form
   $S_{j_1}(x,y) \vee S_{j_2}(x,y) \vee \cdots$.
 \item Symmetric clauses on the right.
 \end{itemize}

 Furthermore, if $r > 1$, then for each $i=1,\ldots,r$ there exists at
 least one middle clause that contains $U_i$ and no other left
 ubiquitous symbol.  Similarly, on the right.
\end{theorem}

\begin{example} \label{ex:running:query:type:2:cont:2} We illustrate
  here a simple forbidden query:
  \begin{align*}
&\forall x \big(\forall y (U(x,y) \vee S_1(x,y)) \vee \forall y (U(x,y) \vee S_2(x,y))\big)
&\wedge & \forall x \forall y \big(S_1(x,y) \vee S_2(x,y) \vee S_3(x,y) \vee S_4(x,y)\big)
&\wedge & \forall y \big(\forall x (V(x,y) \vee S_3(x,y)) \vee \forall x (V(x,y)  \vee S_4(x,y))\big)
  \end{align*}
  Here $U$ is a left-ubiquitous symbol, and $V$ is a right-ubiquitous
  symbol.
\end{example}

\begin{example} \label{ex:many:ubiquitous:symbols} The ubiquitous
  symbols need not be unique, and may occur in middle clauses.
  Consider:
  \begin{align*}
    Q = & \forall x (\forall y (U(x,y) \vee U'(x,y) \vee S_1(x,y) \vee   S_2(x,y))\vee \forall y (U(x,y) \vee U'(x,y) \vee S_2(x,y)\vee S_3(x,y)) \vee \forall y (U(x,y) \vee U'(x,y) \vee S_1(x,y) \vee S_3(x,y))) \\
  \wedge & \forall x \forall y (S_1(x,y) \vee S_2(x,y) \vee S_3(x,y) \vee S_4(x,y) \vee S_5(x,y)) \\
\wedge & \forall y (\forall x (V(x,y) \vee S_4(x,y)) \vee \forall x  (V(x,y) \vee S_5(x,y)))\\
\wedge & \forall x \forall y (U(x,y) \vee S_1(x,y)\vee S_2(x,y) \vee S_3(x,y))\wedge \forall x  \forall y (U'(x,y) \vee S_1(x,y) \vee S_2(x,y) \vee S_3(x,y))
  \end{align*}
%
\end{example}

Before we prove the two theorems, we show the two consequences that
will need later in this paper.

\begin{lemma}[Connected] \label{lemma:simple:lattices:2} Let $Q$ be a forbidden
  query.  Then
  $\forall \alpha \in L(\mb G), \forall \beta \in L(\mb H)$, the
  queries $\forall x \forall y Q_{\alpha \beta}(x,y)$ are connected,
  and depend on all relational symbols in $Q$.
\end{lemma}

\begin{proof}
  Fix a left-to-right path $C_0, C_1, \ldots, C_k$ of minimal length;
  by Lemma~\ref{lemma:everything:on:path}, all symbols in $Q$ occur on
  this path.  Referring to the expressions in
  Eq.~\eqref{eq:q:alpha:beta:2}, all clauses $C_1, \ldots, C_{k-1}$
  occur in $C(x,y)$.  Assume first that $\alpha, \beta \neq \hat 1$,
  then
  $\forall x \forall y Q_{\alpha \beta}(x,y) = \forall x \forall y
  (G_\alpha(x,y) \wedge C(x,y) \wedge H_\beta(x,y))$, and we prove
  that none of the clauses $C_1, \ldots, C_{k-1}$ becomes redundant.
  Recall that $G_\alpha(x,y)$ is a conjunction CNF expressions
  $G_i(x,y)$, each of which is a conjunction of subclauses
  $S_{J_k}(x,y)$ of some left clause of $Q$ (see Example~\ref{ex:QLeft}).  Therefore every clause
  $S_{J_k}(x,y)$ of $G_{\alpha}(x,y)$ contains {\em all} ubiquitous
  symbols, while none of the clauses $C_1, \ldots, C_{k-1}$ contains
  {\em all} ubiquitous symbols, proving that none of the middle clauses on the minimal-length left-to-right path are not redundant. (In fact, by Theorem~\ref{th:forbidden:syntax}, the middle clause $C_1$ does not contain \e{any} ubiquitous symbol. Since the path is minimal then none of the middle clauses contain any ubiquitous symbol.) 
  Some clauses $S_{J_k}(x,y)$ of $G_{\alpha}(x,y)$ may become
  redundant, but the only homomorphisms $C' \rightarrow S_{J_k}$ must
  be from some other clause $C'$ of $G_{\alpha}(x,y)$: otherwise, if
  $C'$ is a middle clause, then we obtain a homomorphism
  $C' \rightarrow C_0'$, where $C_0'$ is the left clause that contains
  $S_{J_k}$.  Thus, at least one subclause of $G_{\alpha}(x,y)$ has to
  be non-redundant, proving that $Q_{\alpha \beta}$ depends on all
  relational symbols in $Q$.  It remains to consider the cases
  $\alpha = \hat 1$ or $\beta = \hat 1$; assuming $\alpha = \hat 1$,
  by Eq.~\eqref{eq:q:alpha:1:2}
  $\forall x \forall y Q_{\hat 1 \alpha}(x,y) = Q \wedge \forall y
  H_\beta(y)$ and the argument is similar.x 
\end{proof}

\begin{lemma}[Invertible] \label{lemma:simple:lattices:1} The mapping
  $(\alpha,\beta) \mapsto Q_{\alpha \beta}(x,y)$ is invertible.  More
  precisely: if the logical implication
  $\forall x \forall y Q_{\alpha_1 \beta_1}(x,y) \Rightarrow \forall x
  \forall y Q_{\alpha_2 \beta_2}(x,y)$ holds, then
  $\alpha_1 \leq \alpha_2$ and $\beta_1 \leq \beta_2$, in the lattices
  $\hat L(\mb G)$ and $\hat L(\mb H)$ respectively.
\end{lemma}

\begin{proof}
  We expand
  $\forall x \forall y Q_{\alpha_1 \beta_1} \Rightarrow \forall x
  \forall y Q_{\alpha_2 \beta_2}$ and drop the quantifiers, to obtain:
  \begin{align*}
  G_{\alpha_1}(x,y) \wedge C(x,y) \wedge H_{\beta_1}(x,y) \Rightarrow & G_{\alpha_2}(x,y) \wedge C(x,y) \wedge H_{\beta_2}(x,y)
  \end{align*}
  Let $\mb V$ be all right ubiquitous symbols.  If we set them to
  $\mb V:=1$, then $H_\beta[\mb V:=1]=1$ for every $\beta$, therefore
  we obtain:
  \begin{align*}
  G_{\alpha_1}(x,y) \wedge C[\mb V:=1](x,y) \Rightarrow & G_{\alpha_2}(x,y) \wedge C[\mb V:=1](x,y)
  \end{align*}
  We conjoin both terms with $C(x,y)$ and notice that $C[\mb V:=1]
  \wedge C \equiv C$ and therefore we obtain:
  \begin{align*}
  G_{\alpha_1}(x,y) \wedge C(x,y) \Rightarrow & G_{\alpha_2}(x,y) \wedge C(x,y)
  \end{align*}
  By definition this means $\alpha_1 \leq \alpha_2$.  We prove
  similarly that $\beta_1 \leq \beta_2$. 
\end{proof}

We will now give the proof of Theorem~\ref{th:forbidden}, but first
illustrate the basic idea on an example.

\begin{example} \label{ex:running:query:type:2:cont}
  The query $Q$ in Example~\ref{ex:running:query:type:2} is not
  forbidden; we repeat it here:
  \begin{align*}
    Q = & \forall x (\forall y S_1(x,y) \vee \forall y S_2(x,y))
\wedge \forall x \forall y (S_1(x,y) \vee  S_3(x,y))
\wedge \forall y (\forall x S_3(x,y) \vee \forall x S_4(x,y))
  \end{align*}
  Notice that $S_2$ does not occur in the middle clause.  Define the
  following query, obtained by replacing $S_2$ with a unary symbol
  $R(x)$:
  \begin{align*}
    Q' = & \forall x \forall y (S'_1(x,y) \vee R(x)) \wedge \forall x
           \forall y (S'_1(x,y) \vee S'_3(x,y)) \wedge \forall
           y(\forall x S_3'(x,y) \vee \forall x S_4'(x,y))
  \end{align*}
  We claim that $\gmc_{\text{bi}}(Q') \leq^P_m \gmc_{\text{bi}}(Q)$;
  since $Q'$ is a query of Type I-II, we have already shown that
  $\gmc(Q')$ is \#P-hard, and this implies that $\gmc(Q)$ is also
  \#P-hard.  To prove the claim, consider any probabilistic database
  $\Delta' = (\Dom', p')$ for $Q'$.  Define the following
  probabilistic database $\Delta= (\Dom, p)$ for $Q$, where
  $\Dom \defeq \Dom' \cup \set{b_1}$ for a fresh constant $b_1$, and
  where the probabilities are defined as follows, for all
  $a,b\in \Dom$:
  \begin{align*}
    p(S_2(a,b_1)) \defeq & p'(R(a)) & p(S_2(a,b)) \defeq & 1 \\
j=1,3,4:\ p(S_j(a,b_1)) \defeq & 1 & p(S_j(a,b)) \defeq & p'(S'_j(a,b))
  \end{align*}
  In $Q$, we have
  $\forall y S_1(x,y) \equiv (\forall y\neq b_1 S_1(x,y)) \wedge
  S_1(x,b_1) \equiv \forall y S_1'(x,y)$, where here the variable $y$
  in $\forall y S_1'(x,y)$ ranges over $\Dom'$, i.e. without $b_1$.
  Similarly,
  $\forall y S_2(x,y) \equiv (\forall y\neq b_1 S_2(x,y))\wedge
  (S_2(x,b_1)) \equiv R(x)$, etc, and the query $Q$ becomes:
\begin{align*}
  \begin{array}{rlll}
    Q= & \forall x(\forall y S_1(x,y) \vee \forall y S_2(x,y))\wedge
    & \forall x \forall y (S_1(x,y) \vee S_3(x,y)) \wedge 
    & \forall y (\forall x S_3(x,y) \vee \forall x S_4(x,y))\\
    \equiv & \forall x (R(x) \vee \forall y S_1'(x,y)) \wedge
    & \forall x \forall y(S_1'(x,y) \vee S_3'(x,y))
      \wedge
    & \forall y (\forall x S_3'(x,y)\vee \forall x S_4'(x,y))
  \end{array}
\end{align*}
which is equivalent to $Q'$, proving $\pr(Q) = \pr'(Q')$.
\end{example}

\begin{proof} (of Theorem~\ref{th:forbidden}) Fix a left-right path in
  $Q$, not necessarily of minimal length, denote it
  $C_0, C_1, \ldots, C_k$, and recall that $C_0$ is a union of of
  subclauses
  $C_0 = \forall x (\forall y S_{J_1}(x,y) \vee \forall y S_{J_2}(x,y)
  \vee \cdots)$.  We start with the following:

  \begin{claim}[Variant of Lemma 8.36
    in~\cite{DBLP:journals/jacm/DalviS12}] \label{claim:lemma:8:36}
    Suppose $S_1$ is a symbol that
    occurs in both $C_0, C_1$, $S_0$ is a symbol that occurs in $C_0$
    and does not occur in $C_1, C_2, \ldots, C_k$.  Then, if there
    exists a subclause $S_J(x,y)$ of $C_0$ that contains $S_1$ but not
    $S_0$, then there exists an unsafe query $Q'$ with strictly fewer
    binary symbols such that
    $\gmc_{\text{bi}}(Q') \leq^P_m \gmc_{\text{bi}}(Q)$.
  \end{claim}

  The intuition is that, since $S_0$ only occurs in $C_0$, it should
  be a ubiquitous symbol, but fails to be one; then we can simplify
  $Q$ to $Q'$.  The query $Q'$ will have left clauses that are
  slightly more general than those introduced in
  Definition~\ref{def:types:of:clauses}: it may contain left clauses
  of the form:
  \begin{align*}
 & \forall x \big(R_1(x) \vee R_2(x)  \vee \cdots \vee \forall y S_{J_1}(x,y) \vee \forall y S_{J_2}(x,y) \vee \cdots\big)
  \end{align*}
  Its middle and right clauses are as given in
  Definition~\ref{def:types:of:clauses}.  That is, its left clauses
  may contain multiple unary symbols and/or multiple subclauses.  Such
  left clauses can be further simplified to either Type I left
  clauses, or Type II leff clauses, see Propositions 8.6 and 8.7
  in~\cite{DBLP:journals/jacm/DalviS12}.

  \begin{proof} (Of Claim~\ref{claim:lemma:8:36}) Let
    $S_{J_1}, \ldots, S_{J_t}$ be all the subclauses of $C_0$ that
    contain $S_0$; by assumption there exists at least one other
    sub-clause that contains $S_1$ and not $S_0$.  Define $Q'$ the
    query obtained from $Q$ as follows.  The vocabulary consists of
    (a) for every symbol $S_j$ in $Q$ other than $S_0$, there is a
    fresh binary symbol $S_j'(x,y)$ in $Q'$.  (b) for every
    $q=1,\ldots,t$ and every symbol $S_j$ that occurs in the subclause
    $S_{J_q}$ (including $S_0$) there is a fresh unary symbol denoted
    $R_j^{(q)}(x)$.  Notice that $Q'$ has one less binary symbol,
    since there is no $S_0'(x,y)$.  For any subclause
    $S_{J_0}(x,y) = S_{j_1}(x,y) \vee S_{j_2}(x,y) \vee \cdots \vee
    S_{j_m}(x,y)$ that occurs anywhere in $Q$, we denote the following
    expressions:
    \begin{align*}
      S'_{J_0} \defeq & S_{j_1}'(x,y) \vee \cdots \vee S_{j_m}'(x,y)
      &
        R^{(q)}_{J_0} = 
      &  \begin{cases}
        R_{j_1}^{(q)}(x) \vee \cdots \vee R_{j_m}^{(q)}(x) & \mbox{if $\symb(S_{J_0})\subseteq\symb(S_{J_q})$}\\
        \texttt{true} & \mbox{otherwise}
      \end{cases}
    \end{align*}
    We construct the query $Q'$ from $Q$ by replacing each subclause
    $S_{J_0}(x,y)$ with a new expression, according to the following
    two cases (the justification will become clear below, when we
    describe the mapping from $\Delta$ to $\Delta'$):
    \begin{description}
    \item[Case 1:] $S_0 \in \symb(S_{J_0})$.  Then replace
      $S_{J_0}(x,y)$ with $\bigwedge_{q=1,t} R_{J_0}^{(q)}(x)$.  (Note
      that this might be $\texttt{true}$.)
    \item[Case 2:] $S_0\not\in \symb(S_{J_0})$.  Then replace
      $S_{J_0}(x,y)$ with $\bigwedge_{q=1,t} R_{J_0}^{(q)}(x) \wedge S_{J_0}'(x,y)$
    \end{description}
    We show now that
    $\gmc_{\text{bi}}(Q') \leq^P_m \gmc_{\text{bi}}(Q)$.  Given a
    database $\Delta'=(\Dom', p')$, we define
    $(\Dom \defeq \Dom \cup \set{b_1, b_2, \ldots, b_t}, p)$, where
    $b_1, \ldots, b_t$ are fresh constants, and define:
    \begin{align*}
      \forall q=1,t:\ \ p(S_0(a,b_q))\eqdef & p'(R_0^{(q)}(a)) & p(S_0(a,b)) \defeq & 1 \\
      \forall q=1,t: \forall S_j \in \symb(S_{J_q}) - \set{S_0}\ \ p(S_j(a,b_q)) \defeq & p'(R_j^{(q)}(a)) & p(S_j(a,b)) \defeq & p'(S_j'(a,b))\\
      \forall q=1,t: \forall S_j \in \symb(Q) - \symb(S_{J_q})\ \  p(S_j(a,b_q)) \defeq & 1 & p(S_j(a,b)) \defeq & p'(S_j'(a,b))
    \end{align*}
    We check that the lineage of $Q$ on $\Delta$ is the same as that
    of $Q'$ on $\Delta'$.  Consider any subclause of $Q$ and assume
    first that it contains $S_0$:  $S_{J_0}(x,y) = S_0(x,y) \vee
    S_{j_1}(x,y) \vee \cdots$  When we substitute $y:=b_q$, then
    $S_0(x,b_q)\equiv R_0^{(q)}(x)$, and for every other symbol
    $S_{j_i}(x,b_q)$ is either $R_{j_i}^{(q)}(x)$ or $\texttt{true}$,
    hence the expression is equivalent to $R^{(q)}_{J_0}$; when we
    substitute $y:=b$ for some $b \in \Dom'$, then $S_0(x,b)\equiv
    \texttt{true}$ and entire expression vanishes.  This justifies
    Case 1 above.  Case 2 is justified similarly and ommitted.

    It remains to prove that $Q'$ is an unsafe query.  For that we
    prove that, once converted into CNF, $Q'$ has a left-right path of
    clauses that is non-redundant.  We start by observing that every
    clause $C_1, C_2, \ldots, C_k$ on our path (other than $C_0$) is
    converted into an isomorphic clause, where each symbol $S_j(x,y)$
    is replaced by $S_j'(x,y)$.  Indeed, by assumption, no subclause
    $S_{J_0}(x,y)$ on this path contains $S_0$, hence they are converted
    according to case 2.  Moreover,
    $R_{J_0}^{(q)}\equiv \texttt{true}$, because there must exist a
    symbol $S_j$ in $S_{J_0}$ that does not occur in $C_0$: indeed,
    when $S_{J_0}$ is any middle clause $C_i$, $i=1, \ldots, k-1$, and
    $\symb(C_i) = \symb(S_{J_0}) \subseteq \symb(S_J)$ then there is a
    homomorphism $C_i \rightarrow S_J \rightarrow C_0$, contradiction.
    Since the length of $Q$ is $\geq 2$, $S_J$ has no common symbols
    with the right clause $C_k$.  Consider any other clause $C$ of
    $Q$.  Each of its subclauses $S_{J_0}$ is converted to a
    conjunction of up to $t+1$ expressions, i.e.  either
    $\bigwedge_{q=1,t} R_{J_0}^{(q)}(x)$ or
    $\bigwedge_{q=1,t} R_{J_0}^{(q)}(x) \wedge S_{J_0}'(x,y)$; notice
    that some expressions $R_{J_0}^{(q)}$ may be
    $\equiv \texttt{true}$. We convert the resulting expression into
    CNF, thus from $C$ we obtain a conjunction
    $C' \wedge C'' \wedge \cdots$, each obtained by making one choice
    out of the up to $t+1$ choices for each subclause of $C$.
    Considering now $C_0$, we define $C_0'$ to be obtained as follows:
    (a) for each subclause $S_{J_q}(x,y)$ of $C_0$, choose
    $R_{J_0}^{(q)}(x)$ (actually one can check that no other choice
    exists here) (b) for each other subclause $S_{J_0}$, choose
    $S_{J_0}'$.  That is, $C_0$ looks like this:
    $R_{J_1}^{(1)}(x) \vee R_{J_1}^{(2)}(x) \vee \cdots \vee
    R_{J_t}^{(t)}(x) \vee \forall y S'_{J_0}(x,y) \vee \cdots$ By
    assumption there exists at least one subclause $S_{J_0}$ other
    than $S_{J_1}, \ldots, S_{J_t}$, hence $C_0'$ is a left clause,
    i.e. it does not degenerate to $\bigvee_q R_{J_q}^{(q)}(x)$.  (If
    we applied this construction to
    Example~\ref{ex:running:query:type:2:cont:2} trying to remove
    $U(x,y)$, then the left clause degenerates to
    $R_0^{(1)}(x) \vee R_1^{(1)}(x) \vee R_0^{(2)}(x) \vee
    R_2^{(2)}(x)$.)  Furthermore, the symbol $S_1'$ is common in
    $C_0'$ and $C_1'$, hence $C_0', C_1', \ldots, C_k'$ is a
    left-right path.  It remains to prove that it is not redundant.
    For that, assume the contrary, that there exists a homomorphism
    $f : C' \rightarrow C_i'$, for some clause
    $C' \not\equiv \texttt{true}$.  Let $C$ be the original clause
    from which we derived $C'$ (thus $C$ became
    $C' \wedge C'' \wedge \cdots$), and consider any subclause
    $S_{J_0}(x,y)$ of $C$.  We will construct a homomorphism
    $g : S_{J_0}(x,y) \rightarrow C_i$; by taking their union we
    obtain a homomorphism $C\rightarrow C_i$, contradicting the fact
    that $Q$ has no redundant clauses.  To construct $g$, we consider
    the two cases above.  Case 1: $S_0 \in \symb(S_{J_0})$, then $C'$
    must contain some expression $R^{(q)}_{J_0}(x)$ for some choice of
    $q=1,\ldots,t$, and since
    $R^{(q)}_{J_0}(x)\not\equiv\texttt{true}$ (otherwise
    $C' \equiv \texttt{true}$), we have that
    $\symb(S_{J_0}) \subseteq \symb(S_{J_q})$.  Given the homomorphism
    $f:C'\rightarrow C_i'$, we have that $C_i'$ contains
    $R^{(q)}_{J_0}(x)$, hence it must be that $i=0$.  We simply define
    $g$ to map $S_{J_0}$ to the subclause $S_{J_q}$ of $C_0$.  Case 2:
    $S_0 \not\in \symb(S_{J_0})$.  Its translation is
    $\bigwedge_{q=1,t} R_{J_0}^{(q)}(x) \wedge S_{J_0}'(x,y)$, hence
    $C'$ must contain either some $R_{J_0}^{(q)}(x)$, in which case we
    argue as in Case 1, or contains $S'_{J_0}(x,y)$.  In that case we
    use the homomorphism $f$: it maps $S'_{J_0}(x,y)$ to $C_i'$, we
    simply define $g$ similarly from $S'_{J_0}(x,y)$ to $C_i$.

    This concludes our proof of Claim~\ref{claim:lemma:8:36}.
  \end{proof}

  \begin{example}
    We illustrate with an example showing some of the subtleties of
    the proof of Claim~\ref{claim:lemma:8:36}.  Consider two clauses
    $C_0, C_1$ that form the beginning of a left-right path.  We apply
    Claim~\ref{claim:lemma:8:36} to the symbol $S_0$ in $C_0$, noting
    that it does not co-occur with $S_1$, and show their conversions
    to $C_0', C_1'$ below:
    \begin{align*}
      C_0 = & \forall x \big(\forall y S_1(x,y) \vee \forall y (S_0(x,y) \vee S_2(x,y))\vee \forall y (S_0(x,y) \vee S_3(x,y))\big)  & C_1 = & \forall x \forall y \big(S_1(x,y) \vee S_2(x,y) \vee S_4(x,y)\big)\\
      C_0' = & \forall x \big(\forall y S_1'(x,y) \vee R_0^{(1)}(x) \vee R_2^{(1)}(x) \vee R_0^{(2)}(x) \vee R_3^{(2)}(x)\big)  & C_1' = & \forall x \forall y \big(S'_1(x,y) \vee S'_2(x,y) \vee S'_4(x,y)\big)
    \end{align*}
    There are two subclauses containing $S_0$.  Thus, there are two
    unary symbols $R_0^{(1)}, R_0^{(2)}$ because $S_0$ occurs in both,
    and only one symbol $R_2^{(1)}$ and one symbol $R_3^{(2)}$ because
    $S_2, S_3$ occur only in one subclause respectively.  We
    illustrate why $C_0', C_1$' do not become redundant, by
    considering several other clauses, which are not part of the left
    right path, and thus we denote them $D_1, D_2, \ldots$:
    \begin{align*}
      D_1 = & \forall x\big(\forall y(S_0(x,y) \vee S_1(x,y)) \vee \forall y S_2(x,y)\big)
      & D_1' = & \forall x\big(\texttt{true} \vee \big(R_2^{(1)}(x)\wedge \forall y S_2'(x,y)\big)\big) \equiv \texttt{true}\\
      D_2 = & \forall x \big(\forall y (S_1(x,y) \vee S_2(x,y)) \vee \forall y S_3(x,y)\big)
      & D_2' = & \forall x\big(\forall y (S_1'(x,y) \vee S_2'(x,y)) \vee \big(R_3^{(2)}(x)\wedge\forall y S_3'(x,y)\big)\big)\\
      D_3 = & \forall x \big(\forall y S_1(x,y) \vee \forall y (S_2(x,y)  \vee S_3(x,y))\big)
      & D_3' = & \forall x \big(\forall y S_1'(x,y) \vee \forall y  (S_2'(x,y) \vee S_3'(x,y))\big)
    \end{align*}
    Clause $D_1$ rewrites to $\texttt{true}$ hence does not lead to
    any redundancies.  Clause $D_2$ rewrites to the expression $D_2'$,
    which further rewrites to two clauses:
    $\forall x (\forall y (S_1'(x,y) \vee S_2'(x,y)) \vee
    R_3^{(2)}(x))$ and
    $\forall x (\forall y (S_1'(x,y) \vee S_2'(x,y)) \vee \forall y
    S_3'(x,y))$.  Considering the first clause, we notice that there
    exists a homomorphism from $R_3^{(2)}(x)$ to $C_0'$, but this is
    because there was a homomorphism from $\forall y S_3(x,y)$ to
    $\forall y (S_0(x,y) \vee S_3(x,y))$; this does not extend to the
    entire clause (otherwise there would be a homomorphism
    $D_2 \rightarrow C_0$).  Finally, in $D_3$ we observe that the
    subclause $S_{J_0} \defeq S_2(x,y) \vee S_3(x,y)$ rewrites to
    $S'_2(x,y) \vee S'_3(x,y)$, because
    $R_{J_0}^{(1)}(x) \equiv R_{J_0}^{(2)}(x) \equiv \texttt{true}$
    since $S_2, S_3$ occur in separate clauses with $S_0$; one
    subtlety of the proof of the claim is precisely to ensure that and
    expression like this does not rewrite to
    $R_2^{(1)}(x) \wedge R_3^{(2)}(x) \wedge \forall y S_3'(x,y)$,
    because after converting to CNF it will render $C_0'$ redundant.
  \end{example}

  Continuing the proof of Theorem~\ref{th:forbidden}, we can now
  assume w.l.o.g. that for every $S_1$ common to $C_0, C_1$ and for
  every $S_0$ occurring only in $C_0$, every subclause of $C_0$
  containing $S_1$ must also contain $S_0$.  Next we show:

  \begin{claim} \label{claim:lemma:8:38}
    Every subclause $S_J(x,y)$ of $C_0$ has some common symbol with
    $C_1$ (this is Lemma 8.38 in~\cite{DBLP:journals/jacm/DalviS12}).
  \end{claim}

  Indeed, if $S_J'$ is any other subclause of $C_0$ that does have
  some common symbol $S_1$ with $C_1$, then {\em all} symbols $S_0$ of
  $S_J$ must occur in $S_J'$, therefore
  $\forall y S_J(x,y) \vee \forall y S_J'(x,y) \equiv \forall y
  S_J(x,y)$, contradicting the assumption that $C_0$ has non redundant
  subclauses, proving claim~\ref{claim:lemma:8:38}.

  At this point we will restrict the left-right path $C_0, C_1,
  \ldots, C_k$ to be of minimal length.

  \begin{claim} \label{claim:x}
    Every symbol in $C_0$ is either $C_0$-ubiquitous, or occurs in $C_1$.
  \end{claim}

  \begin{proof}
    Let $S_0 \in \symb(C_0) - \symb(C_1)$.  By minimality, $S_0$ does
    not occur in $C_1, C_2, \ldots, C_k$.  We prove that $S_0$ is
    $C_0$-ubiquitous.  If it is not, then there exists a subclause
    $S_J$ that does not contain $S_0$.  By
    Claim~\ref{claim:lemma:8:38}, $S_J$ contains some symbol $S_1$
    common with $C_1$, thus, by Claim~\ref{claim:lemma:8:36} $S_J$
    must contain $S_0$, contradiction.
  \end{proof}

  \begin{claim} \label{claim:at:least:one:u}
  There exists at least one $C_0$-ubiquitous symbol.
  \end{claim}

  Indeed, otherwise, all symbols of $C_0$ occur in $C_1$, which
  implies that there exists a homomorphism $C_0 \rightarrow C_1$
  (because $C_1$ is a middle clause), contradicting the fact that $Q$
  has no redundant clauses; this proves the claim.  Finally, we prove:

  \begin{claim}
  	\label{claim:sameUbiquitousSymbls}
    If $C_0, C_0'$ are two left clauses then the $C_0$-ubiquitous
    symbols are the same as the $C_0'$-ubiquitous symbols.
  \end{claim}

  We first prove that $C_0'$ shares some common symbols with $C_1$.
  Assuming otherwise, since all symbols in $Q$ must occur on the path
  $C_0, C_1, \ldots, C_k$, and this path has minimal length, it
  follows that all symbols of $C_0'$ occur in $C_0$, none occur in
  $C_1$ (by the assumption in the claim), thus they are $C_0$-ubiquitous.  In particular,
  $\symb(C_0')\subseteq\symb(S_J)$ for any subclause $S_J$ of $C_0$,
  which implies that there exists homomorphisms
  $C_0' \rightarrow S_J \rightarrow C_0$ contradicting the assumption
  that $C_0$ is non-redundant.  Thus, $C_0'$ shares some common
  symbols with $C_1$, and therefore $C_0', C_1, C_2, \ldots, C_k$ is
  also a left-right path of minimal length.  Since the previous
  argument applies to this path as well, symbols in $C_0'$ are also
  partitioned into those common with $C_1$ and symbols ubiquitous in
  $C_0'$.  Consider a symbol ubiquitous in $C_0'$: it must occur on
  the path $C_0, C_1, \ldots, C_k$ and, since it doesn't occur in
  $C_1$, it must occur in $C_0$ (since the path is of minimal length),
  implying that it is ubiquitous in $C_0$.  By symmetry, every
  ubiquitous symbol in $C_0$ is also a ubiquitous symbol in $C_0'$.
  This proves the claim.
\end{proof}

Finally, we prove Theorem~\ref{th:forbidden:syntax}, which follows
from the next technical lemma.

\begin{lemma} \label{lemma:forbidden:properties} Let $Q$ be a
  forbidden query and let $C_0, C_1, \ldots, C_k$ be a left-right path
  of minimal length.  Then (1) There exists at least one left
  ubiquitous symbol that does not occur in $C_1$; in particular, the
  query $Q$ has at least one left ubiquitous symbol $U$.  (2) More: no
  ubiquitous symbol occurs in $C_1$.  (3) for every left clause
  $C = \forall y S_{J_1}(x,y) \vee \forall y S_{J_2}(x,y) \vee
  \cdots$, every subclause $S_{J_i}$ has some common symbol with
  $C_1$. (4) If there are more than one left ubiquitous symbols, then
  each of them occurs in some middle clause.  (5) For any clause $C$
  other than $C_1$, if $\symb(C) \subseteq \symb(C_0)\cup \symb(C_1)$
  then $C$ contains some ubiquitous symbol.  (6)  If $C$ is a middle
  clause containing some left ubiquitous symbol $U$, then $\symb(C)
  \cap \symb(C_2) = \emptyset$; equivalently, $\symb(C) \subseteq
  \symb(C_0) \cup \symb(C_1)$.
\end{lemma}

\begin{proof} (1) Assuming otherwise, then by Claim~\ref{claim:x}, every symbol of $C_0$ must occur
  in $C_1$, $\symb(C_0) \subseteq \symb(C_1)$ which implies that there
  exists a homomorphism $C_0\rightarrow C_1$.  (2) We strengthen the
  claim by showing that no ubiquitous symbol can occur in $C_1$.  Let
  $U$ occur in $C_1$.  We claim that $Q[U:=0]$ is an unsafe query,
  contradicting the assumption that $Q$ is final. For that we
  prove that none of the clauses
  $C_0[U:=0], C_1[U:=0], C_2, \ldots, C_k$ is redundant in $Q[U:=0]$.
  It is easy to see that there is no homomorphism
  $C[U:=0] \rightarrow C_i[U:=0]$ for $i=0,1$, because both $C_0, C_1$
  contain $U$ and that would imply the existence of a homomorphism
  $C \rightarrow C_i$.  Assuming there exists a homomorphism
  $C[U:=0] \rightarrow C_i$ for $i \geq 2$; we must have $i=2$,
  because the path has minimal length.  Let $S$ be a non-ubiquitous
  symbol in $\symb(C_0) \cap \symb(C_1)$, thus $S \not\in \symb(C_2)$
  since the path is of minimal length.  Consider now the left-right
  path $C_0, C, C_2, C_3, \ldots, C_k$ in $Q$.  This is also of
  minimal length, hence by the definition of a forbidden query, $S$
  must occur in $C$, but this contradicts the existence of a
  homomorphism $C[U:=0] \rightarrow C_2$. (3) From Claim~\ref{claim:sameUbiquitousSymbls}, we know that the $C_0$-ubiquitous symbols are the same as the $C$-ubiquitous symbols. Since $C_0,C_1,\dots,C_k$ is of minimal length and contains all relational symbols of $Q$, then $\symb(C) \subseteq \symb(C_0)\cup \symb(C_1)$.
  Now, suppose that $S_{J_i}$ is a
  sub-clause of $C$ that does not contain any symbol from $C_1$.  
  This means that $S_{J_i}$ contains only left ubiquitous symbols. Since the query is minimized, then $S_{J_i}$ is the only subclause of $C$. But then, $C\rightarrow C_0$ because every subclause of $C_0$ contains all ubiquitous symbols (i.e., $\symb(S_{J_i})$), which is a contradiction.
  
  \eat{
  therefore
  $S_{J_i}$ is redundant, because there exists at least one other
  subclause. \batya{Shouldn't it be that there is a homomorphism from
    $S_{J_i}$ to every other left clause (because $S_{J_i}$ contains
    only ubiquitous symbols) ?} \dan{correct, we must argue that.  We
    must also argue that $S_{J_i}$ cannot be the single subclause in
    $C$; perhaps that's obvious, maybe we don't need it?}
}
 (4) Suppose there are at least two left ubiquitous
  symbols $U_1, U_2$, and suppose $U_1$ does not occur in any middle
  clause.  We claim that $Q[U_1 := 0]$ is an unsafe query,
  contradicting the fact that $Q$ is final.  To prove the claim, we
  show that none of the clauses $C_0[U_1 := 0], C_1, \ldots, C_k$ in
  $Q[U_1:=0]$ is redundant.  Assume that there exists a homomorphism
  $C[U_1:=0] \rightarrow C_i$ for $i \geq 1$; then $C$ must contain
  $U_1$ (otherwise there exists a homomorphism $C \rightarrow C_i$),
  hence $C$ is a left clause, but in that case it also contains $U_2$,
  hence the homomorphism is not possible.  Similarly, if there exists
  a homomorphism $C[U_1:=0] \rightarrow C_0[U_1:=0]$, then there
  exists a homomorphism $C \rightarrow C_0$, because $C_0$ contains
  $U_1$ in every subclause.  (5) Suppose
  $\symb(C) \subseteq \symb(C_0) \cup \symb(C_1)$.  If $C$ has no
  ubiquitous symbols, then we claim that
  $\symb(C) \subseteq \symb(C_1)$, but this implies that there exists
  a homomorphism $C \rightarrow C_1$ which is a contradiction.  To
  prove the claim, let $S$ be any symbol in $C$.  Since $S$ is not
  ubiquitous, if $S \in \symb(C_0)$ then it must also belong to $C_1$,
  proving the claim.  (6) Assume otherwise, then $C_0, C, C_2, \ldots,
  C_k$ is also a left-right path of minimal length, hence by item (2)
  $C$ cannot contain any ubiquitous symbol, contradiction.  
\end{proof}

\subsection{Computing $Q$ on a Block-database}

Recall that $\Phi_\pdb(Q)$ denotes the lineage of $Q$ on the TID
$\pdb$.  As before, a block $B(u,v)$ is a bipartite TID with two
distinguished constants $u, v$.  We denote by $U(B), V(B)$ the two
partitions of the domain of $B$, that is $\Dom(B) = U(B) \cup V(B)$.
By definition, $u \in U(B), v \in V(B)$.  We define:
\begin{align}
    Y_{\alpha \beta}(u,v) \defeq & \Phi_{B(u,v)}(Q_{\alpha \beta}(u,v))  \label{eq:y:alpha:beta}
\end{align}
where $Q_{\alpha \beta}$ was introduced in
Eq.~\eqref{eq:q:alpha:beta}.  We prove the following theorem, which is
the analogous of Theorem~\ref{th:disjointBlocks} for type II queries:
Recall that $L(\mb G), L(\mb H)$ are the supports of the left lattice,
and the right lattice of the query, and this includes $\hat 1$.  In
this section we want to remove $\hat 1$, and denote $L_0(\mb G) \defeq
L(\mb G)-\set{\hat 1}$, $L_0(\mb H) \defeq L(\mb H) - \set{\hat 1}$.

\begin{theorem} \label{th:pq:mobius} Let $U,V$ be two disjoint sets,
  and let $\pdb = \bigcup_{u\in U, v \in V} B(u,v)$ be a disjoint
  union of blocks (that is, no two blocks share any tuple or any
  constant, expect for endpoints $u,v$).  Then:
  \begin{align}
    \Pr(Q) &=  (-1)^{|U|+|V|}\sum_{\substack{\sigma : U \rightarrow L_0(\mb{G}), \\\tau: V  \rightarrow L_0(\mb{H})}}
    \left(\prod_{u\in U} \mu(\sigma(u))\right) \cdot \left(\prod_{v\in V} \mu(\tau(v))\right) \cdot  
    \prod_{u \in U, v \in V}  \Pr(Y_{\sigma(u) \tau(v)}(u,v)) \label{eq:p:q:type:2}
  \end{align}
\end{theorem}

The sum above has exponentially many terms, namely
$2^{|L_0(\mb G)| + |L_0(\mb H)|}$; we will show later how to use this formula.

\begin{proof}
  The TID $\pdb$ is bipartite, hence its domain is partitioned into
  $\Dom(\pdb) = \bar U \cup \bar V$.  We have $U \subseteq \bar U$,
  $V \subseteq \bar V$, and the inclusions may be strict since each
  block $B(u,v)$ may have its own left and right constants, other than
  $u,v$.

  By definition,
  $\Phi_\pdb(Q)=\Phi_\pdb(Q_{\text{left}})\wedge
  \Phi_\pdb(Q_{\text{middle}}) \wedge \Phi_\pdb(Q_{\text{right}})$.
  Since
  $Q_{\text{left}} =\forall x(\bigvee_{i=1,m} \forall y G_i(x,y))$,
  its lineage is given by
  $\Phi_\pdb(Q_{\text{left}}) = \bigwedge_{a \in \bar U}
  \Phi_\pdb(\bigvee_{i=1,m}\forall y G_i(a,y))$.  When
  $a \in \bar U - U$, then the lineage
  $\Phi_\pdb(\bigvee_{i=1,m}\forall y G_i(a,y))$ is equal to the
  lineage restricted to the unique block $B(u,v)$ that contains $a$,
  i.e.
  $\Phi_\pdb(\bigvee_{i=1,m}\forall y G_i(a,y)) =
  \Phi_{B(u,v)}(\bigvee_{i=1,m}\forall y G_i(a,y))$, because, by
  construction, $a$ can only be connected to constants in the same
  block.  When $a \in U$, then the lineage may span multiple blocks.
  Using this observation, and repeating it for $Q_{\text{right}}$, we
  derive the following expressions for the lineages:
  \begin{align}
    \Phi_\pdb(Q_{\text{left}}) = & \bigwedge_{a \in \bar U} \Phi_\pdb(\bigvee_i\forall y(G_i(a,y)))
= \left(\bigwedge_{u \in  U} \Phi_\pdb(\bigvee_i\forall y(G_i(u,y)))\right)\wedge \left(\bigwedge_{a \in \bar U - U}\Phi_\pdb(\bigvee_i\forall y(G_i(a,y))) \right)\nonumber \\
= & \left(\bigwedge_{u \in U} \Phi_\pdb(\bigvee_i\forall y(G_i(u,y)))\right)\wedge \left(\bigwedge_{u\in U, v \in V}\Phi_{B(u,v)}(Q_{\text{left}})\right)\nonumber \\
    \Phi_\pdb(Q_{\text{right}}) = & \left(\bigwedge_{v \in V} \Phi_\pdb(\bigvee_j\forall x(H_j(x,v)))\right)\wedge \left(\bigwedge_{u\in U, v \in V}\Phi_{B(u,v)}(Q_{\text{right}})\right)\nonumber \\
    \Phi_\pdb(Q_{\text{middle}}) = & \bigvee_{u\in U, v \in V} \Phi_{B(u,v)}(Q_{\text{middle}})\nonumber\\
    \Phi_\pdb(Q) = & \Phi_\pdb(Q_{\text{left}}) \wedge \Phi_\pdb(Q_{\text{middle}}) \wedge \Phi_\pdb(Q_{\text{right}})\nonumber \\
 = & \left(\bigwedge_{u \in U} \Phi_\pdb(\bigvee_i\forall y(G_i(u,y) \underbrace{\wedge C(u,y)}_{\mbox{part of $Q$}}))\right)\wedge\bigwedge_{u\in U, v \in V} \Phi_{B(u,v)}(Q) \wedge \left(\bigwedge_{v \in V} \Phi_\pdb(\bigvee_j\forall x(H_j(x,v) \underbrace{\wedge C(x,v)}_{\mbox{part of $Q$}}))\right)\label{eq:change:1}
  \end{align}
  In the last line we added the redundant terms $C(u,y)$ and $C(x,v)$
  (recall that $Q_{\texttt{middle}}=\forall x \forall y C(x,y)$).  Now
  we will apply M\"obius' inversion formula on the expression for
  $\Phi_\pdb(Q)$ above, repeatedly, once for each $u \in U$, then once
  for each $v \in V$.  We show how to do it for one fixed constant
  $u_1 \in U$.  First, we separate $u_1$ from the conjunction
  $\bigwedge_{u \in U}$ above:

\begin{align*}
    \Phi_\pdb(Q) = & \Phi_\pdb(\bigvee_i\forall y(G_i(u_1,y)\wedge C(u_1,y)))\\
\wedge  & \underbrace{\left(\bigwedge_{u \in U-\set{u_1}} \Phi_\pdb(\bigvee_i\forall y(G_i(u,y)\wedge C(u,y)))\right)\wedge\bigwedge_{u\in U, v \in V} \Phi_{B(u,v)}(Q) \wedge \left(\bigwedge_{v \in V} \Phi_\pdb(\bigvee_j\forall x(H_j(x,v)\wedge C(x,v)))\right)}_{\defeq\text{REST}}\\
 = & \Phi_\pdb(\bigvee_i\forall y(G_i(u_1,y)\wedge C(u_1,y)))\wedge \text{REST} =  \bigvee_i \left(\Phi_\pdb(\forall y G_i(u_1,y)\wedge C(u_1,y)) \wedge \text{REST}\right) \\
\pr(Q) = \pr(\Phi_\pdb(Q)) =  & - \sum_{\alpha_1 \in L_0(\mb G)}\mu(\alpha_1) \Pr\left(\Phi_\pdb(\forall y(G_{\alpha_1}(u_1,y)\wedge C(u_1,y)))\wedge\text{REST}\right)
\end{align*}
Recall that the lattice $\hat L_0(\mb G)$ was defined by the formulas
$G_1\wedge C, G_2 \wedge C, \ldots$, and not by $G_1, G_2, \ldots$ We
introduced earlier the term $C(u_1,y)$ in order to be able to apply
the M\"obius formula for this lattice.  Next, we consider formula
$\text{REST}$, which contains the conjunction
$\bigwedge_{u \in U-\set{u_1}}$.  We separate a second
$u_2 \in U-\set{u_1}$, and repeat this argument for
$u_2, u_3, \ldots \in U$, reducing the $\text{REST}$ formula, until we
arrive at:

\begin{align}
  \pr(Q) = & (-1)^{|U|} \sum_{\sigma: U \rightarrow L_0(\mb G)}\left(\prod_{u \in U}\mu(\sigma(u))\right) \pr(\bigwedge_{u \in U}\Phi_\pdb(\forall y(G_{\sigma(u)}(u,y)\wedge C(u,y))) \wedge \text{REST})\nonumber\\
\mbox{where } \text{REST} \defeq & \bigwedge_{u\in U, v \in V} \Phi_{B(u,v)}(Q) \wedge \left(\bigwedge_{v \in V} \Phi_\pdb(\bigvee_j\forall x(H_j(x,v)\wedge C(x,v)))\right)\label{eq:rest:right}
\end{align}

We repeat the same process on the right clauses $\bigwedge_{v \in V} \Phi_\pdb\left(\bigvee_j\forall x(H_j(x,v)\wedge C(x,v))\right)$, and obtain:

\begin{align}
  \pr(Q) = & (-1)^{|U|+|V|} \sum_{\sigma: U \rightarrow L_0(\mb G),  \tau: V \rightarrow L_0(\mb H)}
\left(\prod_{u \in U}\mu(\sigma(u))\right)\left(\prod_{v \in V}\mu(\tau(v))\right)\nonumber\\
  & \pr(\bigwedge_{u \in U}\Phi_\pdb(\forall y(G_{\sigma(u)}(u,y)\wedge C(u,y))) \wedge \bigwedge_{u\in U, v \in V} \Phi_{B(u,v)}(Q)\wedge \bigwedge_{v \in V} \Phi_\pdb(\forall x(H_{\tau(v)}(x,v)\wedge C(x,v))))\label{eq:change:2}
\end{align}

Next, we observe that:
\begin{align}
\Phi_\pdb(\forall y(G_{\sigma(u)}(u,y)\wedge C(u,y)))= & \bigwedge_{v \in V} \Phi_{B(u,v)}(\forall y(G_{\sigma(u)}(u,y)\wedge C(u,y)))\label{eq:change:4}\\
\Phi_\pdb(\forall x(H_{\tau(v)}(x,v)\wedge C(x,v))) = & \bigwedge_{u \in U} \Phi_{B(u,v)}(\forall x(H_{\tau(v)}(x,v)\wedge C(x,v)))\nonumber
\end{align}
Since $C$ is the middle part of $Q$, we have  $\Phi_{B(u,v)}(C(u,y)
\wedge Q) \equiv \Phi_{B(u,v)}(Q)$ and therefore we can eliminate
$C(u,v)$ and obtain:
\begin{align}
  \pr(Q) = & (-1)^{|U|+|V|} \sum_{\sigma: U \rightarrow L_0(\mb G),  \tau: V \rightarrow L_0(\mb H)}
\left(\prod_{u \in U}\mu(\sigma(u))\right)\left(\prod_{v \in V}\mu(\tau(v))\right)
  \pr(\bigwedge_{u\in U, v \in V} \Phi_{B(u,v)}(\forall y(G_{\sigma(u)}(u,y)))\wedge \Phi_{B(u,v)}(Q) \wedge \Phi_{B(u,v)}(\forall x(H_{\tau(v)}(x,v))))\nonumber\\
  = &(-1)^{|U|+|V|} \sum_{\sigma: U \rightarrow L_0(\mb G),  \tau: V \rightarrow L_0(\mb H)}
\left(\prod_{u \in U}\mu(\sigma(u))\right)\left(\prod_{v \in V}\mu(\tau(v))\right)
  \pr(\bigwedge_{u\in U, v \in V} \Phi_{B(u,v)}(\forall y(G_{\sigma(u)}(u,y))\wedge Q \wedge \forall x(H_{\tau(v)}(x,v))))\label{eq:change:3}\\
  = &(-1)^{|U|+|V|} \sum_{\sigma: U \rightarrow L_0(\mb G),  \tau: V \rightarrow L_0(\mb H)}
\left(\prod_{u \in U}\mu(\sigma(u))\right)\left(\prod_{v \in V}\mu(\tau(v))\right)
  \prod_{u\in U, v \in V}\pr(\Phi_{B(u,v)}(\forall y(G_{\sigma(u)}(u,y))\wedge Q \wedge \forall x(H_{\tau(v)}(x,v))))\nonumber
\end{align}
This completes the proof of Theorem~\ref{th:pq:mobius}.
\end{proof}

Let $G(U,V,E)$ be a bipartite graph, defining a $\ccp$ problem.  As
for Type I queries (Eq.~\eqref{eq:fTypeIQ}) we construct a
TID that is a union of disjoint blocks
$\bigcup_{u\in U, v \in V} B(u,v)$, where for every non-edge
$(a,b) \not\in E$ we define $B(a,b)$ to be the trivial block where all
tuples have probability 1.  However, unlike the
formula~\eqref{eq:fTypeIQ} for type I queries, now we need
to cope the products of M\"obius functions in
Eq.~\eqref{eq:p:q:type:2}.  To remove those, we extend the graph by
adding, for each node $u\in U$,  one outgoing edge $(u,u')$ where $u'$ is a
fresh node with no other incoming edges; similarly, we add edges
$(v',v)$, one for each node $v \in V$.  More precisely, the new graph
is $(U\cup V', V \cup U', E \cup E')$, where $V' = \setof{v'}{v \in
  V}$, $U' = \setof{u'}{u \in U}$ and $E' = \setof{(u,u')}{u \in U}
\cup \setof{(v,v')}{v \in V}$.  Then, we have:
\begin{corollary} \label{cor:pq:mobius}  Given the notations above:
	\begin{align*}
	\Pr(Q) =  \sum_{\substack{\sigma : U \rightarrow L_0(\mb{G}), \\ \tau: V  \rightarrow L_0(\mb{H})}}
	& \left(\prod_{u\in U} \mu(\sigma(u))\right) \cdot \left(\prod_{v\in V} \mu(\tau(v))\right) \cdot
	  \prod_{(u,v)\in E}  \Pr(Y_{\sigma(u) \tau(v)}(u,v)) \cdot
	  \prod_{u \in  U} \Pr(Y_{\sigma(u),\hat 1}(u,u')) \cdot
	  \prod_{v \in V} \Pr(Y_{\hat 1 \tau(v)}(v',v)) \\
	\mbox{where }   Y_{\alpha \beta}(u,v) \defeq & \Phi_{B(u,v)}(G_\alpha(u) \wedge Q \wedge H_\alpha(v))\\
          Y_{\alpha \hat 1}(u,u') \defeq & \Phi_{B(u,u')}(G_\alpha(u) \wedge Q)=\Phi_{B(u,u')}(G_\alpha(u) \wedge Q \wedge H_{\hat 1}(u'))\ \ \ \ \mbox{see Eq.~\eqref{eq:q:alpha:1:2}}\\
          Y_{\hat 1 \beta}(v',v) \defeq & \Phi_{B(v',v)}(Q \wedge H_\alpha(v))= \Phi_{B(v',v)}(G_{\hat 1}(v') \wedge Q \wedge H_\alpha(v))
	\end{align*}
\end{corollary}

\begin{proof}
  (Sketch) The proof consists of a straightforward extension of the
  proof of Theorem~\ref{th:pq:mobius}.  We begin the proof similarly,
  but in Eq.\eqref{eq:change:1} we replace the sets $U$ and $V$ by
  $U\cup V'$ and $V \cup U'$ respectively.  Next, we apply the
  M\"obius inversion formula repeatedly, once to each $u \in U$ as
  before, but we do not apply it to any $v' \in V'$.  Therefore, there
  are two changes to the expression $\text{REST}$ in
  \eqref{eq:rest:right}.  First, it will have the following residual:
  \begin{align*}
   & \bigwedge_{v' \in V'} \Phi_\Delta(\bigvee_i \forall y G_i(v',y))    
  \end{align*}
  Since $v'$ has a single outgoing edge $(v',v)$, this residual is
  restricted to the lineage in  the block $B(v',v)$, i.e. it is
  equivalent to:
  \begin{align}
   \bigwedge_{v \in V} \Phi_{B(v',v)}(\bigvee_i \forall y G_i(v',y)) \label{eq:residual:1}
  \end{align}
  The second change in~\eqref{eq:rest:right} is that we need to expand
  the expressions $\bigwedge_{u\in U, v \in V} \Phi_{B(u,v)}(Q)$ with
  $\bigwedge_{u\in U} \Phi_{B(u,u')}(Q)$ and
  $\bigwedge_{v \in V} \Phi_{B(v',v)}(Q)$.  The latter absorbs the
  residual~\eqref{eq:residual:1}, because
  $(\bigvee_i \forall y G_i(x,y))$ is implied by $Q$ (i.e. it is one
  of the conjuncts of $Q$):
  \begin{align*}
    \forall v \in V: && \Phi_{B(v',v)}(\bigvee_i \forall y G_i(v',y)) \wedge\Phi_{B(v',v)}(Q)\equiv&\Phi_{B(v',v)}(Q)
  \end{align*}
  Similarly, when we apply M\"obius inversion formula to each
  $v \in V$ we do not apply it to any $u' \in U'$, and are left with a
  similar residual on the right, which also gets absorbed.  Thus, the
  only change to Eq.~\eqref{eq:change:2} is the addition of the
  boolean formulas $\bigwedge_{u\in U} \Phi_{B(u,u')}(Q)$ and
  $\bigwedge_{v \in V} \Phi_{B(v',v)}(Q)$.  Next, we need to modify
  Eq.~\eqref{eq:change:4} from $\bigwedge_{v \in V}$ to
  $\bigwedge_{v \in (V \cup \set{u'})}$, in effect adding the conjunct
  $\bigwedge_{u \in U} \Phi_{B(u,u')}(\forall y G_{\sigma(u)}(u,y))$,
  and similarly for the line below Eq.~\eqref{eq:change:4}.  The
  Boolean formula under $\pr(\cdots)$ in Eq.~\eqref{eq:change:3}
  becomes:
  \begin{align*}
    & \bigwedge_{u\in U, v \in V} \Phi_{B(u,v)}(\forall y(G_{\sigma(u)}(u,y))\wedge Q \wedge \forall x(H_{\tau(v)}(x,v)))
    \wedge \bigwedge_{u\in U} \Phi_{B(u,u')}(\forall y(G_{\sigma(u)}(u,y))\wedge Q)
    \wedge \bigwedge_{v \in V} \Phi_{B(v',v)}(\wedge Q \wedge \forall x(H_{\tau(v)}(x,v)))
  \end{align*}
  Since distinct blocks do not share any tuples, the conjuncts above
  are independent, hence, as before, the probability is their product,
  $\pr(Q) =$
  \begin{align*}
     & \prod_{u\in U, v \in V} \pr(\Phi_{B(u,v)}(\forall y(G_{\sigma(u)}(u,y))\wedge Q \wedge \forall x(H_{\tau(v)}(x,v))))
    \cdot \prod_{u\in U} \pr(\Phi_{B(u,u')}(\forall y(G_{\sigma(u)}(u,y))\wedge Q))
    \cdot \prod_{v \in V} \pr(\Phi_{B(v',v)}(\wedge Q \wedge \forall x(H_{\tau(v)}(x,v))))
  \end{align*}
  Finally, we notice that $u \in U, v \in V$ in the first product can
  be replaced by $(u,v) \in E$, because all non-edges are trivially
  true.
\end{proof}

We can now describe the reduction from $\ccp$ to $\gmc(Q)$.  Let $Q$
be any bipartite, unsafe type II query (meaning: type II-II), and let
$\bar m, \bar n$ be the sizes of its lattice supports $L_0(\mb{G})$ and
$L_0(\mb{H})$.  Since $Q$ is unsafe, we have $\bar m, \bar n \geq 3$.
Our goal is to prove $\ccp(\bar m, \bar n) \leq^P \gmc(Q)$.  Fix an
instance of a $\ccp(\bar m, \bar n)$ problem $(U,V,E)$.  We extend the
graph as before to $(U\cup V', V \cup U', E \cup E')$.  We define the
TID $\pdb$ to be a union of blocks $B(u,v)$ for all
$(u,v) \in E \cup E'$; as before, $B(a,b)$ is trivially true when
$(a,b)$ is not an edge.  The blocks (to be defined in the next
section) will be isomorphic, and therefore, the following quantities
do not depend on $u,v$:
\begin{align}
  y_{\alpha\beta} \defeq & \Pr(Y_{\alpha,\beta}(u,v)) && u \in U, v \in V \label{def:y:alpha:beta}\\
  y_{\alpha*} \defeq & \Pr(Y_{\alpha,\hat 1}(u,u')) && u \in U\nonumber\\
  y_{\hat 1 \beta} \defeq & \Pr(Y_{\hat 1 \beta}(v',v)) && v \in V\nonumber
\end{align}

By Corollary~\ref{cor:pq:mobius}, we obtain:

\begin{align*}
\Pr(Q)  = & \sum_{\substack{\sigma : U \rightarrow L_0(\mb{G}) \\ \tau: V  \rightarrow L_0(\mb{H})}}
\left(\prod_{u\in U} \mu(\sigma(u))\right) \cdot \left(\prod_{v\in V} \mu(\tau(v))\right)\cdot 
\left(\prod_{(u,v) \in E} y_{\sigma(u),\tau(v)}\right) \cdot
\left(\prod_{u \in U} y_{\sigma(u),\hat 1}\right) \cdot
\left(\prod_{v \in V} y_{\hat 1 \tau(v)}\right)
\end{align*}

Recall that $L_0(\mb G) = L(\mb G) - \set{\hat 1}$ is the strict
support of the left lattice, and $\bar m$ is its size. Consider now
the coloring counting problem given by the the graph $(U,V,E)$.  Every
pair $\sigma, \tau$ defines a coloring of the graph $(U,V,E)$.
Denoting
$\mb{k}\eqdef\set{k_{\alpha,\beta},k_{\alpha,\hat 1},k_{\hat 1,\beta}
  \mid \alpha \subseteq [\bar{m}], \beta \subseteq [\bar{n}]}$ its
signature, the factor $y_{\alpha \beta}$ will occur precisely
$k_{\alpha \beta}$ times, i.e. with the exponent $k_{\alpha \beta}$
where $k_{\alpha \beta}$ is the number of edges in the bipartite graph
whose edges are colored $\alpha$ and $\beta$ respectively.  The factor
$y_{\hat 1,\beta}$ will occur $k_{\hat 1,\beta}$ times where
$k_{\hat 1,\beta}$ represents the number of nodes in $V$ colored
$\beta$. Likewise, $k_{\alpha,\hat 1}$ represents the number of nodes
in $U$ colored $\alpha$. Therefore:
\begin{align}
  \Pr(Q) = \sum_{\mb{k}} \#\mb{k} \left(\prod_{\alpha} (\mu(\alpha))^{k_{\alpha,\hat 1}}\right)\cdot
\left(\prod_{\beta} (\mu(\beta))^{k_{\hat 1,\beta}}\right) \cdot 
\left(\prod_{\alpha, \beta} y_{\alpha \beta}^{k_{\alpha \beta}}  \right) \cdot 
\left(\prod_{\beta} y_{\hat 1,\beta}^{k_{\hat 1,\beta}}\right) \cdot
\left(\prod_{\alpha} y_{\alpha,\hat 1}^{k_{\alpha,\hat 1}}\right)
\label{eq:look:here}
\end{align}
where $\alpha, \beta$ range over the strict supports $L_0(\mb G)$,
$L_0(\mb H)$ (i.e. are $\neq \hat 1$).  Our unknowns are $\#\mb{k}$;
there is one unknown for every coloring signature $\mb{k}$.  We
introduce new variables $x_{\mb{k}}$, indexed by the signatures
$\mb{k}$, as follows:
\begin{align*}
x_{\mb{k}} \defeq \#\mb{k} \cdot \left(\prod_{\alpha} (\mu(\alpha))^{k_{\alpha,\hat 1}}\right)\cdot
\left(\prod_{\beta} (\mu(\beta))^{k_{\hat 1,\beta}}\right)
\end{align*}
Thus, one call to the oracle for $\Pr(Q)$ computes the following linear
combination of the unknowns $x_{\mb{k}}$:
\begin{align*}
  \Pr(Q) = \sum_{\mb{k}} \left(\prod_{\substack{\alpha \in L(\mb{G}) \\ \beta \in L(\mb{H}) \\ (\alpha,\beta)\neq (\hat 1, \hat 1)}} y_{\alpha \beta}^{k_{\alpha \beta}}  \right)x_{\mb{k}}
\end{align*}
Notice that here $\alpha, \beta$ range over the {\em entire} support,
i.e. including $\hat 1$, except of the combination
$\alpha=\hat 1, \beta =\hat 1$, because that does not occur in
$\Pr(Q)$, Eq.~\eqref{eq:look:here}.

Let $h \defeq (\bar m + 1)(\bar n+1) = O(1)$.  The equation above has
$(M+1)^{h}$ unknowns $x_{\mb{k}}$ because for every pair
$\alpha \in L(\mb G)$ and $\beta \in L(\mb H)$ there can be between
$0$ and $M$ blocks with the $\alpha\beta$ configuration. Accordingly,
there are $(M+1)^{h}$ coefficients
$y_{\alpha \beta}^{k_{\alpha \beta}}$.  To simplify the notation, lets
denote the pair $\alpha \beta$ by a single index $i$, where
$i=1,\ldots,h$.  The equation becomes:
\begin{align}
  \Pr(Q) = & \sum_{k_1, \ldots, k_h \in \set{0,\ldots,M}} \left(\prod_{i=1,h} y_i^{k_i}\right) x_{k_1 k_2 \cdots k_h}
\label{eq:p:q:type:2:reference}
\end{align}

Let $\bp=\set{p_1, \ldots, p_h}$ be a set of $h$ natural numbers where
$p_i\geq 1$.  We will show in the next several sections how to
construct a block $B(u,v)$ that depends on $\bp$, hence we denote it
$B^{(\bp)}(u,v)$, where all tuples have probabilities in $\set{0,1/2,1}$,
such that:
\begin{align}
y_{i}=\prod_{j=1}^hy_i^{(p_j)} \mbox{ where }  y_i^{(p_j)}\eqdef a_i\lambda_1^{p_j} + b_i\lambda_2^{p_j} && \forall i\in [1,h]
\label{eq:yi_2}
\end{align}
where $\lambda_1$, $\lambda_2$ and $a_i,b_i$, $i\in [1,h]$ are constants independent of $p_1, \ldots, p_h$ satisfying the following.
\begin{align}
\lambda_1\neq \pm\lambda_2 && \mbox{ and } && \lambda_1\neq 0, \lambda_2 \neq 0 \label{eq:conditionLambda2}\\
b_i \neq 0&& && \forall i\in [h] \label{eq:conditiononZero2}\\
a_ib_j{\neq}a_jb_i&&  && i\neq j \label{eq:conditionCoefficients2}
\end{align}

By Theorem~\ref{th:non-singular}, if we set the values of
$p_1, \ldots, p_h$ independently to $1,2,\ldots, (M+1)$, then we
obtain a system with $(M+1)^h$ equations whose matrix is non-singuar,
from which we can compute the unknowns $\#k$ in polynomial time (by
Gaussian elimination), and thus sholve the instance of the
$\ccp(\bar m, \bar n)$ problem.  In the next section we describe how
to construct the block $B^\bp(u,v)$ to ensure that the
probability~\eqref{eq:yi_2} is given by an expression of the
type~\eqref{eq:yi_2}.  We notice that $B^\bp(u,v)$ has the same
structure for all $u,v$; in what follows we only discuss a single
block $B^\bp(u,v)$.


\subsection{Designing the Block $B^{(\bp)}(u,v)$}

Fix $\bp = (p_1, \ldots, p_h)$ a vector of $h$ natural numbers
$\geq 1$.  We describe here the block $B^{(\bp)}(u,v)$; its tuples and
probabilities are the same for all choices of $u,v$, thus our
discussion below does not depend on $u, v$.  The block
$B^{(\bp)}(u,v)$ will consists of a union of $h$ blocks, each
corresponding to one of the parameters $p_1, \ldots, p_h$:
\begin{align*}
  B^{(\bp)}(u,v) = & \bigcup_{j=1,h} B_j^{p_j}(u,v)
\end{align*}
Since the blocks $B_1^{p_1}, \ldots, B_h^{p_h}$ have disjoint sets of tuples, we
have:
\begin{align*}
\forall \alpha \in L(\mb G), \beta \in L(\mb H): \ \ \  Y^{(\bp)}_{\alpha \beta}(u,v) \defeq & \Phi_{B^{(\bp)}(u,v)}(G_\alpha(u) \wedge Q \wedge H_\beta(v))= \bigwedge_{j=1,h}\Phi_{B_j^{p_j}(u,v)}(G_\alpha(u) \wedge Q \wedge H_\beta(v))
\defeq \bigwedge_{j=1,h} Y^{(p_j)}_{\alpha \beta, j}
\end{align*}
Indeed, each clause of the lineage of $G_\alpha(u)$ lies entirely
within one block $B_j^{p_j}$, because
$G_\alpha(u) = \forall y G_\alpha(u,y)$ has a single variable
$\forall y$ (see Eq.~\eqref{eq:g:alpha:u}).  Consider now the query
$Q$.  The only clauses whose lineage may span multiple blocks are
those in
$Q_{\text{left}} \equiv \forall x \left(\bigvee_{j=1,m} \forall y
  G_j(x,y)\right)$, and only those clauses obtained by mapping $x$ to
$u$.  But that formula is absorbed by $G_\alpha(u)$, in other words:
\begin{align*}
  (\forall y G_\alpha(u,y)) \wedge (\bigvee_{j=1,m} \forall y G_j(u,y))\equiv & \forall y G_\alpha(u,y)
\end{align*}
Therefore, the probability $y_{\alpha \beta} = \pr(Y_{\alpha \beta}(u,v))$ is a product of probabilities
one for each block:
\begin{align*}
\forall \alpha \in L(\mb G), \beta \in L(\mb H): \ \ \   y^{(\bp)}_{\alpha \beta} = & \prod_{j=1,h} \pr(\Phi_{B^{(p_j)}_j(u,v)}(G_\alpha(u) \wedge Q \wedge H_\beta(v))) \defeq \prod_{j=1,h} y^{(p_j)}_{\alpha \beta, j}
\end{align*}
We will design the blocks $B_j^{(p_j)}(u,v)$ to be similar, and differ
only in their parameter $p_j$.  To simplify the notation, we drop the
index $j$: thus, the expressions
$p_j, B_j^{(p_j)}(u,v), y_{\alpha \beta, j}, Y^{(p_j)}_{\alpha \beta,
  j}$, etc become
$p, B^{(p)}(u,v), y_{\alpha \beta}, Y^{(p)}_{\alpha \beta}$.  Our goal
is to design the block $B^{(p)}(u,v)$, where $p \geq 1$ is a natural
number, such that, for all $\alpha, \beta$ (including $\hat 1$), we
have:
\begin{align}
  y^{(p)}_{\alpha \beta} = & \pr(Y_{\alpha \beta}) = a_{\alpha \beta} \lambda_1^{p} + b_{\alpha  \beta} \lambda_2^{p}
\label{eq:y:lambda:single}
\end{align}
where the parameters
$\lambda_1, \lambda_2, a_{\alpha \beta}, a_{\alpha, \hat 1}$, etc, are
independent of $p$ and satisfy the
conditions~\eqref{eq:conditionLambda2}-\eqref{eq:conditionCoefficients2}.

We describe now the block $B^{(p)}(u,v)$, and will refer to
Fig.~\ref{fig:block:type:2}.  Let $m$ be the maximum
number of subclauses in any left or right clause; notice that
$m \geq 2$.  An {\em elementary block} $B(a,b)$ is the set of tuples
$B(a,b) \defeq \set{S_1(a,b),S_2(a,b), \ldots}$, i.e. there is exactly
one tuple $S(a,b)$ for each binary symbol $S \in \calR$. 

\begin{definition}
  The block $B^{(p)}(u,v)$ is the disjoint union of the following blocks:
  \begin{itemize}
  \item A prefix block $B_{\text{pref}}(u,r_0)$, which, in turn, is
    the disjoint union of $r$ parallel blocks:
    \begin{align*}
      B_{\text{pref}}(u,r_0) = & \bigcup_{i=1,r}(B(u,t_{\text{pref},i}) \cup B(r_0,t_{\text{pref},i}))
    \end{align*}
    where $B(u,t_{\text{pref},i}), B(r_0,t_{\text{pref},i})$ are
    elementary blocks.  The number $r$ will be chosen later.
  \item A zig-zag part, which is a union  of $2p+1$  elementary
    blocks:
    \begin{align*}
    B(r_0,t_0) \cup  B(r_1,t_0) \cup B(r_1,t_1)\cup  \ldots\cup B(r_p,t_{p-1})\cup B(r_p,t_p)
    \end{align*}    
  \item A suffix block $B_{\text{suff}}(t_p,v)$, which is the union of
    $r$ parallel blocks (same $r$ as for the prefix):
    \begin{align*}
      B_{\text{suff}}(t_p,v) = & \bigcup_{i=1,r}(B(r_{\text{suff},i},t_p)\cup B(r_{\text{suff},i},v))
    \end{align*}
   where 
    $B(r_{\text{suff},i},t_p), B(r_{\text{suff},i},v)$ are elementary
    blocks.
  \item For each constant $r_i$ (including $r_{i,\text{suff}}$)
    introduced above there are $m-2$ dead-end branches of elementary
    blocks: $B(r_i, e^{(1)}_i)\cup \ldots \cup B(r_i, e^{(m-2)}_i)$.
  \item For each constant $t_i$ introduced above (including
    $t_{\text{pref},i}$) there are $m-2$ dead-end branches of elementary
    blocks: $B(f_i^{(1)},t_i) \cup \ldots \cup B(f^{(m-2)}_i,t_i)$.
  \item For any other pairs of constants $a,b$ not explicitly
    mentioned above, there is a trivial elementary block $B(a,b)$
    where all tuples have probability $1$.
  \end{itemize}

We denote by $B(r_0,t_p)$ the {\em zig-zag} portion of the block:
\begin{align}
  B(r_0,t_p) = & B(r_0, t_0) \cup \bigcup_{i=1,p} \left(B(r_i,t_{i-1}) \cup B(r_i,t_i) \cup \bigcup_j \underbrace{(B(r_i,e^{(j)}_i) \cup B(f_{i-1}^{(j)},t_{i-1}))}_{\mbox{dead ends}}\right)
\label{eq:zig:zag:block}
\end{align}
Therefore the entire block is:
\begin{align*}
  B^{(p)}(u,v) = & B_{\text{pref}}(u,r_0) \cup \left(\underbrace{\bigcup_j B(r_0,e^{(j)}_0)}_{\mbox{dead end}}\right) \cup B(r_0, t_p) \cup \left(\underbrace{\bigcup_j B(f_i^{(j)},t_p)}_{\mbox{dead end}}\right)
\end{align*}
\end{definition}

\begin{figure}
  \centering
      \includegraphics[width=0.7\textwidth]{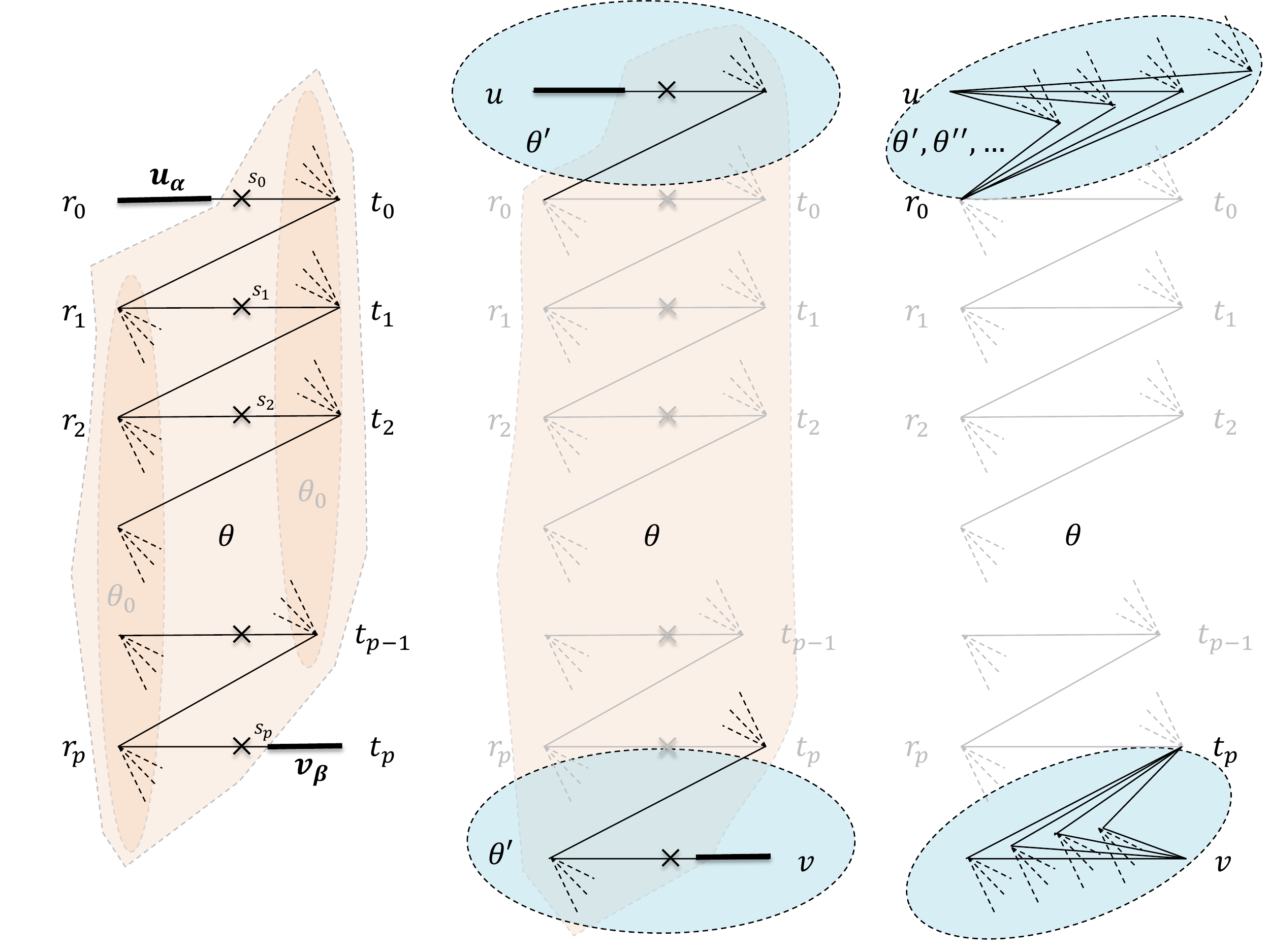}
\hfill
\parbox[t]{0.3\textwidth}{
\centerline{(a)}
Choose $\theta_0 \subseteq \theta$ such that:
\newline $0 < |\lambda_1| < \lambda_2$
\newline $\forall \alpha, \beta: b_{\alpha \beta}>0$
}
\hfill
\parbox[t]{0.3\textwidth}{
\centerline{(b)}
Fix one $(\alpha_1, \beta_1)\neq (\alpha_2,\beta_2)$.
\newline Choose $\theta'$ such that:
\newline $a_{\alpha_1\beta_1}b_{\alpha_2\beta_2}\neq  a_{\alpha_2\beta_2}b_{\alpha_1\beta_1}$
}
\hfill
\parbox[t]{0.3\textwidth}{
\centerline{(b)}
Parallel branches of $\theta'$ such that:
\newline For all $(\alpha_1, \beta_1)\neq (\alpha_2,\beta_2)$:
\newline $a_{\alpha_1\beta_1}b_{\alpha_2\beta_2}\neq  a_{\alpha_2\beta_2}b_{\alpha_1\beta_1}$
}
\hfill
\null
\caption{Illustration of the Block $B^{(p)}(u,v)$ and the progression
  of assignments $\theta_0$, $\theta$, $\theta'$ that, together,
  satisfy conditions~\eqref{eq:conditionLambda2},
  \eqref{eq:conditiononZero2}, and
  \eqref{eq:conditionCoefficients2}. Fig. (a) shows only the zig-zag
  part $B^{(p)}(r_0,t_p)$, where we choose $\theta$ to satisfy
  conditions~\eqref{eq:conditionLambda2} and
  \eqref{eq:conditiononZero2}.  The assignment $\theta$ is independent
  of $\alpha, \beta$ and leaves unassigned some Boolean variables at
  the beginning and that the end ($\mb u_\alpha, \mb v_\beta$).
  Fig. (b) extends the block $B^{(p)}(r_0,t_p)$ with a prefix and a
  suffix consisting of a single branch, thus $B^{(p)}(u,v)$ is
  isomorphic to $B^{(p+2)}(r_0,t_{p+2})$.  Here $\theta's$ extends the
  assignment $\theta$ to the remaining variables in the prefix and
  suffix in order to ensure {\em one} condition
  $a_{\alpha_1\beta_1}b_{\alpha_2\beta_2}\neq
  a_{\alpha_2\beta_2}b_{\alpha_1\beta_1}$.  Fig. (c) extends the
  prefix/suffix with parallel branches in order to satisfy all
  conditions
  $a_{\alpha_1\beta_1}b_{\alpha_2\beta_2}\neq
  a_{\alpha_2\beta_2}b_{\alpha_1\beta_1}$.}
  \label{fig:block:type:2}
\end{figure}


\subsection{Two Properties of $Y^{(p)}_{\alpha \beta}$}

\label{subsec:y:is:connected}

In this section we prove two properties of $Y^{(p)}_{\alpha \beta}$.
First, the mapping $(\alpha, \beta) \mapsto Y^{(p)}_{\alpha \beta}$ is
invertible (see Lemma~\ref{lemma:simple:lattices:1}).  Second, if $Q$
is a forbidden query, then $Y^{(p)}_{\alpha \beta}$ is a connected
Boolean formula, for every $\alpha \in L(\mb G), \beta \in L(\mb H)$
(see Lemma~\ref{lemma:simple:lattices:2}).  To simplify the
discussion, we consider $Y^{(p)}_{\alpha \beta}$ to be the lineage on
the zig-zag block $B^{(p)}(r_0,t_p)$~\eqref{eq:zig:zag:block}, and
will not consider the prefix and the suffix blocks of $B^{(p)}(u,v)$;
the proofs of both properties extend immediately to the complete block
$B^{(p)}(u,v)$.  Thus, $Y^{(p)}_{\alpha \beta}$ means the the lineage
restricted to the zig-zag block $B^{(p)}(r_0,t_p)$:
\begin{align}
  Y^{(p)}_{\alpha \beta} \defeq & \Phi_{B(r_0,t_p)}(G_\alpha(r_0) \wedge Q \wedge H_\beta(t_p))  & y^{(p)}_{\alpha \beta} \defeq & \Pr(Y^{(p)}_{\alpha \beta}) \label{eq:yab:on:zigzag}
\end{align}
As usual, each tuple in $B^{(p)}(r_0,t_p)$ becomes a Boolean variable.

\begin{lemma}
  \label{lemma:type2:invertible} The mapping
  $(\alpha, \beta) \mapsto Y^{(p)}_{\alpha \beta}$ is invertible.
  More precisely, if the logical implication
  $\forall x \forall y Y^{(p)}_{\alpha_1 \beta_1}(x,y) \Rightarrow
  \forall x \forall y Y^{(p)}_{\alpha_2 \beta_2}(x,y)$ holds, then
  $\alpha_1 \leq \alpha_2$ and $\beta_1 \leq \beta_2$, in the lattices
  $\hat L(\mb G)$ and $\hat L(\mb H)$ respectively.
\end{lemma}

\begin{proof}
  The proof is identical to that of
  Lemma~\ref{lemma:simple:lattices:1} and omitted.
\end{proof}

\begin{lemma} \label{lemma:type2:connected} 
  Let $Q$ be a forbidden query of type II.  Then, for all
  $\alpha \in L(\mb G), \beta \in L(\mb H)$, the Boolean formulas
  $Y^{(p)}_{\alpha \beta}$ are connected, and depend on all Boolean
  variables in the block $B^{(p)}(r_0,t_0)$.
\end{lemma}

\begin{proof}
  The proof extends the proof idea in
  lemma~\ref{lemma:simple:lattices:2}, and it is slightly simpler
  because $Q$ is forbidden.  Let $C_0, C_1, \ldots, C_k$ be a
  left-to-right path in $Q$, of minimal length; we will use it to
  construct a long path of clauses in $Y^{(p)}_{\alpha \beta}$ that
  goes through all zig-zag blocks.  We start with one block in the
  zig-zag fragment, say $B(r_i, t_{i-1})$.  The middle clauses
  $C_1, \ldots, C_{k-1}$ have only two logical variables $x,y$, e.g.
  $\forall x \forall y (S_1(x,y) \vee S_2(x,y) \vee \cdots)$, hence
  their groundings
  $C_1(r_i,t_{i-1}), C_2(r_i,t_{i-1}), \ldots C_{p-1}(r_i,t_{i-1})$
  are non-redundant clauses (since they were non-redundant in $Q$) and
  are connected.  Similarly, on the next block in the zig-zag,
  $C_1(r_i,t_i), C_2(r_i,t_i), \ldots C_{k-1}(r_i,t_i)$ are
  non-redundant and connected.  It remains to connect
  $C_1(r_i,t_{i-1})$ and $C_1(r_i,t_i)$ and for that we use the left
  clause $C_0 = \forall x \bigvee_{k=1,\ell} \forall y S_{J_k}(x,y)$.
  Consider the first two sub-clauses $S_{J_1}, S_{J_2}$ in $C_0$:
  since $Q$ is forbidden (see Def.~\ref{def:forbidden}), each has a
  common symbol with $C_1$.  We consider the grounding of $C_0$ that
  maps $S_{J_1}$ to the block $B(r_i,t_{i-1})$, maps $S_{J_2}$ to the
  block $B(r_i,t_i)$, and maps each all other subclauses to distinct
  dead-end branches $B(r_i, e^{(j)}_i)$.  This is possible because
  there are $m-2$ dead-end branches, and $m \geq \ell$, and it is also
  necessary, in order to ensure that the grounded clause is not
  redundant, see Example~\ref{ex:why:we:need:dead:ends} in
  Appendix~\ref{appendix:proof:lemma:long}.  Therefore, this grounding
  is a clause in $Y^{(p)}_{\alpha \beta}$ that is connected to
  $C_1(r_i,t_{i-1})$ via $S_{J_1}(r_i,t_{i-1})$ and is also connected
  to $C_1(r_i,t_i)$ via $S_{J_2}(r_i,t_i)$.  By repeating this for all
  links of the zig-zag chain, we obtain a long sequence of connected
  clauses that start with $C_1(r_0,t_0)$, end with $C_{k-1}(r_p,t_p)$.
  With the same argument we can extend this connected component along
  all dead-end branches, i.e extend it with
  $C_1(r_i,e_i^{(j)}), \ldots, C_{k-1}(r_i,e_i^{(j)})$ for all $i,j$,
  and similarly with
  $C_{k-1}(f_i^{(j)},t_i), \ldots, C_1(f_i^{(j)},t_i)$.  Thus, this
  giant connected component has a zig-zag, with many spikes.  It
  remains to prove that the clauses at the end of the zig-zag and at
  the end of the spikes are also connected.  Now consider the
  beginning of the zig-zag, which is a grounding of
  $G_\alpha(r_0,t_0)$.  When $\alpha \neq \hat 1$, then, as we argued
  in the proof of Lemma~\ref{lemma:simple:lattices:2}, no clause of
  $G_\alpha(r_0,t_0)$ can make $C_1(r_0,t_0)$ redundant, no middle
  clause can make any clause in $G_\alpha(r_0,t_0)$ redundant, and,
  while some clauses within $G_\alpha(r_0,t_0)$ might have
  homomorphisms to others, any remaining  non-redundant clause  of
  $G_\alpha(r_0,t_0)$ contains all ubiquitous symbols.  The case when
  $\alpha = \hat 1$ differs a little from
  Lemma~\ref{lemma:simple:lattices:2}, and here we consider two
  cases.  If some ubiquitous symbol of $Q$ occurs in a middle clause
  $C$, then all ubiquitous symbols of $Q$ occur in some middle clauses
  by Lemma~\ref{lemma:forbidden:properties} (4).  Since each such  middle
  clauses is non-redundant in $C(r_0,t_0)$, and is  connected to the
  path $C_1, C_2, \ldots$ (otherwise it consists only of ubiquitous
  symbols, hence there exists a homomorphism $C \rightarrow C_0$), the
  entire Boolean formula is connected and has all Boolean variables.
  If no ubiquitous symbol occurs in a middle clause, then every clause
  in $G_{\hat 1}(r_0,t_0)$ is non-redundant, because if there were a
  homomorphism from some middle clause $C(r_0,t_0)$, then $C$ must
  contain only non-ubiquitous symbols, hence all are in $C_1$, by the
  Definition~\ref{def:forbidden} of forbidden queries, which implies a
  homomorphism $C \rightarrow C_1$, contradiction.
\end{proof}

Recall from Section~\ref{sec:background:factorization} that the {\em
  distance} of two Boolean variables in a monotone formula is the
smallest number of clauses that connects them.  Fix a left-right path
$C_0, C_1, \ldots, C_k$ in $Q$, of minimal length.  Let $U$ be any
left-ubiquitous symbol in $C_0$; recall that, by
Lemma~\ref{lemma:forbidden:properties}, $U$ does not occur in $C_1$.
Similarly, let $V$ be a right-ubiquitous symbol in $C_k$ and note that
that it does not occur in $C_{k-1}$.  We think of the Boolean
variables (tuples) in the zig-zag block $B^{(p)}(r_0,t_p)$ as being
partially ordered: $U(r_0,t_0)$ is the very ``first'' tuple, and $X$
comes before $Y$ if $d(U(r_0,t_),X) < d(U(r_0,t_0), Y)$ in the Boolean
formula $Y^{(p)}_{\alpha \beta}$.  This is only a pre-order, because
we may have equal distances, it helps understanding the structure of
the block $B^{(p)}(r_0,t_p)$.  The next lemma shows that this
pre-order is independent of the choice of $\alpha, \beta$.

\begin{lemma}
  Let $X$ denote the Boolean variable associated to any tuple of the
  zig-zag block $B(r_0,t_p)$.  Then the distance between $U(r_0,t_0)$
  and $X$ is the same in all formulas $Y^{(p)}_{\alpha \beta}$,
  i.e. it does not depend on $\alpha, \beta$.  Similarly, the distance
  from $X$ to $V(r_p,t_p)$ is the same.
\end{lemma}

\begin{proof}
  Every clause in  $Y^{(p)}_{\alpha \beta}$ that contains $U(r_0,t_0)$
  is connected to the clause $C_1(r_0,t_0)$ (and is not connected to
  $C_2(r_0,t_0)$).  From there, the distance to any variable $X$ is
  the same, regardless of the choice of $\alpha, \beta$.  The same
  argument applies to $Y^{(p)}_{\alpha \hat 1}$ and $Y^{(p)}_{\hat 1 \beta}$.
\end{proof}

\begin{lemma}
  Let $X$ be the Boolean variable associated to a tuple on the main
  branch of the zig-zag block $B(r_0,t_p)$; that is, $X$ has the form
  $X=S(r_i,t_i)$ or $X=S(r_i, t_{i-1})$, but not on a dead-end branch
  like $S(r_i, e_i^{(j)})$.  Assume that $d(U(r_0,t_0), X) \geq 2$ and
  $d(X,V(r_p,t_p)) \geq 2$.  Then $X$ disconnects $U, V$ in
  $Y^{(p)}_{\alpha \beta}$, for all $\alpha, \beta$ (see
  Definition~\ref{def:properties:of:boolean:functions}).
\end{lemma}

\begin{proof}
  Let $S$ be the relational symbol of the tuple $X$, and assume that
  $X = S(r_i,t_i)$; the case $X=S(r_i,t_{i-1})$ is similar.  Since $Q$
  is final, both $Q[S:=0]$ and $Q[S:=1]$ are safe, and this can happen
  in one of two ways.  The first is that the left clause $C_0$ and
  right clause $C_k$ remain left/right clauses in $Q[S:=0]$ (or
  $Q[S:=1]$), but the path $C_1, \ldots, C_{k-1}$ is disconnected,
  i.e. one of the clauses becomes $\texttt{true}$ or becomes
  redundant.  In that case setting $X:=0$ or $X:=1$ also disconnects
  the Boolean variables $U(r_i,t_i)$ from $V(r_i,t_i)$.  The second
  case is when $Q[S:=0]$ or $Q[S:=1]$ has no left clause, or no right
  clause: this happens either because they become $\texttt{true}$, or
  redundant, or they degenerate to middle clauses (e.g. setting
  $S_1:=0$ in
  $\forall y (U(x,y) \vee S_1(x,y)) \vee \forall y (U(x,y) \vee
  S_2(x,y))$).  In that case the connection between $C_1(r_i,t_{i-1})$
  and $C_1(r_i,t_i)$ is broken, again disconnecting $U(r_0,t_0)$ from
  $V(r_p,t_p)$.
\end{proof}

\subsection{Consistent Assignments}

\label{subsec:consistent:assignment}

We want to assign the same probabilities to similar tuples in similar
blocks, e.g. the grounding of $S_3$ should have the same probability
in all blocks $B(r_0,t_0), B(r_1, t_1), B(r_2,t_2), \ldots$ We make
this precise below.  We denote with lower case $s, s', \ldots$ the
real variables representing the probabilities of the ground tuples,
e.g. $S(r_i,t_i)$, $S'(r_i, t_i)$, $\ldots$ We group the Boolean
variables and their associated real variables into equivalence classes
defined follows:

\begin{itemize}
\item For each symbol $S\in \calR$ the {\em odd} equivalence class is
  the set of Boolean variables corresponding to its grounding in the
  odd-numbered zig-zag blocks:
  $S(r_0,t_0), S(r_1,t_1), S(r_2,t_2), \ldots, S(r_p,t_p)$.  
\item For each symbol $S\in \calR$ the {\em even} equivalence class is
  the set of real variables corresponding to its grounding in the
  even-numbered zig-zag blocks:
  $S(r_1,t_0), S(r_2,t_1), S(r_3,t_2), \ldots, S(r_p,t_{p-1})$.
\item For each symbol $S\in \calR$ the $j$'th left dead-end
  equivalence class is the set of variables corresponding to the
  $S(r_1,e_1^{(j)}), \ldots, S(r_p,e_p^{(j)})$; similarly for the
  $j$'th right dead-end equivalence class,
  $S(f_0^{(j)},t_0), S(f_1^{(j)},t_1), \ldots, S(f_{p-1}^{(j)},t_p)$.
\item We will write $\ldots, S_{i-1}, S_i, S_{i+1}, \ldots$ for
  variables in the same equivalence class derived from some symbol
  $S$.  We denote similarly the real variables,
  $\ldots, s_{i-1}, s_i, s_{i+1}, \ldots$ This notation is with some
  abuse, because does not specify whether we mean the odd equivalence
  class, or the even class, or one of the dead end classes.  Depending
  on the type of class, it may contain either $p$ or $p+1$ variables.
\item Finally, we notice that the distance between two consecutive
  variables in the same odd equivalence class or the same even
  equivalence class is exactly $2k$, i.e.  $d(S_{i-1}, S_i) =2k$,
  $\forall i$.  The distance between two consecutive variables in a
  dead-end equivalence class is $\geq 2k$, but in general depends on
  the symbol $S$.
\end{itemize}

\begin{definition}
  Let $\theta$ be any partial assignment from the Boolean variables in
  the zig-zag block $B(r_0, t_p)$ to $\set{0,1/2,1}$.  We say that
  $\theta$ is {\em consistent} if, for every equivalence class $E$ the
  following holds.  If $X, Y$ are two variables in that equivalence
  class and $\theta(X)$ is defined, then either $\theta(Y)$ is also
  defined as $\theta(X)=\theta(Y)$, or $\theta(Y)$ is undefined and
  $Y = S(r_0,t_0)$ or $Y = S(r_p,t_p)$ for some symbol $S$.
\end{definition}

In other words, $\theta$ must act in the same way on the entire
equivalence class, with only exception that it may leave some
variables in the first block $B(r_0, t_0)$ and the last block
$B(r_p,t_p)$ undefined.

We define now a partial, consistent assignment $\theta_0$ as follows,
see also Fig.~\ref{fig:block:type:2} (a).  Let $X$ be a tuple on a
dead-end branch, e.g. $X=S(r_i,e_i^{(j)})$.  Setting $X:=0$ or $X:=1$
may disconnect $U(r_0,t_0)$ from $V(r_p,t_p)$, for example, if $X$
occurs in one grounding of $C_0$ that connects $C_1(r_i,t_{i-1})$ with
$C_1(r_i,t_i)$, then setting $X:=0$ or $X:=1$ may disconnect the main
chain from $U(r_0,t_0)$ to $V(r_p,t_p)$.  If setting $X:=0$ or $X:=1$
does not disconnect $U(r_0,t_0)$ from $V(r_p,t_p)$, then we define
$\theta_0(X):=0$, or $1$ respectively, and do this for all tuples the
equivalence class of $X$; it holds that the tuples $U(r_0,t_0)$ and
$V(r_p,t_p)$ remain connected in all Boolean formulas
$Y^{(p)}_{\alpha \beta}[\theta_0]$, for all $\alpha, \beta$.  We
proceed similarly with the dead-end tuples on the right, i.e. we set
an entire equivalence class to 0 or to 1 if $U(r_0,t_0)$ and
$V(r_p,t_p)$ remain connected.  After this process, $\theta_0$ has the
following property:

\begin{definition} Let $U, V$ be the left- and right-ubiquitous
  symbols introduced above.  We say that the consistent partial
  assigning $\theta_0$ is {\em final} if, forall $\alpha, \beta$, the
  Boolean function $Y^{(p)}_{\alpha \beta}[\theta_0]$ is connected,
  and, for any tuple $X$ in $B(r_0,t_p)$, setting $X:=0$ or $X:=1$
  disconnects $U(r_0,t_0)$ from $V(r_p,t_p)$.
\end{definition}

We will fix $\theta_0$ from now on.

\subsection{{The Eigenvalues of the Zig-zag Block}}

We will now give a closed form formula for the probability of
$Y^{(p)}_{\alpha \beta}$, which, recall, we are using temporarily to
denote the lineage on the zig-zag block $B(r_0,t_p)$, see
Eq.~\eqref{eq:yab:on:zigzag}, i.e. without the prefix/suffix
blocks. Since $\alpha \in L(\mb G)$ and $\beta \in L(\mb H)$, the
probabilities $y^{(p)}_{\alpha \beta}$ form a
$(\bar m+1) \times (\bar n+1)$ matrix.  However, we show that this
matrix has rank 2, hence it can be described by a $2 \times 2$ matrix,
and, as a consequence, $y^{(p)}_{\alpha \beta}$ is given in terms of
two eigenvalues:
$y^{(p)}_{\alpha \beta} = a_{\alpha \beta} \lambda_1^{p} + b_{\alpha
  \beta} \lambda_2^{p}$ for $\lambda_1,\lambda_2 \neq 0$ and
$\lambda_1 \neq \pm \lambda_2$.

Let $C_0, C_1, \ldots, C_k$ be a left-right path in $Q$, of minimal
length, and recall that $U, V$ are two ubiquitous symbols that do not
occur in $C_1, C_{k-1}$.  Fix an index $j$ such that
$3 \leq j \leq k-5$, and fix a symbol
$S \in \symb(C_j) \cap \symb(C_{j+1})$.  Since $Q$ is final, both
$Q[S:=0]$ and $Q[S:=1]$ are safe queries, which implies that $S$
disconnects $Q_{\text{left}}$ from $Q_{\text{right}}$ in both
$Q[S:=0]$ and $Q[S:=1]$.  We will call $S$ an {\em articulation
  symbol}.  We denote by $S_0, S_1, \ldots, S_p$ the Boolean variables
forming the odd equivalence class $S(r_0,t_0), \ldots, S(r_p,t_p)$,
and denote by $s_0, s_1, \ldots, s_p$ their associated real
variables.  In Fig.~\ref{fig:block:type:2} (a) the articulation
variables are shown with an $x$. 

Denote by $B(r_0,r_i)$ and $B(t_i,t_p)$ the following  subsets of the
zig-zag block $B(r_0,t_p)$:
\begin{align*}
  B(r_0,r_i) \defeq & \bigcup_{\ell=1,i} \left(B(r_{\ell-1},t_{\ell-1}) \cup B(r_\ell,t_{\ell-1}) \cup \bigcup_j \left(B(r_\ell,e^{(j)}_\ell) \cup B(f^{(j)}_{\ell-1},t_{\ell-1})\right)\right)\\
  B(t_i,t_p) \defeq & \bigcup_{\ell=i+1,p} \left(B(r_\ell,t_{\ell-1}) \cup B(r_\ell,t_\ell) \cup \bigcup_j \left(B(f^{(j)}_{\ell-1},t_{\ell-1})\cup B(r_\ell,e^{(j)}_\ell)\right)\right)
\end{align*}
These are precisely the two connected components of $B(r_0,t_p)$ after
we remove the single block $B(r_i,t_i)$.

\begin{lemma}
  The Boolean variable $S_i$ disconnects $B(r_0,r_i)$ and $B(t_i,t_p)$
  (see Def.~\ref{def:properties:of:boolean:functions}) in the Boolean
  function $Y^{(p)}_{\alpha \beta}$, for any $\alpha, \beta$.
\end{lemma}

In other words, we have:
\begin{align}
 Y_{\alpha \beta}^{(p)}[S_i := 0] = & A_0 \wedge B_0 &
 Y_{\alpha \beta}^{(p)}[S_i := 1] = & A_1 \wedge B_1 \label{eq:articulation:one:step}
\end{align}
where all Boolean variables from $B(r_0,r_i)$ are in $A_0, A_1$, and
all Boolean variables from $B(t_i,t_p)$ are in $B_0, B_1$.  The proof
follows immediately from the fact that $S$ disconnects
$Q_{\text{left}}$ from $Q_{\text{right}}$ in both $Q[S:=0]$ and
$S[S:=1]$.

Since $S_\ell$ is in $A_0, A_1$ for $\ell < i$, we can repeat this
process and use $S_\ell$ to disconnect $A_0, A_1$, etc.  More
generally, consider any values $v_0, v_1, \ldots, v_p \in \set{0,1}$.
Then, assigning these values to the variables of the articulation
symbol splits $Y_{\alpha \beta}^{(p)}$ into $p+2$ formulas with no
common Boolean variables.
\begin{align}
  Y_{\alpha \beta}^{(p)}[S_0:=v_0, S_1:=v_1, \cdots, S_p:=v_p] = & U_\alpha^{(v_0)} \wedge Z_1^{(v_0v_1)} \wedge \cdots \wedge Z_{p}^{(v_{p-1}v_p)} \wedge V_\beta^{(v_p)}\label{eq:articulate:big:y}
\end{align}
The only expressions that depend on the choice of $\alpha, \beta$ are
$U_\alpha^{(v_0)}$ and $V_\beta^{(v_p)}$.  Since all Boolean
expressions on the RHS in~\eqref{eq:articulate:big:y} have disjoint
Boolean variables, their probabilities are independent, hence their
arithmetization is a product of multilinear polynomials:
\begin{align*}
  y_{\alpha \beta}^{(p)}[s_0:=v_0, s_1:=v_1, \cdots, s_p:=v_p] = & u_\alpha^{(v_0)} \cdot z_1^{(v_0v_1)} \cdots z_p^{(v_{p-1}v_p)} \cdot v_\beta^{(v_p)}
\end{align*}
We express the result in matrix form, where $\diag(a,b)$, denotes the
diagonal matrix
$\left[\begin{array}{cc}a & 0 \\ 0 & b\end{array}\right]$.  This is an
important expression which we define formally:

\begin{definition} \label{def:polynomial:y:alpha:beta} The polynomial $y^{(p)}_{\alpha \beta}$ is defined
  as follows:
\begin{align*}
  y^{(p)}_{\alpha, \beta} =
& \left[
  \begin{array}{cc}
    u_\alpha^{(0)}  & u_\alpha^{(1)}
  \end{array}
\right]
\cdot
\diag(1-s_0, s_0)
\cdot
\left[
 \begin{array}{cc}
  z_1^{(00)} &   z_1^{(01)}\\
  z_1^{(10)} &   z_1^{(11)}
 \end{array}
\right]
\cdot
\diag(1-s_1, s_1)
\cdot
\left[
 \begin{array}{cc}
  z_2^{(00)} &   z_2^{(01)}\\
  z_2^{(10)} &   z_2^{(11)}
 \end{array}
\right]
\cdots
\left[
 \begin{array}{cc}
  z_p^{(00)} &   z_p^{(01)}\\
  z_p^{(10)} &   z_p^{(11)}
 \end{array}
\right]
\cdot
\diag(1-s_p, s_p)
\left[
 \begin{array}{c}
  v_\beta^{(0)}\\
  v_\beta^{(1)}
 \end{array}
\right]
\end{align*}
\end{definition}

Denote by $\mb z_i$ the $2 \times 2$ matrix:
\begin{align*}
\mb z_i \defeq &
\left[
 \begin{array}{cc}
  z_i^{(00)} &   z_i^{(01)}\\
  z_i^{(10)} &   z_i^{(11)}
 \end{array}
\right]
\end{align*}
We view it as matrix of multilinear polynomials, where the variables
represent the (yet unchosen) probabilities of the tuples in all our
blocks.

Our next task is to define a consistent assignment $\theta$ that
extends $\theta_0$ from Sec.~\ref{subsec:y:is:connected}, such that
$\det(\mb z_i[\theta]) \neq 0$.  For that we will use
Lemma~\ref{lemma:determinant:connected} in the introduction. However,
we have a problem: the matrix $\mb z_i$ may contain two variables from
the same equivalence class, and the assignment $\theta$ from
Lemma~\ref{lemma:determinant:connected} might assign them different
values, making $\theta$ is inconsistent.  This happens whenever the
decomposition in Eq.~\eqref{eq:articulation:one:step} has migrating
variables.  Since we chose the articulation variable $S$ in
$\symb(C_j) \cap \symb(C_{j+1})$, the Boolean variable
$S_i = S(r_i,t_i)$ disconnects
$U_i \defeq U(r_i,t_i), V_i \defeq V(r_i,t_i)$, i.e. the
left-ubiquitous symbol and the right-ubiquitous symbol in block
$B(r_i,t_i)$.  The distance from $U_i$ to $S_i$ is $j$, and the
distance from $S_i$ to $V_i$ is $k-j-1$, hence, by
Lemma~\ref{lemma:disconnect:balls} item~\ref{item:disconnect:balls},
$S_i$ also separates all the symbols in $C_0,\ldots,C_{j-2}$ from
$C_{j+2}, \ldots, C_k$ (more precisely: their groundings in the block
$B(r_i,t_i)$), hence the only possible migrating variables are of the
form $S'_i \defeq S'(r_i,t_i)$, with
$S' \in \symb(C_{j-1})\cup \cdots \cup \symb(C_{j+1})$.  In other
words, $d(S'_i, S_i) \leq 1$.  Suppose that $S'_i$ migrates from left
to right,  i.e. it occurs in $A_0$ and in $B_1$ in
Eq.~\eqref{eq:articulation:one:step}.  Then:
\begin{itemize}
\item $s_i'$ occurs in the first column of $\mb z_i$, because this is
  part of $A_0$.
\item $s_i'$ occurs in the second row of $\mb z_{i+1}$, because this
  is part of $B_0$; then it's equivalent variable $s_{i-1}'$ occurs in
  the second row of $\mb z_i$.
\end{itemize}
If $S'_i$ migrates from right to left, i.e. from $B_1$ to $A_0$, then:
\begin{itemize}
\item $s_i'$ occurs in the first row of $\mb z_{i+1}$, because this is
  part of $B_0$; hence $s'_{i-1}$ occurs in the first row of $\mb
  z_i$.
\item $s_i'$ occurs in the second column of $\mb z_i$, because this is
  part of $A_1$.
\end{itemize}
Pictorially, these two cases are illustrated as follows:
\begin{align}
  \mb z_i =
&
\left[
  \begin{array}{ll}
    z_i^{(00)}(s_i') & z_i^{(01)} \\
    z_i^{(10)}(s_{i-1}', s_i') & z_i^{(11)}(s_{i-1}')
  \end{array}
\right]
&
  \mb z_i =
&
\left[
  \begin{array}{ll}
    z_i^{(00)}(s'_{i-1}) & z_i^{(01)}(s'_{i-1},s_i') \\
    z_i^{(10)} & z_i^{(11)}(s_i')
  \end{array}
\right] \label{eq:how:variables:migrate}
\end{align}

We seek a consistent assignment $\theta$, hence we need
$\theta(s_{i-1}')=\theta(s_i')$, and for that we will simply
substitute the real variable $s'_{i-1}$ with the variable $s'_i$.  To
do that, we need the following technical lemma.

\begin{lemma} \label{lemma:degree:2:connected}
  Let $F$ be a connected Boolean function, and
  $\mb A = (A_1, \ldots, A_k)$, $\mb B=(B_1, \ldots, B_k)$ be two
  tuples of $k$ Boolean variables each.  Let $X$ be a variable that
  disconnects $\mb A, \mb B$, such that  $d(\mb A, X) \geq 3$ and
  $d(\mb B, X) \geq 3$.  Let $f$ be the arithmetization of $F$.  Then
  the polynomial $f[\mb b := \mb a]$, where each variable $b_i$ is
  substituted by the variable $a_i$, is irreducible.
\end{lemma}

Notice that, while $f$ is a multilinear polynomial, each variable
$a_i$ has degree 2 in the polynomial $f[\mb b := \mb a]$.  In general,
if $f$ is irreducible, then equating variables does not keep it
irreducible.  For example, if $F=A \vee B$ then $f = a+b-ab$ and
$f[b:=a] = 2a - a^2$ factorizes as $a(2-a)$.  The lemma gives
sufficient conditions for which $f[\mb B := \mb A]$ is irreducible.

\begin{proof}
  By Lemma~\ref{lemma:disconnect:balls} (\ref{item:disconnect:balls})
  $X$ disconnects a ball around of $\mb A$ from a ball around $\mb B$.
  Let $C, D$ be two variables s.t. $d(\mb A, C)=d(\mb B,D)=1$, thus,
  $X$ disconnects $\mb A C, \mb B D$.  Decompose the multilinear
  polynomial according to $X$:
  \begin{align*}
    f[x:=0] = & g_0(\mb a,c) \cdot h_0(\mb b,d) & f[x:=1] =  & g_1(\mb a,c) \cdot h_1(\mb b,d)\\
    f[\mb b:=\mb a, x:=0] = & g_0(\mb a,c) \cdot h_0(\mb a,d) & f[\mb b:=\mb a, x:=1] =  & g_1(\mb a,c) \cdot h_1(\mb a,d)
  \end{align*}
  where we indicated that $\mb a,c$ occurs only in $g_0, g_1$, while
  $\mb b, d$ occurs only in $h_0, h_1$.  While $g_0(\mb a,c)$ may
  further factorize, all variables $\mb a, c$ are in the same
  irreducible factor, because their clauses are connected in
  $F[X:=0]$; similarly for the other three expressions.  Assume now
  that $f[\mb b := \mb a]$ factorizes; since the degree of $X$ in
  $f[\mb b := \mb a]$ is $\leq 1$, there exists an irreducible factor
  $k$ of $f[\mb b := \mb a]$ that does not contain $X$.  The $k$
  divides both expressions in the second line above.  We consider
  three cases.

  \begin{description}
  \item[Case 1:] If $k$ contains the variable $c$, then it must divide
    both $g_0(\mb a, c)$ and $g_1(\mb a, c)$.  It follows that $k$
    divides $f = g_0\cdot h_0 \cdot (1-x) + g_1 \cdot h_1 \cdot x$.
    This is a contradiction because $f$ is irreducible.
  \item[Case 2:] If $k$ contains the variable $d$, the it must divide
    $h_0(\mb a, d)$ and $h_1(\mb a, d)$.  In particular $k$ contains
    the variables $\mb a$, and the degrees of all $\mb a$-variables is
    $1$, because their degree in $h_0(\mb a, d), h_1(\mb a, d)$ is 1.
    Substitute $\mb b$ for $\mb a$ in $k$, we a conclude that
    $k[\mb a:= \mb a]$ divides both $h_0(\mb b, d)$ and
    $h_1(\mb b, d)$. Hence it divides $f$ by the argument in Case 1,
    and we reached a contradiction.
  \item[Case 3] If $k$ contains neither $c$ nor $d$, then by the
    argument above it cannot contain any of $\mb a$.  Since $k$
    divides $g_0(\mb a, c) \cdot h_0(\mb a, d)$, it also divides
    $g_0(\mb a, c) \cdot h_0(\mb b, d)$, and it similarly divides
    $g_1(\mb a, c) \cdot h_1(\mb b, d)$.  This, again, implies that
    $k$ divides $f$, which is a contradiction.
  \end{description}
\end{proof}

We will now prove the existence of a consistent assignment $\theta$
that extends $\theta_0$ such that $\det(\mb z_i[\theta]) \neq 0$.
Recall that $\theta_0$, introduced in the previous section, has the
property that all Boolean functions $Y^{(p)}_{\alpha \beta}[\theta_0]$
are connected, and every variable $X$ disconnects $U(r_0,t_0)$ from
$V(r_p,t_p)$.

\begin{theorem} \label{th:assignment:for:lambdas} Let
  $\mb V \defeq \bigcup_i \vars(\mb z_i) \cup \set{s_0, s_1, \ldots,
    s_p}$, and define $\theta$ the following assignment: if
  $\theta_0(x)$ is defined then $\theta(x) \defeq \theta_0(x)$,
  otherwise $\theta(x) \defeq 1/2$.  Then, for every $i$,
  $\det(\mb z_i[\theta]) \neq 0$.
\end{theorem}

\begin{proof}
  Choose arbitrary $\alpha, \beta$, and recall that the polynomial
  $y^{(p)}_{\alpha \beta}[\theta_0]$ is irreducible
  (Lemma~\ref{lemma:type2:connected}).  Fix $i = 1,\ldots,p$: we will
  first define some consistent $\theta$ that extends $\theta_0$ such
  that $\det(\mb z_i[\theta]) \neq 0$.  Denote by $\mb a$ the set of
  migrating variables $s'_i$, and denote by $\mb b$ the set of the
  predecessor equivalent variables $s'_{i-1}$; as we have seen, both
  $s'_i$ and $s'_{i-1}$ belong to $\mb z_i$.  Let $\mb A, \mb B$ their
  corresponding Boolean variables.  To apply
  Lemma~\ref{lemma:degree:2:connected}, we need to find some variable
  $X$ that disconnects $\mb A, \mb B$ and that is far from both.
  Using the left-right path $C_0, C_1, \ldots, C_k$, choose any symbol
  $S'' \in C_j$ for any $1 < j < k-1$, and define
  $X \defeq S''(r_i, t_{i-1})$.  Since all $\mb B$-variables are in
  the block $B(r_{i-1},t_{i-1})$ and all $\mb A$-variables are in
  $B(r_i, t_i)$, we have $d(\mb B, X) \geq 2$ and
  $d(\mb A, X) \geq 2$. Therefore, by the lemma, the polynomial
  $y^{(p)}_{\alpha \beta}[\mb b := \mb a, \theta_0]$ is irreducible.

  \begin{claim}
    $\det(\mb z_i[\mb b := \mb a, \theta_0]) \not\equiv 0$.
  \end{claim}

  \begin{proof}
    Assume the contrary, that $\det(\mb z_i[\mb b := \mb a, \theta_0]) \equiv
    0$. Then, by Theorem~\ref{th:rank:1:matrix:of:polynomials}, the
    matrix $\mb y_i[\mb b := \mb a, \theta_0]$ has rank 1, more
    precisely there exists polynomials $g_0, g_1, h_0, h_1$, such that:
  \begin{align*}
    \left[
    \begin{array}{c}
      g_0 \\ g_1
    \end{array}
\right] \cdot
    \left[
    \begin{array}{cc}
      h_0 & h_1
    \end{array}
\right] \equiv & \mb z_i[\mb b := \mb a, \theta_0]
  \end{align*}
  Substituting this expression in
  $y^{(p)}_{\alpha \beta}[\mb b := \mb a, \theta_0]$ we obtain a
  factorization:

\begin{align*}
  y^{(p)}_{\alpha_0 \beta_0}[\mb b := \mb a, \theta_0] = 
& 
\underbrace{
\left[
  \begin{array}{cc}
    u_{\alpha_0}^{(0)}  & u_{\alpha_0}^{(1)}
  \end{array}
\right]
\cdots
\left[
  \begin{array}{c}
    g_0 \\ g_1
  \end{array}
\right]}_{\mbox{factor 1}}
\cdot
\underbrace{
\left[
    \begin{array}{cc}
      h_0 & h_1
    \end{array}
\right]
\cdots
\left[
 \begin{array}{c}
  v_\beta^{(0)}\\
  v_\beta^{(1)}
 \end{array}
\right]}_{\mbox{factor 2}}
\end{align*}
This contradicts the fact that
$y^{(p)}_{\alpha \beta}[\mb b := \mb a, \theta_0]$ is irreducible.
This completes the proof of the claim.
  \end{proof}

  \begin{claim}
    There exists an assignment $\theta_i$ of the variables of
    $\mb z_i[\mb b := \mb a]$ with values in $\set{0,1/2,1}$
    s.t. $\theta_i$ extends $\theta_0$ and
    $\det(\mb z_i[\mb b: := \mb a, \theta_i]) \neq 0$.
  \end{claim}

  \begin{proof}
    We use Lemma~\ref{lemma:determinant:connected} in the
    introduction.  To apply it we must verify that
    $\det(\mb z_i[\mb b := \mb a, \theta_0])$ is a polynomial of degree
    $\leq 2$ in each variable.  This follows immediately by inspecting
    Eq.~\eqref{eq:how:variables:migrate}.  When $s_i'$ migrates from
    left to right then the variable $s_i'=s'_{i-1}$ has degree 2 in
    $y_i^{(10)}$ and degree 0 in $y_i^{(01)}$, hence it has degree 2
    in $\det(\mb z_i)$, and similarly for variables that migrate from
    right to left. Therefore, by
    Lemma~\ref{lemma:determinant:connected}, there exists an
    assignment $\theta_i$ of the variables in
    $\mb z_i[\mb b := \mb a, \theta_0]$ (hence: an extension of
    $\theta_0$) such that
    $\det(\mb z_i[\mb b := \mb a, \theta_i]) \neq 0$.
  \end{proof}

  From here we derive immediately:

  \begin{claim}
    There exists a consistent assignment $\theta$ of the variables
    $\mb V$ that extends $\theta_0$ and $\det(\mb z_i[\theta]) \neq 0$
    for all $i=1,p$.
  \end{claim}

  \begin{proof}
    Take $\theta = \theta_1 \cup \theta_2 \cup \cdots \cup \theta_p$.
    While $\theta_{i-1}, \theta_i$ are defined on some common
    variables (the migrating variables) they have the same values.
    Finally, assign
    $\theta(s_0) = \theta(s_1) = \cdots = \theta(s_p)=1/2$ (since the
    separator variables are not part of any matrix $\mb z_i$).
  \end{proof}

  Finally, we prove that $\theta$ assigns the value $1/2$ to every
  variable not in the domain of $\theta_0$.

  \begin{claim}
    Assume $p \geq 3$, and let $\theta$ be any consistent assignment
    of the variables $\mb V$ that extends $\theta_0$.  Suppose that
    there exists some variable $X \not \in \Dom(\theta_0)$, such
    $\theta(X)=0$ or $\theta(X)=1$.  Then there exists $i$ such that
    $\det(\mb z_i[\theta])= 0$.
  \end{claim}

  \begin{proof}
    Let $X$ be any variable $\not \in \Dom(\theta_0)$.  Consider the
    product of matrices in Def.~\ref{def:polynomial:y:alpha:beta} that
    defines the polynomial $y^{(p)}_{\alpha, \beta}$; then
    $y^{(p)}_{\alpha, \beta}[\theta_0]$ is the same product, where
    $\theta_0$ is applied to each matrix, i.e. $\mb z_i[\theta_0]$.
    $X$ is a variable that occurs in either one, or at most two
    consecutive matrices (when it migrates).  Assuming it occurs in
    $\mb z_i, \mb z_{i+1}$, and we split the expression in
    Def.~\ref{def:polynomial:y:alpha:beta} into three parts
    \begin{align*}
      y^{(p)}_{\alpha, \beta}[X:=0,\theta_0]= & \mb a \cdot \diag(1-s_{i-1},s_{i-1})  \cdot \mb w \cdot  \diag(1-s_{i+1},s_{i+1})\cdot \mb b
    \end{align*}
Where:
    \begin{align*}
      \mb a \defeq & \mb u \cdot \prod_{\ell=1,i-1}\diag(1-s_{\ell-1}, s_{\ell-1})\cdot \mb  z_\ell[\theta_0]
      & \mb w \defeq & \mb z_i[X:=0,\theta_0] \cdot \diag(1-s_i)\cdot \mb z_{i+1}[X:=0,\theta_0]
      & \mb b \defeq & \prod_{\ell=i+2,p} \mb z_\ell[\theta_0] \cdot  \diag(1-s_\ell,s_\ell) \cdot \mb v
    \end{align*}
    We prove that, if $y^{(p)}_{\alpha, \beta}[X:=0,\theta_0]$
    factorizes, then $\det(\mb w) \equiv 0$.  This implies that either
    $\det(\mb z_i[X:=0,\theta_0]) \equiv 0$ or
    $\det(\mb z_{i+1}[X:=0,\theta_0]) \equiv 0$.  We will denote
    $s \defeq s_{i-1}$ and $s' \defeq s_{i+2}$ to reduce clutter,
    thus:
    \begin{align*}
      y^{(p)}_{\alpha, \beta}[X:=0,\theta_0]=  & a_0 b_0 w_{00} (1-s)(1-s') + a_0 b_1 w_{01} (1-s)s' + a_1 b_0 w_{10} s (1-s') + a_1 b_1 w_{01} ss' 
    \end{align*}
    Since $y^{(p)}_{\alpha, \beta}[X:=0,\theta_0]$ is reducible, let
    $f$ be an irreducible factor $f$ that contains the variable
    $U(r_0,t_0)$, where $U$ is a left ubiquitous symbol.  That is, $f$
    contains a ``far left'' variable, which only occurs in $a_0$ and
    $a_1$.  In particular, it does not share any variables with
    $b_0, b_1$, because $X$ separates them in the Boolean formula
    $Y^{(p)}_{\alpha, \beta}[\theta_0]$.  We have:
    \begin{align*}
      f \big| & a_0 b_0 w_{00} (1-s)(1-s') + a_0 b_1 w_{01} (1-s)s' + a_1 b_0 w_{10} s (1-s') + a_1 b_1 w_{01} s' 
    \end{align*}
    Consider now the variables $s, s'$: $f$ may contain neither, or
    just $s$, or both $s, s'$ (since $U(r_0,t_0)$ is closer to
    $S_{i-1}$ than to $S_{i+2}$) in the latter case we switch the
    roles of $\mb a, \mb b$, i.e. start with some irreducible factor
    that contains some ``far right'' variable $V(r_p,t_p)$. Hence we
    will assume w.l.o.g. that $f$ does not contain $s'$.  Then, can
    set separately $s'=0$ and $s'=1$ and obtain:
    \begin{align}
      f \big| & (a_0  w_{00} (1-s) + a_1 w_{10} s)b_0 \label{eq:out:of:names}\\
      f \big| & (a_0  w_{01} (1-s) + a_1 w_{11} s)b_1 \nonumber
    \end{align}
    Consider first the case when $f$ does not contain $s$ either.  In
    that case we repeat the argument, and obtain:
    \begin{align*}
       f \big|& a_0b_0 w_{00} & f \big|& a_0b_1 w_{01} & f \big|& a_1b_0 w_{10} &  f \big|& a_1b_1 w_{00}
    \end{align*}
    It follows that $f | a_0$ and $f | a_1$, which implies
    $f | y^{(p)}_{\alpha, \beta}[\theta_0]$ (without setting $X:=0$),
    which contradicts our assumption.  Consider now the case when $f$
    contains $s$.  Notice that we can remove $b_0, b_1$
    from~\eqref{eq:out:of:names}, since $f$ shares no variables with
    them.  We eliminate $a_1$ by multiply the first line by $w_{11}$
    and the second by $-w_{10}$, then eliminate $a_0$ similarly, and
    obtain:
    \begin{align*}
      f \big| & a_0 (1-s) (w_{00}w_{11} - w_{01}w_{10}) \\
      f \big| & a_1 s (w_{00}w_{11} - w_{01}w_{10})
    \end{align*}
    If $(w_{00}w_{11} - w_{01}w_{10})\neq 0$, then $f$ divides both
    $a_0(1-s)$ and $a_1s$, again implying that
    $f | y^{(p)}_{\alpha, \beta}[\theta_0]$, which is a contradiction.
  \end{proof}
\end{proof}

In summary, we have

\begin{align}
\forall i: \ \ \  \mb z_i[\theta] \defeq \mb z =
&
\left[
 \begin{array}{cc}
  z_{00} &   z_{01}\\
  z_{10} &   z_{11}
 \end{array}
\right] \label{eq:the:z}
\end{align}

\begin{lemma}  \label{lemma:zij:not:zero}
  $z_{00}, z_{01}, z_{10}, z_{11} > 0$.
\end{lemma}

\begin{proof} Consider $z_{00}$: this is the probability of the
  Boolean expression $Z_i^{(00)}[\theta_0]$
  in~\eqref{eq:articulate:big:y}, for any choice of $i$.  We claim
  that $Z_i^{(00)}[\theta_0]\not\equiv \texttt{false}$.  By
  definition, $Z_i^{(00)}[\theta_0]$ is the formula obtained by
  factorizing:
  \begin{align*}
    Y_{\alpha \beta}^{(p)}[S_{i-1}:=0, S_i:=0, \theta_0] = & A_0 \wedge Z_i^{(00)}[\theta_0] \wedge B_0
  \end{align*}
  We claim that this expression is not $\equiv \texttt{false}$.
  Indeed, by definition of $\theta_0$,  $Y_{\alpha
    \beta}^{(p)}[\theta_0]$ is a connected monotone Boolean function,
  and by our choice of the variables $S_i$, their distance is
  $d(S_{i-1}, S_i)= 2k$, hence they are neither prime implicants, nor
  do they occur together in a clause.  It follows that by setting both
  to $0$ we not make $Y_{\alpha \beta}^{(p)}[S_{i-1}:=0, S_i:=0,
  \theta_0] \equiv \texttt{false}$. Since all Boolean variables in
  $Z_i^{(00)}[\theta_0]$ have been assigned probability $1/2$, it
  follows that its probability is $> 0$, proving $z_{00} > 0$.  The
  other cases are similar and omitted.
\end{proof}

Denote by $\lambda_1, \lambda_2$ the eigenvalues of the matrix
$\mb z$.    We also assign probabilities $1/2$ to all the articulation points:
$s_1 = s_2 = \cdots = s_p = 1/2$, and obtain:

\begin{align}
  y^{(p)}_{\alpha \beta}[\theta] = &
  \frac{1}{2^{p+1}} 
\left[
  \begin{array}{cc}
    u_\alpha^{(0)}  & u_\alpha^{(1)}
  \end{array}
\right]
\cdot
\left[
 \begin{array}{cc}
  z_{00} &   z_{01}\\
  z_{10} &   z_{11}
 \end{array}
\right]^p
\cdot
\left[
 \begin{array}{c}
  v_\beta^{(0)}\\
  v_\beta^{(1)}
 \end{array}
\right] =
\frac{1}{2}\left(a_{\alpha \beta}(\lambda_1/2)^p + b_{\alpha \beta} (\lambda_2/2)^p\right)
\label{eq:y:alpha:beta:lambda:p}
\end{align}
where $a_{\alpha \beta}$, $b_{\alpha \beta}$ are coefficients that are
independent of $p$.  We prove now
Condition~\eqref{eq:conditionLambda2}:

\begin{theorem} \label{th:lambda:type2} The following hold:
  $0 < |\lambda_1| < \lambda_2$.
\end{theorem}

\begin{proof}
  $\det(\mb z) \neq 0$ implies $\lambda_1, \lambda_2 \neq 0$.
  $\lambda_1 \neq \lambda_2$, because both $z_{01}$ and $z_{10}$ are
  $\neq 0$, since they represent probabilities.
  $\lambda_1 + \lambda_2 > 0$, because the trace of the matrix is
  $z_{00} + z_{11} > 0$ as a sum of two probabilities.
\end{proof}

Next, prove condition~\eqref{eq:conditiononZero2}.

\begin{lemma}
  Assume $\lambda_1 < \lambda_2$, then, for all $\alpha, \beta$,
  $b_{\alpha \beta} > 0$.
\end{lemma}

\begin{proof}
  We first expand the power $p$ of the matrix $\mb z$:
  \begin{align*}
    \mb z^p = &
\left[
 \begin{array}{cc}
  a_1\lambda_1^p + a_2\lambda_2^p &   b_1\lambda_1^p + b_2\lambda_2^p\\
  c_1\lambda_1^p + c_2\lambda_2^p &   d_1\lambda_1^p + d_2\lambda_2^p
 \end{array}
\right]
  \end{align*}
We prove that $a_2, b_2, c_2, d_2 > 0$.  We start by observing that
$\lambda_1, \lambda_2$ are the solutions to:
\begin{align*}
  \lambda^2 - (z_{00}+z_{11})\lambda + (z_{00}z_{11} - z_{01}z_{10}) =  & 0
\end{align*}
We use the fact that the quantities $z_{00}, z_{01}, z_{10}, z_{11}$
represent proabilities, hence they are $> 0$.  It follows that
$\lambda_1 + \lambda_2 = \texttt{tr}(\mb z) = z_{00}+z_{11} > 0$.
Since
$\lambda_1\lambda_2 = \det(\mb z) = (z_{00}z_{11} - z_{01}z_{10})$, we
cannot have $\lambda_1 = z_{00}$, because then
$\lambda_2=\texttt{tr}(\mb z) - \lambda_1 = z_{11}$ and
$z_{00}z_{11} = \lambda_1\lambda_2 = \det(\mb z) =
z_{00}z_{11}-z_{10}z_{01}$ implying $z_{01}=0$ or $z_{10}=0$, which is
impossible by Lemma~\ref{lemma:zij:not:zero}.  Thus,
$\lambda_{1,2} \neq z_{00}, z_{11}$.

Next, since $\mb z^0$ is the identity matrix, we obtain:
\begin{align*}
  a_1+a_2 = & 1 & b_1 + b_2 = & 0 \\
  c_1+c_2 = & 0 & d_1 + d_2 = & 1
\end{align*}
It follows that $b_1 = -b_2$ and $c_1 = -c_2$.  Similarly, we examine
$\mb z^1 = z$, and deduce $a_1 \lambda_1 + a_2 \lambda_2 =z_{00}$,
$d_1 \lambda_1 + d_2 \lambda_2 =z_{11}$.  If $a_1 = 0$ then
$a_2 = 1- a_1 = 1$ which implies $\lambda_2=z_{00}$, contradiction.
Thus, $a_1, a_2, d_1, d_2 \neq 0$.  Finally, we consider the
expression for $\mb z^p$:
  \begin{align*}
    \mb z^p = &
\left[
 \begin{array}{cc}
  a_1\lambda_1^p + a_2\lambda_2^p &   b_2(\lambda_2^p - \lambda_1^p)\\
  c_2(\lambda_2^p - \lambda_2^p) &   d_1\lambda_1^p + d_2\lambda_2^p
 \end{array}
\right]
  \end{align*}
We notice that all entries in $\mb z^p$ are $>0$, because they are
probabilities.  When $p \rightarrow \infty$ then all terms above are
dominated by $\lambda_2^p$, which implies $a_2, b_2, c_2, d_2 > 0$.
Finally, the lemma follows from computing $y^{(p)}_{\alpha \beta}$
using~\eqref{eq:y:alpha:beta:lambda:p}, and obtain:
\begin{align*}
  y^{(p)}_{\alpha \beta} = & \frac{1}{2}
\left(
(
u_\alpha^{(0)}a_1v_\beta^{(0)} +
u_\alpha^{(0)}b_1v_\beta^{(1)} +
u_\alpha^{(1)}c_1v_\beta^{(0)} +
u_\alpha^{(1)}d_1v_\beta^{(1)})\lambda_1^p
+
(
u_\alpha^{(0)}a_2v_\beta^{(0)} +
u_\alpha^{(0)}b_2v_\beta^{(1)} +
u_\alpha^{(1)}c_2v_\beta^{(0)} +
u_\alpha^{(1)}d_2v_\beta^{(1)})\lambda_2^p
\right)
\end{align*}
and the lemma follows by observing that the factor of $\lambda_2^p$ is $>0$.
\end{proof}

Finally, it remains to prove
condition-\eqref{eq:conditionCoefficients2}:
$a_{\alpha_1\beta_1}b_{\alpha_2\beta_2}\neq
a_{\alpha_2\beta_2}b_{\alpha_1\beta_1}$.  To enforce this condition we
need to design carefully the prefix and suffix blocks.  We do this in
the next sections.


\subsection{One Condition  $a_{\alpha_1\beta_1}b_{\alpha_2\beta_2}\neq  a_{\alpha_2\beta_2}b_{\alpha_1\beta_1}$}

\label{sec:not:proportional}

Our end goal is to define the block $B^{(p)}(u,v)$, and its
probabilities, such as to satisfy all three
conditions~\eqref{eq:conditionLambda2}, \eqref{eq:conditiononZero2},
and \eqref{eq:conditionCoefficients2}.  In the previous sections we
have denoted by $Y^{(p)}_{\alpha \beta}$ and $y^{(p)}_{\alpha \beta}$
the lineage on the zig-zag block only~\eqref{eq:yab:on:zigzag}, and
satisfied the first two conditions; in this and the next section we
extend it with the prefix/suffix, and satisfy the third condition,
\eqref{eq:conditionCoefficients2}.  Recall that the complete block
$B^{(p)}(u,v)$ is:
\begin{align}
  B^{(p)}(u,v) = B(u,r_0) \cup \left(\bigcup_j B(r_0, e^{(j)}_0)\right) \cup B(r_0,t_p) \cup\left(\bigcup_j B(f_p^{(j)},t_p) \right)\cup B(t_p,v)
\label{eq:all:blocks}
\end{align}
We will denote by:
\begin{align*}
  Y^{(p)}_{\alpha \beta}(r_0,t_p) \defeq & \Phi_{B(r_0,t_p)}(G_\alpha(r_0) \wedge Q \wedge H_\beta(t_p))\\
  Y^{(p)}_{\alpha \beta}(u,v) \defeq & \Phi_{B^{(p)}(u,v)}(G_\alpha(u) \wedge Q \wedge H_\beta(v))
\end{align*}
and similarly for their probabilities,
$y^{(p)}_{\alpha \beta}(r_0,t_p), y^{(p)}_{\alpha \beta}(u,v)$.  We
have already defined a partial assignment $\theta$ of Boolean
variables in $B^{(p)}(r_0,t_p)$ to probabilities in $\set{0,1/2, 1}$
to satisfy Condition~\eqref{eq:conditionLambda2}
(Condition~\eqref{eq:conditiononZero2} came almost for free).  Now we
will extend $\theta$ to a total assignment, i.e. to all variables in
the block $B^{(p)}(u,v)$, to also satisfy
condition~\eqref{eq:conditionCoefficients2}.

More precisely, let $\mb V \subset B(r_0,t_p)$ be the set of all
Boolean variables that occur in the matrices
$\mb z_1, \ldots, \mb z_p$.  Add to it the articulation variables,
$\mb V' \defeq \mb V \cup \set{s_0, s_1, \ldots, s_p}$.  Then $\theta$
is the assignment of the variables $\mb V'$ given by
Theorem~\ref{th:assignment:for:lambdas}.  We have:

\begin{align}
  y^{(p)}_{\alpha \beta}(r_0,t_p)[\theta] = &
  \frac{1}{2^p} 
\left[
  \begin{array}{cc}
    u_\alpha^{(0)}(r_0,t_0)[\theta]  & u_\alpha^{(1)}(r_0,t_0)[\theta]
  \end{array}
\right]
\cdot
\left[
 \begin{array}{cc}
  z_{00} &   z_{01}\\
  z_{10} &   z_{11}
 \end{array}
\right]^p
\cdot
\left[
 \begin{array}{c}
  v_\beta^{(0)}(r_p,t_p)[\theta]\\
  v_\beta^{(1)}(r_p,t_p)[\theta]
 \end{array}
\right] =
a_{\alpha \beta}(r_0,t_p)(\lambda_1/2)^p + b_{\alpha \beta}(r_0,t_p)(\lambda_2/2)^p\label{eq:y:again:r0:t0}\\
  y^{(p)}_{\alpha \beta}(u,v)[\theta] = &
  \frac{1}{2^p} 
\left[
  \begin{array}{cc}
    u_\alpha^{(0)}(u,t_0)[\theta]  & u_\alpha^{(1)}(u,t_0)[\theta]
  \end{array}
\right]
\cdot
\left[
 \begin{array}{cc}
  z_{00} &   z_{01}\\
  z_{10} &   z_{11}
 \end{array}
\right]^p
\cdot
\left[
 \begin{array}{c}
  v_\beta^{(0)}(r_p,v)[\theta]\\
  v_\beta^{(1)}(r_p,v)[\theta]
 \end{array}
\right] =
a_{\alpha \beta}(t_p,v)(\lambda_1/2)^p + b_{\alpha \beta}(t_p,v)(\lambda_2/2)^p\label{eq:y:again:u:v}
\end{align}

We are interested in the latter expression, where we haven't yet
defined the structure and/or probabilities of the two polynomial
matrices $\mb u_\alpha, \mb v_\beta$.  Notice that, whatever our
choice, conditions~\eqref{eq:conditionLambda2},
\eqref{eq:conditiononZero2} continue to apply, since we proved them
for any polynomials $\mb u_\alpha, \mb v_\beta$.  Now, we will design
the prefix/suffix blocks and assign probabilities to the variables in
$\mb u_\alpha, \mb v_\beta$ to also satisfy
condition~\eqref{eq:conditionCoefficients2}.  We will do this in two
steps.  The first step, described in this section, will satisfy the
condition for {\em one} pair
$(\alpha_1, \beta_1) \neq (\alpha_2, \beta_2)$.  For that we need a
single branch in the prefix and a single branch in the suffix, thus
$B^{(p)}(u,v)$ will be like $B^{(p+2)}(r_0,t_{p+2})$, as illustrated
in Fig.~\ref{fig:block:type:2} (b).  The second step, described in the
next sections, uses multiple parallel branches to satisfy the
condition for {\em all} pairs, illustrated in
Fig.~\ref{fig:block:type:2} (c).

We will start by satisfying a single condition: for a fixed pair
$(\alpha_1, \beta_1) \neq (\alpha_2, \beta_2)$,our goal is to satisfy.
$a_{\alpha_1\beta_1}b_{\alpha_2\beta_2}\neq
a_{\alpha_2\beta_2}b_{\alpha_1\beta_1}$.  We will use a single branch
for the prefix/suffix blocks, hence $B^{(p)}(u,v)$ is isomorphic to
$B^{(p+2)}(r_0,t_{p+2})$, and, $y^{(p)}_{\alpha \beta}$ is given by
\eqref{eq:y:again:r0:t0}.  It suffices to prove how complete the
assignment $\theta$ to all variables in $B^{(p)}(r_0,t_p)$ in order to
satisfy
$a_{\alpha_1\beta_1}b_{\alpha_2\beta_2}\neq
a_{\alpha_2\beta_2}b_{\alpha_1\beta_1}$.  Our construction is
independent of $p$ because, on one hand, the condition that we need to
satisfy,
$a_{\alpha_1\beta_1}b_{\alpha_2\beta_2}\neq
a_{\alpha_2\beta_2}b_{\alpha_1\beta_1}$, is independent of $p$, on the
other hand we can assume w.l.o.g. that the polynomials
$\mb v_\beta(r_p,t_p)[\theta]$ and
$\mb v_\beta(r_{p+1},t_{p+1})[\theta]$ have the same variables, since
the blocks $B(r_p,t_p)$ and $B(r_{p+1},t_{p+1})$ are isomorphic.  In
other words, we assume that the variables of the polynomials
$y_{\alpha \beta}^{(p)}(r_0,t_p)[\theta]$ are the same for all choices
of $p$ (since $\theta$ already assigns values to all variables in the
zig-zag portion of $B^{(p)}(r_0,t_p)$, which depends on $p$).

When there are no migrating variables, then the polynomials
$u_\alpha^{(0)}, u_\alpha^{(1)}, v_\beta^{(0)}, v_\beta^{(1)}$
in~\eqref{eq:y:again:r0:t0} do not contain any variables with
$\mb V'$, thus we can drop the argument $\theta$; then we are free to
assign their probabilities as we need.  However, if a symbol $S'$
migrates from right to left, then variables from its equivalence class
occur in both $\mb u_\alpha$ and $\mb v_\beta$.  The reason is that
$\theta$ assigns the same value to all variables in an equivalence
class, hence all the tuples
$S'(r_0,t_0), S'(r_1,t_1), \ldots, S'(r_p,t_p)$ are associated with
the same real variable $s'$.  Assuming $S'$ migrates from right to
left, then $S'(r_p,t_p)$ appears in $v_\beta^{(0)}$ (on the ``right'')
and $S'(r_0,t_0)$ appears in $u_\alpha^{(1)}$ (on the ``left'').
Similarly, $u_\alpha^{(0)}$ and $v_\beta^{(1)}$ may share common
variables.  The notation
$u_\alpha^{(0)}[\theta], \ldots, v_\beta^{(1)}[\theta]$ indicates that
we apply $\theta$ to all migrating variables $s'$, and recall that
$\theta(s')=1/2$.  Recall that $\theta$ assigns probabilities $1/2$ to
these variables, hence, we extend it to a complete assignment
$\theta'$ we must do it consistently with this assignment.

We start by restating the condition
$a_{\alpha_1\beta_1}b_{\alpha_2\beta_2}\neq
a_{\alpha_2\beta_2}b_{\alpha_1\beta_1}$.

\begin{lemma} \label{lemma:from:matrix:to:chains} Let $\mb z$ be a
  $2 \times 2$ matrix with eigenvalues
  $0 \neq \lambda_1 \neq \lambda_2 \neq 0$, and let
  $\mb u_i, \mb v_i$, $i=1,2$ be four vectors.  Define the following
  two sequences, $y^{(p)}_i$, $p \geq 0$, $i=1,2$:
  \begin{align*}
    y^{(p)}_1 \defeq & \mb u_1 \cdot \mb z^p \cdot \mb v_1 \defeq a_1 (\lambda_1)^p + b_1 (\lambda_2)^p&
    y^{(p)}_2 \defeq & \mb u_2 \cdot \mb z^p \cdot \mb v_2 \defeq a_2 (\lambda_1)^p + b_2 (\lambda_2)^p
  \end{align*}
For any $p \geq 0$, consider the following matrix:
  \begin{align*}
  \mb{D}^{(p)} \defeq &
  \left[
  \begin{array}{cc}
    y_1^{(p)}& y_2^{(p)}\\
    y_1^{(p+1)}& y_2^{(p+1)}
  \end{array}
  \right]
  \end{align*}
Then the following statements are equivalent:
\begin{itemize}
\item $a_1b_2 \neq a_2b_1$,
\item there exists $p\geq 0$ such that  $\det(\mb D^{(p)}) \neq 0$,
\item for all  $p\geq 0$,   $\det(\mb D^{(p)}) \neq 0$.
\end{itemize}
\end{lemma}

\begin{proof}
  We compute $\det(\mb D)$ directly:
\begin{align*}
  \det\left[
  \begin{array}{cc}
    a_1\lambda_1^{p} + b_1\lambda_2^{p} &     a_2\lambda_1^{p} + b_2\lambda_2^{p}\\
    a_1\lambda_1^{p+1} + b_1\lambda_2^{p+1} &     a_2\lambda_1^{p+1} + b_2\lambda_2^{p+1}
  \end{array}
\right] = & 
  \lambda_1^p\lambda_2^p(\lambda_2 - \lambda_1) (a_1 b_2 - a_2 b_1)
\end{align*}
and using the fact that $0 \neq \lambda_1 \neq \lambda_2 \neq 0$.
\end{proof}

Therefore, in order to find an assignment $\theta'$ that satisfies
$a_{\alpha_1\beta_1}b_{\alpha_2\beta_2}\neq
a_{\alpha_2\beta_2}b_{\alpha_1\beta_1}$, we need to construct {\em
  two} blocks, $B^{(p)}(r_0,t_p)$ and $B^{(p+1)}(r_0,t_{p+1})$:
$\theta$ already assigns probabilities to the variables in their
zig-zag part, hence we only need to extend it with $\theta'$ that
assigns probabilities to the remaining variables (which we assumed are
the same in $B^{(p)}(r_0,t_p)$ and $B^{(p+1)}(r_0,t_{p+1})$) such that
$\det(\mb{D}^{(p)}_{\alpha_1\beta_1\alpha_2\beta_2})[\theta']
\neq 0$, where:
\begin{align}
  \mb{D}^{(p)}_{\alpha_1\beta_1\alpha_2\beta_2} \defeq & \left[
  \begin{array}{cc}
    y_{\alpha_1\beta_1}^{(p)}[\theta]& y_{\alpha_2,\beta_2}^{(p)}[\theta]\\
    y_{\alpha_1\beta_1}^{(p+1)}[\theta]& y_{\alpha_2,\beta_2}^{(p+1)}[\theta]
  \end{array} \right] \label{eq:2:2:matrix:d}
\end{align}
%
Furthermore, by lemma~\ref{lemma:from:matrix:to:chains} we can check
the condition for any $p$, so we choose conveniently to check it for
$p=0$, in other words use the blocks $B^{(0)}(r_0,t_0)$ and
$B^{(1)}(r_0,t_1)$.  To find $\theta'$ such that
$\det(\mb{D}^{(p)}_{\alpha_1\beta_1\alpha_2\beta_2})[\theta']
\neq 0$, we proceed as usual: we first prove that the polynomial
$\det(\mb{D}_{\alpha_1\beta_1\alpha_2\beta_2}^{(0)})
\not\equiv 0$, then use this to argue that there exists an assignment
$\theta'$ of its variables such that
$\det(\mb{D}_{\alpha_1\beta_1\alpha_2\beta_2}^{(0)}[\theta']) \neq 0$.

\begin{lemma} \label{lemma:alpha0:beta0:alpha:beta} If
  $(\alpha_1,\beta_1) \not\eq (\alpha_2,\beta_2)$, then
  $\det(\mb{D}^{(0)}_{\alpha_1\beta_1\alpha_2\beta_2})\not\equiv 0$.
\end{lemma}

\begin{proof}
  For arbitrary $\alpha, \beta$, denote the following matrices of
  polynomials (appearing in Eq.~\eqref{eq:y:again:r0:t0}):
  \begin{align*}
    \mb u_{\alpha} = &\left[
              \begin{array}{c}
                u^{(0)}_{\alpha} \\
                u^{(1)}_{\alpha}
              \end{array}
              \right] &
    \mb v_{\beta} = &\left[
              \begin{array}{c}
                v^{(0)}_{\beta} \\
                v^{(1)}_{\beta}
              \end{array}
              \right]
  \end{align*}
Notice that these are polynomials that {\em include} the migrating
variables.  By Eq.~\eqref{eq:y:again:r0:t0}, we have:
\begin{align*}
  y^{(0)}_{\alpha \beta} = & \mb u^T_\alpha \mb v_\beta 
&  y^{(1)}_{\alpha \beta} = & \mb u^T_\alpha \mb z \mb v_\beta
\end{align*}
For the values $\alpha_1, \beta_1, \alpha_2, \beta_2$ given in the
lemma, define the following matrix:
\begin{align*}
  \mb D \defeq &
\left[
  \begin{array}{cc}
    y_{\alpha_1\beta_1}^{(0)}& y_{\alpha_2,\beta_2}^{(0)}\\
    y_{\alpha_1\beta_1}^{(1)}& y_{\alpha_2,\beta_2}^{(1)}
  \end{array} \right] 
= \left[
 \begin{array}{cc}
\mb u_{\alpha_1}^T \mb v_{\beta_1}   & \mb u_{\alpha_2}^T \mb v_{\beta_2}\\
\mb u_{\alpha_1}^T \mb z \mb v_{\beta_1}   & \mb u_{\alpha_2}^T \mb z \mb v_{\beta_2}
 \end{array}
\right]
\end{align*}
Then $\mb D[\theta] = \mb{D}^{(0)}_{\alpha_1\beta_1\alpha_2\beta_2}$,
i.e. $\mb D$ is obtained by exposing the migrating variables, which in
$\mb{D}^{(0)}_{\alpha_1\beta_1\alpha_2\beta_2}$ are assigned by
$\theta$ (all are assigned the value $1/2$).  Thus, it suffices to
prove $\det(\mb D[\theta])\not\equiv 0$.

We denote by $\mb T$ the non-singular matrix that diagonalizes the
matrix $\mb z$, that is $\mb z = \mb T^{-1} \mb \Lambda \mb T$, where
$\Lambda = \diag(\lambda_1, \lambda_2)$,  and define:
\begin{align*}
\mb x_{\alpha_1}^T \defeq & \mb u_{\alpha_1}\mb T & \mb y_{\beta_1} \defeq & \mb T^{-1} \mb v_{\beta_1}\\
\mb x_{\alpha_2}^T \defeq & \mb u_{\alpha_2}\mb T & \mb y_{\beta_2} \defeq & \mb T^{-1} \mb v_{\beta_2}
\end{align*}
The matrix $\mb T$ is a matrix of numbers, while
$\mb x_{\alpha_1}, \ldots, \mb y_{\beta_2}$ are vectors of multilinear
polynomials. We denote the components of the vector $\mb x_{\alpha_1}$
by $x_{\alpha_1}^{(0)}$ and $x_{\alpha_1}^{(1)}$, similar to the
components of the vector $\mb u_{\alpha_1}$.  Notice that
$\vars(x_{\alpha_1}^{(0)}) \subseteq \vars(u_{\alpha_1}^{(0)})\cup
\vars(u_{\alpha_1}^{(1)})$ and
$\vars(x_{\alpha_1}^{(1)})\subseteq \vars(u_{\alpha_1}^{(0)})\cup
\vars(u_{\alpha_1}^{(1)})$.  These variables include the migrating
variables.  Similarly for the other three vectors.  Therefore, we
obtain:
\begin{align*}
 \det(\mb{D}) = &
\left|
 \begin{array}{cc}
\mb u_{\alpha_1}^T \mb v_{\beta_1}   & \mb u_{\alpha_2}^T \mb v_{\beta_2}\\
\mb u_{\alpha_1}^T \mb z \mb v_{\beta_1}   & \mb u_{\alpha_2}^T \mb z \mb v_{\beta_2}
 \end{array}
\right|=
\left|
 \begin{array}{cc}
\mb x_{\alpha_1}^T \mb y_{\beta_1}   & \mb x_{\alpha_2}^T \mb y_{\beta_2}\\
\mb x_{\alpha_1}^T \mb \Lambda \mb y_{\beta_1}   & \mb x_{\alpha_2}^T \mb \Lambda \mb y_{\beta_2}
 \end{array}
\right|\\
= &
\left|
  \begin{array}{cc}
x_{\alpha_1}^{(0)}y_{\beta_1}^{(0)}+x_{\alpha_1}^{(1)}y_{\beta_1}^{(1)}&x_{\alpha_2}^{(0)}y_{\beta_2}^{(0)}+x_{\alpha_2}^{(1)}y_{\beta_2}^{(1)}\\
x_{\alpha_1}^{(0)}\lambda_1y_{\beta_1}^{(0)}+x_{\alpha_1}^{(1)}\lambda_2y_{\beta_1}^{(1)}&x_{\alpha_2}^{(0)}\lambda_1y_{\beta_2}^{(0)}+x_{\alpha_2}^{(1)}\lambda_2y_{\beta_2}^{(1)}
  \end{array}
\right| =
\left(x^{(0)}_{\alpha_1}y^{(0)}_{\beta_1}x^{(1)}_{\alpha_2}y^{(1)}_{\beta_2}
-
x^{(1)}_{\alpha_1}y^{(1)}_{\beta_1}x^{(0)}_{\alpha_2}y^{(0)}_{\beta_2}
\right)\cdot (\lambda_2 - \lambda_1) \defeq f \cdot (\lambda_2 - \lambda_1)
\end{align*}
Since $\lambda_1 \neq \lambda_2$, in order to prove
$\det(\mb D[\theta])\not\equiv 0$, we need to show that
$f[\theta] \not\equiv 0$.  Assuming the contrary, we have the
following identity of polynomials:

\begin{align}
\underbrace{x^{(0)}_{\alpha_1}[\theta]x^{(1)}_{\alpha_2}[\theta]}_{\vars(\mb u_{\alpha_1}[\theta], \mb u_{\alpha_2}[\theta])}
\underbrace{y^{(0)}_{\beta_1}[\theta]y^{(1)}_{\beta_2}[\theta]}_{\vars(\mb v_{\beta_1}[\theta], \mb v_{\beta_2}[\theta])}
  \equiv & 
\underbrace{x^{(1)}_{\alpha_1}[\theta]x^{(0)}_{\alpha_2}[\theta]}_{\vars(\mb u_{\alpha_1}[\theta], \mb u_{\alpha_2}[\theta])}
\underbrace{y^{(1)}_{\beta_1}[\theta]y^{(0)}_{\beta_2}[\theta]}_{\vars(\mb v_{\beta_1}[\theta], \mb v_{\beta_2}[\theta])}
\label{eq:identity:above}
\end{align}
We have indicate above the set of variables that occur in these
multilinear polynomials.  At this point we observe that
$\vars(\mb u_{\alpha_i}[\theta])$ and $\vars(\mb v_{\beta_j}[\theta])$
are disjoint sets of variables, for any $i,j=1,2$.  Indeed, the only
variables shared by $\mb u_{\alpha_i}$ and $\mb v_{\beta_j}$ are the
migrating variables, but these have been replaced by the constant
$1/2$ by $\theta$.  Therefore, assuming the identity
\eqref{eq:identity:above} holds, then {\em both} the following
identities hold too:
\begin{align}
  x^{(0)}_{\alpha_1}[\theta]x^{(1)}_{\alpha_2}[\theta] \equiv &x^{(1)}_{\alpha_1}[\theta]x^{(0)}_{\alpha_2}[\theta]&
  y^{(0)}_{\beta_1}[\theta]y^{(1)}_{\beta_2}[\theta] \equiv & y^{(1)}_{\beta_1}[\theta]y^{(0)}_{\beta_2}[\theta]
\label{eq:the:two:identities}
\end{align}
Now we return to the definition of the vectors $\mb x_{\alpha_i}$,
$\mb y_{\beta_j}$ in terms of $\mb u_{\alpha_i}$, $\mb v_{\beta_j}$
and the non-singular matrix $\mb T$.  Assuming the entries in $\mb T$
are $a,b,c,d$, where $ad-bc\neq 0$, we have:
\begin{align*}
  x_{\alpha_i}^{(0)} = & a u_{\alpha_i}^{(0)} + c u_{\alpha_i}^{(1)}\\
  x_{\alpha_i}^{(1)} = & b u_{\alpha_i}^{(0)} + d u_{\alpha_i}^{(1)}
\end{align*}
and we obtain:
\begin{align*}
  (x^{(0)}_{\alpha_1}x^{(1)}_{\alpha_2} -  x^{(1)}_{\alpha_1}x^{(0)}_{\alpha_2}) =&
\left(
(a u_{\alpha_1}^{(0)} + c u_{\alpha_1}^{(1)})
(b u_{\alpha_2}^{(0)} + d u_{\alpha_2}^{(1)})
-
(b u_{\alpha_1}^{(0)} + d u_{\alpha_1}^{(1)})
(a u_{\alpha_2}^{(0)} + c u_{\alpha_2}^{(1)})
\right)
= (ad-bc)(u_{\alpha_1}^{(0)}u_{\alpha_2}^{(1)} - u_{\alpha_1}^{(1)}u_{\alpha_2}^{(0)})
\end{align*}
We apply a similar change of base from $\mb y_{\beta_i}$ to
$\mb v_{\beta_i}$, and the identities~\eqref{eq:the:two:identities}
become:
\begin{align}
  u_{\alpha_1}^{(0)}[\theta]u_{\alpha_2}^{(1)}[\theta] \equiv &u_{\alpha_1}^{(1)}[\theta]u_{\alpha_2}^{(0)}[\theta]&
  v_{\beta_1}^{(0)}[\theta]v_{\beta_2}^{(1)}[\theta] \equiv &v_{\beta_1}^{(1)}[\theta]v_{\beta_2}^{(0)}[\theta]
\label{eq:both:identities}
\end{align}
We have shown that, if $\det(\mb D[\theta])\equiv 0$, then {\em both}
identities~\eqref{eq:both:identities} hold.  We prove that this is a
contradiction.  For that we show that there exists a total assignment
$\theta'$ of all variables in the polynomials
$\mb u_{\alpha_i}, \mb v_{\beta_j}$ that (1) extends $\theta$, and (2)
make at least one of the quantities in \eqref{eq:both:identities}
$\neq$.  Notice that when $\alpha_1=\alpha_2$ then the first identity
{\em does} hold, but in that case $\beta_1 \neq \beta_2$ and then we
show that the second identity implies a contradiction.  When both
$\alpha_1=\alpha_2$ and $\beta_1=\beta_2$ then both identities
\eqref{eq:both:identities} hold, but we have assumed that
$(\alpha_1,\beta_1)\neq (\alpha_2,\beta_2)$.

To prove our claim, we remove the assignment $\theta$ and start from
the polynomials $\mb u_{\alpha_i}, \mb v_{\beta_j}$.  We claim that
the following non-indentities hold:
\begin{align}
\alpha_1\neq \alpha_2 \Rightarrow &  u_{\alpha_1}^{(0)}u_{\alpha_2}^{(1)} \not\equiv u_{\alpha_1}^{(1)}u_{\alpha_2}^{(0)}&
\beta_1 \neq \beta_2  \Rightarrow &  v_{\beta_1}^{(0)}v_{\beta_2}^{(1)} \not\equiv v_{\beta_1}^{(1)}v_{\beta_2}^{(0)}\label{eq:not:equiv:alpha:beta:here}
\end{align}
Indeed, suppose the first identity holds.  Then by
Theorem~\ref{th:rank:1:matrix:of:polynomials} we can factorize the
polynomials as follows:
\begin{align*}
  \left[
  \begin{array}{cc}
    u_{\alpha_1}^{(0)} & u_{\alpha_1}^{(1)}\\
    u_{\alpha_2}^{(0)} & u_{\alpha_2}^{(1)}
  \end{array}
\right] \equiv
  \left[
  \begin{array}{cc}
    f\cdot h & f \cdot k \\
    g \cdot h & g \cdot k
  \end{array}
\right]
\end{align*}
If $f$ is not a constant polynomial, then $f$ divides both
$u_{\alpha_1}^{(0)}$ and $u_{\alpha_1}^{(1)}$, and therefore it
divides
$y_{\alpha_1 \beta_1}^{(0)} =
u_{\alpha_1}^{(0)}v_{\beta_1}^{(0)}(1-s_0)+u_{\alpha_1}^{(1)}v_{\beta_1}^{(1)}s_0$
(where $s_0$ is the articulation variable), which contradicts the fact
that $y_{\alpha_1 \beta_1}^{(0)}$ is irreducible.  This proves that
$f$ must be a constant.  Similarly, $g$ must be a constant, by the
same argument.  It follows that
$u_{\alpha_1}^{(0)} \equiv c u_{\alpha_2}^{(0)}$ for some constant
$c$.  Since both $u_{\alpha_1}^{(0)}$ and $u_{\alpha_2}^{(0)}$ are
arithmetizations of monotone Boolean functions, when all variables are
set to $1$ then both quantities are $=1$, hence the constant is $c=1$.
Thus, $u_{\alpha_1}^{(0)} \equiv u_{\alpha_2}^{(0)}$, and by the same
argument , $u_{\alpha_1}^{(1)} \equiv u_{\alpha_2}^{(1)}$ which
implies that the two polynomials are identical,
$u_{\alpha_1} \equiv u_{\alpha_2}$, which implies
$\alpha_1 = \alpha_2$ by Lemma~\ref{lemma:type2:invertible}.

At this point we need to treat separately the cases when
$\alpha_1\neq \alpha_2$ and $\beta_1 \neq \beta_2$ and when one of
them is equal.  Assume first that both are different, hence both
Eq.~\eqref{eq:not:equiv:alpha:beta:here} hold.  In that case the
following product of polynomials is not identically zero:
\begin{align}
  f \defeq & \left(u_{\alpha_1}^{(0)}u_{\alpha_2}^{(1)}  - u_{\alpha_1}^{(1)}u_{\alpha_2}^{(0)}\right)\cdot
\left(v_{\beta_1}^{(0)}v_{\beta_2}^{(1)} - v_{\beta_1}^{(1)}v_{\beta_2}^{(0)}\right)
\not\equiv 0 \label{eq:f:with:migrating:variables}
\end{align}
We claim that each variable in $f$ has degree $\leq 2$.  Indeed, the
only variables common in both factors are the migrating variables. Let
$s'$ be a migrating variable, and assume it migrates from right to
left.  Then it occurs only in
$v_{\beta_1}^{(0)}, v_{\beta_2}^{(0)}, u_{\alpha_1}^{(1)},
u_{\alpha_2}^{(1)}$, and therefore it has degree 1 in each of the
factors, hence it total degree in $f$ is 2.  Therefore, by
Lemma~\ref{lemma:three:values}, there exists an assignment $\theta'$
with values in $\set{0,1/2,1}$ such that $f[\theta'] \neq 0$.  In
other words, we have proven that {\em both} the following inequalities
hold:
\begin{align}
  u_{\alpha_1}^{(0)}[\theta']u_{\alpha_2}^{(1)}[\theta'] \neq &u_{\alpha_1}^{(1)}[\theta']u_{\alpha_2}^{(0)}[\theta']&
  v_{\beta_1}^{(0)}[\theta']v_{\beta_2}^{(1)}[\theta'] \neq &v_{\beta_1}^{(1)}[\theta']v_{\beta_2}^{(0)}[\theta']
\label{eq:both:must:hold}
\end{align}
It remains to prove that $\theta'$ assigns $1/2$ to each migrating
variable $s'$.  

Let $s'$ be a variable migrating from right to left.  We claim that
either the following two identities hold:
\begin{align}
  u_{\alpha_1}^{(0)}[s':=0]u_{\alpha_2}^{(1)}[s':=0] = &u_{\alpha_1}^{(1)}[s':=0]u_{\alpha_2}^{(0)}[s':=0]\label{eq:migrating:0}\\
  v_{\beta_1}^{(0)}[s':=1]v_{\beta_2}^{(1)}[s':=1] = &v_{\beta_1}^{(1)}[s':=1]v_{\beta_2}^{(0)}[s':=1]\label{eq:migrating:1}
\end{align}
or the following two identities hold (obtained by switching $s':=0$
and $s':=1$):
\begin{align}
  u_{\alpha_1}^{(0)}[s':=1]u_{\alpha_2}^{(1)}[s':=1] = &u_{\alpha_1}^{(1)}[s':=1]u_{\alpha_2}^{(0)}[s':=1]\label{eq:migrating:2}\\
  v_{\beta_1}^{(0)}[s':=0]v_{\beta_2}^{(1)}[s':=0] = &v_{\beta_1}^{(1)}[s':=0]v_{\beta_2}^{(0)}[s':=0]\label{eq:migrating:3}
\end{align}
The claim completes the proof, because, assuming the first two
equalities hold,~\eqref{eq:migrating:0} and ~\eqref{eq:migrating:1},
then, if $\theta'(s')=0$ then the left inequality in
~\eqref{eq:both:must:hold} becomes an equality, while if
$\theta'(s')=1$, then the right inequality in
~\eqref{eq:both:must:hold} becomes an equality, which is a
contradiction because we have chosen $\theta'$ such that
both~\eqref{eq:both:must:hold} hold.  Similarly for the case when
~\eqref{eq:migrating:2} and ~\eqref{eq:migrating:3} hold.  In either
case, $\theta'(s')$ cannot be either 0 or 1, hence it must be $1/2$
and, since $s'$ was an arbitrary migrating variable, we have that
$\theta'$ is an extension of $\theta$.  

Thus, it remains to prove the claim.  Recall that, for any
$\alpha, \beta$, $y_{\alpha \beta}^{(0)}$ is {\em final}, meaning that
for every symbol $s'$, setting it to $0$ or to $1$ decomposes the
polynomial.\footnote{There is no need for partial assignment
  $\theta_0$ in Sec.~\ref{subsec:consistent:assignment} because when
  $p=0$ then $\theta_0$ is empty; in fact, when $p=0$ then
  $y^{(0)}_{\alpha \beta}$ is isomorphic to $Q_{\alpha \beta}$, and it
  is final because $Q$ is final.}  Thus, for any $\alpha, \beta$:
\begin{align}
   y_{\alpha \beta}^{(0)}[s':=0] = & a_{\alpha} \cdot b_{\beta} \label{eq:y:split:by:migrating:var}
\end{align}
where $a_\alpha$ and $b_\beta$ are polynomials that depend only on
$\alpha$ and $\beta$ respectively.  By assumption, $s_0$, causes $s'$
to migrate, hence, by Corollary~\ref{cor:migrating:symmetric}, $s'$
will cause $s_0$ to migrate.  Assume that $s_0$ migrates from right to
left: that is $s_0$ occurs in $b_{\beta}$, and when we decompose
$y^{(0)}_{\alpha \beta}[s':=1]$ then it occurs on the left.  Then, we
apply~\eqref{eq:y:split:by:migrating:var} to
$y^{(0)}_{\alpha_1 \beta}$ and $y^{(0)}_{\alpha_2 \beta}$ where
$\alpha_1,\alpha_2$ are the values given by the lemma, and $\beta$ is
arbitrary, and obtain:
\begin{align*}
   y_{\alpha_1 \beta}^{(0)}[s':=0,s_0:=0] = & \overbrace{a_{\alpha_1} \cdot b_0}^{u_{\alpha_1}^{(0)}[s':=0]}\cdot \overbrace{c_{0\beta}}^{v_{\beta}^{(0)}[s':=0]}&
   y_{\alpha_1 \beta}^{(0)}[s':=0,s_0:=1] = & \overbrace{a_{\alpha_1} \cdot b_1}^{u_{\alpha_1}^{(1)}[s':=0]}\cdot \overbrace{c_{1\beta}}^{v_{\beta}^{(1)}[s':=0]}\\
   y_{\alpha_2 \beta}^{(0)}[s':=0,s_0:=0] = & \underbrace{a_{\alpha_2} \cdot b_0}_{u_{\alpha_2}^{(0)}[s':=0]}\cdot \underbrace{c_{0\beta}}_{v_{\beta}^{(0)}[s':=0]}&
   y_{\alpha_2 \beta}^{(0)}[s':=0,s_0:=1] = & \underbrace{a_{\alpha_2} \cdot b_1}_{u_{\alpha_2}^{(1)}[s':=0]}\cdot \underbrace{c_{1\beta}}_{v_{\beta}^{(1)}[s':=0]}
\end{align*}
and both sides of~\eqref{eq:migrating:0} become equal to
$a_{\alpha_1}a_{\alpha_2}b_0b_1$, thus we have proven the
identity~\eqref{eq:migrating:0}.  Applying the same reasoning to the
decomposition $y_{\alpha \beta}^{(0)}[s':=1]$ (where $s_0$ occurs on
the left) we deduce the identity~\eqref{eq:migrating:1}.  Thus, when
$s_0$ migrates from right to left, then both
identities~\eqref{eq:migrating:0} and~\eqref{eq:migrating:1} hold.
Similarly, when it migrates from right to left then
~\eqref{eq:migrating:0} and~\eqref{eq:migrating:1} hold, proving the claim.

Next, assume that $\alpha_1 \neq \alpha_2, \beta_1 = \beta_2$.   In
that case only the first condition
in~\eqref{eq:not:equiv:alpha:beta:here} holds, thus we have:
\begin{align*}
    f \defeq & u_{\alpha_1}^{(0)}u_{\alpha_2}^{(1)}  -   u_{\alpha_1}^{(1)}u_{\alpha_2}^{(0)}\not\equiv 0
\end{align*}
We prove that $f[\theta]\not\equiv 0$, where $\theta(s')=1/2$ for all
migrating variables.  Here we notice that every migrating variable in
$f$ has degree 1, because it occurs either only in
$u_{\alpha_1}^{(0)}$ and $u_{\alpha_2}^{(0)}$ or only in
$u_{\alpha_1}^{(1)}$ and $u_{\alpha_2}^{(1)}$.  We prove that either
$f[s':=0]\equiv 0$ or $f[s':=1]\equiv 0$.  Using the same argument as
before, we derive that either~\eqref{eq:migrating:0} holds or
~\eqref{eq:migrating:2} holds; equations~\eqref{eq:migrating:1}
and~\eqref{eq:migrating:3} hold vacuously because $\beta_1=\beta_2$.
In the first case, when ~\eqref{eq:migrating:0} holds, then
$f[s':=0]\equiv 0$; in the second case $f[s':=1]\equiv 0$.  Thus, $f$
is divisible by either $s'$ or by $1-s'$.  It follows that $f$ is a
product of the form $s' (1-s'') s''' \cdots$ i.e. there is one factor
for each migrating variable $s'$, and that factor is either $s'$ or
$1-s'$.  It follows that, if $\theta$ assigns values $1/2$ to all
migrating variables, then $f[\theta]\not\equiv 0$, completing the
proof.
\end{proof}

\begin{corollary} \label{cor:no:idea:what:name:to:give:here} If
  $(\alpha_1, \beta_1) \neq (\alpha_2,\beta_2)$ then there exists an
  assignment $\theta'$ to all variables in
  $\mb{D}^{(0)}_{\alpha_1\beta_1\alpha_2\beta_2}$ such that
  $\det(\mb{D}^{(0)}_{\alpha_1\beta_1\alpha_2\beta_2}[\theta'])\neq
  0$.  Notice that $\theta'$ depends on the choices of
  $\alpha_1, \beta_1, \alpha_2, \beta_2$.
\end{corollary}

\begin{proof}
  This is an immediate consequence of the previous
  Lemma~\ref{lemma:alpha0:beta0:alpha:beta} and of
  Lemma~\ref{lemma:three:values}, because
  $\det(\mb{D}^{(0)}_{\alpha_1\beta_1\alpha_2\beta_2})$ is a degree-2
  multivariate polynomial that is $\not\equiv 0$.
\end{proof}

This completes our goal for this subsection. For any fixed
$(\alpha_1, \beta_1) \neq (\alpha_2,\beta_2)$, we can construct a
block $B^{(p)}(u,v)$ isomorphic to $B^{(p+2)}(r_0,t_{p+2})$ and define
an assignment $\theta'$ of its variables such that the polynomials
$y^{(p)}_{\alpha \beta}(u,v)$ satisfy
condition-\eqref{eq:conditionCoefficients2} for the given pair.
Importantly, while $\theta'$ depends on
$\alpha_1, \beta_1, \alpha_2, \beta_2$, its restriction to the zig-zag
block $B^{(p)}(r_0,t_p)$ agrees with $\theta$, and is thus independent
on $\alpha_1, \beta_1, \alpha_2, \beta_2$, see
Fig.~\ref{fig:block:type:2} (b). 


\subsection{All Conditions  $a_{\alpha_1\beta_1}b_{\alpha_2\beta_2}\neq  a_{\alpha_2\beta_2}b_{\alpha_1\beta_1}$}

In the second part, we show that, if we modify the prefix/suffix
blocks by constructing many parallel branches, in order to satisfy
condition-\eqref{eq:conditionCoefficients2} for {\em all} pairs
$(\alpha_1, \beta_1) \neq (\alpha_2, \beta_2)$.  

Consider a block $B^{(p)}(u,v)$ defined by Eq.~\eqref{eq:all:blocks}.
Let $\mb V_{\text{pref}}$, $\mb V_{\text{suff}}$ be the set of
variables (tuples) in the prefix block $B(u,r_0)$ and suffix block
$B(t_0, v)$ respectively.  These sets will depend on how many branches
we choose for these blocks.  Let $\mb V_{\text{zigzag}}$ be the
remaining variables, in the zig-zag block $B^{(p)}(r_0,t_p)$ and the
two remaining sets of dead-end branches at $r_0$ and $t_p$
respectively.  We will fix the following assignment $\theta$ on $\mb
V_{\text{zigzag}}$.  Consider a prefix with a single branch, and a
suffix with a single branch, thus $B^{(p)}(u,v)$ is isomorphic to
$B^{(p+2)}(r_0,t_{p+2})$, then apply
Corollary~\ref{cor:no:idea:what:name:to:give:here}.  This gives us an
assignment $\theta'$ to all variables in $B^{(p)}(u,v)$ such that
$\det(\mb{D}^{(p)}_{\alpha_1\beta_1\alpha_2\beta_2})[\theta']\neq 0$,
where $\mb{D}^{(p)}_{\alpha_1\beta_1\alpha_2\beta_2}$ is defined by
Eq.~\eqref{eq:2:2:matrix:d} w.r.t. the entire block $B^{(p)}(u,v)$.
While $\theta'$ depends on the choices of
$\alpha_1\beta_1\alpha_2\beta_2$, its restriction to
$B^{(p)}(r_0,t_p)$ is independent of
$\alpha_1\beta_1\alpha_2\beta_2$.  Let $\theta$ be that restriction.
See Fig.~\ref{fig:block:type:2} (c) for an illustration.

With the assignment to $\mb V_{\text{zigzag}}$ fixed, we recompute the
probabilities $y^{(p)}_{\alpha \beta}$, by separating the quantities
that depend on the prefix/suffix from the rest.  For that we use
M\"obius inversion formula applied to the points $r_0$ and $t_p$,
which separate the two sets of blocks, and obtain:
\begin{align}
  y^{(p)}_{\alpha \beta}(u,v) = &
\sum_{\gamma \in L_0(\mb G),\delta \in L_0(\mb H)} p_{\alpha \gamma} c_{\gamma \delta}^{(p)} q_{\delta \beta}\label{eq:prefix:zigzag:suffix}
\end{align}
where:
\begin{align}
  p_{\alpha \gamma} \defeq  & \Pr(Y_{\alpha \gamma}(u,r_0))&
  c_{\gamma \delta}^{(p)} \defeq  & \mu(\gamma)\mu(\delta)y^{(p)}_{\gamma \delta}(r_0,t_p)
\prod_jy^{(0)}_{\gamma \hat 1}(r_0,e^{(j)}_0) \prod_jy^{(0)}_{\hat 1 \delta}(f^{(j)}_0,t_0) &
  q_{\delta \beta}  \defeq  & \Pr(Y_{\delta \beta}(t_p,v)) \label{eq:p:and:q:vars}
\end{align}
We think of these quantities as follows.  The values
$c_{\gamma \delta}$ are constants, since they are defined by the
assignment $\theta$ to $\mb V_{\text{zigzag}}$, and do not depend on
the pair $(\alpha_1, \beta_1) \neq (\alpha_2, \beta_2)$.  The
quantities $p_{\alpha \gamma}$ are defined by the prefix block
$B(u,r_0)$.  The value $p_{\alpha \gamma}$ is defined exactly like
$y_{\alpha \beta}$, the only difference is that it goes from
left-to-left, thus $\alpha, \gamma$ come from the same lattice
$L(\mb G)$, with the only restriction that $\gamma\neq \hat 1$.
Similarly, the values $q_{\delta \beta}$ are defined by the suffix
block.

We compute $\det(\mb{D}_{\alpha_1\beta_1\alpha_1\beta_2})$ in terms of
these new variables $p_{\alpha \gamma}$ and $q_{\delta \beta}$, and
denote it by $f_{\alpha_1\beta_1\alpha_1\beta_2}$:

\begin{align}
f_{\alpha_1\beta_1\alpha_1\beta_2} =&
\left|
    \begin{array}{cc}
    y_{\alpha_1\beta_1}^{(0)}& y_{\alpha_2,\beta_2}^{(0)}\\
    y_{\alpha_1\beta_1}^{(1)}& y_{\alpha_2,\beta_2}^{(1)}
  \end{array}
\right| =
\left|
    \begin{array}{cc}
    \sum_{\gamma_1 \in L_0(\mb G),\delta_1 \in L_0(\mb H)} p_{\alpha_1 \gamma_1} c_{\gamma_1 \delta_1}^{(0)} q_{\delta_1 \beta_1}&
    \sum_{\gamma_2 \in L_0(\mb G),\delta_2 \in L_0(\mb H)} p_{\alpha_2 \gamma_2} c_{\gamma_2 \delta_2}^{(0)} q_{\delta_2 \beta_2}\\
    \sum_{\gamma_3 \in L_0(\mb G),\delta_3 \in L_0(\mb H)} p_{\alpha_1 \gamma_3} c_{\gamma_3 \delta_3}^{(1)} q_{\delta_3 \beta_1}&
    \sum_{\gamma_4 \in L_0(\mb G),\delta_4 \in L_0(\mb H)} p_{\alpha_2 \gamma_4} c_{\gamma_4 \delta_4}^{(1)} q_{\delta_4 \beta_2}\\
  \end{array}
\right|\nonumber\\
= & \sum_{\gamma_1,\delta_1,\ldots,\gamma_4,\delta_4}
\left(
c_{\gamma_1 \delta_1}^{(0)} c_{\gamma_4 \delta_4}^{(1)} p_{\alpha_1  \gamma_1} q_{\delta_1 \beta_1} p_{\alpha_2 \gamma_4}  q_{\delta_4 \beta_2}-
c_{\gamma_2 \delta_2}^{(0)} c_{\gamma_3 \delta_3}^{(1)} p_{\alpha_1 \gamma_2} q_{\delta_2 \beta_1} p_{\alpha_2 \gamma_3} q_{\delta_3 \beta_2}
\right) \nonumber \\
= & \sum_{\gamma_1,\delta_1,\ldots,\gamma_4,\delta_4}
\left(
c_{\gamma_1 \delta_1}^{(0)} c_{\gamma_4 \delta_4}^{(1)} p_{\alpha_1  \gamma_1} q_{\delta_1 \beta_1} p_{\alpha_2 \gamma_4}  q_{\delta_4 \beta_2}
\right)
-
\sum_{\gamma_1,\delta_1,\ldots,\gamma_4,\delta_4}
\left(
c_{\gamma_2 \delta_2}^{(0)} c_{\gamma_3 \delta_3}^{(1)} p_{\alpha_1 \gamma_2} q_{\delta_2 \beta_1} p_{\alpha_2 \gamma_3} q_{\delta_3 \beta_2}
\right) \nonumber \\
= & \bar m \bar n\sum_{\gamma_1,\delta_1,\gamma_4,\delta_4}
\left(
c_{\gamma_1 \delta_1}^{(0)} c_{\gamma_4 \delta_4}^{(1)} p_{\alpha_1  \gamma_1} q_{\delta_1 \beta_1} p_{\alpha_2 \gamma_4}  q_{\delta_4 \beta_2}
\right)
-
\bar m \bar n \sum_{\gamma_2,\delta_2,\gamma_3,\delta_3}
\left(
c_{\gamma_2 \delta_2}^{(0)} c_{\gamma_3 \delta_3}^{(1)} p_{\alpha_1 \gamma_2} q_{\delta_2 \beta_1} p_{\alpha_2 \gamma_3} q_{\delta_3 \beta_2}
\right) \nonumber \\
= & \bar m \bar n \sum_{\gamma_1, \delta_1, \gamma_2, \delta_2}
\left(
c^{(0)}_{\gamma_1 \delta_1}c^{(1)}_{\gamma_2\delta_2}-c^{(0)}_{\gamma_1\delta_2}c^{(1)}_{\gamma_2 \delta_1}
\right)
p_{\alpha_1 \gamma_1} p_{\alpha_2 \gamma_2} q_{\delta_1 \beta_1} q_{\delta_2 \beta_2} \nonumber \\
\defeq & \bar m \bar n \sum_{\gamma_1, \delta_1, \gamma_2, \delta_2}
\left(
\Gamma_{\gamma_1 \gamma_2 \delta_1 \delta_2}
\right)
p_{\alpha_1 \gamma_1} p_{\alpha_2 \gamma_2} q_{\delta_1 \beta_1} q_{\delta_2 \beta_2} 
\label{eq:ugly:polynomial}
\end{align}

We make the following observations about this polynomial.
\begin{itemize}
\item The quantities $p_{\alpha \gamma}$ and $q_{\delta \beta}$ are
  defined in~\eqref{eq:p:and:q:vars}, and thus depend on the
  assignment $\theta$ on the variables $\mb V_{\text{pref}} \cup \mb
  V_{\text{suff}}$.
\item On the other hand, we can view the quantities $p_{\alpha
    \gamma}$ and $q_{\delta \beta}$ as independent variables, and
  thus,   $f_{\alpha_1\beta_1\alpha_1\beta_2}$ in
  Eq.~\eqref{eq:ugly:polynomial} is a multivariate polynomial in these
  variables.
\item The polynomial $f_{\alpha_1\beta_1\alpha_1\beta_2}$ is not
  identically 0.  This follows from
  Corollary~\ref{cor:no:idea:what:name:to:give:here} since we can
  construct prefix/suffix blocks consisting of a single branch and
  extend $\theta$ to a total assignment $\theta'$ such that
  $f_{\alpha_1\beta_1\alpha_1\beta_2} =
  \det(\mb{D}_{\alpha_1\beta_1\alpha_1\beta_2})[\theta']\neq 0$.
\item The coefficients of the polynomial
  $f_{\alpha_1\beta_1\alpha_1\beta_2}$ are
  $\Gamma_{\gamma_1 \gamma_2 \delta_1 \delta_2}$.  The total number of
  variables is $2\bar m + 2\bar n$, where $\bar m = |L_0(\mb G)|$,
  $\bar n = |L_0(\mb H)|$.  
\item When $\alpha_1 = \alpha_2=\alpha$, $\beta_1 \neq \beta_2$, then
  the polynomial has only $\bar m + 2\bar n$ variables, and it
  simplifies to:
  \begin{align*}
    & \bar m \bar n \sum_{\gamma, \delta_1, \delta_2}
      \left(
      c^{(0)}_{\gamma \delta_1}c^{(1)}_{\gamma\delta_2}-c^{(0)}_{\gamma\delta_2}c^{(1)}_{\gamma \delta_1}
      \right)
      p_{\alpha \gamma}^2 q_{\delta_1 \beta_1} q_{\delta_2 \beta_2} 
  \end{align*}
  The reason is that in Eq.~\eqref{eq:ugly:polynomial}, the monomials
  $p_{\alpha \gamma_1}p_{\alpha \gamma_2}$ and
  $p_{\alpha \gamma_2}p_{\alpha \gamma_1}$ are of course the same, but
  the coefficients given by $\gamma_1, \gamma_2$ and
  $\gamma_2, \gamma_1$ cancel out, except when $\gamma_1= \gamma_2$.
  This polynomial is also not identically zero,
  $\det(\mb{D}_{\alpha\beta_1\alpha\beta_2})[\theta']\neq 0$, by the
  same argument.
\item Similarly, when $\alpha_1 \neq \alpha_2$,
  $\beta_1 = \beta_2=\beta$ the polynomial has only $2\bar m + \bar n$
  variables, simplifies similarly, and is not $\equiv 0$.
\item If both $\alpha_1 = \alpha_2$ and $\beta_1 = \beta_2$, then
  there are only $\bar m + \bar n$ variables, but the polynomial is
  $\equiv 0$.  We are not interested in this case.
\item Suppose we have built two separate blocks $B_1^{(p)}(u,v)$ and
  $B_2^{(p)}(u,v)$, with the same $p$, but different
  prefixes/suffixes, i.e. different number of parallel branches, and
  different probability assignments.  The values of the
  quantities~\eqref{eq:p:and:q:vars} are $p^{(1)}_{\alpha \gamma}$ and
  $q^{(1)}_{\delta \beta}$ in the first block, and
  $p^{(2)}_{\alpha \gamma}$ and $q^{(2)}_{\delta \beta}$ in the second
  block; the coefficients $c_{\gamma \delta}^{(p)}$ are the same,
  since we use the same assignment $\theta$ for the variables
  $\mb V_{\text{zigzag}}$ in blocks.  Construct a new block
  $B^{(p)}(u,v)$ whose prefix consists of the union of all parallel
  branches in $B^{(p)}_1(u,v)$ and $B^{(p)}_2(u,v)$, and similarly its
  suffix consists of the union of all parallel branches of the
  suffices of the two blocks; see Fig.~\ref{fig:block:type:2} (c).
  Then the quantities~\eqref{eq:p:and:q:vars} of the new block are
  given by the products, i.e.
  $p^{(1)}_{\alpha \gamma}p^{(2)}_{\alpha \gamma}$ and
  $q^{(1)}_{\delta \beta}q^{(2)}_{\delta \beta}$ respectively.  This
  follows immediately from the fact that, in the M\"obius expansion,
  the formulas for the parallel branches in the suffix/postfix become
  independent; this is in fact a special case of
  Theorem~\ref{th:pq:mobius}.
\end{itemize}

Next, we state a theorem, of possible independent interest, which we
prove in the next section.

\begin{theorem} \label{th:products} Let $f_1, \ldots, f_m$ be
  multivariate polynomials in the variables
  $\mb{x} = (x_1, \ldots, x_n)$.  Suppose that, for each $i=1,m$,
  there exists a set of values
  $\mb{v}_i = (v_{i1}, \ldots, v_{in}) \in \R^n$, $v_{ij} > 0$, such
  that $f_i(\mb{v}_i) \neq 0$.  Then there exists natural numbers
  $k_1, k_2, \ldots, k_n \geq 1$ such that, setting
  $\mb{v} = (v_1, v_2, \ldots, v_n)$ where
  $v_j \defeq v_{i1}^{k_1} v_{i2}^{k_2} \cdots v_{in}^{k_n}$ we have
  $f_1(\mb{v}) \neq 0, \ldots, f_m(\mb{v}) \neq 0$.
\end{theorem}

If $\mb v_1, \mb v_2$ are two vectors, then we write $\mb v_1 \mb v_2$
for their element-wise product.  Thus, the theorem states that, if
$f_i(\mb v_i) \neq 0$, then there exists exponents such that, defining
$\mb v \defeq \prod_i \mb v_i^{k_i}$, then for all $i=1,\ldots,m$,
$f_i(\mb v) \neq 0$. We prove the theorem in the next section.  

We use this theorem as follows.  We need to satisfy several conditions
$a_{\alpha_1\beta_1}b_{\alpha_2\beta_2}\neq
a_{\alpha_2\beta_2}b_{\alpha_1\beta_1}$.  Let $m$ be the number of
such conditions, we will refer to them using an index
$i=1,2,\ldots,m$.  For each condition $i$ we know that we can satisfy
that condition, using a single prefix/suffix branch and some
assignment $\theta_i$ of the variables $\mb V_{\text{pref}}$,
$\mb V_{\text{suff}}$; this follows from
Corollary~\ref{cor:no:idea:what:name:to:give:here}.  The assignment
$\theta_i$ results in numerical values $p_{\alpha \gamma}[\theta_i]$,
$q_{\delta \beta}[\theta_i]$, and, in particular, on these values,
$f_{\alpha_1\beta_1\alpha_1\beta_2}\neq 0$. Next, construct parallel
branches of the prefix/suffix.  In the new block, the quantities
$p_{\alpha \gamma}$ are products
$\prod_i (p_{\alpha \gamma}[\theta_i])^{k_i}$ where $k_i$ is the
number of copies of the branch with assignment $\theta_i$, and
similarly for $q_{\delta \beta}$.  To apply the
Theorem~\ref{th:products} we need to prove that
$p_{\alpha \gamma}[\theta_i]>0$ for all $\alpha, \gamma, \theta_i$.
Corollary~\ref{cor:no:idea:what:name:to:give:here} seems insufficient
for that.  Instead, we will use the corollary only to prove that the
polynomials $f_{\alpha_1\beta_1\alpha_1\beta_2}$ are not identically
0, then prove that we can satisfy each condition $i$ with an
assignment that guarantees $p_{\alpha \gamma}[\theta_i]>0$ and
$q_{\delta \gamma}[\theta_i] > 0$ for all
$\alpha,\beta,\gamma,\delta$.

\begin{lemma} \label{lemma:two:monomials}
  Consider two distinct monomials
  $p_{\alpha_1 \gamma_1}p_{\alpha_2 \gamma_2}q_{\delta_1 \beta_1} q_{\delta_2 \beta_2}$ and
  $p_{\alpha_1 \gamma_3}p_{\alpha_2 \gamma_4}q_{\delta_3 \beta_1}q_{\delta_4 \beta_2}$
  occurring in the polynomial
  $f_{\alpha_1 \alpha_2 \beta_1 \beta_2}$~\eqref{eq:ugly:polynomial}.
  Then there exists an assignment $\theta$ to the variables in
  $\mb V_{\text{pref}}\cup \mb V_{\text{suff}}$ such that:
  \begin{enumerate}[(a)]
  \item The following inequality holds
      $p_{\alpha_1 \gamma_1}[\theta]p_{\alpha_2 \gamma_2}[\theta]q_{\delta_1 \beta_1}[\theta]q_{\delta_2 \beta_2}[\theta]\neq 
       p_{\alpha_1 \gamma_3}[\theta]p_{\alpha_2 \gamma_4}[\theta]q_{\delta_3 \beta_1}[\theta]q_{\delta_4 \beta_2}[\theta]$
  \item For every $\alpha, \gamma$, $p_{\alpha \gamma}[\theta] > 0$
    and for every $\delta, \beta$, $q_{\delta \beta}[\theta] > 0$.
  \end{enumerate}
\end{lemma}

\begin{proof}
  We view the two monomials as multilinear polynomials over the
  variables $\mb V_{\text{pref}}\cup \mb V_{\text{suff}}$.  Consider
  any polynomials $p_{\alpha \gamma}$, $q_{\delta \beta}$, defined by
  Eq.~\eqref{eq:p:and:q:vars}.  The Boolean formula
  $Y_{\alpha \gamma}$ is similar to $Y_{\alpha \beta}$ in
  Eq.~\eqref{eq:y:alpha:beta}, the only difference is that we are now
  moving from left to right to left.  In particular,
  $p_{\alpha \gamma}$ is an irreducible polynomial.  If $U$ is a
  left-ubiquitous symbol in $Q$, then we write $U', U''$ for the
  Boolean variables $U' = U(u,t_{\text{pref}})$ and
  $U''=U(r_0,t_{\text{pref}})$.  We say that $U'$ is a ubiquitous
  variable at the ``start'' and $U''$ is at the ``end'', since their
  distance in $Y_{\alpha \gamma}$ is $2k$.  We write
  $Y_{\alpha \gamma}$ as:
  \begin{align*}
    Y_{\alpha \gamma} = & G'_\alpha \wedge Y \wedge G''_\gamma
  \end{align*}
  where $G'_\alpha$ consists of  clauses that contain some
  ubiquitous variable $U'$ at the start, $G''_\gamma$ consists of
  clauses with some ubiquitous variable $U''$ at the end, and $Y$
  contains all clauses without ubiquitous symbols.  Notice that $Y$ is
  the same formula for all choices of  $\alpha$ and $\gamma$.  

  We will use the following simple fact.  If $F$ is a monotone Boolean
  function in CNF  with $n$ variables, and $f$ is its arithmetization, then
  for any $\theta \in [0,1]^n$, if $f[\theta]=0$ then there exists a
  clause $X_1 \vee X_2 \vee \cdots$ in $F$ such that
  $\theta(X_1)=\theta(X_2)=\ldots = 0$.  In this case, with some
  abuse, we write $F[\theta]= 0$, although $F[\theta]$ is not a well
  defined value since $\theta$ assigns real values to Boolean
  variables.

  To prove the lemma, we will prove that there exists an assignment
  $\theta$ that satisfies item (a) of the lemma, and also satisfies
  the following properties: (1) $Y[\theta]\neq 0$ (2) for any start-
  or end-ubiquitous symbol $U$, $\theta(U) \neq 0$.  This implies
  that, for any $\alpha, \gamma$, $p_{\alpha \gamma}[\theta]\neq 0$,
  because the clauses $Y$ will be the same, while each clause in
  $G'_\alpha$ or $G''_\gamma$ has at least one ubiquitous symbol,
  which is not assigned the value 0.  We enforce similar properties on
  $Y_{\delta \beta}$, and therefore these imply condition (b) of the
  lemma.

  We start by proving that the two monomials, when viewed as
  multilinear polynomials over the variables
  $\mb V_{\text{pref}}\cup \mb V_{\text{suff}}$ are not identical
  polynomials.  Suppose otherwise, then:
  \begin{align}
  p_{\alpha_1 \gamma_1}p_{\alpha_2 \gamma_2}q_{\delta_1 \beta_1}  q_{\delta_2 \beta_2}\equiv p_{\alpha_1 \gamma_3}p_{\alpha_2 \gamma_4}q_{\delta_3 \beta_1} q_{\delta_4 \beta_2}
\label{eq:four:nonequivalent}
  \end{align}
  Since the polynomials $p_{\alpha \gamma}$ depend only on variables
  from the prefix while $q_{\delta \beta}$ depend only on variables
  from the suffix, we obtain that both two identities below must hold:
  \begin{align}
    p_{\alpha_1 \gamma_1}p_{\alpha_2 \gamma_2}\equiv &p_{\alpha_1 \gamma_3}p_{\alpha_2 \gamma_4}&
q_{\delta_1 \beta_1}  q_{\delta_2 \beta_2} \equiv & q_{\delta_3 \beta_1} q_{\delta_4 \beta_2}\label{eq:last:equation}
  \end{align}
  We claim that if the first identity holds, then, when viewed as
  monomial terms in $f_{\alpha_1 \beta_1 \alpha_2 \beta_2}$, the
  monomials $p_{\alpha_1 \gamma_1}p_{\alpha_2 \gamma_2}$ and
  $p_{\alpha_1 \gamma_3}p_{\alpha_2 \gamma_4}$ are the same.  If the
  second identity holds as well, then the other two monomials
  $q_{\delta_1 \beta_1} q_{\delta_2 \beta_2}$ and
  $q_{\delta_3 \beta_1} q_{\delta_4 \beta_2}$ are also identical, but
  both cannot hold by the assumption of the lemma.  To prove the
  claim, assume that the first identity above holds.  Since each of
  the four polynomials is irreducible, there are two cases.  The first
  is when $p_{\alpha_1 \gamma_1}\equiv p_{\alpha_1 \gamma_3}$ and
  $p_{\alpha_2 \gamma_2}\equiv p_{\alpha_2 \gamma_4}$.  In that case,
  using the argument in Lemma~\ref{lemma:simple:lattices:1}, we
  conclude that $\gamma_1=\gamma_3$ and $\gamma_2=\gamma_4$, meaning
  that monomials $p_{\alpha_1 \gamma_1}p_{\alpha_2 \gamma_2}$ and
  $p_{\alpha_1 \gamma_3}p_{\alpha_2 \gamma_4}$ are the same.  The
  second case is $p_{\alpha_1 \gamma_1}\equiv p_{\alpha_2 \gamma_4}$
  and $p_{\alpha_2 \gamma_2}\equiv p_{\alpha_1 \gamma_3}$, in which
  case $\alpha_1=\alpha_2$ and $\gamma_1=\gamma_4$ and
  $\gamma_2=\gamma_3$, and, again, the two monomials are identical.

  Therefore, at least one of the identities
  in~\eqref{eq:last:equation} does not hold.  Assume w.l.o.g. that it
  is the first identity.  We will construct below a certain assignment
  $\theta$ on $\mb V_{\text{pref}}$ such that:
  \begin{align}
    p_{\alpha_1 \gamma_1}[\theta]p_{\alpha_2 \gamma_2}[\theta]\neq
    p_{\alpha_1 \gamma_3}[\theta]p_{\alpha_2 \gamma_4}[\theta] \label{eq:p:alpha:gamma}
  \end{align}
  and satisfies conditions (1) and (2) above.  Then we will extend it
  by setting $\theta(X)=1$ for all variables $X$ in
  $\mb V_{\text{suff}}$; this assignment satisfies both conditions (a)
  and (b) of the lemma.  We distinguish two cases:

  \begin{description}
  \item[Case 1:] $(\gamma_1,\gamma_2)=(\gamma_4,\gamma_3)$.  By
    Lemma~\ref{lemma:three:values} in the introduction there exists
    that $\theta\in \set{0,1/2,1}$ that satisfies:
    \begin{align}
      p_{\alpha_1 \gamma_1}[\theta]p_{\alpha_2 \gamma_2}[\theta]\neq &
      p_{\alpha_1 \gamma_2}[\theta]p_{\alpha_2 \gamma_1}[\theta] \label{eq:really:last}
    \end{align}
    By Corollary~\ref{cor:values:half:or:one}, we may assume
    w.l.o.g. that neither side $=0$, which implies that
    $Y[\theta]\neq 0$.  We prove that for any ubiquitous symbol $U'$
    at the start, $\theta(U') \in \set{1/2,1}$.  Since the query is
    final, $Q[U':=0]$ is a safe query.  Since all left clauses remain
    intact in $Q[U':=0]$ (none becomes redundant and no subclause
    becomes redundant), it follows that in $Q_{\text{left}}[U:=0]$ all
    left clauses are disconnected from the right clauses; referring to
    a left-to-right path $C_0,C_1,\ldots,C_k$ of minimal length, the
    clause $C_1$ becomes redundant in $Q[U:=0]$.  It follows that, for
    any $\alpha, \gamma$,
    \begin{align*}
      Y_{\alpha \gamma}[U':=0] = & G'_\alpha[U':=0] \wedge Y_0 \wedge G''_\gamma
&   \vars(G'_\alpha[U':=0]) \cap \vars(Y_0 \wedge G''_\gamma) = & \emptyset
    \end{align*}
    where $Y_0$ are clauses of $Y$ that are not redundant in
    $Y_{\alpha \gamma}[U':=0]$.  Therefore, the polynomials
    $p_{\alpha \gamma}$ factorize as
    $p_{\alpha \gamma}[U' \asn 0] = f_\alpha \cdot g_\gamma$, and the following
    equivalence holds
    $p_{\alpha_1 \gamma_1}[U' \asn 0] p_{\alpha_2 \gamma_2}[U' \asn 0] \equiv p_{\alpha_1
      \gamma_2}[U' \asn 0] p_{\alpha_2 \gamma_1}[U' \asn 0]$.  Thus, if $\theta(U')=0$ then
    \eqref{eq:really:last} becomes an equality, which is a
    contradiction.  If $U''$ is a ubiquitous symbol at the end, then
    we prove similarly that $\theta(U'') \neq 0$.  This completes the
    proof of Case 1.

  \item[Case 2:] $(\gamma_1, \gamma_2) \neq (\gamma_4,\gamma_3)$; by
    our earlier discussion we also have
    $(\gamma_1, \gamma_2) \neq (\gamma_3,\gamma_4)$.  Here we will
    define $\theta(U')=1$ for all ubiquitous symbols at the start.
    The four Boolean formulas become:
    \begin{align*}
      F_{\gamma_1} \defeq & Y_{\alpha_1 \gamma_1}[\mb U':=1] = Y \wedge G''_{\gamma_1}
&     F_{\gamma_3} \defeq & Y_{\alpha_1 \gamma_3}[\mb U':=1] = Y \wedge G''_{\gamma_3}\\
      F_{\gamma_2} \defeq & Y_{\alpha_2 \gamma_2}[\mb U':=1] = Y \wedge G''_{\gamma_2}
&     F_{\gamma_4} \defeq & Y_{\alpha_2 \gamma_4}[\mb U':=1] = Y \wedge G''_{\gamma_4}\\
    \end{align*}
    Since all these Boolean formulas are connected, their
    arithmetizations are irreducible polynomials.  It follows that
    $f_{\gamma_1}f_{\gamma_2} \not\equiv f_{\gamma_3}f_{\gamma_4}$,
    since otherwise we have $(\gamma_1,\gamma_2)=(\gamma_3,\gamma_4)$
    or $(\gamma_1,\gamma_2)=(\gamma_4,\gamma_3)$.  By
    Lemma~\ref{lemma:three:values} there exists an assignment $\theta$
    in $\set{0,1/2,1}$ such that
    $f_{\gamma_1}[\theta]f_{\gamma_2}[\theta] \not\equiv
    f_{\gamma_3}[\theta]f_{\gamma_4}[\theta]$.  We check that these
    four functions, together with ubiquitous symbols at the end,
    $U''_1, \ldots, U''_m$, satisfy the conditions in
    Lemma~\ref{lemma:values:half:or:one:part2}.  Indeed, by setting
    any $U''_i := 1$, all the subclausess in $G''_\gamma$ that depend
    on the choice of $\gamma$ become 1, since they contain {\em all}
    ubiquitous symbols, and therefore
    $G''_{\gamma_1}[U''_i:=1] \equiv G''_{\gamma_2}[U''_i:=1] \equiv
    G''_{\gamma_3}[U''_i:=1] \equiv G''_{\gamma_4}[U''_i:=1]$.
    Therefore, we can apply Lemma~\ref{lemma:values:half:or:one:part2}
    and obtain an assignment $\theta$ such that
    $f_{\gamma_1}[\theta]f_{\gamma_2}[\theta] \not\equiv
    f_{\gamma_3}[\theta]f_{\gamma_4}[\theta]$ and
    $\theta(U'') \in \set{1/2,1}$ for all $U''$.  This concludes the
    proof.
  \end{description}
\end{proof}

We now proceed to re-prove
Corollary~\ref{cor:no:idea:what:name:to:give:here}.

\begin{lemma} \label{lemma:strictly:positive}
  Fix any $(\alpha_1, \beta_1) \neq (\alpha_2, \beta_2)$, and consider
  a   prefix/suffix block with a single branch.   Then there exists an
  assignment $\theta$ of the variables in $\mb V_{\text{pref}} \cup
  \mb V_{\text{suff}}$ such that (1) $f_{\alpha_1, \beta_1, \alpha_2,
    \beta_2}[\theta]\neq 0$ and (2) for any $\alpha, \gamma$,
  $p_{\alpha \gamma}[\theta]>0$ and for any $\beta,
  \delta$, $q_{\delta \beta}[\theta]>0$. 
\end{lemma}

\begin{proof}
  Denote by $m_1, m_2, \ldots, m_n$ the monomials occurring in {\em
    all} polynomials $f_{\alpha_1, \beta_1, \alpha_2, \beta_2}$.
  That, each $m_i$ has the form
  $m_i =p_{\alpha_1 \gamma_1}p_{\alpha_2 \gamma_2}q_{\delta_1 \beta_1}
  q_{\delta_2 \beta_2}$ for some choice of
  $\gamma_1, \gamma_2, \delta_1, \delta_2$, and we can write
  $f_{\alpha_1, \beta_1, \alpha_2, \beta_2} = \sum_i \Gamma_i m_i$,
  where $\Gamma_i$ is the coefficient of the $i$'th monomial.  Recall
  that the polynomial is not identically 0 (this follows from
  Corollary~\ref{cor:no:idea:what:name:to:give:here}).  For each
  $i\neq j$, let $\theta_{ij}$ be the assignment given by
  Lemma~\ref{lemma:two:monomials} for this pair of monomials.  Denote
  by $\mb p_{ij}$ be the vector consisting of all values
  $p_{\alpha \gamma}[\theta_{ij}]$ and
  $q_{\delta \beta}[\theta_{ij}]$: many do not occur in the polynomial
  $f_{\alpha_1, \beta_1, \alpha_2, \beta_2}$, but we include all of
  them in these vectors.  Notice that all components of all these
  vectors are $>0$.
  By Theorem~\ref{th:products}, we can construct prefix/suffix blocks
  consisting of multiple parallel copies of these blocks, such that
  (1) the new vectors $\mb p$ are element-wise products of the vectors
  $\mb p_{ij}$, and similarly for $\mb q$; in particular all their
  components are $>0$, and (2) {\em all} pairs of monomials in
  $f_{\alpha_1, \beta_1, \alpha_2, \beta_2}$ are distinct:
  $m_i \neq m_j$.  Call this block the {\em starting} block; notice
  that it has several parallel branches in both the prefix and the
  suffix (the same number of branches in the prefix and suffix).
  Thus, we have:
  \begin{align*}
    f_{\alpha_1, \beta_1, \alpha_2, \beta_2}(\mb p, \mb q) = & \sum_{i=1,n} \Gamma_i m_i
  \end{align*}
  Next, we make  $k$ parallel copies of the starting block; on this
  even larger block, the vector $\mb p$ is replaced by $\mb p^k$,
  i.e. each component is raised to the power $k$, and similarly for
  $\mb q$.  Thus, in the new block, the polynomial is:
  \begin{align*}
    f_{\alpha_1, \beta_1, \alpha_2, \beta_2}(\mb p^k, \mb q^k) = & \sum_{i=1,n} \Gamma_i m_i^k
  \end{align*}
  If this value is $=0$ for all $k = 1, 2, \ldots, n+1$, then
  $\Gamma_1 = \cdots = \Gamma_n=0$, because the values $m_i$ are
  distinct and thus the matrix of the system of linear equation is
  non-singular (it is a Vandermonde matrix).  But that implies that
  the polynomial is identically 0, which is a contradiction.  Thus,
  there exists $k$ such that the value of this polynomial is $\neq 0$.
  Since all components of $\mb p^k, \mb q^k$ are $>0$, this proves the
  claim.
\end{proof}

Finally, we pove:

\begin{corollary}
  There exists a choice of the prefix/suffix blocks such that, for
  every pair $(\alpha_1, \beta_1) \neq (\alpha_2, \beta_2)$, the
  polynomial~\eqref{eq:ugly:polynomial} is $\neq 0$.
\end{corollary}

\begin{proof}
  By the previous lemma we can construct a prefix/suffix block that
  satisfies one polynomial $\neq 0$.  By Theorem~\ref{th:products}, we
  can construct parallel branches of these prefix/suffix blocks to
  satisfy all polynomials $\neq 0$, as required.
\end{proof}

The corollary completes the proof: with this choice of prefix/suffix
blocks we have $a_{\alpha_1\beta_1}b_{\alpha_2\beta_2}\neq
a_{\alpha_2\beta_2}b_{\alpha_1\beta_1}$ for all  $(\alpha_1, \beta_1)
\neq (\alpha_2, \beta_2)$, and this we proved
condition-\eqref{eq:conditionCoefficients2}.

It remains to prove Theorem~\ref{th:products}.

\subsection{Proof of Theorem~\ref{th:products}}

Thus, in the rest of this section it remains to prove
Theorem~\ref{th:products}.  Here we will refer to the polynomials
$f_1, f_2, \ldots$ as $p_1, p_2, \ldots$

We write a multivariate polynomial as
\begin{align}
  p(x_1, \ldots, x_n) = p(\mb x) = & \sum_{\mb e: \mb e \leq d} a_{\mb e} \mb x^{\mb e}
\label{eq:poly} 
\end{align}
Here $\mb e = (e_1, \ldots, e_n) \in \N^n$ denotes a vector of
exponents, $\mb x^{\mb e} \defeq \prod_i x_i^{e_i}$, and
$\mb e \leq d$ means $e_i \leq d$ for all $i$; in other words, each
variable $x_i$ has degree $\leq d$. We assume that the coefficients
$a_{\mb e}$ are real numbers.

We will consider vectors of values
$\mb v = (v_1, \ldots, v_n) \in \R^n_+$, where $v_i > 0$ for $i=1,n$.
We denote by $\mb v^k \defeq (v_1^k, \ldots, v_n^k)$, and denote by
$\mb v \mb w \defeq (v_1w_1, \ldots, v_nw_n)$, where
$\mb w = (w_1, \ldots, w_n)$.  Then, Theorem~\ref{th:products} says
that: if $p_1(\mb v_1)\neq 0$, $\ldots$, $p_m(\mb v_m)\neq 0$, then
there exists $\mb u \defeq \mb v_1^{k_1} \cdots \mb v_m^{k_m}$
s.t. $p_1(\mb u)\neq 0$, $\ldots$, $p_m(\mb u)\neq 0$.

We prove the theorem through a sequence of lemmas.  Let $p(\mb x)$ be
a multivariate polynomial in $n$ variables, of degree $d$, and let
$\mb v \in \R_+^n$ be a vector of non-zero values.  To compute
$p(\mb v^k)$ we will group the terms of \eqref{eq:poly} as follows.
Let:
\begin{align*}
  U \defeq & \setof{\mb v^{\mb e}}{\mb e \leq d} \subseteq \R_+
\end{align*}
That is, $U$ is the set of all {\em distinct} values $\mb v^{\mb e}$
that will occur in the expansion of $p(\mb v)$.  Assume $U$ has $m$
distinct values, $U = \set{u_1, \ldots, u_m}$.  For all $u_i \in U$,
define:

\begin{align*}
  E_{\mb v, i} \defeq & \setof{\mb e}{\mb e\leq d, \mb v^{\mb e} = u_i}
\end{align*}

Thus, $E_{\mb v, 1} \cup E_{\mb v, 2} \cup \ldots \cup E_{\mb v, m}$
forms a partition of the set of all exponents occurring in $p(\mb x)$,
and we can write it as a sum of $m$ polynomials:

\begin{align*}
  p(\mb x) = & \sum_{i=1,m} (\sum_{\mb e \in E_{\mb v, i}} a_{\mb e} \mb x^{\mb e})\defeq  \sum_{i=1,m} p_{\mb v, i}(\mb x)
\end{align*}

Then, for all $i$, $p_{\mb v, i}(\mb v^k) = u_i^k p_{\mb v, i}(\mb 1)$,
because all terms $(\mb v^k)^{\mb e}$ in $p_{\mb v, i}(\mb v^k)$ are equal
to $u_i^k$.  Therefore,
\begin{align}
p(\mb v^k) = \sum_{i=1,m} p_{\mb v, i}(\mb v^k) = & \sum_{i=1,m} u_i^k p_{\mb v, i}(\mb 1)  \label{eq:balanced}
\end{align}
Notice that $p_{\mb v, i}(\mb 1)=p_{\mb v, i}(1,1,\ldots,1)$ are just
the sum of all coefficients of the polynomial $p_{\mb v, i}(\mb x)$.

Let's call a polynomial $p(\mb x)$ {\em balanced} if $p(\mb 1)=0$;
otherwise it is {\em imbalanced}.  We prove:

\begin{lemma}
  Let $\mb v \in \R^n_+$ be such that $p(\mb v) \neq 0$.  Then there
  exists $i$ such that $p_{\mb v, i}(\mb x)$ is imbalanced.
\end{lemma}

\begin{proof}
  We prove the converse: if each $p_{\mb v, i}(\mb x)$ is balanced,
  then $p(\mb v)=0$.  This follows immediately from
  Eq.~\eqref{eq:balanced}.
\end{proof}

\begin{lemma}
  If at least one of the polynomials $p_{\mb v, i}(\mb x)$ is
  imbalanced, then there exists $k_0\geq 0$ such that forall
  $k \geq k_0$, $p(\mb v^k)\neq 0$.  In other words,
  $p(\mb v^k)\neq 0$, for all $k$ that are ``large enough''.
\end{lemma}

\begin{proof}
  Assume w.l.o.g. that $m > 0$ and all polynomials $p_{\mb v, i}(\mb
  x)$ are imbalanced (otherwise we simply remove the balanced
  polynomials and corresponding values $u_i$).  Also assume $u_1 > u_2
  > \cdots > u_m$.  Then   Eq.~\eqref{eq:balanced} becomes:
  \begin{align*}
    p(\mb v^k) = & \sum_{i=1,m} u_i^k p_{\mb v, i}(\mb 1) =
u_1^k \left(p_{\mb v,1}(\mb 1) + \underbrace{\sum_{i = 2,m} \left(\frac{u_i}{u_1}\right)^kp_{\mb  v,i}(\mb 1)}_{\rightarrow 0}\right)
  \end{align*}
  When $k\rightarrow \infty$, then the expression under $\sum_{i=2,m}$
  goes to 0, hence when $k$ is large enough,
  $p_{\mb v, 1}(\mb 1) + \sum_{i=2,n}(\cdots) \neq 0$.
\end{proof}

For a simple example, consider $f(x_1,x_2,x_3) = 2x_1^2x_2-x_3^2$, and
assume $\mb v = (3, 2, 6)$.  Then $f(3,2,6) =2\cdot 18 - 36=0$.
However, the reader may verify that, for every $k \geq 2$,
$f(3^k,2^k,6^k) > 0$.

Next, we will examine combinations of the form $\mb v^s \mb w^t$ for
natural numbers $s, t \geq 1$.  We will require some simple inequality
constraints of these pairs $s,t$, which we define next.

\begin{definition}
  An {\em inequality constraint} is a pair of real numbers
  $(\alpha, \beta)$, such that $(\alpha,\beta) \neq (0,0)$ (i.e. not
  both $\alpha,\beta$ can be 0).  We say that two natural numbers
  $s, t \in \N$ {\em satisfy} the constraint if $s, t \geq 1$ and
  $s \alpha + t \beta \neq 0$.  If $\Gamma$ is a set of inequality
  constraints, then we write $\Gamma \models (s,t)$ when $s,t$ satisfy
  every constraint in $\Gamma$.
\end{definition}

The intuition behind an inequality constraint is the following.  We
have two numbers $v,w$, and want to find exponents $s,t$ such that
$v^s \neq w^t$. Any pair $(s,t)$ satisfying the constraint
$(\log v, -\log w)$ will also satisfy the inequality $v^s \neq w^t$.
We need two very simple facts:

\begin{lemma} \label{lemma:trivial1} If $\Gamma$ is a finite set of
  inequality constraints, then there exists infinitely many pairs of
  natural numbers $s,t$ s.t. $s \geq 1, t \geq 1$ that satisfy all
  constraints in $\Gamma$.
\end{lemma}

\begin{proof} Let
  $\Gamma = \set{(\alpha_1,\beta_1), \ldots, (\alpha_m, \beta_m)}$,
  and define the set
  $S \defeq \setof{-\alpha_i/\beta_i}{(\alpha_i,\beta_i) \in \Gamma,
    \beta_i\neq 0}$.  This is a finite set of real numbers.  Then, any
  pair of natural numbers $s,t$ such that $s, t \geq 1$ and
  $t/s \in \Q - S$ satisfies all constraints in $\Gamma$: indeed, if
  $\beta_i\neq 0$ then $t/s \neq -\alpha_i/\beta_i$ implies
  $s \alpha_i + t \beta_i \neq 0$, and if $\beta_i=0$ then
  $s\alpha_i + t \beta_i= s \alpha_i \neq 0$ because $s > 0$.
\end{proof}

\begin{lemma} \label{lemma:trivial2} Let
  $\alpha_1> \alpha_2 > \cdots > \alpha_m$ and
  $\beta_1 > \beta_2 > \cdots > \beta_n$ be two sequences of distinct
  real values.  Then there exists a finite set of inequality
  constraints $\Gamma$ such that, for any numbers $s,t$, if
  $\Gamma \models (s,t)$, then the $m\cdot n$ values
  $q\alpha_i + t \beta_j$, $i=1,m$, $j=1,n$, are distinct.
\end{lemma}

\begin{proof} For all tuples $i_1, i_2, j_1, j_2$ such that
  $1 \leq i_1 \leq i_2 \leq m$ and $1 \leq j_1 \leq j_2 \leq n$ and
  $(i_1,j_1) \neq (i_2,j_2)$, define
  $\gamma_{i_1i_2j_1j_2} = \alpha_{i_1} - \alpha_{i_2}$ and
  $\delta_{i_1i_2j_1j_2} = \beta_{j_1} - \beta_{j_2}$.  Let $\Gamma$
  be the set of constraints
  $(\gamma_{i_1i_2j_1j_2},\delta_{i_1i_2j_1j_2})$.  If
  $\Gamma \models (s,t)$, then
  $s \gamma_{i_1i_2j_1j_2} + t \delta_{i_1i_2j_1j_2} \neq 0$ for all
  $i_1,i_2,j_1,j_2$, which implies
  $s\alpha_{i_1} + t \beta_{j_1} \neq s \alpha_{i_2} + t \beta_{j_2}$.
\end{proof}

Let
$p(\mb x, \mb y) = \sum_{\mb e, \mb f} a_{\mb e, \mb f}\mb x^{\mb e} \mb y^{\mb f}$ 
be a polynomial in two sets of variables, and let
$\mb v, \mb w$ be two sequences of positive real numbers.  We define:

\begin{align*}
  U \defeq &\setof{\mb v^{\mb e}}{\mb e \leq d} = \set{u_1 > u_2 >  \cdots > u_m (> 0)} \\
  Z \defeq &\setof{\mb w^{\mb e}}{\mb e \leq d} = \set{z_1 > z_2 > \cdots > z_s (> 0)}\\
  E_{\mb v,\mb w, i,j} \defeq & \setof{(\mb e,\mb f)}{\mb e\leq d, \mb f \leq d, \mb v^{\mb e} = u_i, \mb w^{\mb f} = z_j}
\end{align*}

As before, for any two sequences of positive real numbers
$\mb v, \mb w$, the sets $E_{\mb v,\mb w, i,j}$ partition the set of
exponents occurring in $p$, and we can write:

\begin{align*}
  p(\mb x, \mb y) = & 
\sum_{i,j}\left( \sum_{(\mb e, \mb f) \in E_{\mb v,\mb w, i,j}} a_{\mb e, \mb f}\mb x^{\mb e} \mb y^{\mb f}\right)
\defeq
\sum_{i=1,m; j=1,s} p_{\mb v, \mb w, i, j}(\mb x, \mb y)
\end{align*}

\begin{lemma} \label{lemma:polly:xy}
  Fix $p(\mb x, \mb y)$, and let $\mb v \in \R_+^n$ be a vector of
  values $>0$, such that $p(\mb v, \mb 1) \neq 0$.  Let
  $\mb w \in \R_+^n$ be any other vector of values $>0$.  Then there
  exists a finite set of inequality constraints $\Gamma$, such that,
  forall $s, t$, if $\Gamma \models (s,t)$ then there exists
  $k_0 \geq 0$ such that forall $k \geq k_0$,
  $p(\mb u^{kq}, \mb v^{kt}) \neq 0$.
\end{lemma}

\begin{proof} Let
  $U, Z, E_{\mb v,\mb w, i,j}, p_{\mb v, \mb w, i, j}(\mb x, \mb y)$
  be defined as above.  We start by noticing that, for any numbers
  $a, b \geq 0$, we have
  $p_{\mb v, \mb w, i, j}(\mb v^a, \mb w^b) = u_i^az_j^bp_{\mb v, \mb w, i, j}(\mb 1, \mb 1)$.
  For any three numbers $q,t,k \geq 0$, we have:
  \begin{align*}
    p(\mb v^{kq}, \mb w^{kt}) = & \sum_{i=1,m; j=1,s} p_{\mb v, \mb w, i, j}(\mb v^{kq}, \mb w^{kt})\\
= & \sum_{i=1,m; j=1,s} u_i^{kq}z_j^{kt} p_{\mb v, \mb w, i, j}(\mb 1,  \mb 1)\\
= & \sum_{i=1,m; j=1,s} r_{ij}^k p_{\mb v, \mb w, i, j}(\mb 1,  \mb 1)
  \end{align*}
  where $r_{ij} \defeq u_i^qz_j^t$.  We notice that there exists $i,j$
  such that $p_{\mb v, \mb w, i, j}(\mb 1, \mb 1)\neq 0$.  Indeed, if
  we choose $k=q=1, t=0$, then the quantity above becomes
  $p(\mb v^{kq}, \mb w^{kt}) = p(\mb v, \mb 1)$, which, by assumption
  of the lemma is $\neq 0$, proving that at least one quantity
  $p_{\mb v, \mb w, i, j}(\mb 1, \mb 1)\neq 0$.

  We will define a set of constraints $\Gamma$ such that
  $\Gamma \models (q,t)$ implies that all values $r_{ij} = u_i^qz_j^t$
  are distinct or, equivalently, the quantities
  $q \log u_i + t \log z_j$ are distinct.  To obtain such a $\Gamma$,
  we apply Lemma~\ref{lemma:trivial2} to the sequences $\log u_i$ and
  $\log z_j$ respectively.  Considering only those values $r_{ij}$ for
  which $p_{\mb v, \mb w, i, j}(\mb 1, \mb 1)\neq 0$, let $r_{i_0j_0}$
  be the largest number.  Then we have:
  \begin{align*}
     p(\mb v^{kq}, \mb w^{kt}) = & \sum_{i=1,m; j=1,s} r_{ij}^k p_{\mb v, \mb w, i, j}(\mb 1,  \mb 1)\\
 = & r_{i_0j_0} \left(p_{\mb v, \mb w, i_0, j_0}(\mb 1,  \mb 1)+
\sum_{(i,j)\neq (i_0,j_0)} \left(\frac{r_{ij}}{r_{i_0j_0}}\right)^k p_{\mb v, \mb w, i, j}(\mb 1,  \mb 1)
\right)
  \end{align*}
  Since
  $\lim_{k \rightarrow \infty}
  \left(\frac{r_{ij}}{r_{i_0j_0}}\right)^k =0$ it follows that, for
  $k$ large enough, $p(\mb v^{kq}, \mb w^{kt})\neq 0$, as required.
\end{proof}

Finally, we can prove Theorem~\ref{th:products}.

\begin{proof}
  (Of Theorem~\ref{th:products}) We proceed by induction on $m$.  When
  $m=1$ then the theorem holds trivially.  Assume $m \geq 2$, and
  denote
  $p(\mb x) \defeq p_1(\mb x) \cdot p_2(\mb x) \cdots p_{m-1}(\mb x)$.
  By induction hypothesis, there exists
  $\mb v = \mb{v}_1^{k_1}\cdots\mb{v}_{m-1}^{k_{m-1}}$ such that
  $p(\mb v)\neq 0$, and there exists $\mb w$ such that
  $p_m(\mb w) \neq 0$.

Define the following polynomials $f(\mb x, \mb y)$ and $g(\mb x, \mb y)$:

\begin{align*}
  f(x_1, \ldots, x_n, y_1, \ldots, y_n) \defeq & p(x_1y_1, \ldots, x_ny_n)\\
  g(x_1, \ldots, x_n, y_1, \ldots, y_n) \defeq & p_m(x_1y_1, \ldots, x_ny_n)
\end{align*}

We apply Lemma~\ref{lemma:polly:xy} to the polynomial $f$ and the
sequences $\mb v, \mb w$: the assumption $f(\mb v,\mb 1) \neq 0$ holds
because $f(\mb v, \mb 1)= p(\mb v) \neq 0$.  Therefore, we obtain a
finite set of constraints $\Gamma_1$ s.t. for all $s,t$, if $\Gamma_1 \models
(s,t)$, then $f(\mb v^{ks},\mb w^{kt})\neq 0$ for all $k$ ``large
enough''.

Similarly, we apply Lemma~\ref{lemma:polly:xy} to the polynomial $g$
and the same sequences $\mb v, \mb w$.  The condition in the lemma
holds, because $g(\mb 1, \mb w)=p_m(\mb w) \neq 0$.  Therefore, there
exists a finite set of constraints $\Gamma_2$ s.t. for all $s,t$, if
$\Gamma_2 \models (s,t)$ then $g(\mb v^{ks},\mb w^{kt}) \neq 0$ for
all $k$ ``large enough''.  

Let $(s,t)$ be any pair that satisfies both $\Gamma_1$ and $\Gamma_2$:
such a pair exists by Lemma~\ref{lemma:trivial1}, because
$\Gamma_1 \cup \Gamma_2$ is a finite set.  Therefore, if $k$ is large
enough, then, denoting $\mb u \defeq \mb{v}^{ks} \mb{w}^{kt}$, we
notice that this has the required form of the theorem, i.e.
$\mb u = \prod_i \mb{v}_i^{k_i}$ for appropriate exponents $k_i$, and
we prove that it satisfies the conditions of the theorem.  Indeed, on
one hand
$p(\mb{u}) = p(\mb{v}^{ks}\mb{w}^{kt})=f(\mb{v}^{ks},\mb{w}^{kt})\neq
0$ which implies $p_i(\mb{u}) \neq 0$ for all $i=1,m-1$, and on the
other hand
$p_m(\mb{u})=p_m(\mb{v}^{ks}\mb{w}^{kt}) = g(\mb{v}^{ks},\mb{w}^{kt})
\neq 0$, proving the theorem.
\end{proof}

\end{document}